\documentclass[11pt,colorinlistoftodos]{amsart}
\usepackage[utf8]{inputenc}
\usepackage[titletoc]{appendix}
\usepackage[LGR,T1]{fontenc}
\usepackage{bm,bbm}
\usepackage[bbgreekl]{mathbbol}
\usepackage{dsfont}
 \usepackage{relsize}

\usepackage[backend=biber, style=numeric, date=year, url=false, doi=false, isbn=false, eprint=true, firstinits=true, maxcitenames=5, maxbibnames=5, style=alphabetic]{biblatex}
\AtEveryBibitem{%
  \ifentrytype{online}{%
  }{%
    \clearfield{eprint}%
    \clearfield{urldate}%
  }%
}
\usepackage{blindtext}
\usepackage{amssymb}
\usepackage{amsfonts}
\usepackage{amsthm}
\usepackage{tikz}
\usepackage{caption}
\usepackage{graphics}
\usepackage{xcolor,colortbl} 
\usepackage{subfigure}
\usepackage{epigraph}
\usepackage[top=1in, bottom=1.25in, left=1.25in, right=1.25in]{geometry}
\usepackage{placeins}
\usetikzlibrary{snakes}
\usepackage[new]{old-arrows}
\usepackage{ctable}

\numberwithin{equation}{subsection}

\usepackage{chngcntr}
\counterwithin{table}{subsection}
\counterwithin{figure}{subsection}

\usepackage{enumerate}
\usepackage{verbatim}
\usetikzlibrary{trees}
\usetikzlibrary{decorations.markings}
\usetikzlibrary{arrows}
\usetikzlibrary{cd}
\usepackage{mathrsfs}
\usepackage{xparse}

\usetikzlibrary{positioning,arrows,patterns}
\usetikzlibrary{decorations.markings}
\usetikzlibrary{calc}
\tikzset{
  photon/.style={decorate, decoration={snake}, draw=black},
  fermion/.style={draw=black, postaction={decorate},decoration={markings,mark=at position .55 with {\arrow{>}}}},
  fermion2/.style={dashed, dash phase=0.1pt, draw=black, postaction={decorate},decoration={markings,mark=at position .55 with {\arrow{>}}}},
  vertex/.style={draw,shape=circle,fill=black,minimum size=5pt,inner sep=0pt},
particle/.style={thick,draw=black},
particle2/.style={thick,draw=blue},
avector/.style={thick,draw=black, postaction={decorate},
    decoration={markings,mark=at position 1 with {\arrow[black]{triangle 45}}}},
gluon/.style={decorate, draw=black,
    decoration={coil,aspect=0}}
 }
\NewDocumentCommand\semiloop{O{black}mmmO{}O{above}}
{%
\draw[#1] let \p1 = ($(#3)-(#2)$) in (#3) arc (#4:({#4+180}):({0.5*veclen(\x1,\y1)})node[midway, #6] {#5};)
}

\newcommand*\circled[1]{\tikz[baseline=(char.base)]{
            \node[shape=circle,draw,inner sep=2pt] (char) {#1};}}

\usepackage{relsize}
\usepackage{url}
\usepackage{hyperref}
\hypersetup{colorlinks=true, citecolor=purple, linktocpage, linkcolor=blue, urlcolor=cyan} 

\theoremstyle{plain}
\newtheorem{thm}{Theorem}[subsection]
\newtheorem{lem}[thm]{Lemma}
\newtheorem{prop}[thm]{Proposition}

\theoremstyle{definition}
\newtheorem{defn}[thm]{Definition}
\newtheorem{Conj}[thm]{Conjecture}

\newtheorem*{thm*}{Theorem}
\newtheorem*{lem*}{Lemma}
\newtheorem*{prop*}{Proposition}
\newtheorem*{cor*}{Corollary}
\newtheorem*{exe*}{Exercise}
\newtheorem*{defn*}{Definition}
\newtheorem{rem}[thm]{Remark}

\newtheorem{ex}[thm]{Example}

\theoremstyle{remark}

\newcommand{\R}{\mathbb{R}}
\newcommand{\Z}{\mathbb{Z}}
\newcommand{\Q}{\mathbb{Q}}

\newcommand{\E}{\mathbb{E}}

\newcommand{\calD}{\mathcal{D}}
\newcommand{\calY}{\mathcal{Y}}

\newcommand{\dd}{{\mathrm{d}}}

\newcommand{\calR}{\mathcal{R}}

\DeclareMathOperator{\GL}{GL}

\newcommand{\id}{\mathrm{id}}

\DeclareMathOperator{\tr}{Tr}

\DeclareMathOperator{\Div}{\textnormal{div}}

\DeclareMathOperator{\gh}{gh}

\newcommand{\C}{\mathbb{C}}

\DeclareMathOperator{\ad}{ad}

\DeclareMathOperator{\Ad}{Ad}

\DeclareMathOperator{\Aut}{Aut}
\DeclareMathOperator{\End}{End}
\DeclareMathOperator{\Hom}{Hom}

\newcommand{\de}{\partial}

\newcommand{\calA}{\mathcal{A}}
\newcommand{\calB}{\mathcal{B}}
\newcommand{\calH}{\mathcal{H}}
\newcommand{\calS}{\mathcal{S}}

\newcommand{\calG}{\mathcal{G}}
\newcommand{\calI}{\mathcal{I}}
\newcommand{\calO}{\mathcal{O}}
\newcommand{\calL}{\mathcal{L}}
\newcommand{\calM}{\mathcal{M}}

\newcommand{\calW}{\mathcal{W}}
\newcommand{\calT}{\mathcal{T}}

\newcommand{\calE}{\mathcal{E}}
\newcommand{\calP}{\mathcal{P}}
\newcommand{\calF}{\mathcal{F}}

\newcommand{\calN}{\mathcal{N}}

\def\gpd{\,\lower1pt\hbox{$\longrightarrow$}\hskip-.24in\raise2pt
               \hbox{$\longrightarrow$}\,}

\let\Hat=\widehat

\DeclareMathOperator{\Map}{Map}

\newcommand{\I}{\mathrm{i}}

\newcommand{\calV}{\mathcal{V}}

\newcommand{\Sym}{\textnormal{Sym}}

\makeatletter
\newcommand{\upint}{\DOTSI\upintop\ilimits@}
\newcommand{\upoint}{\DOTSI\upointop\ilimits@}
\makeatother

\usepackage[colorinlistoftodos]{todonotes}

\makeatletter

\pgfdeclareshape{genus}{
 \anchor{center}{\pgfpointorigin}
\backgroundpath{
    \begingroup
    \tikz@mode
    \iftikz@mode@fill

         \pgfpathmoveto{\pgfqpoint{-0.78cm}{-.17cm}}
     \pgfpathcurveto %
            {\pgfpoint{-0.35cm}{-.44cm}}
        {\pgfpoint{0.35cm}{-.44cm}}
        {\pgfpoint{.78cm}{-0.17cm}} 
     \pgfpathmoveto{\pgfqpoint{-0.78cm}{-0.17cm}}
     \pgfpathcurveto %
            {\pgfpoint{-0.25cm}{.25cm}}
        {\pgfpoint{.25cm}{.25cm}}
        {\pgfpoint{0.78cm}{-0.17cm}}
        \pgfusepath{fill}
        \fi

        \iftikz@mode@draw
     \pgfpathmoveto{\pgfqpoint{-1cm}{0cm}}
     \pgfpathcurveto %
            {\pgfpoint{-0.5cm}{-.5cm}}
        {\pgfpoint{0.5cm}{-.5cm}}
        {\pgfpoint{1cm}{0cm}}

     \pgfpathmoveto{\pgfqpoint{-0.75cm}{-0.15cm}}
     \pgfpathcurveto %
            {\pgfpoint{-0.25cm}{.25cm}}
        {\pgfpoint{.25cm}{.25cm}}
        {\pgfpoint{0.75cm}{-0.15cm}}
              \pgfusepath{stroke}
        \fi
        \endgroup
    }
    }

    \makeatother

\tikzset{residual/.style={draw, shape=circle, black,inner sep=1pt}}

\title[Higher Gauge Theory Methods in the BV-BFV Formalism]{4-Manifold Topology, Donaldson--Witten Theory, Floer Homology and Higher Gauge Theory Methods in the BV-BFV Formalism}
\author[N. Moshayedi]{Nima Moshayedi}
\address{Institut f\"ur Mathematik\\ Universit\"at Z\"urich\\ 
Winterthurerstrasse 190
CH-8057 Z\"urich \linebreak Departement of Mathematics, University of California, Berkeley California 94305, USA}
\email[N.~Moshayedi]{nima.moshayedi@math.uzh.ch}
\dedicatory{Dedicated to J\"urg Fr\"ohlich on the occasion of his 75th birthday}

\begin{document}
\makeatletter
\providecommand\@dotsep{5}
\def\listtodoname{List of Todos}
\def\listoftodos{\@starttoc{tdo}\listtodoname}
\makeatother

\maketitle

\begin{abstract}
We study the behavior of Donaldson's invariants of 4-manifolds based on the moduli space of anti self-dual connections (instantons) in the perturbative field theory setting where the underlying source manifold has boundary. It is well-known that these invariants take values in the instanton Floer homology groups of the boundary 3-manifold. Gluing formulae for these constructions lead to a functorial topological field theory description according to a system of axioms developed by Atiyah, which can be also regarded in the setting of perturbative quantum field theory, as it was shown by Witten, using a version of supersymmetric Yang--Mills theory, known today as Donaldson--Witten theory. One can actually formulate an AKSZ model which recovers this theory for a certain gauge-fixing. We consider these constructions in a perturbative quantum gauge formalism for manifolds with boundary that is compatible with cutting and gluing, called the BV-BFV formalism, which was recently developed by Cattaneo, Mnev and Reshetikhin. We prove that this theory satisfies a modified Quantum Master Equation and extend the result to a global picture when perturbing around constant background fields. These methods are expected to extend to higher codimensions and thus might help getting a better understanding for fully extendable $n$-dimensional field theories (in the sense of Baez--Dolan and Lurie) in the perturbative setting, especially when $n\leq 4$.
Additionally, we relate these constructions to Nekrasov's partition function by treating an equivariant version of Donaldson--Witten theory in the BV formalism.
Moreover, we discuss the extension, as well as the relation, to higher gauge theory and enumerative geometry methods, such as Gromov--Witten and Donaldson--Thomas theory and recall their correspondence conjecture for general Calabi--Yau 3-folds. In particular, we discuss the corresponding (relative) partition functions, defined as the generating function for the given invariants, and gluing phenomena.
\end{abstract}

\tableofcontents

\section{Introduction}

\subsection{Overview and motivation}
There is no doubt that the study of gauge theories had a big influence on modern mathematics and physics. A particularly influential and important one is the study of the topology of 4-manifolds by using Yang--Mills theory and the notion of anti self-dual connections, also called \emph{instantons}, which can be considered as a class of critical points of the Yang--Mills action functional. An early attempt, maybe even one of the starting points, was the construction of instantons given in \cite{AtiyahDrinfeldHitchinManin1978}. Shortly thereafter, Donaldson introduced in \cite{Donaldson1983,Donaldson1984,Donaldson1990} his famous polynomial invariants which are a type of topological invariants based on the theory of characteristic classes on vector bundles (topological $K$-theory) and the construction of the moduli space of anti self-dual connections. In particular, these invariants are described as integrals of a product of certain cohomology classes over the moduli space. Thus, defining these invariants relies very much on the behaviour of the moduli space as a sufficiently ``nice'' manifold. Many different people, including Uhlenbeck \cite{Uhlenbeck1982,Uhlenbeck1982b}, Freed \cite{FreedUhlenbeck1984}, Freedman \cite{Freedman1982} and Taubes \cite{Taubes1982}, provided results which contributed to the fact that these moduli spaces are indeed ``nice'' enough. However, at first, these invariants have been only defined for the case when the 4-manifold is closed. 

A new drive came into play when Floer introduced in \cite{Floer1988,Floer1989} a type of topological invariant of closed 3-manifolds by using methods of gauge theories based on ideas of Witten regarding the constructions of Morse theory in the setting of perturbative quantum field theory \cite{Witten1982}. In particular, he constructed an infinite-dimensional version of Morse theory based on the moduli space of anti self-dual connections together with the Chern--Simons action functional playing the role of the Morse function. The main invariants are then given by the homology groups, called \emph{instanton Floer homology groups}, similarly as the Morse homology groups in the finite setting.
Braam and Donaldson then realized in \cite{BraamDonaldson1995,Donaldson2002} that the polynomial invariants defined by Donaldson can be extended to 4-manifolds with boundary by imposing that the invariants are then exactly valued in these Floer groups. Moreover, they proved a gluing formula which endows the invariant with the structure of a functorial TQFT according to Atiyah's axioms \cite{Atiyah1988}.

Another type of a similar approach to a homology theory for Floer's construction was considered in \cite{Floer1988b} by using the symplectic manifolds and corresponding transversal Lagrangian submanifolds thereof. Besides important insights regarding the Arnold conjecture (see e.g. \cite{HoferZehnder1994}), it was soon realized that this type of homology theory plays a fundamental role in order to formulate and understand the mirror symmetry appearing in string theory \cite{Yau1992,HZKKTVPV2003} from a homological point of view. In particular, it was Fukaya who constructed an $A_\infty$-category in \cite{Fukaya1993} (see also \cite{FukayaOhOhtaOno2009_1,FukayaOhOhtaOno2009_2}) from out of this notion and Kontsevich who used these categories to formulate a conjecture which, for two mirror Calabi--Yau manifolds, relates this category of one side of the mirror to the derived category of coherent sheaves of the other side of the mirror through an equivalence of triangulated categories. This is known famously today as the \emph{homological mirror symmetry conjecture} \cite{Kontsevich1994_2}. 

On the other hand, based on ideas of Atiyah \cite{Atiyah1987}, Witten gave a way of obtaining Donaldson's polynomials by considering the perturbative expansion of the expectation value (path integral quantization) for a certain observable with respect to a local action functional. He also gave an argument involving the case of manifolds with boundary by interpreting the boundary states as the Floer groups in agreement with Donaldson's observation. 
The perturbative methods of treating quantum gauge theories have developed through time by different approaches. 

The \emph{Batalin--Vilkovisky (BV) formalism} \cite{BV1,BV2,BV3} provides a nice way of dealing with quantum gauge theories in a cohomological symplectic formalism \cite{KijowskiTulczyjew1979} by using methods of functional integrals. In fact, the gauge-fixing there is equivalent to the choice of a Lagrangian submanifold which, by similar methods as the \emph{BRST formalism} \cite{BRS1,BRS2,Tyutin1976} and \emph{Faddeev--Popov ghosts} \cite{FP}, gives a way of computing the partition function by the perturbative expansion into Feynman graphs. The methods described by Batalin and Vilkovisky are considered from a Lagrangian point of view, whereas the Hamiltonian counterpart was described in the work of Batalin, Fradkin, Fradkina and Vilkovisky \cite{BF1,BF2,FV1,FF}, usually called the \emph{BFV formalism}. Note that these constructions can be considered separately for closed spacetime manifolds. However, one is often interested in the quantization picture for manifolds with boundary in order to use the concept of locality for the simplification through cutting and gluing properties. Additionally, everything should be consistent with Atiyah's TQFT axioms and Segal's axioms regarding conformal field theories \cite{Segal1988}. Such an extension has been recently provided by Cattaneo, Mnev and Reshetikhin to deal with the classical and quantum formalism of local gauge theories on manifolds with boundary in the cohomological symplectic setting by coupling the BV construction in the bulk to the BFV construction on the boundary \cite{CMR1,CMR2}. These constructions can be easily extended to higher codimensions in the classical setting, but need more sophisticated techniques in the quantum setting. Nevertheless, this formalism is expected to give a reasonable candidate for a perturbative formulation of (fully) extended TQFTs (in the sense of Baez--Dolan \cite{BaezDolan1995} or Lurie \cite{Lurie2009}) for the case of interest. 

As it was shown in \cite{Ikeda2011}, it turns out that the field theory constructed by Witten can be naturally formulated by using a type of BV theory developed by Alexandrov, Kontsevich, Schwarz and Zaboronsky (AKSZ) in \cite{AKSZ} for a special gauge-fixing. This formulation extends in a nice way to the BV-BFV formalism since AKSZ theories often appear as suitable deformations of abelian $BF$ theory \cite{Mnev2019}. This is also the case for Donaldson--Witten theory, formulated as a 4-dimensional AKSZ theory with target $\mathfrak{h}[1]\oplus\mathfrak{h}[2]$ for some Lie algebra $\mathfrak{h}$. 
An approach to globalize a special type of AKSZ theories has been given in \cite{CMW3}. There one starts with an AKSZ theory of any dimension which is of \emph{split-type}. Choosing a background field as an element of the moduli space of classical solutions, one can vary the local theories over the target manifold by considering the \emph{Grothendieck connection} as in the setting of formal geometry developed by Gelfand--Fuks \cite{GelfandFuks1969,GelfandFuks1970}, Gelfand--Kazhdan \cite{GK} and Bott \cite{B}. Within the BV-BFV formalism, this leads to more general gauge conditions, such as the modified \emph{differential} Quantum Master Equation. 

Another approach, which is related to Donaldson's construction, to obtain topological invariants of 4-manifolds is due to Seiberg and Witten \cite{SeibergWitten1994,SeibergWitten1994_2} (see also \cite{Nicolaescu2000} for a more mathematical introduction). They considered $\calN=2$ supersymmetric Yang--Mills theory\footnote{Recall that $\calN$ denotes the number of irreducible real spin representations in a supersymmetry (SUSY) algebra.} and formulated a set of equations (\emph{Seiberg--Witten equations}) which contains the same information of the 4-manifold as the Yang--Mills equations but has the advantage that it is much easier to deal with. Solutions of these equations are usually called \emph{monopoles}\footnote{According to the equations for magnetic monopoles in the theory of electromagnetism.}. 

The field-theoretic approach of Seiberg--Witten theory is in its nature given by an $A$-model as in \cite{Witten1988a}. The low energy effective action can actually be written in terms of a holomorphic function $\mathsf{F}$, called the \emph{Seiberg--Witten prepotential}. In \cite{Nekrasov2003}, Nekrasov showed how one can compute this prepotential $\mathsf{F}$ as a certain limit by using techniques of equivariant localization for a given torus action to define a partition function proportional to the volume of the moduli space of instantons. In particular, this partition function is given as a generating function with coefficients given by the volumes of the connected components of the moduli space of instantons controlled through the \emph{instanton numbers}. This is achieved by formulating a 6-dimensional theory (i.e. a given 4-manifold, or locally $\R^4$, times a torus $\mathbb{T}$) and then perform dimensional reduction. 
This is also called the \emph{$\calN=2$ supersymmetric Yang--Mills theory in the $\Omega$-background}, where $\Omega$ is a $4\times 4$ matrix with entries given by one of the two generators of the torus $\mathbb{T}$ with a sign or zero in a certain way. This background in fact allows one to integrate out all the \emph{fluctuations (high energy modes)} appearing in the functional integral quantization, and hence get rid of any divergencies. 

Since the underlying $\calN=2$ supersymmetric Yang--Mills theory in this setting is based on the techniques of equivariant localization, one can study an equivariant version of Donaldson--Witten theory in the AKSZ-BV setting and consider a (regularized) perturbative quantization in order to obtain the mentioned partition function. In order to deal with with the quantization of such an equivariant field theory, it is natural to formulate an equivariant version of the BV formalism as it was developed in \cite{CattZabz2019}. However, the extension to manifolds with boundary is still open but following the constructions of Atiyah, Donaldson and Witten for the non-equivariant setting, the boundary states, hence the geometric quantization procedure for the state space in the BV-BFV formalism, should recover an equivariant version of instanton Floer homology as discussed in \cite{AustinBraam1996}. 

An approach to the construction of Donaldson for higher dimensional gauge theories was initiated through the work of Donaldson and Thomas in \cite{DonaldsonThomas1998} and developed further by Thomas in \cite{Thomas2000}. In particular, using a holomorphic version of Chern--Simons theory \cite{Chern1974}, Thomas constructed topological invariants for Calabi--Yau 3-folds by constructing a holomorphic version of the Casson invariant \cite{Saveliev1999} which counts the bundles over the given Calabi--Yau 3-fold and extends to the counting of curves in algebraic 3-folds. The formulation of these invariants in terms of a weighted Euler characteristic was given by Behrend in \cite{Behrend2009}. One can consider the Donaldson--Thomas partition function, similarly as for Gromov--Witten theory \cite{Witten1991,Behrend1997}, as a generating function for the corresponding invariants on any Calabi--Yau 3-fold $X$. Moreover, one can prove gluing formulae for these partition functions, first considered in \cite{LiWu2015} by considering an associated divisor and the properties of the Hilbert scheme of 1-dimensional subschemes of $X$. By the nature of its construction, Donaldson--Thomas theory is related to the construction of Nekrasov theory through its way of instanton counting (the modified moduli space of instantons in Nekrasov's theory is replaced by the Hilbert scheme of 1-dimensional subschemes of $X$ in Donaldson--Thomas theory). 
Moreover, it seems plausible that the Donaldson--Thomas partition function is related to the Gromov--Witten partition function by a certain correspondence. Such a correspondence for general Calabi--Yau 3-folds was conjectures by Maulik, Nekrasov, Okounkov and Pandharipande in \cite{MaulikNekreasovOkounkovPandharipande2006_1,MaulikNekreasovOkounkovPandharipande2006_2} and proven for toric 3-folds in \cite{MaulikOblomkovOkounkovPandharipande2011} by Maulik, Oblomkov, Okounkov and Pandharipande. The correspondence seems to formulate a version of the homological mirror symmetry conjecture, which was already known to be true for toric 3-folds. 

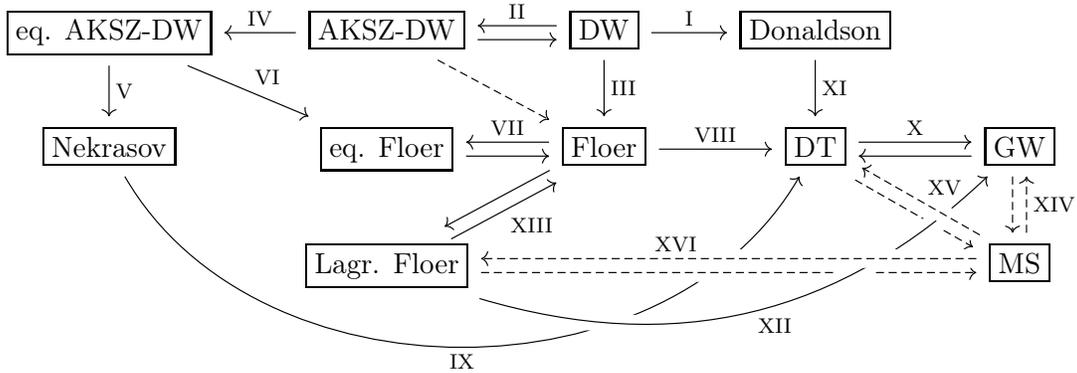
\begin{figure}[h!]
\begin{center}
\begin{tikzcd}
\boxed{\text{eq. AKSZ-DW}}  \arrow[dr,"\mathrm{VI}"]\arrow[d,"\mathrm{V}"]  & \boxed{\text{AKSZ-DW}}\arrow[r,shift right]\arrow[dr,dashed] \arrow[l,swap,"\mathrm{IV}"]  & \boxed{\text{DW}}\arrow[l,shift right,swap,"\mathrm{II}"] \arrow[r,"\mathrm{I}"] \arrow[d,"\mathrm{III}"]      & \boxed{\text{Donaldson}} \arrow[d,"\mathrm{XI}"] &                         \\
\boxed{\text{Nekrasov}} \arrow[rrr,swap,bend right=60,"\mathrm{IX}"] & \boxed{\text{eq. Floer}}\arrow[r,shift right]  & \boxed{\text{Floer}}\arrow[dl,shift right]\arrow[r,"\mathrm{VIII}"] \arrow[l,swap,shift right,"\mathrm{VII}"] & \boxed{\text{DT}}\arrow[dr,dashed,shift right] \arrow[r,shift left,"\mathrm{X}"]        & \boxed{\text{GW}}\arrow[l,shift left]\arrow[d,dashed,shift right]\\
& \boxed{\text{Lagr. Floer}}\arrow[rrr,dashed,crossing over,shift right]\arrow[ur,swap,shift right,"\mathrm{XIII}"]\arrow[urrr, bend right=30,crossing over,swap,"\mathrm{XII}"]& & & \boxed{\text{MS}}\arrow[lll,dashed,crossing over,swap, shift right,"\mathrm{XVI}" {xshift=-20pt}]\arrow[u,dashed,swap,shift right,"\mathrm{XIV}"]\arrow[ul,dashed,swap,shift right,crossing over,"\mathrm{XV}"]
\end{tikzcd}
\end{center}
\caption{Diagrammatic illustration of the relations discussed in this paper.}
\label{fig:graphic}
\end{figure}

\textbf{Legend for Figure \ref{fig:graphic}:}

\begin{enumerate}[I]
    \item Witten's construction to obtain Donaldson polynomials via field theory,
    \item AKSZ-BV formulation of Donaldson--Witten theory,
    \item When the source manifold has boundary. Boundary states are given by Floer groups,
    \item Equivariant formulation of the AKSZ-BV formulation of Donaldson--Witten theory through an equivariant BV construction,
    \item BV partition function of equivariant Donaldson--Witten theory cooresponds to Nekrasov's partition function,
    \item BV-BFV quantization of equivariant Donaldson--Witten theory produces equivariant Floer groups on the boundary,
    \item Equivariant formulation of Floer (co)homology,
    \item Through Taubes' construction: the Euler characteristic with respect to Floer homology for any homology 3-sphere coincides with its Casson invariant. In particular, Donaldson--Thomas invariants are a complex version of such Euler characteristic when using the Floer groups with respect to the holomorphic Chern--Simons action functional for certain Calabi--Yau manifolds,
    \item Through enumerative counting methods and the nature of its formulation, 
    \item Conjectured correspondence by Maulik, Nekrasov, Okounkov and Pandharipande,
    \item Higher gauge theory approach when passing to complex geometry,
    \item Formulation of Gromov--Witten invariants through the moduli space of pseudo-holomorphic curves as an $A$-model path integral,
    \item The Atiyah--Floer conjecture,
    \item Ingredients for mirror symmetry,
    \item Ingredients for mirror symmetry,
    \item Ingredients for mirror symmetry.
    \end{enumerate}

\subsection{Main purpose of the paper}
The purpose of this paper is to provide the reader with: 
\begin{enumerate}[$(i)$]
\item a concise overview of the field of 4-manifold topology, both from the pure mathematical and field-theoretic point of view. Especially, to explain some of the relations of the literature as well as to new developements such as the BV-BFV techniques developed recently. 
\item a concise overview of (instanton and Lagrangian) Floer (co)homology and how it fits into the 4D-3D bulk-boundary correspondence through Chern--Simons theory. 
\item a detailed formulation of a (global) perturbative quantization of the AKSZ formulation of Donaldson--Witten theory on source manifolds with boundary which fits into the gauge theory setting of the BV-BFV formalism by using methods of configuration space integrals. 
\item a short discussion of the quantization of higher defect theories through the recently developed methods of shifted (symplectic or Poisson) structures.
\item a concise overview of the relation to Seiberg--Witten theory through Nekrasov's construction by using equivariant methods of localization and an equivariant version of the BV formalism in combination with Donaldson--Witten theory and a description of the relation with Donaldson's invariants through an equivariant version of instanton Floer (co)homology.  
\item a description of the relation to constructions of Donaldson--Thomas theory.
\item a concise overview of the conjectured correspondence of Donaldson--Thomas theory and Gromov--Witten theory and connections to other fields, such as e.g. topological recursion or supergeometry.
\end{enumerate}

\subsection{Structure of the paper}
The paper is structured as follows:
\begin{itemize}
    \item In Section \ref{sec:Donaldson-Witten_theory} we recall the definition of the moduli space of anti self-dual connections (instantons) together with some properties and define Donaldson's invariants. Moreover, we briefly discuss Witten's construction by using methods of field theory in order to obtain the aforementioned invariants.
    \item In Section \ref{sec:instanton_and_Lagrangian_Floer_homology} we recall the construction of Morse theory and Morse homology in order to construct instanton Floer homology by using the Chern--Simons action functional as a Morse function. We describe the relation of Donaldson's invariants with this homology theory through the setting of 4-manifolds with boundary.
    Moreover, we recall the construction of Lagrangian Floer homology, the Atiyah--Floer conjecture and give a remark on an approach for proving it.
    \item In Section \ref{sec:BV-BFV_formalism} we recall the classical and quantum BV-BFV formalism, the notion of BV algebras and give the important examples of AKSZ theories and abelian $BF$ theories. Moreover, we briefly discuss higher codimension (defects, branes) extensions and give some ideas for the quantization approach.
    \item In Section \ref{sec:AKSZ_formulation_of_DW_theory} we give an AKSZ construction of Donaldson--Witten theory and show how it fits into the classical BV-BFV setting.
    \item In Section \ref{sec:quantization_of_DW_theory_in_the_BV-BFV_setting} we consider the quantization of the classical setting obtained in Section \ref{sec:AKSZ_formulation_of_DW_theory} by describing the Feynman rules and the partition function. In particular, we prove that the Donaldson--Witten partition function satisfies a modified version of the Quantum Master Equation by using configuration space integrals and show that we get a well-defined cohomology theory on the boundary state space. Moreover, we construct a perturbative globalization approach by using methods of formal geometry and argue that the partition function constructed through the globalization approach lies in the kernel of a certain operator that squares to zero which also leads to a cohomology theory in this setting, similarly as for the nonglobal case.
    \item In Section \ref{sec:Nekarsov_partition_function_and_equivariant_BV_formalism} we introduce Nekrasov's partition function in terms of polynomial counts through the Hilbert scheme of monomial ideals in $\C[u_1,u_2]$ and the $\mathrm{Ext}^1$-groups by using the modular interpretation of the tangent bundle. We also consider an equivariant version of the BV gauge formalism and hence consider an equivariant version of Donaldson--Witten theory. On closed 4-manifolds the quantum picture corresponds to Nekrasov's partition function, whereas for 4-manifolds with boundary it induces an equivariant version of instanton Floer (co)homology as the boundary state. We describe the notion of equivariant Floer (co)homology following the construction of Austin--Braam in the Cartan model. Moreover, we discuss its relation to Donaldson's invariants. In particular, we consider the fact that these equivariant (co)homology groups appear as the (dual) image of the invariants.
    \item In Section \ref{sec:Relation_to_Donaldson-Thomas_theory} we recall Donaldson--Thomas invariants through their description as the relative topological Euler characteristic of the Hilbert scheme of 1-dimensional subschemes of a Calabi--Yau 3-fold by using Behrend's construction. We also recall Taubes' construction for their relation to Floer homology by considering  them as a holomorphic version of the Casson invariant. Moreover, we construct the Donaldson--Thomas partition function as a generating function with respect to the Donaldson--Thomas invariants, extend it to a relative version when considering appropriate divisors of the underlying Calabi--Yau 3-fold and describe their gluing properties.  
    \item In Section \ref{sec:Relation_to_Gromov-Witten_theory} we recall Kontsevich's moduli space of stable maps for a Calabi--Yau 3-fold and define the corresponding Gromov--Witten invariants. We also recall the partition function for Gromov--Witten invariants, similarly as before, as a generating series for these invariants and extend also this to a relative version for appropriate divisors. We briefly explain some ideas of topological recursion and its connection to Gromov--Witten invariants. To close the circle, we recall the Gromov--Witten/Donaldson--Thomas correspondence conjecture for Calabi--Yau 3-folds. Finally, we give some ideas for Gromov--Witten invariants on (graded) supermanifolds by considering ideas of Ke\ss ler--Sheshmani--Yau in the case of super Riemann surfaces. 
\end{itemize}

\subsection*{Notation}
Throughout the paper, $\Sigma$ will denote a 4-manifold and $N$ a 3-manifold. It might happen that $N$ appears as the boundary 3-manifold of a 4-manifold $\Sigma$, otherwise the boundary of $\Sigma$ is denoted by $\de\Sigma$. Riemann surfaces of genus $g$ will be denoted by $\Sigma_g$. Moduli spaces, of different flavor however, will be denoted by the calligraphic letter $\calM$ and specified through different decorations. We will denote real numbers by $\R$, complex numbers by $\C$ and integers by $\Z$.
If a (co)homology group takes values in $\R$, we will not further emphasize it, i.e. we write $H^\bullet(\Sigma)$ instead of $H^\bullet(\Sigma,\R)$, whereas for integer-valued groups we will try to indicate it, i.e. we write $H^\bullet(\Sigma,\Z)$. The exterior product between differential forms is sometimes explicitly written and sometimes left out to avoid any cumbersome notation. The Einstein summation convention is assumed, i.e. we sum over repeating indices, whenever the summation sign is not written explicitly. Algebro-geometrical objects (e.g. (projective) algebraic varieties, schemes, etc.) are usually denoted by $X$, which should be distinguished from the classical base field in the AKSZ formulation of Donaldson--Witten theory, which will be also denoted by $X$. Classical action functionals are usually denoted by $S$ (with additional decorations, depending on the given theory), whereas the BV action is denoted by the calligraphic version $\calS$. The same holds for the space of fields, usually denoted by $F$, and its BV counterpart $\calF$. The space of connections on some principal bundle $P$ is usually denoted by $\calA$ (sometimes also $\calA(P)$ to emphasize the bundle). Smooth maps on a manifold $M$ will be denoted either by $C^\infty(M)$ or by $\calO_M$ (as the structure sheaf) depending on the context. Sections of a bundle $E$ are denoted by $\Gamma(E)$.
The space of smooth differential forms on a manifold $M$ will be denoted by $\Omega^\bullet(M)$. Differential forms on $M$ with values in some bundle $E$ will be denoted by $\Omega^\bullet(M,E)$.
Finally, we will denote by $\I:=\sqrt{-1}$ the imaginary unit.  

\subsection*{Acknowledgements} I would like to thank A. S. Cattaneo, A. Sheshmani, P. Safronov and N. Reshetikhin for discussions and comments.
This research was supported by the NCCR SwissMAP, funded by the Swiss National Science Foundation, and by the SNF grant No. 200020\_192080.

\section{Donaldson--Witten theory}
\label{sec:Donaldson-Witten_theory}
\subsection{Moduli space of anti self-dual connections}
\label{subsec:Moduli_space_of_anti_self-dual_connections}
Let $(\Sigma,g)$ be a 4-dimensional Riemannian manifold and let $\mathfrak{g}=\mathrm{Lie}(G)$ be the Lie algebra of a Lie group $G$. Moreover, let $P\to \Sigma$ be a principal $G$-bundle and consider its adjoint bundle $\ad P:=P\times_G\mathfrak{g}$.
Using the Hodge star operator $*$, we can define for a 2-form $\chi\in \Omega^2(\Sigma)$ its \emph{self-dual} and \emph{anti self-dual} parts by
\begin{equation}
\label{eq:SD_ASD}
    \chi^\pm=\frac{1}{2}(\chi\pm *\chi).
\end{equation}
Note that in four dimensions we have $*\colon H^2(\Sigma)\to H^2(\Sigma)$. We can extend the splitting of \eqref{eq:SD_ASD} naturally to differential forms with values in $\ad P$. In particular we can extend this to the curvature 2-form $F_A$ of a connection 1-form $A\in \Omega^1(\Sigma,\ad P)$. A connection $A\in \Omega^1(\Sigma,\ad P)$ is called \emph{anti self-dual} or \emph{instanton} if $F^+_A=0$. We define the \emph{instanton number} of $A$ to be 
\begin{equation}
\label{eq:instanton_number}
    k_A:=\frac{1}{8\pi^2}\int_\Sigma\tr(F_A\land F_A).
\end{equation}
This number is an integer only if $G=\mathrm{SU}(2)$. Moreover, if the corresponding $\mathrm{SU}(2)$-bundle lifts to a 2-dimensional complex vector bundle $E$, we have $k_A=\int_\Sigma c_2(E)\in\Z$, i.e. the instanton number is determined by the second Chern class since $c_1(E)=0$ (because in this case $\tr(F_A\land F_A)=0$). Note that if $A$ is an anti self-dual connection, the instanton number is positive. Consider the \emph{Yang--Mills action functional}\footnote{We will not always indicate the measure whenever it is clear. When we write $\int \|F_A\|^2$, we actually mean $\int \|F_A\|^2\dd\mu$.} 
\begin{equation}
    S^{\mathrm{YM}}_\Sigma(A)=\int_\Sigma\|F_A\|^2,
\end{equation}
where we have considered the curvature locally as $F_A=\frac{1}{2}F_{ij}\dd x^i\land \dd x^j$. Hence, we get that
\begin{equation}
    S^\mathrm{YM}_\Sigma(A)=\int_\Sigma\|F_A\|^2=\int_\Sigma\|F_A^-\|^2+\int_\Sigma \|F_A^+\|^2\geq 8\pi^2k_A.
\end{equation}
This shows that the critical points of $S^\mathrm{YM}_\Sigma$ are bounded by $8\pi^2 k_A$, and in fact the minimum is attained if $A$ is anti self-dual, i.e. $F_A^+=0$. Note that $\tr(F_A\land F_A)=-\|F_A\|^2$ and thus $\tr(F_A\land F_A)=-\left(\|F_A^+\|^2-\|F_A^-\|^2\right)$.
Let us denote by $\mathcal{A}:=\Omega^1(\Sigma,\ad P)$ the space of connections on $P$. One can check that this is an infinite-dimensional space. Define the map 
\begin{align}
    \begin{split}
        s\colon \mathcal{A}&\to \Omega^2_{\mathrm{ASD}}(\Sigma,\ad P),\\
        A &\mapsto F_A^+,
    \end{split}
\end{align}
where we have denoted by $\Omega^2_{\mathrm{ASD}}(\Sigma,\ad P)$ the anti self-dual 2-forms on $\Sigma$ with values in $\ad P$.
Moreover, denote by $\mathcal{G}:=\Gamma(\mathrm{Aut}(\ad P))$ the infinite-dimensional Lie group of gauge transformations of $\ad P$. Then we can define the \emph{moduli space of anti self-dual connections}\footnote{Sometimes it is also called \emph{moduli space of instantons}. Moreover, we drop the dependence on the instanton number $k$, the 4-manifold $\Sigma$ and the group $G$ from the notation and just write $\calM_\mathrm{ASD}$ whenever it is clear.} by
\begin{equation}
\label{eq:moduli_ASD}
    \calM^k_{\mathrm{ASD}}(\Sigma,G):=s^{-1}(0)=\{[A]\in \mathcal{A}/\mathcal{G}\mid s(A)=0,\,k_A=k\}.
\end{equation}

\subsection{Donaldson polynomials}\label{subsec:Donaldson_polynomials}
In \cite{Donaldson1983}, Donaldson defined topological invariants for 4-manifolds by using methods of gauge theory. If $\Sigma$ is an oriented, compact and simply connected 4-manifold, we can split $H^2(\Sigma)$ into eigenspaces $H^2_+(\Sigma)\oplus H^2_-(\Sigma)$ by using the Hodge star (dual and anti self-dual part). Let $b^2_+:=\dim H^2_+(\Sigma)$, $P\to \Sigma$ a principal $\mathrm{SU}(2)$-bundle with second Chern class $c_2(P)$ and $k:=\int_\Sigma c_2(P)\in \mathbb{Z}$, the instanton number \eqref{eq:instanton_number}. The moduli space of anti self-dual connections $\calM_{\mathrm{ASD}}$ in this setting\footnote{For $k>0$, there is a more complicated notion of dimension called the \emph{virtual dimension} (see Footnote \ref{foot:virtual_fundamental_class} for more details on the \emph{virtual fundamental class}) for the moduli space $\calM_\mathrm{ASD}$. The virtual dimension and the dimension might be different for certain metrics. However, it was argued in \cite{FreedUhlenbeck1984} that for the subspace of Riemannian metrics on $\Sigma$ consisting of \emph{generic metrics}, the virtual dimension and the actual dimension of $\calM_\mathrm{ASD}$ coincide, i.e. \[\mathrm{vdim}\,\calM_\mathrm{ASD}=\dim\calM_\mathrm{ASD}=8k-3(1+b_+^2).\]}, as defined in \eqref{eq:moduli_ASD}, is of dimension\footnote{A general formula for the dimension is given by $\dim \calM_\mathrm{ASD}=4a(G)k-\dim G(1+b_+^2)$, where $a(G)\in\Z$ is an integer depending on the group $G$ (in our case $a(\mathrm{SU}(2))=2$). Note that we have $\dim \mathrm{SU}(2)=3$. A first approach to the dimension was considered in \cite{AtiyahHitchinSinger1977,AtiyahHitchinSinger1978} for $S^4$ which is given by $8k-3$ since then $b_+=0$.} 
\begin{equation}
\label{eq:dimension_moduli}
\dim \calM_{\mathrm{ASD}}=8k-3(1+b_+^2),
\end{equation}
which is even if $b_+^2$ is odd. Let us assume that $b_+^2$ is indeed odd and let $\dim \calM_\mathrm{ASD}=2d$. We also need the notion of an \emph{irreducible} connection. A connection $A\in \mathcal{A}$ is called \emph{irreducible}, if the image of the holonomy representation is not contained\footnote{This is equivalent to saying that the holonomy group of the connection is precisely $\mathrm{SU}(2)$ and not a proper subgroup.} in a one parameter subgroup of $\mathrm{SU}(2)$. Denote the space of irreducible connections by $\mathcal{A}^*\subset \mathcal{A}$. Consider the quotient $\mathcal{A}^*/\mathcal{G}$ by the gauge transformations and note that $\calM_{\mathrm{ASD}}\subset \mathcal{A}^*/\mathcal{G}$ is a subspace of the quotient. Moreover, it is quite compact and almost defines an element of $H_{2d}(\mathcal{A}^*/\mathcal{G},\mathbb{Z})$. 

Let $\mathcal{R}$ denote the space of all Riemannian metrics on $\Sigma$ of class $C^r$ for some $r\geq0$. Freed and Uhlenbeck have shown in \cite{FreedUhlenbeck1984} that for all $k>0$, the moduli space 
\[
\calM_\mathcal{R}:=(\mathcal{A}^*/\mathcal{G})\times \mathcal{R}
\]
is in fact a manifold. Moreover, for $k>0$, there exists a subset of $\mathcal{R}'\subset\mathcal{R}$ (elements of $\mathcal{R}'$ are called \emph{generic metrics}) such that $\calM_{\mathrm{ASD}}$ is a smooth submanifold of $\mathcal{A}^*/\mathcal{G}$ for all $g\in\mathcal{R}'$ and with virtual dimension equal to the expected dimension.

\subsubsection{Uhlenbeck compactification}
\label{thm:Uhlenbeck_removable_singularities}
By a combination of two theorems in \cite{Uhlenbeck1982,Uhlenbeck1982b}, Uhlenbeck showed that there exists a natural compactification $\overline{\calM}_{\mathrm{ASD}}$ of the moduli space $\calM_{\mathrm{ASD}}$. If we denote the moduli space of anti self-dual connections by $\calM_\mathrm{ASD}^k$ to indicate its dependence of the instanton number $k$, we get that
\begin{equation}
    \label{eq:Uhlenbeck_compactification}
    \overline{\calM}_\mathrm{ASD}\subset \bigcup_{j=0}^k \calM_{\mathrm{ASD}}^{k-j}\cup\Sym^j(\Sigma).
\end{equation}
This holds for any generic metric $g\in\mathcal{R}'$. The following theorem due to Uhlenbeck will be very important also for later constructions:

\begin{thm}[Uhlenbeck\cite{Uhlenbeck1982b}]
Let $A$ be an anti self-dual connection in a principal bundle $P$ over the punctured $4$-ball $B^4\setminus\{0\}$. If the $L^2$-norm of the curvature $F_A$ of $A$ is finite, i.e. if 
\[
\int_{B^4\setminus\{0\}}\|F_A\|^2<\infty,
\]
then there exists a gauge in which the bundle $P$ extends to a smooth bundle $\widetilde{P}$ over $B^4$ and the connection $A$ extends to a smooth anti self-dual connection $\widetilde{A}$ in $B^4$.
\end{thm}

Note that Theorem \ref{thm:Uhlenbeck_removable_singularities} implies in particular that every anti self-dual connection over $\R^4$ with bounded Yang--Mills action functional with respect to the $L^2$-norm can be obtained from an anti self-dual connection over $S^4=\R^4\cup\{\infty\}$.

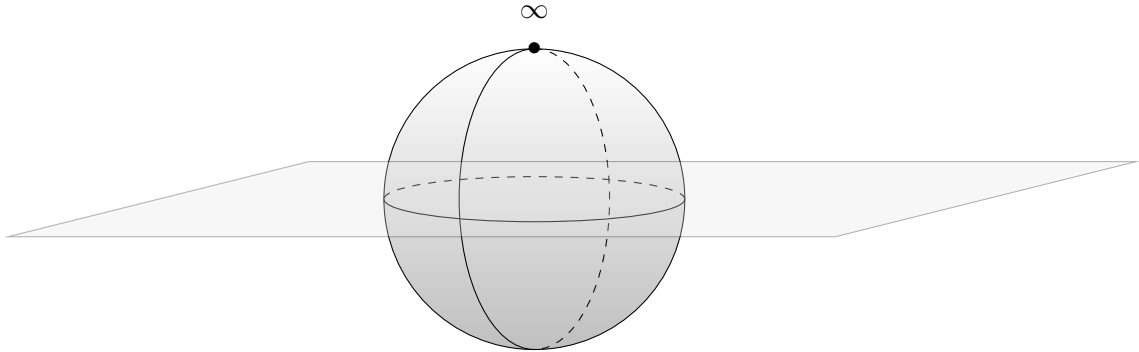
\begin{figure}[h!]
    \centering
    \begin{tikzpicture}
    \shadedraw[bottom color=gray!50, top color=gray!.5] (-1,0) circle (2);
    \draw[dashed,fill=gray!10,opacity=.1] (1,0) arc (0:360:2cm and 0.3cm);
    \draw (1,0) arc (0:-180:2cm and 0.3cm);
    \draw[dashed] (1,0) arc (0:180:2cm and 0.3cm);
    \node(a) at (-1,2.5) {$\infty$}; 
    \node(a) at (-1,2) {$\bullet$};
    \draw (-1,2) arc (90:270:1cm and 2cm);
    \draw[dashed] (-1,2) arc (90:-90:1cm and 2cm);
    \draw[fill=gray!20,opacity=.3] (-8,-.5)--(-4,.5)--(7,.5)--(3,-.5)--cycle;
    \end{tikzpicture}
    \caption{Compactification of $\R^4$ to $S^4$.}
    \label{fig:S4_compact}
\end{figure}

\subsubsection{Construction of invariants}
On $(\mathcal{A}^*/\mathcal{G})\times \Sigma$ one can define a principal $\mathrm{SO}(3)$-bundle $\widetilde{P}$ as $(\pi^*P\to \mathcal{A}^*\times \Sigma)/\mathcal{G}$ with $\pi\colon \mathcal{A}\times \Sigma\to \Sigma$ being the projection onto the second factor. Note that an $\mathrm{SO}(3)$-bundle $P\to \Sigma$ is classified by the second Stiefel--Whitney class $w_2(P)\in H^2(\Sigma,\mathbb{Z}/2)$ and the first Pontryagin class $p_1(P)\in H^4(\Sigma,\mathbb{Z})$. When $w_2(P)=0$, then the bundle lifts to an $\mathrm{SU}(2)$-bundle with second Chern class $c_2(P)=-\frac{1}{4}p_1(P)$. However, also if $w_2(P)\not=0$, everything can be done more or less similarly as before (although the orientability of the moduli space will be more complicated\footnote{In fact, if $w$ is a class in $H^2(\Sigma,\Z/2)$ and devide the integral lifts $c$ of $w$ into equivalence classes through the relation $c\sim c$ if and only if $\frac{1}{2}(c-c')$ is even. Note that when $\Sigma$ is spin, then there is only one equivalence class and there are two otherwise. The orientation on the moduli space is then induced by the orientation of $H_+^2$ and a choice of equivalence class of the integral lifts of the second Stiefel--Whitney class $w_2$ of the given $\mathrm{SO}(3)$-bundle.}).
Using the first Pontryagin class $p_1(\widetilde{P})\in H^4((\mathcal{A}^*/\mathcal{G})\times \Sigma,\mathbb{Z})$, we can define a map 
\begin{align}
\begin{split}
\mu\colon H_2(\Sigma,\mathbb{Z})&\to H^2(\mathcal{A}^*/\mathcal{G},\mathbb{Z}),\\
[C]&\mapsto \int_{C}p_1(\widetilde{P}).
\end{split}
\end{align}

Thus, one can define polynomials $\calD(\Sigma)$ of degree $d$ on $H_2(\Sigma,\mathbb{Z})$ as a map  
\begin{align}
\label{eq:Donaldson_polynomials}
\begin{split}
\calD(\Sigma)\colon H_2(\Sigma,\Z)\times\dotsm \times H_2(\Sigma,\Z)&\to \Z,\\
([C_1],\ldots,[C_d])&\mapsto\calD(\Sigma)([C_1],\ldots,[C_d]):=\int_{\overline{\calM}_\mathrm{ASD}}\prod_{i=1}^d\mu([C_i]).
\end{split}
\end{align}

\begin{thm}[Donaldson\cite{Donaldson1983}]
Suppose that $b_+^2>1$ and $k>\frac{3}{2}\left(\frac{b_+^2+1}{2}\right)$. Then the polynomials $\calD(\Sigma)$ are independent of the metric and indeed only depend on the homology classes of $C_1,\ldots, C_d$, hence $\calD(\Sigma)$ define topological invariants of $\Sigma$.
\end{thm}
 
\begin{thm}[Donaldson\cite{Donaldson1990}]
Suppose $\Sigma$ is a simply connected, oriented 4-manifold with $b_+^2$ odd and there is an orientation preserving diffeomorphism between $\Sigma$ and an oriented connected sum of manifolds $\Sigma_1$, $\Sigma_2$ both having $b_+^2>0$. Then $\calD(\Sigma)=0$ for all $k$.
\end{thm}

\subsection{Field theory formulation}
\label{subsec:field_theory_formulation}
In \cite{Witten1988}, Witten provided a way of obtaining the Donaldson polynomials by formulating a topological quantum field theory and using methods of functional integrals. In particular, he computed the expectation value of a certain observable by a perturbative expansion for a suitable Lagrangian density.
The action that was used is given in components by 
\begin{multline}
    \label{eq:Witten_action}
    S^\mathrm{DW}_\Sigma=\int_\Sigma\dd^4u\sqrt{g}\tr\bigg(\frac{3}{8}F_{ij}F^{ij}+\frac{1}{2}\phi D_iD^i\lambda-\I\eta D_i\psi^i+\I D_i\psi_j\chi^{ij}-\\-\frac{\I}{8}\phi[\chi_{ij},\chi^{ij}]-\frac{\I}{2}\lambda[\psi_i,\psi^i]-\frac{\I}{2}\phi[\eta,\eta]-\frac{1}{8}[\phi,\lambda]^2\bigg),
\end{multline}
where the fields are defined as in \cite{Witten1988}. We will refer to this action as the \emph{classical Donaldson--Witten (DW) action functional}. 
Denote by $\Phi$ the collection of all fields of the theory. Then, Donaldson's polynomials \eqref{eq:Donaldson_polynomials} can be obtained by using the correlation function with respect to \eqref{eq:Witten_action} as
\[
\left\langle O_{a_1}\dotsm O_{a_d}\right\rangle:=\int \exp(\I S^\mathrm{DW}_\Sigma(\Phi)/\hbar)\prod_{j=1}^dO_{a_j}(\Phi)\mathscr{D}[\Phi],
\]
for some observables $O_{a_1},\ldots, O_{a_d}$. The observables of interest are given by 
\begin{equation}
    \label{eq:Witten_observables}
    O^{(\gamma)}:=\int_\gamma W_{k_\gamma},
\end{equation}
where $\gamma\in H^{k_\gamma}(\Sigma,\Z)$ and, for $k_\gamma=0,\ldots, 4$, we have $W_0:=\frac{1}{2}\tr(\phi\land \phi)$, $W_1:=\tr(\phi\land \psi)$, $W_2:=\tr(\frac{1}{2}\psi\land\psi+\I\phi\land F)$, $W_3:=\I\tr(\psi\land F)$ and $W_4:=-\frac{1}{2}\tr(F\land F)$. To each $O^{(\gamma)}$, one actually associates a $(4-k_\gamma)$-form $\mu(\gamma)$ on $\overline{\calM}_\mathrm{ASD}$ as in Section \ref{subsec:Donaldson_polynomials}. Explicitly, we get 
\begin{align}
    \label{eq:correlation_function}
    \begin{split}
    \left\langle O^{(\gamma_1)}\dotsm O^{(\gamma_d)}\right\rangle&=\int \exp(\I S^\mathrm{DW}_\Sigma(\Phi)/\hbar)\prod_{j=1}^dO^{(\gamma_j)}(\Phi)\mathscr{D}[\Phi]\\
    &=\int \exp(\I S^\mathrm{DW}_\Sigma(\Phi)/\hbar)\prod_{j=1}^d\int_{\gamma_j}W_{k_j}(\Phi)\mathscr{D}[\Phi]\\
    &=\int_{\overline{\calM}_\mathrm{ASD}}\prod_{j=1}^d\mu(\gamma_j)=\mathcal{D}(\Sigma)(\gamma_1,\ldots,\gamma_d).
    \end{split}
\end{align}
\begin{rem}
The integral in \eqref{eq:correlation_function} is rigorously defined since the differential forms $\mu(\gamma)$ can be obtained from each observable $O^{(\gamma)}$ by integrating out the high energy modes. 
\end{rem}

\section{Instanton and Lagrangian Floer homology}
\label{sec:instanton_and_Lagrangian_Floer_homology}
\subsection{Morse homology}
\label{subsec:Morse_homology}
Let us recall the construction of \emph{Morse homology}. Fix a compact, closed manifold $\Sigma$. A function $f\colon \Sigma\to \R$ is called \emph{Morse function}, if all its critical points are non-degenerate, i.e. for all critical points, the Hessian of $f$ is invertible. In other words, the section $\dd f$ is transverse to the zero section of $T^*\Sigma$.
Since the Hessian is self-adjoint, it has real spectrum and hence we define $\mathrm{ind}_pf$ to be the dimension of the sum of all negative eigenspaces. Moreover, we assume that for any two critical points $p,q$ of $f$ we have that $\mathrm{ind}_pf>\mathrm{ind}_qf$ implies $f(p)>f(q)$. The function $f$ is then said to be \emph{self-indexing}. If we choose a Riemannian metric on $\Sigma$, we can look at the gradient $\nabla f$ and the corresponding flow equations
\begin{equation}
\label{eq:morse_flow}
    \dot{\gamma}(t)=-\nabla_{\gamma(t)}f.
\end{equation}
For two critical points $p,q$ of $f$, we define the moduli space $\calM(p,q)$ of solutions of \eqref{eq:morse_flow} such that
\begin{equation}
    \lim_{t\to -\infty}\gamma(t)=p,\qquad \lim_{t\to+\infty}\gamma(t)=q.
\end{equation}
Note that $\calM(p,q)$ is empty unless $f(p)>f(q)$. In fact, we have the following theorem:
\begin{thm}[Morse--Smale]
Let $p,q$ be two distinct critical points of a Morse function $f$.
Then $\calM(p,q)$ is a manifold of dimension $\mathrm{ind}_pf-\mathrm{ind}_qf$.
\end{thm}
There is a free and proper $\R$-action on $\calM(p,q)$ given by reparametrization $\gamma(t)\mapsto \gamma(t-a)$ for some $a\in \R$. We consider then the quotient 
\[
\overline{\calM}(p,q):=\calM(p,q)/\R.
\]
One can check that if $\mathrm{ind}_pf=\mathrm{ind}_qf+1$, then $\overline{\calM}(p,q)$ is a compact, oriented 0-dimensional manifold. Hence, by counting its points by signs, we can deduce
\begin{equation}
    \#\overline{\calM}(p,q)\in \Z.
\end{equation}
We can now construct a chain complex as follows. Define the chain groups to be given by 
\begin{equation}
    CM_k(\Sigma,f):=\bigoplus_{\substack{\text{$p$ critical point of $f$}\\ \mathrm{ind}_pf=k}}\Z\langle p\rangle.
\end{equation}
The boundary operator $\de$ will be constructed by counting flow lines. Namely, for critical points $p$ of $f$ with $\mathrm{ind}_pf=k$ we define the boundary operator by 
\begin{equation}
    \de p=\sum_{\substack{\text{$q$ critical point of $f$}\\ \mathrm{ind}_qf=k-1}}\#\overline{\calM}(p,q)\cdot q
\end{equation}
One can show (although it is non-trivial) that $\de^2=0$ and thus indeed defines a differential. 
Finally, \emph{Morse homology} is given by the homology of this complex, and we denote it by 
\[
HM_\bullet(\Sigma,f):=\ker \de/\mathrm{im}\,\de.
\]

\subsection{The (holomorphic) Chern--Simons action functional}
\label{subsec:Chern-Simons}
An important action functional for further discussions is given by Chern--Simons theory \cite{Chern1974}, which is a topological field theory on a 3-manifold with many connections to other mathematical theories \cite{Witten1989,Reshetikhin1991,Cattaneo2008}. Let us briefly recall its construction and extension to a holomorphic version which will be important later (see Section \ref{sec:Relation_to_Donaldson-Thomas_theory}). Let $N$ be a real 3-manifold and consider a vector bundle $E\to N$ with structure group $G$. The curvature $F_A$ of a connection $A\in\calA$ defines a closed 1-form
\begin{equation}
\label{eq:1-form}
a\mapsto \frac{1}{4\pi^2}\int_N\tr(a\land F_A),\qquad a\in\Omega^1(N,\ad E)
\end{equation}
on $\calA$. In fact, we can extend \eqref{eq:1-form} to gauge equivalence classes since this expression is gauge invariant. Fix a base point $A_0$ in the space of gauge equivalence classes. One can show that \eqref{eq:1-form} actually appears as the exterior derivative of a local action functional given by 
\begin{equation}
    \label{eq:Chern-Simons_action}
    S^\mathrm{CS}_N(A):=\int_N\tr\left(\dd_{A_0}a\land a+\frac{2}{3}a\land a\land a\right),\qquad A=A_0+a.
\end{equation}
This is the Chern--Simons action functional. One can check that this is gauge invariant under transformations connected to the identity. Moreover, on the gauge equivalence classes it is well-defined modulo $\Z$. Note also that critical points are given by flat connections. One can extend this picture to Calabi--Yau 3-folds, i.e. smooth compact K\"ahler 3-folds $X$ with trivial canonical bundle $K_X\cong \calO_X$. Consider now the space $\calA_\mathrm{hol}:=\{\bar\de\text{-operators on a fixed smooth bundle $E\to X$}\}$ and the closed 1-form
\begin{equation}
\label{eq:holomorphic_1-form}
a\mapsto \frac{1}{4\pi^2}\int_X\tr\left(a\land F^{0,2}_A\right)\land \dd\text{vol}_X,\qquad a\in\Omega^{0,1}(X,\ad E),
\end{equation}
where $F^{0,2}_A$ denotes the antiholomorphic curvature (i.e. with respect to $\bar\de$) of a holomorphic connection $A\in\calA_\mathrm{hol}$ and $\dd\text{vol}_X$ denotes the complex volume form on $X$. similarly as before, the 1-form \eqref{eq:holomorphic_1-form} is gauge invariant and thus descends to the space of gauge equivalence classes. Again, fixing a basepoint $A_0\in\calA_\mathrm{hol}$, \eqref{eq:holomorphic_1-form} appears as the exterior derivative of a local holomorphic action functional given by 
\begin{equation}
    \label{eq:holomorphic_Chern-Simons_action}
    S^{\mathrm{hol},\mathrm{CS}}_X(A):=\frac{1}{4\pi^2}\int_X\tr\left(\bar\de_{A_0}a\land a+\frac{2}{3}a\land a\land a\right)\land \dd\text{vol}_X,\qquad A=A_0+a.
\end{equation}
This the \emph{holomorphic} Chern--Simons action functional. Again, one can check that \eqref{eq:holomorphic_Chern-Simons_action} is gauge invariant with respect to transformations connected to the identity. The critical points of \eqref{eq:holomorphic_Chern-Simons_action} are given by integrable holomorphic structures on the bundle $E$.

\subsection{A 4D-3D bulk-boundary correspondence on (infinite) cylinders}
\label{subsec:4D-3D}
Let $\Sigma$ be a 4-manifold. Moreover, let $G=\mathrm{SU}(2)$ and consider a principal $G$-bundle $P\to \Sigma$. Then $P$ has one characteristic class, the second \emph{Chern class} $c_2(P)\in H^2(\Sigma,\Z)$. Using the \emph{Chern--Weil formalism} \cite{Weil1949,Chern1952}, we can identify it with its de Rham cohomology representative in the image of $H^2(\Sigma,\R)$, which is given by
\begin{equation}
    c_2(P)=\frac{1}{8\pi^2}\tr(F_A\land F_A),
\end{equation}
for some connection 1-form $A\in \mathcal{A}$ on $P$ and its corresponding curvature 2-form $F_A=\dd A+\frac{1}{2}[A,A]$. Indeed, this cohomology class is independent of the choice of connection and it is a closed 2-form. Note also that its integral over any 4-manifold $\Sigma$ is an integer, since $c_2(P)$ is an integral class.
If $G=\mathrm{SO}(3)$, we get that the corresponding class is given by the first Pontryagin class of the associated bundle $P^{\mathrm{SO}(3)}:=P\times_{\mathrm{SO}(3)}\R^3$, which is given by 
\begin{equation}
    p_1(P^{\mathrm{SO}(3)})=-\frac{1}{2\pi^2}\tr(F_A\land F_A).
\end{equation}
Note that these classes vanish on 3-manifolds. However, there is a way how we can construct invariants on 3-manifolds by using Chern--Simons theory (see Section \ref{subsec:Chern-Simons}). One can actually check that
\begin{equation}
\label{eq:Chern-Simons_inv}
    \left(\int_\Sigma\tr(F_A\land F_A)\right)\Big/8\pi^2\Z
\end{equation}
depends only on the gauge equivalence class of $A\big|_{\de\Sigma}$ and not on $\Sigma$ nor on $A$ in the bulk. If $A$ extends\footnote{In fact, one can always extend an $\mathrm{SU}(2)$- or $\mathrm{SO}(3)$-bundle on a closed oriented 3-manifold over some compact oriented 4-manifold. In particular, this follows from the fact that for these groups the cobordism group is given by the third homology group of their classifying spaces.} to a connection $A'$ on an extended bundle $P'\to \Sigma'$, then we can glue $\Sigma$ to $\Sigma'$ along their common boundary $N$ and obtain a new bundle $P''\to \Sigma'':=\Sigma\cup_{N}\Sigma'$. Hence
\begin{equation}
    \int_\Sigma\tr(F_{A}\land F_A)-\int_{\Sigma'}\tr(F_{A'}\land F_{A'})=\int_{\Sigma''}\tr(F_{A''}\land F_{A''})\in 8\pi^2\Z.
\end{equation}
Let $B$ be a connection on $P\to N$. Then the Chern--Simons action functional $S^\mathrm{CS}_{N}(B)$ at $B$ is given by \eqref{eq:Chern-Simons_inv} where $F_A$ is the curvature of a connection $A$ that is given as an extension of $B$ to some bundle over $\Sigma$.
Let $B_0$ be a fixed connection on $P$ and consider a family of connections $(B_t)_{t\in[0,1]}$ such that $B_1=B$, which we can regard as a connection $A$ on $I\times P\to I\times N$. Then we have 
\begin{equation}
    S^\mathrm{YM}_{I\times N,B_0}(B)=\int_{I\times N}\|F_A\|^2=-\int_{I\times N}\tr(F_A\land F_A).
\end{equation}
Consider the path $B_t=B_0+tb$ for some $b$ such that $B=B_0+b$, and assume that $B_0$ is a trivial connection by choosing a trivialization for $G=\mathrm{SU}(2)$. Then
\begin{equation}
    F_A=\dd(tb)+\frac{t^2}{2}[b\land b]=\dd t\land b+t\dd b+\frac{t^2}{2}[b\land b],
\end{equation}
which gives
\begin{equation}
    \tr(F_A\land F_A)=\dd t\land \tr(b\land (2t\dd b+t^2[b\land b])).
\end{equation}
Hence, we get the following chain of equality:
\begin{align}
\begin{split}
    S^\mathrm{YM}_{I\times N}(B)&=-\int_{I\times N}\tr(F_A\land F_A)\\
    &=-\int_{I\times N}\dd t\land \tr(b\land(2t\dd b+t^2[b\land b]))\\
    &=-\int_{N}\tr\left(b\land \dd b+\frac{2}{3}b\land b\land b\right)=-S^\mathrm{CS}_{N,B_0}(b).
\end{split}
\end{align}
Moreover, if $B_0$ does not arise from a trivialization, we get 
\begin{equation}
   S^\mathrm{YM}_{I\times N}(B)=-\int_{N}\tr\left(2b\land F_{B_0}+b\land \dd b+\frac{2}{3}b\land b\land b\right),
\end{equation}
where $F_{B_0}$ is the curvature of the connection $B_0$.

\subsection{Instanton Floer homology}\label{subsec:Instanton_Floer_homology}
We will restrict ourselves to the case where $G=\mathrm{SU}(2)$ and $N$ is an oriented integral homology 3-sphere\footnote{This is a 3-manifold $N$ whose homology groups are the same as for $S^3$. Namely, $H_0(N,\Z)=H_3(N,\Z)\cong \Z$ and $H_1(N,\Z)=H_2(N,\Z)=0$.} endowed with a Riemannian metric. Denote by $\mathcal{A}$ the set of all connections on $P\to N$, which is an affine space over $\Omega^1(N,\ad P)$. Recall that a gauge transformation $g$ is a bundle automorphism of $P$ covering the identity map of $\Sigma$. For such a transformation $g\colon P\to P$, we can construct a new map $\hat g\colon P\to G$ by $g(p)=p\hat g(p)$, with the property $\hat g(ph)=h^{-1}\hat g(p)h$ for all $p\in P$ and $h\in G$. We call the set of all gauge transformations $\mathcal{G}$. The elements of $\mathcal{G}$ are identified with sections of $\Ad P=P\times_{\Ad} G$. We would like to construct a Morse chain complex (as in Section \ref{subsec:Morse_homology}) over $\mathcal{A}/\mathcal{G}$. Indeed, this can be done since by completing $\mathcal{A}$ and $\mathcal{G}$ with respect to the \emph{Sobolev norm topology}\footnote{For the analytic sides of the construction, one needs the notion of a Sobolev space $L^p_\ell$ of sections of bundles associated to $P$, i.e. sections locally represented by functions with first $\ell$ derivatives in $L^p$. There one can define a norm $\|\enspace\|_{L^p_\ell(A)}$ for any smooth connection $A$, for example, if $p=2$ and $\ell=1$, we have 
\[
\|\sigma\|^2_{L^2_1(A)}=\|\dd_A\sigma\|^2_{L^2}+\|\sigma\|^2_{L^2}=\int_\Sigma\dd\mu\left(\| \dd_A\sigma\|^2+\|\sigma\|^2\right).
\]}, the spaces $\mathcal{A}$ and $\mathcal{G}$ have the structure of an infinite-dimensional manifold. Passing to the quotient will produce a manifold by a local slice theorem for the natural action of $\mathcal{G}$ on $\mathcal{A}$, which is given by pulling back 1-forms on $P$, i.e. $A\mapsto g^*A$. Note that the curvature $F_A$ of $A$ would transform to the curvature of $g^*A$ by $F_{g^*A}=\Ad_{\hat g}F_A$. 
Denote by $\mathcal{A}^*\subset \mathcal{A}$ the subspace of irreducible connections and consider the tangent space of $\calA^*$ at a reference connection $B$, which is given by 
\begin{equation}
    T_{B}\mathcal{A}^*\cong\Omega^1(N,\ad P).
\end{equation}
Next, we would like to consider the tangent space at some equivalence class $[B]\in \mathcal{A}^*/\mathcal{G}$. For this purpose, we use the Hodge star $*\colon \Omega^j(N)\to \Omega^{3-j}(N)$ induced by the metric on $N$. Moreover, recall that on $j$-forms we have $*^2=(-1)^{j(3-j)}$. Denoting by $\dd_B\colon \Omega^j(N,\ad P)\to \Omega^{j+1}(N,\ad P)$ the covariant derivative with respect to $B$, we can define its formal adjoint by $\dd^*_B:=-*\dd_{B}*$. Then we can obtain
\begin{equation}
    T_{[B]}(\mathcal{A}^*/\mathcal{G})\cong \ker(\dd_B^*).
\end{equation}
The key point for our Morse complex is that we want to consider the Chern--Simons action functional
\[
S^\mathrm{CS}_{N}\colon \mathcal{A}^*/\mathcal{G}\to \R/8\pi^2\Z
\]
to play the role of a Morse function\footnote{Of course, this requires the Chern--Simons action functional to actually be Morse, i.e. with critical points being nondegenerate. In particular, one can consider small perturbations $\varepsilon$ (\emph{holonomy perturbations}) such that $S^\mathrm{CS}+\varepsilon$ is indeed sufficiently nice. Here, $\varepsilon\colon \calA^*/\calG\to \R/8\pi^2\Z$ is some admissible perturbation function. 
In particular, the existence of regular values for the perturbed action is connected to the properties of being smooth and its differential to be a \emph{Fredholm operator}.}. For convenience, we actually want to consider $-S^\mathrm{CS}_{N}$.
Recall from Section \ref{subsec:Chern-Simons} that critical points of the Chern--Simons action functional on $\mathcal{A}^*$ are given by flat connections and on $\mathcal{A}^*/\mathcal{G}$ by gauge equivalence classes of flat connections. 

The Floer chains $CF_\bullet(N)$ are given by the $\Z$-module with generators $[B]$ being gauge equivalence classes of flat connections of the trivial $\mathrm{SU}(2)$-bundle over $N$. For two flat connections $B_0,B_1$ consider a path $t\mapsto B_t$ connecting them. The corresponding path of operators has a spectral flow which is defined by pursuing the net number of eigenvalues crossing zero. Denote by $\mathrm{Eig}^{\mp}(B_t)$ the set of eigenvalues that pass along the path of operators $B_t$ from negative to positive and by $\mathrm{Eig}^{\pm}(B_t)$ the set of eigenvalues that pass along the path of operators $B_t$ from positive to negative.
Following \cite{AtiyahPatodiSinger75-76}, we define the \emph{spectral flow} to be the number
\begin{equation}
\mathrm{sf}(B_0,B_1):=\#\mathrm{Eig}^\mp(B_t)-\#\mathrm{Eig}^\pm(B_t).
\end{equation}
Note that in finite dimensions, this corresponds to the difference of the index of critical points. This will replace the grading of the Floer chains for the infinite-dimensional setting.
Consider the space of connections $A$ on $\R\times P$, which is a bundle over $\R\times N$, satisfying:
\begin{enumerate}
    \item $A$ is anti-self dual, i.e. $F_A=-*F_A$,
    \item $\lim_{t\to -\infty}[A\vert_{\{t\}\times N}]=:A_-$,
    \item $\lim_{t\to +\infty}[A\vert_{\{t\}\times N}]=:A_+$,
    \item $A$ has finite energy, i.e. the curvature $F_A$ has finite $L_2$-norm:
    \[
    \|F_A\|_2^2=\int_{\R\times N} \|F_A\|^2<\infty.
    \]
\end{enumerate}
If we denote the space of such connections by $\widehat{\calM}(A_-,A_+)$, we can define a moduli space 
\[
\calM(A_-,A_+):=\widehat{\calM}(A_-,A_+)/\mathcal{G},
\]
where $\mathcal{G}:=\Aut(\R\times P)=\{f\colon \R\times N\to \mathrm{SU}(2)\}$ denotes the gauge transformations. We can define an $\R$-action on $\calM(A_-,A_+)$ by shifting the $t$ variable and define 
\[
\overline{\calM}(A_-,A_+):=\calM(A_-,A_+)/\R.
\]
In fact, there exists (see \cite{Floer1988,Donaldson2002,Saveliev2001}) a small (holonomy) perturbation $\varepsilon>0$ of $S^\mathrm{CS}_{N}$ which leads to a moduli space $\calM_\varepsilon(A_-,A_+)$ such that $\overline{\calM}_\varepsilon(A_-,A_+)$ is a smooth oriented manifold with 
\[
\dim \overline{\calM}_\varepsilon(A_-,A_+)=\mathrm{sf}(A_-,A_+)-1.
\]
Moreover, if $\dim\overline{\calM}_\varepsilon(A_-,A_+)=0$, then $\overline{\calM}_\varepsilon(A_-,A_+)$ is compact.
Finally, the boundary operator is defined through the counting of instantons as 
\[
\de A_-:=\sum_{A_+\in \mathrm{\mathcal{A}_\mathrm{flat}}\atop\mathrm{sf}(A_-,A_+)=1}\,\,\#\overline{\calM}_\varepsilon(A_-,A_+)\cdot A_+,
\]
where $\mathcal{A}_\mathrm{flat}$ denotes the space of flat connections.
One can show that the resulting complex $CF_\bullet(N,\de)$ is actually independent of the metric and the perturbation and that $\de^2=0$. Therefore, we get a well-defined homology theory \cite{Floer1988,Donaldson2002}.
The corresponding homology, denoted by $HF_\bullet(N):=\ker \de/\mathrm{im}\, \de$, is called \emph{instanton Floer homology}\footnote{We will sometimes also just call it \emph{Floer homology}, i.e. dropping the word \emph{instanton}, whenever it is clear. We mainly consider this type of Floer homology in this paper. There are different types of Floer homology constructions such as e.g. Lagrangian Floer homology (see Section \ref{subsec:Lagrangian_Floer_homology}) which, e.g., plays an important role in the formulation of Kontsevich's homological mirror symmetry conjecture (Conjecture \ref{conj:HMS})}.

\subsection{Relation to Donaldson polynomials}\label{subsec:relation_to_Donaldson_polynomials}
After Floer defined his homology groups, Donaldson soon realized how they were related to the polynomials he has constructed. Good expositions can be also found in \cite{Atiyah1987,Braam1991}. Assume that $\Sigma=\Sigma_1\cup_{N} \Sigma_2$, where $N$ is an oriented homology 3-sphere and $\Sigma_1$ (resp. $\Sigma_2$) is a simply connected 4-manifold with boundary $N$ (resp. $N^\mathrm{opp}$)\footnote{Here we denote $N^\mathrm{opp}$ to be $N$ with opposite orientation.}. By the assumption that $b^2_+>0$ for both $\Sigma_1$ and $\Sigma_2$, Donaldson defined polynomials
\begin{align}
    \calD(\Sigma_1)\colon H_2(\Sigma_1,\mathbb{Z})\times \dotsm\times H_2(\Sigma_1,\mathbb{Z})&\to (HF_\bullet(N))^*,\\
    \calD(\Sigma_2)\colon H_2(\Sigma_2,\mathbb{Z})\times \dotsm \times H_2(\Sigma_2,\mathbb{Z})&\to (HF_\bullet(N^\mathrm{opp}))^*,
\end{align}
that is that the polynomials $\calD$ are valued in the dual of the Floer homology on the boundary. In fact, one can define a pairing $\langle\enspace,\enspace\rangle_{HF}$ between elements of $(HF(N))^*$ and $(HF(N^\mathrm{opp}))^*$ 
\begin{equation}
\label{eq:HF_pairing}
\langle\enspace,\enspace\rangle_{HF}\colon (HF_j(N))^*\times (HF_{3-j}(N^\mathrm{opp}))^*\to \mathbb{Z}
\end{equation}
by using the fact that $CF_j(N)=CF_{3-j}(N^\mathrm{opp})$ and for irreducible flat connections $A\in CF_j(N)$ and $B\in CF_{j-1}(N)$, that both, $\langle \dd_N A,B\rangle$ and $\langle A,\dd_{N^\mathrm{opp}}B\rangle$, are the number of flow lines from $A$ to $B$ counted with sign. Here we have denoted by $\dd_N$ the de Rham differential on $N$ and by $\dd_{N^\mathrm{opp}}$ the de Rham differential on $N^\mathrm{opp}$.
In fact, the signs are the same and hence these numbers do agree. Note that by defining a cochain complex $CF^j:=\Hom(CF_j,\mathbb{Z})$, we get
\[
HF_j(N)\cong HF^{3-j}(N^\mathrm{opp}).
\]

\begin{thm}[Braam--Donaldson\cite{BraamDonaldson1995,Donaldson2002}]
\label{thm:Donaldson_gluing}
We have 
\[
\calD(\Sigma_1\cup_N\Sigma_2)=\langle \calD(\Sigma_1),\calD(\Sigma_2)\rangle_{HF},
\]
where $\langle\enspace,\enspace\rangle_{HF}$ denotes the pairing as in \eqref{eq:HF_pairing}.
\end{thm}
Theorem \ref{thm:Donaldson_gluing} tells us how Donaldson polynomials glue along boundaries of 4-manifolds. This will be interesting in connection to the perturbative field-theoretic approach of quantum gauge theories on manifolds with boundary which in the cohomological symplectic setting is compatible with cutting and gluing \cite{CMR2}. We will see how this result fits into this framework. 

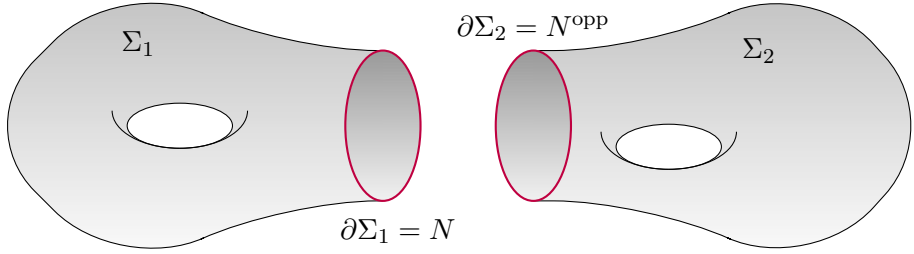
\begin{figure}[h!]
    \centering
    \begin{tikzpicture}
    \shadedraw[rounded corners=35pt,top color=gray!50, bottom color=gray!.5](5.5,-1)--(4.2,-1)--(2,-2)--(0,0)--(2,2)--(4.2,1)--(5.5,1);
    \shadedraw[rounded corners=35pt,top color=gray!50, bottom color=gray!.5](7.5,1)--(9,1)--(11.2,2)--(13,0)
    --(11.2,-2)--(9,-1)--(7.5,-1);
    \draw[fill=white] (3.5,0) arc (0:360:0.7cm and 0.3cm);
    \draw (2.8,-0.3) arc (-90:0:0.9cm and 0.5cm);
    \draw (2.8,-0.3) arc (-90:-180:0.9cm and 0.5cm);
    \shadedraw[color=white, top color=gray!80, bottom color=gray!10.5] (6,0) arc (0:360:0.5cm and 1cm);
    \draw[thick, color=purple] (6,0) arc (0:360:0.5cm and 1cm);
    \shadedraw[color=white, top color=gray!80, bottom color=gray!10.5] (8,0) arc (0:360:0.5cm and 1cm);
    \draw[thick, color=purple] (8,0) arc (0:360:0.5cm and 1cm);
    \draw[fill=white] (10,-0.28) arc (0:360:0.7cm and 0.3cm);
    \draw (9.3,-0.58) arc (-90:0:0.9cm and 0.5cm);
    \draw (9.3,-0.58) arc (-90:-180:0.9cm and 0.5cm);
    \node (a) at (5.7,-1.4) {$\de\Sigma_1=N$};
    \node (a) at (7.5,1.3) {$\de\Sigma_2=N^\mathrm{opp}$};
    \node (a) at (26:2.5) {$\Sigma_1$};
    \node (a) at (10.5,1) {$\Sigma_2$};
    \end{tikzpicture}
    \caption{Gluing of two manifolds along a common boundary with opposite orientations. Think of $\Sigma_1$ and $\Sigma_2$ as two 4-manifolds with boundary 3-manifolds $\de\Sigma_1$ and $\de\Sigma_2$, respectively.}
    \label{fig:gluing_1}
\end{figure}


\subsection{Field theory approach to instanton Floer homology}
\label{subsec:field_theory_approach_to_floer_homology}
Based on ideas of Atiyah \cite{Atiyah1987}, Witten gave an approach for a quantum field theoretic construction regarding the appearance of Floer homology in Donaldson theory for 4-manifolds with boundary in \cite{Witten1988a}. His construction uses a supersymmetric field theory approach to Morse theory developed in \cite{Witten1982}. Moreover, he gave a quantum field theoretic interpretation for the gluing of 4-manifolds along a common boundary 3-manifold with pairing ground states on the boundary contained in the instanton Floer homology groups which recovers the result of Theorem \ref{thm:Donaldson_gluing}. 
Recall that, in general, we are considering expectation values to be given by an integral of the form 
\begin{equation}\label{eq:expectation_path_integral}
\langle O\rangle=\int \exp(\I S(\Phi)/\hbar)O(\Phi)\mathscr{D}[\Phi],
\end{equation}
where $S$ is the action of the theory and $\Phi$ denotes the collection of all integration variables. Let $\Sigma$ be the underlying 4-manifold with boundary $\de\Sigma$ and consider the field theory as in Section \ref{subsec:field_theory_formulation}. Choosing boundary conditions on $\de\Sigma$ is usually required for the computation of a path integral as in \eqref{eq:expectation_path_integral}. For the 3-manifold $\de\Sigma$ (more precisely, the infinite cylinder $\de\Sigma\times\R$), we can consider the associated state space $\calH_{\de\Sigma}$. Let $\Phi\vert_{\de\Sigma}$ be the restriction of all integration variables to $\de\Sigma$ and note that then $\calH_{\de\Sigma}$ denotes the space of functionals depending on $\Phi\vert_{\de\Sigma}$ and a state corresponds to a functional $\Psi(\Phi\vert_{\de\Sigma})$.
If we consider the state $\Psi(\Phi\vert_{\de\Sigma})$ to define a boundary condition, we can consider the path integral \eqref{eq:expectation_path_integral} in terms of this condition to be 
\begin{equation}\label{eq:expectation_boundary_condition}
\langle O\Psi(\Phi\vert_{\de\Sigma})\rangle=\int\exp(\I S^\mathrm{DW}_\Sigma(\Phi)/\hbar)O(\Phi)\Psi(\Phi\vert_{\de\Sigma})\mathscr{D}[\Phi],
\end{equation}
where $S^\mathrm{DW}_\Sigma$ is defined as in \eqref{eq:Witten_action} and 
\begin{equation}
\label{eq:observable}
O:=\prod_{j=1}^d\int_{\gamma_j}W_{k_j},
\end{equation}
with $\gamma_j\in H_{k_j}(\Sigma,\Z)$ and $W_{k_j}$ defined as in Section \ref{subsec:field_theory_formulation}.
In fact, \eqref{eq:expectation_boundary_condition} turns out to be a topological invariant if $\Psi$ represents an instanton Floer cohomology class and it depends only on the cohomology class represented by $\Psi$. The observables are chosen similarly as for the Donaldson polynomials (see Section \ref{subsec:field_theory_formulation}). Hence, one obtains the Donaldson polynomials with values in the (dual of the) instanton Floer cohomology groups as in Section \ref{subsec:Donaldson_polynomials}.

\subsection{Lagrangian Floer homology}
\label{subsec:Lagrangian_Floer_homology}
Besides the instanton construction \cite{Floer1989}, there is another type of Floer homology theory which uses the data of a symplectic manifold \cite{Floer1988b}. We will mainly use the excellent introductory paper \cite{Auroux2014} for this section and we refer to it and the references within for more details and further constructions using Lagrangian Floer homology. We should remark that there the construction is dual to Floer's original construction and thus the grading convention will be reversed (so we should rather speak of cohomology instead of homology, but we decide to keep the original term).
Let $(\Sigma,\omega)$ be a compact symplectic manifold and consider two compact Lagrangian submanifolds $\calL_0,\calL_1\subset \Sigma$. Then we can associate to the pair $(\calL_0,\calL_1)$ of Lagrangians a chain complex $CF(\calL_0,\calL_1)$ which is freely generated by the intersection points of $\calL_0$ and $\calL_1$ together with a differential $\de\colon CF(\calL_0,\calL_1)\to CF(\calL_0,\calL_1)$ such that $\de^2=0$ and thus we can consider its corresponding homology $HF(\calL_0,\calL_1):=\ker \de/\mathrm{im}\, \de$.
Moreover, if there is a Hamiltonian isotopy between two compact Lagrangian submanifolds $\calL_1$ and $\calL_1'$, we get an isomorphism $HF(\calL_0,\calL_1)\cong HF(\calL_0,\calL_1')$ and if there is a Hamiltonian isotopy between $\calL_0$ and $\calL_1$, then $HF(\calL_0,\calL_1)\cong H^\bullet(\calL_0)$. Similarly as before, Lagrangian Floer homology can be formally viewed as an infinite version of Morse homology with respect to the Morse function given by the functional defined on the universal cover of the path space $\mathrm{Path}(\calL_0,\calL_1):=\{\gamma\colon[0,1]\to \Sigma\mid \gamma(0)\in \calL_0,\, \gamma(1)\in \calL_1\}$ which is given by 
\[
S(\gamma,[\Gamma])=-\int_\Gamma \omega,
\]
with $\gamma\in \mathrm{Path}(\calL_0,\calL_1)$ and $[\Gamma]$ the equivalence class of a homotopy $\Gamma\colon [0,1]\times[0,1]\to \Sigma$ between $\gamma$ and a fixed base point in the connected component of $\mathrm{Path}(\calL_0,\calL_1)$ containing $\gamma$. Usually, we want that the two Lagrangian submanifolds intersect transversally, such that there is a finite set of intersection points. The Lagrangian Floer chain complex is then given by 
\[
CF(\calL_0,\calL_1):=\bigoplus_{p\in \calL_0\cap \calL_1}\Lambda\cdot p,
\]
where $\Lambda:=\{\sum_{i\geq 0}a_iT^{\lambda_i}\mid a_i\in \mathbb{K},\, \lambda_i\in\R,\, \lim_{i\to \infty}=+\infty\}$ denotes the \emph{Novikov field} for some field $\mathbb{K}$. 
Let $J$ be a $\omega$-compatible almost-complex structure on $\Sigma$. We can define the differential $\de$ by counting \emph{pseudo-holomorphic strips} in $\Sigma$ with boundary contained in $\calL_0$ and $\calL_1$. This is parametrized through the following construction: Let $p,q\in \calL_0\cap \calL_1$. Then the coefficient of $q$ in $\de p$ is given by the moduli space of maps $u\colon \R\times[0,1]\to \Sigma$ 
such that:
\begin{enumerate}
\item it solves the Cauchy--Riemann equation 
\[
\bar \de_J u:=\frac{\de u}{\de s}+J(u)\frac{\de u}{\de t}=0,
\]
with respect to the boundary conditions
\[
\begin{cases}u(s,0)\in \calL_0\text{ and }u(s,1)\in \calL_1,\,\, \forall s\in \R,\\
\lim_{s\to +\infty}u(s,t)=p,\,\, \lim_{s\to -\infty}u(s,t)=q\end{cases}
\]
\item it has finite energy (symplectic area of the strip):
\[
\int u^*\omega=\iint \left\vert \frac{\de u}{\de s}\right\vert^2\, \dd s\dd t<\infty.
\]
\end{enumerate}

Let $\widehat{\calM}(p,q;[u],J)$ denote the moduli space defined through the conditions (1) and (2) above and $\calM(p,q;[u],J)$ the moduli space after taking the quotient by the $\R$-action given by reparametrization (i.e. $u(s,t)\mapsto u(s-a,t)$ for some $a\in\R$). We have denoted by $[u]$ the homotopy class of $u$ in $\pi_2(\Sigma,\calL_0\cap\calL_1)$. Moreover, define the \emph{Maslov index} of the homotopy class $[u]$ as 
\[
\mathrm{ind}([u]):=\mathrm{ind}_\R(D_{\bar\de_J,u})=\dim\ker D_{\bar\de_J,u}-\dim \mathrm{coker}\, D_{\bar\de_J,u}, 
\]
where $D_{\bar\de_J,u}$ denotes the linearization of $\bar\de_J$ at a given solution $u$. One can show that $D_{\bar \de_J,u}$ is indeed a Fredholm operator and thus one can compute its Fredholm index. In particular, it can be shown that $\widehat{\calM}(p,q;[u],J)$ is a smooth manifold of dimension $\mathrm{ind}([u])$ if all solutions of (1) and (2) are regular, i.e. the operator $D_{\bar\de_J,u}$ is surjective at each point $u\in\widehat{\calM}(p,q;[u],J)$. Thus, $\calM(p,q;[u],J)$ is an oriented 0-dimensional manifold whenever $\mathrm{ind}([u])=1$. 
Compactness of the moduli space is given through \emph{Gromov's compactness theorem} \cite{Gromov1985}. Indeed, Gromov showed that any sequence of $J$-holomorphic curves with uniformly bounded energy admits a subsequence which converges, up to reparametrization, to a nodal tree of $J$-holomorphic curves. 
Denote by $\mathrm{LGr}(n)$ the Grassmannian of Lagrangian $n$-planes in the standard symlectic space $(\R^{2n},\omega_0)$ and by $\mathrm{LGr}(T\Sigma)$ the Grassmannian of Lagrangian planes in $T\Sigma$ as an $\mathrm{LGr}(n)$-bundle over $\Sigma$.
The $\Z$-grading of the complex is obtained by ensuring that the difference of the index of a strip depends only on the difference between the degrees of the two generators that it connects and not on its homotopy class\footnote{This is actually obstructed by the following two conditions: The first Chern class of $\Sigma$ has to be 2-torsion and the \emph{Maslov class} $\mu_\calL\in \Hom(\pi_1(\calL),\Z)=H^1(\calL,\Z)$ of $\calL$ vanishes, i.e. $2c_1(T\Sigma)=0$, which allows to lift the Grassmannian $\mathrm{LGr}(T\Sigma)$ to a fiberwise universal cover $\widetilde{\mathrm{LGr}}(T\Sigma)$ given as the Grassmannian of \emph{graded} Lagrangian planes in $T\Sigma$. In particular, if we have a nowhere vanishing section $\sigma\in \Gamma(\bigwedge^n_\C T^*\Sigma\otimes \bigwedge^n_\C T^*\Sigma)$, the argument of $\sigma$ assigns to each Lagrangian plane $\ell$ a phase $\phi(\ell):=\arg(\sigma\vert_\ell)\in S^1=\R/2\pi\Z$. Thus, one defines a graded lift of $\ell$ to be the choice of a real lift $\tilde{\varphi}(\ell)\in\R$ of $\varphi(\ell)$.
The Maslov class is defined as the obstruction for choosing graded lifts of the tangent spaces to $\calL$, i.e. lifting the section of $\mathrm{LGr}(T\Sigma)$ over $\calL$ given by $p\mapsto T_p\calL$ to a section of the infinite cyclic cover $\widetilde{\mathrm{LGr}}(T\Sigma)$. The Lagrangian $\calL$ together with such a choice of lift is called \emph{graded Lagrangian submanifold} of $\Sigma$.}. Then we can construct the grading as follows: For all $p\in\calL_0\cap\calL_1$ we can obtain a homotopy class of a path connecting $T_p\calL_0$ and $T_p\calL_1$ in $\mathrm{LGr}(T_p\Sigma)$ by connecting the chosen graded lifts of the tangent spaces at $p$ through a path in $\widetilde{\mathrm{LGr}}(T_p\Sigma)$. Composing this path with the canonical short path from $T_p\calL_0$ and $T_p\calL_1$, denoted by $-\lambda_p$, we get a closed loop in $\mathrm{LGr}(T_p\Sigma)$. We can then define the degree of $p$ to be the Maslov index of this loop. One can check that for any strip $u$ connecting $p$ and $q$, we get 
\[
\mathrm{ind}(u)=\deg(q)-\deg(p),
\]
The differential can be then defined as 
\[
\de p=\sum_{q\in\calL_0\cap\calL_1\atop \mathrm{ind}([u])=1}\#\calM(p,q;[u],J) T^{\omega([u])}q,
\]
where $\#\calM(p,q;[u],J)\in\Z$ (or $\Z_2$) is the signed (or unsigned) count of pseudo-holomorphic strips connecting $p$ to $q$ in the class $[u]$, and $\omega([u])=\int u^*\omega<\infty$ is the symplectic area of these strips. One can observe that $\de$ is indeed of degree $+1$. The fact that it squares to zero is a non-trivial observation which requires some assumptions. Let us go back to the compactness argument. In our case, we consider $J$-holomorphic strips $u\colon\R\times[0,1]\to \Sigma$ with boundary on Lagrangian submanifolds $\calL_0$ and $\calL_1$. Then we have three situations that can appear: the first situation is called \emph{strip breaking} (see Figure \ref{fig:broken_strip}) and occurs if energy concentrates at either end $s\to \pm\infty$, i.e. there exists a sequence $(a_n)$ with $a_n\to \pm\infty$ such that the sequence of strips $u_n(s-a_n,t)$ converges to a non-constant limit strip. The second situation is called \emph{disk bubbling} (see Figure \ref{fig:bubbling}) and occurs if energy concentrates at a point on the boundary of the strip, i.e. when $t\in\{0,1\}$, where suitable rescalings of $u_n$ converge to a $J$-holomorphic disk in $\Sigma$ with boundary completely contained in either $\calL_0$ or $\calL_1$. the third situation occurs if energy concentrates at an interior point of the strip, where suitable rescalings of $u_n$ converge to a $J$-holomorphic sphere in $\Sigma$. 

\begin{figure}[h!]
    \centering
    \begin{tikzpicture}
    \tikzset{Bullet/.style={fill=blue,draw,color=#1,circle,minimum size=3pt,scale=0.5}}
    \draw[rounded corners=12pt,fill=gray!20](0,0)--(1,1)--(3,1)--(4,0);
    \draw[rounded corners=12pt, fill=gray!20](4,0)--(3,-1)--(1,-1)--(0,0);
    \draw[rounded corners=12pt,fill=gray!20](4,0)--(5,1)--(7,1)--(8,0);
    \draw[rounded corners=12pt, fill=gray!20](8,0)--(7,-1)--(5,-1)--(4,0);
    \draw (0,0) to (-1,1); 
    \draw (0,0) to (-1,-1); 
    \draw (8,0) to (9,1); 
    \draw (8,0) to (9,-1);
    \node[label=above:{$\calL_1$}] at (2,1){};
    \node[label=above:{$\calL_1$}] at (6,1){};
    \node[label=below:{$\calL_0$}] at (2,-1){};
    \node[label=below:{$\calL_0$}] at (6,-1){};
    \node[Bullet=black, label=above:{$q$}] at (0,0){};
    \node[Bullet=black, label=above:{$p$}] at (8,0){};
    \node[Bullet=black, label=above:{$r$}] at (4,0){};
    \end{tikzpicture}
    \caption{Example of a broken strip situation.}
    \label{fig:broken_strip}
\end{figure}
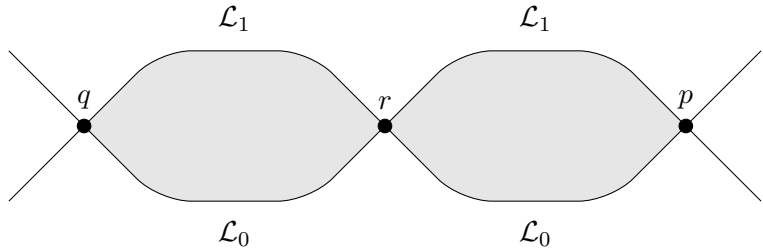

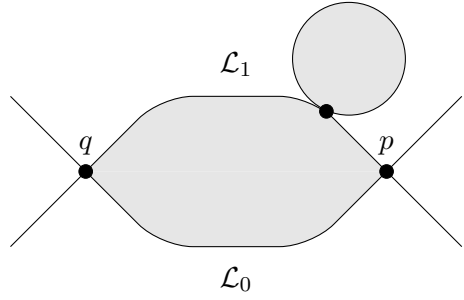
\begin{figure}[h!]
    \centering
    \begin{tikzpicture}
    \tikzset{Bullet/.style={fill=blue,draw,color=#1,circle,minimum size=3pt,scale=0.5}}
    \draw[rounded corners=12pt,fill=gray!20](0,0)--(1,1)--(3,1)--(4,0);
    \draw[rounded corners=12pt, fill=gray!20](4,0)--(3,-1)--(1,-1)--(0,0);
    \draw (0,0) to (-1,1); 
    \draw (0,0) to (-1,-1); 
    \draw (4,0) to (5,1); 
    \draw (4,0) to (5,-1);
    \node[label=above:{$\calL_1$}] at (2,1){};
    \node[label=below:{$\calL_0$}] at (2,-1){};
    \node[Bullet=black, label=above:{$q$}] at (0,0){};
    \node[Bullet=black, label=above:{$p$}] at (4,0){};
    \draw[fill=gray!20] (3.5,1.5) circle (.75);
    \node[Bullet=black] at (3.2,.8){};
    \end{tikzpicture}
    \caption{Example of a disk bubbling situation.}
    \label{fig:bubbling}
\end{figure}

Strip breaking is in fact the ingredient needed for the differential $\de$ to square to zero whenever disk bubbling can be excluded. We can make sure that disk and sphere bubbles do not appear by imposing the condition $[\omega]\cdot \pi_2(\Sigma,\calL_j)=0$ for $j\in\{0,1\}$. There are also other more general ways to avoid disk and sphere bubbling, e.g. to impose a lower bound on the Maslov index by considering the case when the symplectic area of disks and their Maslov index are proportional to each other. One then speaks of \emph{monotone} Lagrangian submanifolds in monotone symplectic manifolds.

\subsection{The Atiyah--Floer conjecture}
In \cite{Atiyah1987}, Atiyah conjectured that instanton Floer homology should be related to Lagrangian Floer homology in the following way (see also \cite{Salamon1995,Wehrheim2005_1,Wehrheim2005_2,SalamonWehrheim2008} for a slightly different approach to the same conjecture using Lagrangian boundary conditions):

\begin{Conj}[Atiyah--Floer\cite{Atiyah1987}]\label{conj:Atiyah-Floer}
Let $\Sigma_g$ be a Riemann surface of genus $g\geq 1$. Then the space of flat $\mathrm{SU}(3)$-connections on $\Sigma_g$, denoted by $\calA_\mathrm{flat}^{\mathrm{SU}(2)}(\Sigma_g)$, up to isomorphism has a symplectic structure. Suppose $N$ is an integral homology sphere with Heegaard splitting along the surface $\Sigma_g$, given by 
\[
N=H^1_g\cup_{\Sigma_g} H_g^2,
\]
where $H^i_g$ denotes a handle-body of genus $g$. Then the space $\calA_\mathrm{flat}^{\mathrm{SU}(2)}(H^i_g)$ of flat connections on $\Sigma_g$ which extend to flat connections on $H^i_g$ determines a Lagrangian subspace of $\calA_\mathrm{flat}^{\mathrm{SU}(2)}(\Sigma_g)$. In particular, we have 
\[
H_\mathrm{inst}F_\bullet(N)\cong H_\mathrm{Lagr}F_\bullet(\calA_\mathrm{flat}^{\mathrm{SU}(2)}(H^1_g),\calA_\mathrm{flat}^{\mathrm{SU}(2)}(H^2_g)),
\]
where $H_\mathrm{inst}F_\bullet$ denotes the instanton Floer homology and $H_\mathrm{Lagr}F_\bullet$ denotes the Lagrangian Floer homology. Here, $\calA_\mathrm{flat}^{\mathrm{SU}(2)}(H^1_g)$ and $\calA_\mathrm{flat}^{\mathrm{SU}(2)}(H^2_g)$ are considered as Lagrangian submanifolds of $\calA_\mathrm{flat}^{\mathrm{SU}(2)}(\Sigma_g)$. Usually, $H_\mathrm{Lagr}F_\bullet(\calA_\mathrm{flat}^{\mathrm{SU}(2)}(H^1_g),\calA_\mathrm{flat}^{\mathrm{SU}(2)}(H^2_g))$ is called the \emph{symplectic instanton Floer homology of $N$}.
\end{Conj}

\begin{rem}
The Atiyah--Floer conjecture leads to a better understanding of the instanton Floer homology for 3-manifolds with boundary. In particular, for a 4-manifold together with a principal $\mathbb{P}\mathrm{U}(2)$-bundle, where $\mathbb{P}\mathrm{U}$ denotes the \emph{projective unitary group}, one can define numerical invariants as in \eqref{eq:Donaldson_polynomials} due to Donaldson. For a 3-manifold together with a principal $\mathbb{P}\mathrm{U}(2)$-bundle, one can define the instanton Floer homology group as in Section \ref{subsec:Instanton_Floer_homology}. For a 2-manifold together with a principal $\mathbb{P}\mathrm{U}(2)$-bundle, we can construct an $A_\infty$-category\footnote{An $A_\infty$-category is, roughly speaking, a category whose associativity condition for the composition of morphisms is relaxed in a higher unbounded homotopical way (see e.g. \cite{KontsevichSoibelman2008}). Usually, the morphisms are given by chain complexes as for the meaning of a \emph{linear category}.}. Moreover, as we have seen in Section \ref{subsec:relation_to_Donaldson_polynomials}, the Donaldson invariants of a 4-manifold $\Sigma$ with boundary is valued in the instanton Floer homology of $\de\Sigma$. Finally, it is expected that the Floer homology of a 3-manifold $N$ with boundary gives an object in the $A_\infty$-category associated to $\de N$.
\end{rem}

\begin{rem}[A way of proving Conjecture \ref{conj:Atiyah-Floer}]
In \cite{DaemiFukaya2017}, Daemi and Fukaya have proposed a proof for the Atiyah--Floer conjecture. In particular, they construct a different version of Lagrangian Floer homology and translate the Atiyah--Floer conjecture to the equivalent conjecture which states that this new version is actually isomorphic to $H_\mathrm{inst}F_\bullet$. They proposed a solution to this equivalent conjecture by considering a mixture of the moduli space of anti self-dual connections and the moduli space of pseudo-holomorphic curves. 
Moreover, they gave a formulation in terms of $A_\infty$-categories in order to state conjectures for stronger versions of some properties for instanton Floer homology of 2- and 3-manifolds. 
\end{rem}

\begin{rem}[Fukaya category]
The $A_\infty$-category whose objects are given by certain Lagrangian submanifolds of a given symplectic manifold $(\Sigma,\omega)$ and morphisms between objects are given by the Lagrangian Floer chain complexes is called the \emph{Fukaya category} \cite{Fukaya1993} and is denoted by $\mathrm{Fuk}(\Sigma,\omega)$ (see also \cite{FukayaOhOhtaOno2009_1,FukayaOhOhtaOno2009_2,Kontsevich1994_2,Auroux2014}). In particular, one could also construct a more general version of this category by using Lagrangian foliations as the objects. Such a construction would be interesting in order to understand the boundary structure when combining with the methods of geometric quantization as for the bounary state space construction in the BV-BFV setting. 
\end{rem}

\begin{rem}[Homological mirror symmetry]
Based on the mirror symmetry conjecture ($A$- and $B$-model mirror construction) considered in string theory for Calabi--Yau 3-folds (see e.g. \cite{Yau1992,HZKKTVPV2003}), Kontsevich formulated a homological version in terms of equivalences of triangulated categories. For some projective variety $X$, let $D^b\mathrm{Coh}(X)$ denote the derived category of coherent sheaves on $X$. Kontsevich's conjecture is then formulated as follows:
\begin{Conj}[Homological mirror symmetry\cite{Kontsevich1994_2}]
\label{conj:HMS}
Let $X$ and $Y$ be mirror dual Calabi--Yau varieties. Then there is an equivalence of triangulated categories 
\begin{equation}
    \label{eq:HMS}
    \mathrm{Fuk}(X)\cong D^b\mathrm{Coh}(Y),
\end{equation}
\end{Conj}
Conjecture \ref{conj:HMS} has been proven by various people for different types of mirror varieties (e.g. toric 3-folds, quintic 3-folds). However, a general direct proof still remains open.  
\end{rem}

\section{The BV-BFV formalism}
\label{sec:BV-BFV_formalism}
In this section, we want to recall the most important concepts of \cite{CMR2}. Good introductory references for the BV and BV-BFV formalisms are \cite{Cost,Mnev2019,CattMosh1}. In order to avoid large formulae, we will not always write out the exterior product between differential forms, hence we remark that from now on this is automatically understood.

\subsection{BV formalism}
\label{subsec:BV_formalism}
Let us start with the BV construction on a closed $d$-dimensional source manifold $\Sigma$. 
\begin{defn}[BV manifold]
A \emph{BV manifold} is a triple $(\calF,\omega,\calS)$, where  $\calF$ is a $\Z$-graded supermanifold\footnote{Recall that a \emph{supermanifold} is a locally ringed space $(\calF,\calO_\calF)$ such that the structure sheaf is locally, for some open set $U\subset \R^d$, given by $C^\infty(U)\otimes \bigwedge^\bullet V^*$, where $V$ is some finite-dimensional real vector space.}, $\omega$ an odd symplectic form on $\calF$ of degree $-1$ and $\calS$ an even function on $\calF$ of degree 0 such that 
\begin{equation}
\label{eq:CME}
(\calS,\calS)=0,
\end{equation}
where $(\enspace,\enspace)$ denotes the odd Poisson bracket induced by the odd symplectic form $\omega$.
\end{defn}
Equation \eqref{eq:CME} is called \emph{Classical Master Equation (CME)}.
\begin{rem}
We call $\calF$ the \emph{BV space of fields}\footnote{Usually, the BV space of fields is given by the $(-1)$-shifted cotangent bundle of the BRST space of fields, i.e. $\calF_{\mathrm{BV}}=T^*[-1]\calF_{\mathrm{BRST}}$. The fiber fields are usually called \emph{anti-fields}. For each field $\Phi$, there is a corresponding notion of anti-field through duality, denoted by $\Phi^\dagger$. Recall that $\gh(\Phi)+\gh(\Phi^\dagger)=-1$. Moreover, if $\calF_\Sigma=\Omega^\bullet(\Sigma)$, then $\deg(\Phi)+\deg(\Phi^\dagger)=\dim\Sigma$.}, $\omega$ the \emph{BV symplectic form} and $\calS$ the \emph{BV action (functional)}. Moreover, $(\enspace,\enspace)$ is often called \emph{BV bracket or anti-bracket}. In the physics literature, the $\Z$-grading is called \emph{ghost number}. We will denote the ghost number by $\mathrm{gh}$. When considering differential forms, we will denote the form degree by $\deg$. 
\end{rem}

\begin{defn}[BV theory]
The assignment of a source manifold $\Sigma$ to a BV manifold 
\[
(\calF_\Sigma,\omega_\Sigma,\calS_\Sigma)
\]
is called a \emph{BV theory}.
\end{defn}

\subsubsection{Quantization} For the quantum picture one considers a canonical second-order differential operator $\Delta$ on half-densities on $\calF$ with the properties
\begin{align}
    \Delta^2&=0,\\
    \Delta(fg)&=\Delta f g\pm f\Delta g\pm (f,g),\quad \forall f,g\in \mathrm{Dens}^{\frac{1}{2}}(\calF),
\end{align}
where $\mathrm{Dens}^{\frac{1}{2}}(\calF)$ denotes the space of half-densities on $\calF$. The operator $\Delta$ is called \emph{BV Laplacian}. If $\Phi^i$ and $\Phi^\dagger_i$ denote field and anti-field respectively, we have 
\[
\Delta f=\sum_{i}(-1)^{\mathrm{gh}(\Phi^i)+1}f\left\langle \frac{\overleftarrow{\delta}}{\delta \Phi^i},\frac{\overleftarrow{\delta}}{\delta \Phi^\dagger_i}\right\rangle, \quad f\in \mathrm{Dens}^{\frac{1}{2}}(\calF).
\]
We denote by $\frac{\overleftarrow{\delta}}{\delta\Phi}$ and $\frac{\overrightarrow{\delta}}{\delta\Phi}$ the \emph{left-} and \emph{right-derivatives}. Namely, we have 
\begin{align*}
    \frac{\overrightarrow{\delta}}{\delta\Phi^i}f&=(-1)^{\gh(\Phi^i)(\gh(f)+1)}f\frac{\overleftarrow{\delta}}{\delta\Phi^i},\\
    \frac{\overrightarrow{\delta}}{\delta\Phi^\dagger_i}f&=(-1)^{(\gh(\Phi^i)+1)(\gh(f)+1)}f\frac{\overleftarrow{\delta}}{\delta\Phi^\dagger_i}.
\end{align*}
\begin{rem}
In particular, one can show that for any odd symplectic supermanifold $\calF$, there is a supermanifold $\calN$ such that $\calF\cong T^*[1]\calN$. Hence, functions on $\calF$ are given by multivector fields on $\calN$. Moreover, the Berezinian bundle on $\calF$ is given by 
\[
\mathrm{Ber}(\calF)\cong \bigwedge^\mathrm{top}T^*\calN\otimes \bigwedge^\mathrm{top}T^*\calN.
\]
Thus, half-densities on $\calF$ are defined by 
\[
\mathrm{Dens}^\frac{1}{2}(\calF):=\Gamma\left(\mathrm{Ber}(\calF)^\frac{1}{2}\right).
\]
One can show that there exists a canonical operator $\Delta^\frac{1}{2}_\calF$ on $\mathrm{Dens}^\frac{1}{2}(\calF)$ which squares to zero \cite{Khudaverdian2004}. Then, fixing a non-vanishing, $\Delta^\frac{1}{2}_\calF$-closed reference half-density $\sigma\in\mathrm{Dens}^\frac{1}{2}(\calF)$, one can define a Laplacian on functions on $\calF$ by 
\begin{equation}
\label{eq:twisting_BV_Laplacian}
\Delta_\sigma f:=\frac{1}{\sigma}\Delta^\frac{1}{2}_\calF(f\sigma).
\end{equation}
It is easy to check that $\Delta_\sigma$ squares to zero. For convenience, we write $\Delta$ when we actually mean $\Delta_\sigma$ in order to not emphasize on the choice of reference half-density. It is important to mention that this construction needs a suitable regularization in the case where $\calF$ is infinite-dimensional. See also \cite{Severa2006} for describing the BV Laplacian naturally through a spectral sequence approach of a double complex using the odd symplectic structure.
\end{rem}

\begin{thm}[Batalin--Vilkovisky--Schwarz\cite{BV2,S}]
\label{thm:BV}
Let $f,g\in \mathrm{Dens}^{\frac{1}{2}}(\calF)$ be two half-densities on $\calF$. Then 
\begin{enumerate}
    \item if $f=\Delta g$ (\emph{BV exact}), we get that 
    \[
    \int_\calL f=0,
    \]
    for any Lagrangian submanifold $\calL\subset \calF$. 
    \item if $\Delta f=0$ (\emph{BV closed}), we get that 
    \[
    \frac{\dd}{\dd t}\int_{\calL_t}f=0,
    \]
    for any smooth family $(\calL_t)$ of Lagrangian submanifolds of $\calF$.
\end{enumerate}
\end{thm}

For the application to quantum field theory, we want to consider\footnote{More precisely, we want to consider $f=\exp(\I\calS/\hbar)\rho$, where $\rho$ is some $\Delta$-closed reference half-density on $\calF$.}
\[
f=\exp(\I\calS/\hbar).
\]
The choice of Lagrangian submanifold is in fact equivalent to the choice of gauge-fixing. Thus, the second point of Theorem \ref{thm:BV} tells us that the the gauge-independence condition is encoded in the equation
\[
\Delta\exp(\I\calS/\hbar)=0,
\]
which is equivalent to 
\begin{equation}
\label{eq:QME}
(\calS,\calS)-2\I\hbar\Delta\calS=0.
\end{equation}
Equation \eqref{eq:QME} is called \emph{Quantum Master Equation (QME)}. Note that the QME reduces to the CME for $\hbar\to 0$.

\begin{ex}[AKSZ theory\cite{AKSZ}]
\label{ex:AKSZ}
A natural example of BV theories was formulated by Alexandrov, Kontsevich, Schwarz and Zaboronsky in \cite{AKSZ}. A \emph{differential graded symplectic manifold} of degree $d$ consists of a graded manifold $\calM$ endowed with an exact symplectic form $\omega=\dd_\calM \alpha$ of degree $d$ and a smooth Hamiltonian function $\Theta$ on $\calM$ of degree $d+1$ satisfying 
\[
\{\Theta,\Theta\}=0,
\]
where $\{\enspace,\enspace\}$ is the Poisson bracket induced by $\omega$. Such a triple $(\calM,\omega=\dd_\calM\alpha,\Theta)$ is sometimes also called a \emph{Hamiltonian manifold}.
A $d$-dimensional \emph{AKSZ sigma model} with target a Hamiltonian manifold $(\calM,\omega=\dd_\calM\alpha,\Theta)$ of degree $d-1$ is the BV theory, which associates to a $d$-manifold $\Sigma$ the BV manifold $(\calF_\Sigma,\omega_\Sigma,\calS_\Sigma)$, where\footnote{This is the infinite-dimensional graded manifold adjoint to the Cartesian product (internal morphisms).} 
\begin{equation}
    \calF_\Sigma=\Map(T[1]\Sigma,\calM),  
\end{equation}
the BV symplectic form $\omega_\Sigma$ is of the form 
\begin{equation}
    \omega_\Sigma=\int_{T[1]\Sigma}\boldsymbol{\mu}_d\omega_{\mu\nu}\delta \Phi^\mu\delta\Phi^\nu,
\end{equation}
and the BV action is given by 
\begin{equation}
    \calS_\Sigma=\int_{T[1]\Sigma}\boldsymbol{\mu}_d\left(\alpha_\mu(\Phi)\dd_\Sigma \Phi^\mu+\Theta(\Phi)\right).
\end{equation}
we have denoted by $\Phi\in\calF_\Sigma$ the superfields, by $\omega_{\mu\nu}$ the components of the symplectic form $\omega$, by $\alpha_\mu$ the components of $\alpha$ and $\Phi^\mu$ are the components of $\Phi$ in local coordinates $(u,\theta)$. Moreover, $\boldsymbol{\mu}_d:=\dd^du\dd^d\theta$ denotes the measure on $T[1]\Sigma$ and $\delta$ is the de Rham differential on $\calF_\Sigma$. Indeed, one can check that $\omega_\Sigma$ is symplectic of degree $-1$ and that $\calS_\Sigma$ is of degree 0 satisfying the CME
\[
(\calS_\Sigma,\calS_\Sigma)=0.
\]
This is obtained by considering the \emph{transgression} map 
\[
\mathscr{T}\colon \Omega^\bullet(\calM)\to\Omega^\bullet(\Map(T[1]\Sigma,\calM)), 
\]
given by the following diagram:
\[
\begin{tikzcd}
  \Map(T[1]\Sigma,\calM)\times T[1]\Sigma \arrow[r,"\mathrm{ev}"] \arrow[d,swap,"\pi_1"]
      & \calM\arrow[dl,dashed]\\
\Map(T[1]\Sigma,\calM) & \end{tikzcd}
\]
where $\pi_1$ denotes the projection to the first factor and $\mathrm{ev}$ denotes the evaluation map. The transgression map is then defined as $\mathscr{T}:=(\pi_1)_*\mathrm{ev}^*$. Note that the map $(\pi_1)_*$ denotes fiber integration. The $(-1)$-shifted symplectic form on $\Map(T[1]\Sigma,\calM)$ is then given by $\omega_\Sigma=\mathscr{T}(\omega)=(\pi_1)_*\mathrm{ev}^*\omega$.
Many theories of interest are in fact of AKSZ-type such as e.g. Chern--Simons theory \cite{Witten1989,AS,AS2}, the Poisson sigma model \cite{CF1,CF4}, Witten's $A$- and $B$-model \cite{Witten1988a,AKSZ} and $2$-dimensional Yang--Mills theory \cite{IM}.
\end{ex}

\subsection{BV algebras}
\label{subsec:BV_algebras}
We want to recall some notions on BV algebras as in \cite{Getzler1994}, and how it is related to the original BV formalism. A \emph{braid algebra} $\mathcal{B}r$ is a commutative DG algebra endowed with a Lie bracket $[\enspace,\enspace]$ of degree $+1$ satisfying the Poisson relations
\begin{equation}
    [a,bc]=[a,b]c+(-1)^{\vert a\vert(\vert b\vert -1)}b[a,c],\qquad \forall a,b,c\in \mathcal{B}r
\end{equation}
An identity element in $\mathcal{B}r$ is an element $\boldsymbol{1}$ of degree 0 such that it is an identity for the product and $[\boldsymbol{1},\enspace]=0$.
A \emph{BV algebra} $\mathcal{BV}$ is a commutative DG algebra endowed with an operator $\Delta\colon \mathcal{BV}_\bullet\to \mathcal{BV}_{\bullet+1}$ such that $\Delta^2=0$ and 
\begin{align}
\label{app:eq:BV_Laplacian_2}
\begin{split}
    \Delta(abc)&=\Delta(ab)c+(-1)^{\vert a\vert}a\Delta(bc)+(-1)^{(\vert a\vert-1)\vert b\vert}b\Delta(ac)\\
    &-\Delta abc-(-1)^{\vert a\vert}a\Delta bc-(-1)^{\vert a\vert+\vert b\vert}ab\Delta c,\qquad \forall a,b,c\in \mathcal{BV}.
\end{split}
\end{align}
An identity in $\mathcal{BV}$ is an element $\boldsymbol{1}$ of degree 0 such that it is an identity for the product and $\Delta\boldsymbol{1}=0$. One can show that a BV algebra is in fact a special type of a braid algebra. More precisely, a BV algebra is a braid algebra endowed with an operator $\Delta\colon \mathcal{BV}_\bullet\to \mathcal{BV}_{\bullet+1}$ such that $\Delta^2=0$ and such that the bracket and $\Delta$ are related by
\begin{equation}
    [a,b]=(-1)^{\vert a\vert}\Delta(ab)-(-1)^{\vert a\vert}\Delta ab-a\Delta b,\qquad\forall a,b\in \mathcal{BV}.
\end{equation}
Moreover, in a BV algebra we have 
\begin{equation}
    \Delta([a,b])=[\Delta a,b]+(-1)^{\vert a\vert-1}[a,\Delta b],\qquad \forall a,b\in \mathcal{BV}.
\end{equation}
The motivation for such an algebra structure comes exactly from the approach of the BV formalism in quantum field theory. Let $(\calF,\omega)$ be an odd symplectic (super)manifold.
Let $f\in \calO_\calF$, where $\calO_\calF$ denotes the space of functions on $\calF$, and consider its Hamiltonian vector field $X_f$. One can check that $\calO_\calF$ endowed with the BV anti-bracket
\begin{equation}
\label{eq:odd_Poisson_bracket}
    (f,g):=(-1)^{\vert f\vert -1}X_f(g)
\end{equation}
is a braid algebra. Let $\mu\in \Gamma(\mathrm{Ber}(\calF))$ be a nowhere-vanishing section of the Berezinian bundle of $\calF$. As we have seen before, this represents a density which is characterized by the integration map $\int\colon \Gamma_0(\calF,\mathrm{Ber}(\calF))\to \R$, where $\Gamma_0$ denotes sections with compact support. Hence, $\mu$ induces an integration map on functions with compact support
\begin{equation}
    \calO_\calF\ni f\mapsto \int_{\calL\subset\calF} f\mu^{1/2},
\end{equation}
for some Lagrangian submanifold $\calL\subset \calF$, where the integral exists.
Then one can define a divergence operator $\Div_\mu X$ by 
\begin{equation}
    \int_\calF (\Div_\mu X)f\mu=-\int_\calF X(f)\mu.
\end{equation}

\begin{lem}
For a vector field $X$, define $X^*:=-X-\Div_\mu X$. Then 
\begin{equation}
    \int_\calF fX(g)\mu=(-1)^{\vert f\vert\vert X\vert}\int_\calF X^*(f)g\mu.
\end{equation}
Moreover, $\Div_\mu(fX)=f\Div_\mu X-(-1)^{\vert f\vert \vert X\vert}X(f)$ and if $\calS\in \calO_\calF$ is an even function, then
\[
\Div_{\exp(\calS)\mu}X=\Div_\mu X+X(\calS).
\]
\end{lem}
One can then define the BV Laplacian $\Delta$ to be the odd operator on $\calO_\calF$ given by
\begin{equation}
\label{eq:Delta_div}
    \Delta f=\Div_\mu X_f.
\end{equation}
A BV (super)manifold $(\calF,\omega,\mu)$ is then an odd symplectic (super)manifold with Berezinian $\mu$ such that $\Delta^2=0$.
\begin{prop}[Getzler\cite{Getzler1994}]
\label{prop:Getzler}
Let $(\calF,\omega,\mu)$ be a BV (super)manifold.
\begin{enumerate}
    \item The algebra $(\calO_\calF,(\enspace,\enspace),\Delta)$ is a BV algebra, where $\Delta$ is given as in \eqref{eq:Delta_div} and $(\enspace,\enspace)$ is the odd Poisson bracket coming from the odd symplectic form $\omega$ as in \eqref{eq:odd_Poisson_bracket}.
    \item The Hamiltonian vector field associated to some $f\in \calO_\calF$ is given by the formula $X_f=-[\Delta,f]+\Delta f$, where $[\enspace,\enspace]$ denotes the commutator of operators.
    \item If $\calS\in \calO_\calF$ and $\Delta_\calS$ is the operator associated to the Berezinian $\exp(\calS)\mu$, then $\Delta_\calS=\Delta-X_\calS$ and $\Delta_\calS^2=X_{\Delta \calS+\frac{1}{2}(\calS,\calS)}$.
\end{enumerate}
\end{prop}
Note that point (3) of Proposition \ref{prop:Getzler} is exactly the case that we have in quantum field theory. Moreover, if 
\begin{equation}
\label{eq:QME_BV_alg}
    \Delta \calS+\frac{1}{2}(\calS,\calS)=0,
\end{equation}
we get that $\Delta_\calS^2=0$, which ensures a BV algebra structure. Note that here we have set $\I\hbar=1$.

\subsection{BFV formalism}
\label{subsec:BFV_formalism}
We want to continue with the Hamiltonian approach of the BFV formalism for closed source manifolds. 
\begin{defn}[BFV manifold]
A \emph{BFV manifold} is a triple $(\calF^\de,\omega^\de,Q^\de)$ such that $\calF^\de$ is a $\Z$-graded supermanifold, $\omega^\de$ is an even symplectic form of degree 0 and $Q^\de$ is a cohomological ($Q^2=0$) and symplectic ($L_Q\omega=0$) vector field on $\calF$ of degree $+1$. Moreover, the Hamiltonian function $\calS^\de$ of $Q^\de$ defined by the equation $\iota_{Q^\de}\omega^\de=\delta\calS^\de$, is required to satisfy the CME
\begin{equation}
\label{eq:CME_2}
\{\calS^\de,\calS^\de\}=0,
\end{equation}
where $\{\enspace,\enspace\}$ denotes the Poisson bracket of degree 0 induced by the even symplectic form $\omega^\de$.
\end{defn}

Note that we have denoted by $\delta$ the de Rham differential on $\calF^\de$. Moreover, we can, equivalently, express Equation \eqref{eq:CME_2} as $Q^\de\calS^\de=0$. Similarly, for a BV manifold $(\calF,\omega,\calS)$ we can consider the Hamiltonian vector field $Q$ associated to $\calS$ through the equation $\iota_Q\omega=\delta\calS$ (by abuse of notation, we also denote the de Rham differential on $\calF$ by $\delta$). Hence, we can express Equation \eqref{eq:CME} as $Q\calS=0$.

\begin{defn}[Exact BFV manifold]
We call a BFV manifold $(\calF^\de,\omega^\de,Q^\de)$ \emph{exact}, if $\omega^\de$ is exact, i.e. there is a 1-form $\alpha$ on $\calF^\de$ such that $\omega^\de=\delta\alpha^\de$.
\end{defn}

\begin{defn}[BFV theory]
The assignment of a manifold $\Sigma$ to a BFV manifold 
\[
(\calF^\de_{\de\Sigma},\omega^\de_{\de\Sigma},Q^\de_{\de\Sigma})
\]
is called a \emph{BFV theory}.
\end{defn}

\begin{rem}
Given an AKSZ theory $(\calF_\Sigma=\Map(T[1]\Sigma,\calM),\omega_\Sigma,\calS_\Sigma)$ as in Example \ref{ex:AKSZ} associated to a manifold with boundary $\Sigma$, we can easily construct its BFV data by restriction to the boundary. The BFV action thus given by 
\[
\calS^\de_{\de\Sigma}=\calS_\Sigma\big|_{\de\Sigma}.
\]
Similarly we can obtain the BFV space of boundary fields 
\[
\calF^\de_{\de\Sigma}=\Map(T[1]\de\Sigma,\calM)=\calF_{\Sigma}\big|_{\de\Sigma}
\]
and the BFV symplectic form 
\[
\omega^\de_{\de\Sigma}=\omega_\Sigma\big|_{\de\Sigma}.
\]
\end{rem}

\subsection{BV-BFV formalism}
\label{subsec:BV-BFV_formalism}
Consider now a source manifold $\Sigma$ with boundary $\de\Sigma$. 
\begin{defn}[BV-BFV manifold]
A \emph{BV-BFV manifold} over an exact BFV manifold $(\calF^\de,\omega^\de=\delta\alpha^\de,Q^\de)$ is a quintuple $(\calF,\omega,\calS,Q,\pi)$ such that $\calF$ is a $\Z$-graded supermanifold, $\omega$ is an odd symplectic form of degree $-1$ on $\calF$, $\calS$ is an even functional of degree 0, $Q$ is a cohomological, symplectic vector field on $\calF$ and $\pi\colon \calF\to \calF^\de$ a surjective submersion such that 
\begin{enumerate}
    \item $\delta\pi Q=Q^\de$,
    \item $\iota_Q\omega=\delta\calS+\pi^*\alpha^\de$,
\end{enumerate}
where $\delta\pi$ denotes the tangent map of $\pi$. In fact, condition (1) and (2) together imply the \emph{modified Classical Master Equation (mCME)}
\begin{equation}
    \label{eq:mCME}
    Q\calS=\pi^*(2\calS^\de-\iota_{Q^\de}\alpha^\de).
\end{equation}
\end{defn}

\begin{defn}[BV-BFV theory]
The assignment of a source manifold with boundary $\Sigma$ to a BV-BFV manifold $(\calF_\Sigma,\omega_\Sigma,\calS_\Sigma,Q_\Sigma,\pi_\Sigma)$ over an exact BFV manifold $(\calF^\de_{\de\Sigma},\omega^\de_{\de\Sigma},Q^\de_{\de\Sigma})$ is called a \emph{BV-BFV theory}.
\end{defn}

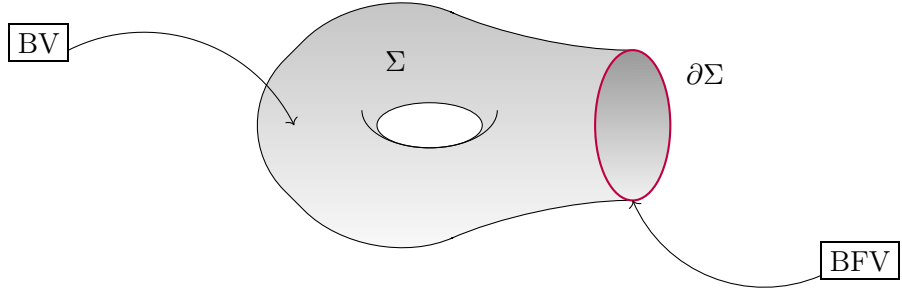
\begin{figure}[h!]
    \centering
    \begin{tikzpicture}
    \shadedraw[rounded corners=35pt,top color=gray!50, bottom color=gray!.5](5.5,-1)--(4.2,-1)--(2,-2)--(0,0) -- (2,2)--(4.2,1)--(5.5,1);
    \draw[fill=white] (3.5,0) arc (0:360:0.7cm and 0.3cm);
    \draw (2.8,-0.3) arc (-90:0:0.9cm and 0.5cm);
    \draw (2.8,-0.3) arc (-90:-180:0.9cm and 0.5cm);
    \shadedraw[color=white, top color=gray!80, bottom color=gray!10.5] (6,0) arc (0:360:0.5cm and 1cm);
    \draw[thick, color=purple] (6,0) arc (0:360:0.5cm and 1cm);
    \node (a) at (20:2.5) {$\Sigma$};
    \node (a) at (6:6.5) {$\de\Sigma$};
    \draw[->, bend angle=45, bend left] (8,-2) to (5.5,-1);
    \node (a) at (8.52,-1.8) {\boxed{\textnormal{BFV}}};
    \draw[->, bend angle=45, bend left] (-2,1) to (1,0);
    \node (a) at (-2.4,1.1) {\boxed{\textnormal{BV}}};
    \end{tikzpicture}
    \caption{The BV theory is associated to the bulk and the BFV theory to the boundary such that we have a coherent coupling.}
    \label{fig:gluing_2}
\end{figure}

\subsubsection{Quantization} For the quantization, one chooses a polarization\footnote{This is an involutive $\omega^\de$-Lagrangian subbundle of $T\calF^\de$.} $\calP$ on the boundary and assume a symplectic splitting\footnote{Actually, this is only needed locally, so one can also allow $\calF$ to be a BV bundle over $\calB^\calP$. A fiber bundle $\calF$ over a base $\calB$ with odd-symplectic fiber $(\calY,\omega)$ is called a \emph{BV bundle} if the transition maps of $\calF$ are given by locally constant fiberwise symplectomorphisms. An even more general setting is given by a \emph{hedgehog fibration} where the fibers of the fiber bundle are given by \emph{hedgehogs} \cite{CMR2}.} of the BV space of fields with respect to it 
\begin{equation}
    \label{eq:splitting1}
    \calF=\calB^\calP\times \calY.
\end{equation}
Here $\calB^\calP$ denotes the leaf space for the polarization $\calP$ which we assume to be smooth and $\calY$ is a complement. Note that we can always split the BV space of fields in such a way but it is important to assume that it is symplectic, i.e. that the BV symplectic form $\omega$ is constant on $\calB^\calP$. In fact, we consider $\omega$ to be a weakly nondegenerate\footnote{In the finite-dimensional case the BV-BFV formalism is not consistent with the nondegeneracy of $\omega$ on the whole space and thus one has to exactly assume nondegeneracy along the fibers.} 2-form on $\calY$ which extends to the product $\calB^\calP\times\calY$.
For $BF$-like theories\footnote{These are perturbations of $BF$ theories.}, which includes AKSZ-type theories, one can think of $\calB^\calP$ to be the fields restricted to the boundary and $\calY$ to be the bulk part. Using the fact that the BFV space of boundary fields $(\calF^\de,\omega^\de)$ is a symplectic manifold, we consider a \emph{geometric quantization} (see e.g. \cite{Wood97,BatesWeinstein2012,Moshayedi2020_geomquant}) on the boundary to obtain a vector space $\calH^\calP$. In order to take care of regularization, we consider another splitting of the bulk part 
\begin{equation}
    \label{eq:splitting2}
    \calY=\calV\times \calY',
\end{equation}
where $\calV$ denotes the space of \emph{residual fields}\footnote{$\calV$ is often also called the space of \emph{low energy fields} or \emph{zero modes} which we assume to be finite-dimensional. The complement $\calY'$ is accordingly often called space of \emph{high energy fields} or \emph{fluctuation fields}}. Moreover, let $\mathcal{Z}:=\calB^\calP\times \calV$ denote the bundle of residual fields over $\calB^\calP$. Note that we also assume a splitting of the symplectic manifold $(\calY,\omega)$ into two symplectic manifolds 
\begin{align*}
    \calY&=\calV\times \calY',\\
    \omega&=\omega_\calV+\omega_{\calY'}.
\end{align*}
Gauge-fixing is then equivalent to the choice of a Lagrangian submanifold $\calL$ of $(\calY',\omega_{\calY'})$. Define $\Hat{\calH}^\calP:=\mathrm{Dens}^\frac{1}{2}(\mathcal{Z})=\mathrm{Dens}^\frac{1}{2}(\calB^\calP)\widehat{\otimes}\mathrm{Dens}^{\frac{1}{2}}(\calV)$ and the BV Laplacian $\Hat{\Delta}^\calP:=\id\otimes \Delta_{\calV}$ on $\Hat{\calH}^\calP$. We have denoted by $\Delta_\calV$ the BV Laplacian on $\calV$.
Note that the space $\mathrm{Dens}^\frac{1}{2}(\calB^\calP)$ coincides with the vector space $\calH^\calP$ constructed by geometric quantization on the boundary. Thus we have 
\[
\Hat{\calH}^\calP=\calH^\calP\widehat{\otimes}\mathrm{Dens}^\frac{1}{2}(\calV).
\]
\begin{rem}
\label{rem:assumption}
Another important assumption is that for any $\Phi\in \mathcal{Z}$, the restriction of the action to $\{\Phi\}\times \calL$ has isolated critical points on $\{\Phi\}\times\calL$.
\end{rem}
The state is then defined through the perturbative expansion into Feynman graphs of the BV integral \begin{equation}
    \label{eq:state}
    \mathsf{Z}^{\scriptscriptstyle\mathrm{BV-BFV}}(\Phi)=\int_\calL \exp\left(\I\calS(\Phi)/\hbar\right)\mathscr{D}[\Phi]\in\widehat{\calH}^\calP,\quad \Phi\in\mathcal{Z}.
\end{equation}

\begin{rem}
If one uses a different choice of residual fields $\calV'$ such that $\calV$ fibers over $\calV'$ as a \emph{hedgehog} (see \cite{CMR2} for the definition of a hedgehog fibration), then the corresponding quantum theories are \emph{BV equivalent} in the sense of \cite{CMR2}. however, there is a minimal choice for the residual fields for which the assumption in Remark \ref{rem:assumption} is satisfied by a good choice of Lagrangian submanifold $\calL$. For the case of abelian $BF$ theory, the minimal space of residual fields is given by the de Rham cohomology of the underlying source manifold.
\end{rem}

If $\Delta\calS=0$, which we usually want to assume, the QME is replaced by the \emph{modified} QME (mQME) 
\begin{equation}
\label{eq:mQME_1}
\big(\hbar^2\Delta+\Omega^\calP\big)\exp\left(\I\calS/\hbar\right)=0,
\end{equation}
where, for local coordinates $(b_i)\in\calB^\calP$, we have that 
\[
\Omega^\calP:=\calS^\de\left(b_i,-\I\hbar\frac{\delta}{\delta b_i}\right)
\]
is the standard ordering quantization of $\calS^\de$. 
If $\calS$ depends on $\hbar$ and (or) $\Delta\calS\not=0$, we get the mQME from the assumption of the QME in the bulk by defining $\calS^\de_\hbar:=\calS^\de+O(\hbar)$ via the equation
\[
\pi^*\calS^\de_\hbar=\frac{1}{2}(\calS,\calS)-\I\hbar\Delta\calS
\]
and setting $\Omega^\calP$ to be the standard ordering quantization of $\calS^\de_\hbar$.
Note that the operator $\hbar^2\Delta+\Omega^\calP$ is of order $+1$. 
Moreover, we assume that the operator $\hbar^2\Delta+\Omega^\calP$ squares to zero in order to have a well-defined BV cohomology. In the finite-dimensional setting we automatically get that $\Hat{\mathsf{Z}}$ solves the mQME: 
\begin{equation}
\label{eq:mQME_2}
    \big(\hbar^2\Delta_\calV+\Omega^\calP\big)\Hat{\mathsf{Z}}=0.
\end{equation}
In the infinite-dimensional setting, where we have to compute the Feynman graphs instead of an integral, the mQME is only expected to hold and requires a separate checking. 

\subsection{Gluing of BV-BFV partition functions}
Let $\Sigma$ be a $d$-manifold and consider two $d$-manifolds $\Sigma_1,\Sigma_2$ such that $\Sigma=\Sigma_1\cup_{N}\Sigma_2$ where $N$ is the $(d-1)$-manifold which is identified with the common boundaries $\de\Sigma_1$ and $\de\Sigma_2$ (see Figure \ref{fig:gluing_1}). Then in general, following Atiyah's TQFT axioms \cite{Atiyah1988}, we can define the glued partition function by the pairing
\begin{equation}
\label{eq:gluing_general}
\mathsf{Z}_{\Sigma}=\int_N \mathsf{Z}_{\Sigma_1,\de\Sigma_1}\mathsf{Z}_{\Sigma_2,\de\Sigma_2}.
\end{equation}

\begin{figure}[h!]
    \centering
    \begin{tikzpicture}
    \shadedraw[rounded corners=35pt,top color=gray!50, bottom color=gray!.5](6,-1)--(4.2,-1)--(2,-2)--(0,0)--(2,2)--(4.2,1)--(7,1)--(9.2,2)--(11,0)
    --(9,-2)--(6,-1);
    \draw[fill=white] (3.5,0) arc (0:360:0.7cm and 0.3cm);
    \draw (2.8,-0.3) arc (-90:0:0.9cm and 0.5cm);
    \draw (2.8,-0.3) arc (-90:-180:0.9cm and 0.5cm);
    \shadedraw[dashed, color=white, top color=gray!100, bottom color=gray!20, opacity=.3] (5.8,0) arc (0:360:0.5cm and 1cm);
    \draw[dashed, color=purple] (5.8,0) arc (0:360:0.5cm and 1cm);
    \draw[fill=white] (8,-0.28) arc (0:360:0.7cm and 0.3cm);
    \draw (7.3,-0.58) arc (-90:0:0.9cm and 0.5cm);
    \draw (7.3,-0.58) arc (-90:-180:0.9cm and 0.5cm);
    \node (a) at (-13:5.8) {$N$};
    \node (a) at (26:2.5) {$\Sigma_1$};
    \node (a) at (8.5,1) {$\Sigma_2$};
    \end{tikzpicture}
    \caption{Gluing of $\Sigma_1$ and $\Sigma_2$ along $N$.}
    \label{fig:gluing_1}
\end{figure}
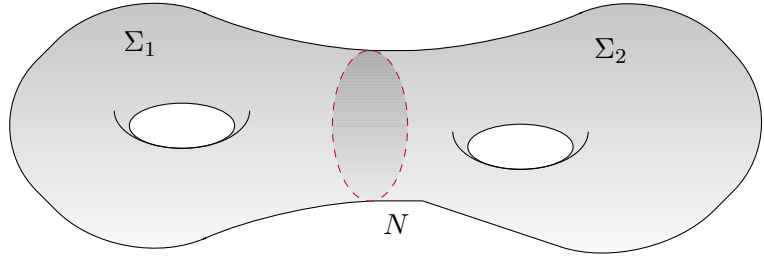

We want to see how this is adapted to the BV-BFV partition function.
Let $(\calY,\omega)$ be a direct product of two odd symplectic manifolds $(\calY',\omega')$ and $(\calY'',\omega'')$, i.e. $\calY=\calY'\times\calY''$, $\omega=\omega'+\omega''$. Then the space of half-densities on $\calY$ factorizes as 
\[
\mathrm{Dens}^\frac{1}{2}(\calY)=\mathrm{Dens}^\frac{1}{2}(\calY')\widehat{\otimes}\mathrm{Dens}^\frac{1}{2}(\calY''). 
\]
If we consider BV integration on the second factor for some Lagrangian submanifold $\calL\subset \calY''$, we can define a pushforward map on half-densities
\[
\int_\calL\colon \mathrm{Dens}^\frac{1}{2}(\calY)\xrightarrow{\id\otimes \int_\calL}\mathrm{Dens}^\frac{1}{2}(\calY').
\]
If we consider the space of fields $\calF=\calB\times\calY$ given by a product as in Section \ref{subsec:BV-BFV_formalism}, and assuming that $\calY$ splits into a product of two odd symplectic manifolds $(\calV,\omega_\calV)$ and $(\calY',\omega_{\calY'})$, we can consider a version of the BV pushforward in terms of families over $\calB$ by 
\[
\int_\calL\colon \mathrm{Dens}^\frac{1}{2}(\calF)\xrightarrow{\id\otimes \int_\calL}\mathrm{Dens}^\frac{1}{2}(\mathcal{Z}),
\]
where $\mathcal{Z}=\calB^\calP\times\calV$ for some chosen polarization $\calP$ and $\calL\subset \calY'$ is a Lagrangian submanifold. In particular, we have the coboundary operator $\hbar^2\Delta_\calV+\Omega^\calP$ on $\mathcal{Z}$, where $\Delta_\calV$ is the canonical BV Laplacian on $\calV$. In fact, the BV-BFV partition function is then defined through the family version of the BV pushforward. 

The gluing of BV-BFV partition functions as in \eqref{eq:state} is given by 
\begin{equation}
\label{eq:BV-BFV_gluing}
\mathsf{Z}^{\scriptscriptstyle\mathrm{BV}}_{\Sigma}=\int_{\mathcal{L}} \left\langle\mathsf{Z}^{\scriptscriptstyle\mathrm{BV-BFV}}_{\Sigma_1,\de\Sigma_1},\mathsf{Z}^{\scriptscriptstyle\mathrm{BV-BFV}}_{\Sigma_2,\de\Sigma_2}\right\rangle_N\in\mathrm{Dens}^\frac{1}{2}(\calV_\Sigma),
\end{equation}
where $\int_\mathcal{L}$ denotes the \emph{BV pushforward} with respect to the map 
\[
\calV_{\Sigma_1}\times\calV_{\Sigma_2}\to \calV_\Sigma, 
\]
for some Lagrangian submanifold $\calL$ in the bulk complement of $\calV_\Sigma$ and $\langle\enspace,\enspace\rangle_N$ is the pairing in $\calH^\calP_N$. 
Note that the product of the space of residual fields $\calV_{\Sigma_1}\times\calV_{\Sigma_2}$ is a \emph{hedgehog fibration} as in \cite{CMR2}. The BV-BFV partition function depends only on residual fields when $\Sigma$ is closed. If we assume that the space of residual fields $\calV_\Sigma$ is finite-dimensional, which is the case for many theories of interest including AKSZ theories, we can compute the number valued partition function by integrating out the residual fields $\int_{\calV_\Sigma}\mathsf{Z}^{\scriptscriptstyle\mathrm{BV}}_{\Sigma}\in\C$. 
If the glued manifold $\Sigma$ has boundary itself (see Figure \ref{fig:gluing_2}), the BV-BFV partition function depends on the boundary fields $\calB^\calP_{\de\Sigma}$. In particular, we get 
\[
\mathsf{Z}^{\scriptscriptstyle\mathrm{BV-BFV}}_{\Sigma,\de\Sigma}=\int_{\mathcal{L}} \left\langle\mathsf{Z}^{\scriptscriptstyle\mathrm{BV-BFV}}_{\Sigma_1,\de\Sigma_1},\mathsf{Z}^{\scriptscriptstyle\mathrm{BV-BFV}}_{\Sigma_2,\de\Sigma_2}\right\rangle_N\in \calH^\calP_{\de\Sigma}\widehat{\otimes}\mathrm{Dens}^\frac{1}{2}(\calV_\Sigma).
\]


\begin{figure}[h!]
    \centering
    \begin{tikzpicture}
    \shadedraw[rounded corners=35pt,top color=gray!50, bottom color=gray!.5](9.5,-1)--(7,-2)--(4.2,-1)--(2,-2)--(0,0)--(2,2)--(4.2,1)--(7,2)--(9.5,1);
    \draw[fill=white] (3.5,0) arc (0:360:0.7cm and 0.3cm);
    \draw (2.8,-0.3) arc (-90:0:0.9cm and 0.5cm);
    \draw (2.8,-0.3) arc (-90:-180:0.9cm and 0.5cm);
    \shadedraw[dashed, color=white, top color=gray!100, bottom color=gray!20, opacity=.3] (5.2,0) arc (0:360:0.5cm and 1.25cm);
    \draw[dashed,color=purple] (5.2,0) arc (0:360:0.5cm and 1.25cm);
    \shadedraw[color=white, top color=gray!80, bottom color=gray!10.5] (10,0) arc (0:360:0.5cm and 1cm);
    \draw[thick, color=purple] (10,0) arc (0:360:0.5cm and 1cm);
    \draw[fill=white] (8,-0.28) arc (0:360:0.7cm and 0.3cm);
    \draw (7.3,-0.58) arc (-90:0:0.9cm and 0.5cm);
    \draw (7.3,-0.58) arc (-90:-180:0.9cm and 0.5cm);
    \node (a) at (20:2.5) {$\Sigma_1$};
    \node (a) at (6:6.5) {$\Sigma_2$};
    \node (a) at (-8:10) {$\partial \Sigma$};
    \end{tikzpicture}
    \caption{Example of the case when $\Sigma=\Sigma_1\cup\Sigma_2$ has boundary.}
    \label{fig:gluing_2}
\end{figure}
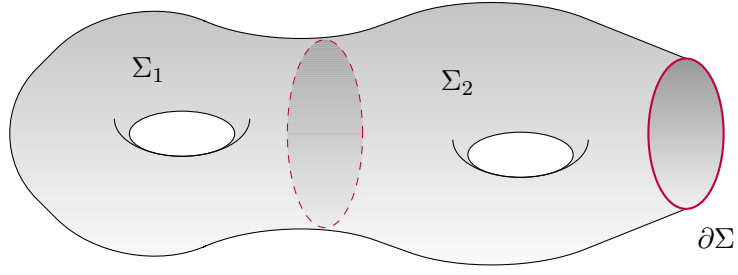

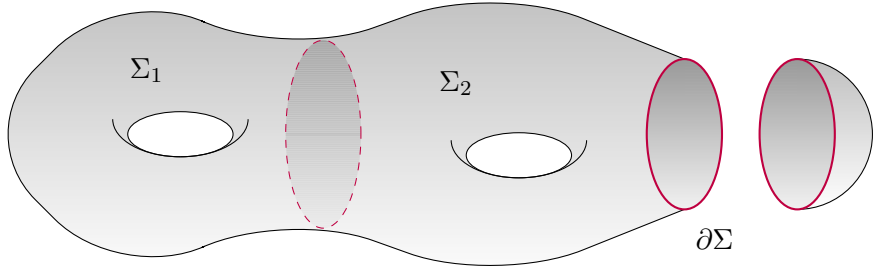
\begin{figure}[h!]
    \centering
    \begin{tikzpicture}
    \shadedraw[rounded corners=35pt,top color=gray!50, bottom color=gray!.5](9.5,-1)--(7,-2)--(4.2,-1)--(2,-2)--(0,0)--(2,2)--(4.2,1)--(7,2)--(9.5,1);
    \draw[fill=white] (3.5,0) arc (0:360:0.7cm and 0.3cm);
    \draw (2.8,-0.3) arc (-90:0:0.9cm and 0.5cm);
    \draw (2.8,-0.3) arc (-90:-180:0.9cm and 0.5cm);
    \shadedraw[dashed, color=white, top color=gray!100, bottom color=gray!20, opacity=.3] (5.2,0) arc (0:360:0.5cm and 1.25cm);
    \draw[dashed,color=purple] (5.2,0) arc (0:360:0.5cm and 1.25cm);
    \shadedraw[color=white, top color=gray!80, bottom color=gray!10.5] (10,0) arc (0:360:0.5cm and 1cm);
    \draw[thick, color=purple] (10,0) arc (0:360:0.5cm and 1cm);
    \shadedraw[top color=gray!50,bottom color=gray!.5] (11,-1) arc (-90:90:1cm and 1cm);
    \shadedraw[color=white, top color=gray!80, bottom color=gray!10.5] (11.5,0) arc (0:360:0.5cm and 1cm);
    \draw[thick, color=purple] (11.5,0) arc (0:360:0.5cm and 1cm);
    \draw[fill=white] (8,-0.28) arc (0:360:0.7cm and 0.3cm);
    \draw (7.3,-0.58) arc (-90:0:0.9cm and 0.5cm);
    \draw (7.3,-0.58) arc (-90:-180:0.9cm and 0.5cm);
    \node (a) at (20:2.5) {$\Sigma_1$};
    \node (a) at (6:6.5) {$\Sigma_2$};
    \node (a) at (-8:10) {$\partial \Sigma$};
    \end{tikzpicture}
    \caption{Closing the manifold $\Sigma$.}
    \label{fig:gluing_3}
\end{figure}

\subsection{Example: abelian $BF$ theory}
\label{subsec:abelian_BF_theory}
Let $\Sigma$ be a $d$-manifold with compact boundary $\de\Sigma$. The BV space of fields is then given by $\calF_\Sigma=\Omega^\bullet(\Sigma)[1]\oplus\Omega^\bullet(\Sigma)[d-2]$. Denote a field in $\calF_\Sigma$ by $\mathbf{X}\oplus\mathbf{Y}$, where $\mathbf{X}\in \Omega^\bullet(\Sigma)[1]$ and $\mathbf{Y}\in \Omega^\bullet(\Sigma)[d-2]$. The BV symplectic form is then given by 
\[
\omega_\Sigma=\int_\Sigma\sum_i\delta\mathbf{Y}_i\delta\mathbf{X}^i,
\]
the BV action by 
\[
\calS_\Sigma=\int_\Sigma\sum_i\mathbf{Y}_i\dd_\Sigma\mathbf{X}^i,
\]
and the cohomological vector field by 
\[
Q_\Sigma=\int_\Sigma\sum_i\left(\dd_\Sigma\mathbf{Y}_i\frac{\delta}{\delta\mathbf{Y}_i}+\dd_\Sigma\mathbf{X}^i\frac{\delta}{\delta\mathbf{X}^i}\right),
\]
where $\delta$ denotes the de Rham differential on $\calF_\Sigma$ and $\dd_\Sigma$ the one on $\Sigma$.
The exact BFV manifold $(\calF^\de_{\de\Sigma},\omega^\de_{\de\Sigma}=\delta\alpha^\de_{\de\Sigma},Q^\de_{\de\Sigma})$ assigned to the boundary $\de\Sigma$ is given by the BFV space of fields $\calF^\de_{\de\Sigma}=\Omega^\bullet(\de\Sigma)[1]\oplus\Omega^\bullet(\de\Sigma)[d-2]$, the primitive 1-form
\[
\alpha^\de_{\de\Sigma}=(-1)^d\int_{\de\Sigma}\sum_i\mathds{Y}_i\delta\mathds{X}^i,
\]
the BFV boundary action
\[
\calS^\de_{\de\Sigma}=\int_{\de\Sigma}\sum_i\mathds{Y}_i\dd_{\de\Sigma}\mathds{X}^i,
\]
and cohomological vector field
\[
Q^\de_{\de\Sigma}=\int_{\de\Sigma}\sum_i\left(\dd_{\de\Sigma}\mathds{Y}_i\frac{\delta}{\delta\mathds{Y}_i}+\dd_{\de\Sigma}\mathds{X}^i\frac{\delta}{\delta\mathds{X}^i}\right).
\]
We have denoted by $\mathds{X}\oplus\mathds{Y}=i_\Sigma^*(\mathbf{X}\oplus\mathbf{Y})\in \calF^\de_{\de\Sigma}$ for $\mathbf{X}\oplus\mathbf{Y}\in \calF_\Sigma$, where $i_\Sigma\colon \de\Sigma\hookrightarrow \Sigma$ denotes the inclusion. The surjective submersion $\pi_\Sigma\colon \calF_\Sigma\to \calF^\de_{\de\Sigma}$ is given by restricting to the boundary, i.e. $\pi_\Sigma=i_\Sigma^*$. 

Consider the case where $\de\Sigma$ is given by the disjoint union of two compact manifolds $\de_1\Sigma$ and $\de_2\Sigma$ such that we can consider a splitting $\calF^\de_{\de\Sigma}=\calF^{\de}_{\de_1\Sigma}\times \calF^\de_{\de_2\Sigma}$. Then we can consider a polarization $\calP$ on $\calF^\de_{\de\Sigma}$ which is given by a direct product of polarizations on each factor. One can choose the convenient $\frac{\delta}{\delta\mathds{Y}}$-polarization on $\de_1\Sigma$ and identify the quotient leaf space with $\calB_{\de_1\Sigma}\cong\Omega^\bullet(\de_1\Sigma)[1]$ whose coordinates are given by the $\mathds{X}$ fields. On $\de_2\Sigma$ one can choose the convenient $\frac{\delta}{\delta\mathds{X}}$-polarization and similarly as before identify the quotient leaf space with $\calB_{\de_2\Sigma}\cong\Omega^\bullet(\de_2\Sigma)[d-2]$ whose coordinates are given by the $\mathds{Y}$ fields. The leaf space associated to the boundary polarization is then given by $\calB^\calP_{\de\Sigma}=\calB_{\de_1\Sigma}\times \calB_{\de_2\Sigma}$. 

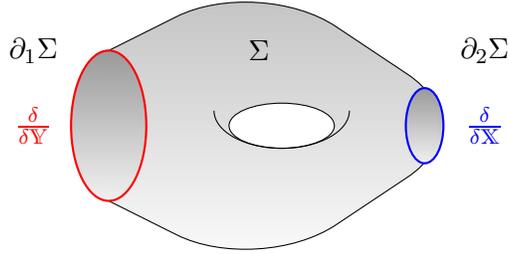
\begin{figure}[h!]
    \centering
    \begin{tikzpicture}
    \shadedraw[rounded corners=35pt,top color=gray!50, bottom color=gray!.5](0,0)--(2,1)--(5,-1)--(2,-3)--(0,-2);
    \draw[fill=white] (3,-1) arc (0:360:0.7cm and 0.3cm);
    \draw (2.3,-1.3) arc (-90:0:0.9cm and 0.5cm);
    \draw (2.3,-1.3) arc (-90:-180:0.9cm and 0.5cm);
    \shadedraw[thick, color=red, top color=gray!80, bottom color=gray!10.5] (0,-1) ellipse (.5cm and 1cm);
    \shadedraw[thick, color=blue, top color=gray!80, bottom color=gray!10.5] (4.2,-1) ellipse (.25cm and .5cm);
    \node(Y) at (-1,-1) {\textcolor{red}{$\frac{\delta}{\delta\mathds{Y}}$}};
    \node(X) at (5,-1) {\textcolor{blue}{$\frac{\delta}{\delta\mathds{X}}$}};
    \node(1) at (-1,0) {$\de_1\Sigma$};
    \node(2) at (5,0) {$\de_2\Sigma$};
    \node(2) at (2,0) {$\Sigma$};
    \end{tikzpicture}
    \caption{Example of $\Sigma$ with two boundaries $\de_1\Sigma$ and $\de_2\Sigma$. On $\de_1\Sigma$ we choose the $\frac{\delta}{\delta\mathds{Y}}$-polarization (i.e. the $\mathds{X}$-representation) and on $\de_2\Sigma$ we choose the $\frac{\delta}{\delta\mathds{X}}$-polarization (i.e. the $\mathds{Y}$-representation).}
    \label{fig:boundary_polarization_BF}
\end{figure}

The case when we have more than two boundary components is illustrated and explained in Figure \ref{fig:more_boundary_components}. The case for boundary components with mixed polarization is explained in Figure \ref{fig:mixed_boundary_polarization}. However, we will not consider this case in this paper\footnote{This case was studied for the Poisson sigma model (2-dimensional AKSZ theory) in \cite{CMW3} for the quantization of the relational symplectic groupoid. Another (a bit different) approach was considered for 2-dimensional Yang--Mills theory in \cite{IM}.}.

\begin{figure}[h!]
    \centering
    \begin{tikzpicture}
    \tikzset{Bullet/.style={fill=blue,draw,color=#1,circle,minimum size=3pt,scale=0.5}}
    \shadedraw[rounded corners=20pt,top color=gray!50, bottom color=gray!.5](5.5,1)--(8,0)--(10,-1)--(12,-1)--(12,-3)--(10,-3)--(8,-4)--(5.2,-5.2)--(5,-3)--(7,-2)--(5,-1)--(5.5,1);
    \draw[fill=white] (9.5,-2) arc (0:360:0.7cm and 0.3cm);
    \draw (8.8,-2.3) arc (-90:0:0.9cm and 0.5cm);
    \draw (8.8,-2.3) arc (-90:-180:0.9cm and 0.5cm);
    \shadedraw[color=white, top color=gray!80, bottom color=gray!10.5] (6,0) arc (0:360:0.5cm and 1cm);
    \draw[thick, color=red] (6,0) arc (0:360:0.5cm and 1cm);
    \shadedraw[color=white, bottom color=gray!80, top color=gray!10.5] (6,-4) arc (0:360:0.5cm and 1cm);
    \draw[thick, color=red] (6,-4) arc (0:360:0.5cm and 1cm);
    \shadedraw[color=white, top color=gray!80, bottom color=gray!10.5] (12.1,-2) arc (0:360:0.5cm and 1cm);
    \draw[thick, color=blue] (12.1,-2) arc (0:360:0.5cm and 1cm);
    \node[label=left:{\textcolor{red}{$\mathds{X}$}}] at (5,0){};
    \node[label=left:{\textcolor{red}{$\mathds{X}$}}] at (5,-4){};
    \node[label=right:{\textcolor{blue}{$\mathds{Y}$}}] at (12.1,-2){};
    \node[label=above:{$\de_1^{(1)}\Sigma$}] at (5.5,1){};
    \node[label=below:{$\de_1^{(2)}\Sigma$}] at (5.5,-5){};
    \node[label=below:{$\de_2\Sigma$}] at (11.6,-3){};
    \node[label=above:{$\Sigma$}] at (7,-1){};
    \end{tikzpicture}
    \caption{Example of a manifold with three boundary components where two components have the $\mathds{X}$-representation and one component has the $\mathds{Y}$-representation. Note that in this case we have $\de_1\Sigma=\de_1^{(1)}\Sigma\cup \de_1^{(2)}\Sigma$. In general we can have $r_1$ boundary components in the $\mathds{X}$-representation and $r_2$ boundary components in the $\mathds{Y}$-representation. In this case we have $\de_1\Sigma=\bigcup_{k=1}^{r_1}\de_1^{(k)}\Sigma$ and $\de_2\Sigma=\bigcup_{k=1}^{r_2}\de_2^{(k)}\Sigma$.}
    \label{fig:more_boundary_components}
\end{figure}
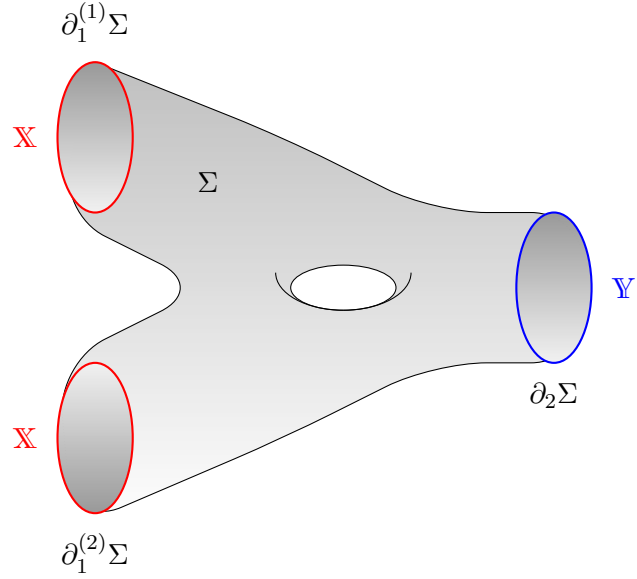

\begin{figure}[h!]
    \centering
    \begin{tikzpicture}
    \tikzset{Bullet/.style={fill=blue,draw,color=#1,circle,minimum size=3pt,scale=0.5}}
    \shadedraw[color=white, top color=gray!80, bottom color=gray!10.5] (0,0) arc (0:360:2cm and 2cm);
    \draw[thick, color=blue] (0,0) arc (0:90:2cm and 2cm);
    \draw[thick, color=red] (-2,2) arc (90:180:2cm and 2cm);
    \draw[thick, color=blue] (-4,0) arc (180:270:2cm and 2cm);
    \draw[thick, color=red] (-2,-2) arc (270:360:2cm and 2cm);
    \node[Bullet=black] at (0,0){};
    \node[Bullet=black] at (-2,2){};
    \node[Bullet=black] at (-2,-2){};
    \node[Bullet=black] at (-4,0){};
    \node[label=right:{\textcolor{blue}{$\mathds{Y}$}}] at (-0.5,1.5){};
    \node[label=left:{\textcolor{red}{$\mathds{X}$}}] at (-3.5,1.5){};
    \node[label=left:{\textcolor{blue}{$\mathds{Y}$}}] at (-3.5,-1.5){};
    \node[label=right:{\textcolor{red}{$\mathds{X}$}}] at (-0.5,-1.5){};
    \end{tikzpicture}
    \caption{Example of a boundary component with mixed boundary polarization.}
    \label{fig:mixed_boundary_polarization}
\end{figure}
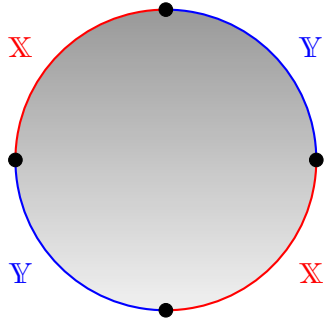

In particular, we want to assume for the leaf space $\calB^\calP_{\de\Sigma}$ the property that for a suitably chosen local functional $f^\calP_{\de\Sigma}$, the restriction of the adapted BFV 1-form $\alpha^{\de,\calP}_{\de\Sigma}:=\alpha^\de_{\de\Sigma}-\delta f^\calP_{\de\Sigma}$ to the fibers of the polarization $\calP$ vanishes. In this case we can identify the space of boundary states $\calH^\calP_{\de\Sigma}$ with $\mathrm{Dens}^\frac{1}{2}(\calB^\calP_{\de\Sigma})$ by multiplying with $\exp(\I f^\calP_{\de\Sigma}/\hbar)$. Moreover, we can then modify the BV action $\calS_\Sigma$ to $\calS^\calP_\Sigma:=\calS_\Sigma+\pi^*_\Sigma f^\calP_{\de\Sigma}$. In our case, we consider the functional
\[
f^{\calP}_{\de\Sigma}=(-1)^{d-1}\int_{\de_2\Sigma}\sum_i\mathds{Y}_i\mathds{X}^i,
\]
the adapted BFV 1-form
\[
\alpha^{\de,\calP}_{\de\Sigma}=(-1)^d\int_{\de_1\Sigma}\sum_i\mathds{Y}_i\delta\mathds{X}^i-\int_{\de_2\Sigma}\sum_i\delta\mathds{Y}_i\mathds{X}^i
\]
and the modified BV action 
\[
\calS^\calP_\Sigma=\int_\Sigma \sum_i\mathds{Y}_i\dd_{\de\Sigma}\mathds{X}+(-1)^{d-1}\int_{\de_2\Sigma}\sum_i\mathds{Y}_i\mathds{X}^i.
\]
Denote by $\widetilde{\mathds{X}}\oplus\widetilde{\mathds{Y}}:=\pi^*_\Sigma(\mathds{X}\oplus\mathds{Y})$ an extension of $\mathds{X}\oplus\mathds{Y}$ to $\calF_\Sigma$ and denote by $\widehat{\mathbf{X}}\oplus\widehat{\mathbf{Y}}$ the complement such that 
\begin{align*}
    \mathbf{X}&=\widetilde{\mathds{X}}+\widehat{\mathbf{X}},\\
    \mathbf{Y}&=\widetilde{\mathds{Y}}+\widehat{\mathbf{Y}}.
\end{align*}
Note that we require that $i_1^*\widehat{\mathbf{X}}=0$ and $i_2^*\widehat{\mathbf{Y}}=0$, where $i_j\colon \de_j\Sigma\hookrightarrow \Sigma$ denotes the inclusion of the corresponding boundary components for $j=1,2$. Let us choose this section of the bundle $\calF_\Sigma\to \calB^\calP_{\de\Sigma}$. The modified BV action is then given by 
\begin{multline*}
\calS^\calP_\Sigma=\int_\Sigma\sum_i\left(\widetilde{\mathds{Y}}_i\dd_{\de\Sigma}\widetilde{\mathds{X}}^i+\widetilde{\mathds{Y}}_i\dd_{\Sigma}\widehat{\mathbf{X}}^i+\widehat{\mathbf{Y}}_i\dd_{\de\Sigma}\widetilde{\mathds{X}}^i+\widehat{\mathbf{Y}}_i\dd_\Sigma\widehat{\mathbf{X}}^i\right)\\
+(-1)^{d-1}\int_{\de_2\Sigma}\sum_i\left(\mathds{Y}_i\widetilde{\mathds{X}}^i+\mathds{Y}_i\widehat{\mathbf{X}}^i\right).
\end{multline*}
Consider the last bulk term which we denote by $\widehat{\calS}_\Sigma:=\int_\Sigma\sum_i\widehat{\mathbf{Y}}_i\dd_\Sigma\widehat{\mathbf{X}}^i$. Note that, by the vanishing boundary conditions, we get the critical points defined by the equations $\dd_\Sigma\widehat{\mathbf{X}}=\dd_\Sigma\widehat{\mathbf{Y}}=0$. Moreover, consider the subcomplex of the de Rham complex $\Omega^\bullet(\Sigma)$ given by 
\[
\Omega^\bullet_{\mathrm{D}j}(\Sigma):=\{\gamma\in\Omega^\bullet(\Sigma)\mid i^*_j\gamma=0\},
\]
where D stands for \emph{Dirichlet}. The residual fields are then given by
\[
\calV_\Sigma=H^\bullet_{\mathrm{D}1}(\Sigma)[1]\oplus H^\bullet_{\mathrm{D}2}(\Sigma)[d-2],
\]
which is a finite-dimensional BV manifold. Note that we have a natural identification with relative cohomology $H^\bullet_{\mathrm{D}j}(\Sigma)\cong H^\bullet(\Sigma,\de_j\Sigma)$ for $j=1,2$.
Hence, we can define a canonical BV Laplacian by choosing coordinates $(z^j,z^\dagger_j)$ on $\calV_\Sigma$. Choosing a basis $([\kappa_j])$ of $H^\bullet_{\mathrm{D}1}(\Sigma)$ with representative $\kappa_j\in \Omega^\bullet_{\mathrm{D}1}(\Sigma)$ and the corresponding dual basis $([\kappa^j])$ of $H^\bullet_{\mathrm{D}2}(\Sigma)$ with representative $\kappa^j\in \Omega^\bullet_{\mathrm{D}2}(\Sigma)$ satisfying the relation $\int_\Sigma\kappa^i\kappa_j=\delta^i_j$, we can write the residual fields as 
\begin{equation}
    \mathsf{x}=\sum_jz^j\kappa_j,\qquad\mathsf{y}=\sum_jz^\dagger_j\kappa^j.
\end{equation}
The BV Laplacian on $\calV_\Sigma$ is given by
\[
\Delta_{\calV_\Sigma}=\sum_j(-1)^{(d-1)\gh(z^j)+1}\frac{\de^2}{\de z^i\de z^\dagger_j}
\]
and the BV symplectic form on $\calV_\Sigma$ is given by 
\[
\omega_{\calV_\Sigma}=\sum_j(-1)^{(d-1)\gh(z^j)+1}\delta z^\dagger_j\delta z^j.
\]
Note that $\gh(z^j)=1-\gh(\kappa_j)$ and $\gh(z^\dagger_j)=\gh(\kappa_j)-2$. Consider another splitting of the fields
\begin{align*}
    \widehat{\mathbf{X}}&=\mathsf{x}+\mathscr{X},\\
    \widehat{\mathbf{Y}}&=\mathsf{y}+\mathscr{Y},
\end{align*}
where $\mathscr{X}\oplus\mathscr{Y}$ denotes the corresponding complement of the field $\mathsf{x}\oplus \mathsf{y}\in \calV_\Sigma$ as elements of a symplectic complement $\calY'_\Sigma$ of $\calV_\Sigma$. Note that $i_1^*\mathscr{X}=i_2^*\mathscr{Y}=0$. Using this, we can see that $\widehat{\calS}_\Sigma=\int_\Sigma\mathscr{Y}\dd_\Sigma\mathscr{X}$ which can be regarded as a quadratic function on $\Omega^\bullet_{\mathrm{D}1}(\Sigma)[1]\oplus\Omega^\bullet_{\mathrm{D}2}(\Sigma)[d-2]$ and has critical points given by closed forms. One can now choose a Lagrangian submanifold $\calL$ of $\calY'_\Sigma$ where $\widehat{\calS}_\Sigma$ has isolated critical points at the origin, which is to say that the de Rham differential $\dd$ has trivial kernel. This can be done by Hodge theory for manifolds with boundary. Consider the Hodge star operator $*$ for a chosen metric on $\Sigma$ with product structure near the boundary, i.e. there is a diffeomorphism $\varphi\colon \de\Sigma\supset U\to \de\Sigma\times [0,\varepsilon)$ for some $\varepsilon>0$ such that the restriction of $\varphi$ to $\de\Sigma$ is the identity on $\de\Sigma$ and the metric on $\Sigma$ restricted to the neighborhood $U$ has the form $\varphi^*(g_{\de\Sigma}+\dd t^2)$, where $g_{\de\Sigma}$ denotes the metric on the boundary and $t\in[0,\varepsilon)$ the coordinate on $\de\Sigma\times [0,\varepsilon)$. One can define the Lagrangian submanifold 
\[
\calL=(\dd^*\Omega^{\bullet+1}_{\mathrm{N}2}(\Sigma)\cap\Omega^\bullet_{\mathrm{D}1}(\Sigma))[1]\oplus (\dd^*\Omega^{\bullet+1}_{\mathrm{N}1}(\Sigma)\cap\Omega^\bullet_{\mathrm{D}2}(\Sigma))[d-2],
\]
where $\dd^*:=*\dd*$ and 
\[
\Omega^\bullet_{\mathrm{N}j}(\Sigma):=\{\gamma\in\Omega^\bullet(\Sigma)\mid i^*_j*\gamma=0\}
\]
is the space of \emph{Neumann} forms relative to $\de_j\Sigma$. One can check that the restriction of $\widehat{\calS}_\Sigma$ to $\calL$ is nondegenerate. Moreover, one can show (see \cite{CMR2} for more details) that $\calL$ is is indeed Lagrangian in the symplectic complement $\calY'_\Sigma$ of $\calV_\Sigma$. One can then construct a \emph{propagator} $\mathscr{P}$ as the integral of the chain contraction $K$ of $\Omega^\bullet_{\mathrm{D}1}(\Sigma)$ onto $H^\bullet_{\mathrm{D}1}(\Sigma)$. The gauge-fixing Lagrangian is then given as 
\[
\calL=\mathrm{im}(K)[1]\oplus\mathrm{im}(K^*)[d-2].
\]
This can be done by choosing a Hodge chain contraction $K\colon \Omega^\bullet_{\mathrm{D}1}(\Sigma)\to \Omega^{\bullet-1}_{\mathrm{D}1}(\Sigma)$. We call a differential form $\beta\in \Omega^\bullet(\Sigma)$ \emph{ultra-Dirichlet relative to $\de_j\Sigma$} if the pullbacks to $\de_j\Sigma$ of all even normal derivatives of $\beta$ and pullbacks of all odd normal derivatives of $*\beta$ vanish. We call a differential form $\beta\in\Omega^\bullet(\Sigma)$ \emph{ultra-Neumann relative to $\de_j\Sigma$} if pullbacks to $\de_j\Sigma$ of all even normal derivatives of $*\beta$ and pullbacks of all odd normal derivatives of $\beta$ vanish. Denote by $\Omega^\bullet_{\widehat{\mathrm{D}}j}(\Sigma)$ and $\Omega^\bullet_{\widehat{\mathrm{N}}j}(\Sigma)$ the space of ultra-Dirichlet and ultra-Neumann forms, respectively. We call a differential form $\beta\in \Omega^\bullet(\Sigma)$ \emph{ultra-harmonic} if it is closed with respect to $\dd$ and $\dd^*$. Denote by $\widehat{\mathrm{Harm}}^\bullet(\Sigma)$ the space of ultra-harmonic forms on $\Sigma$. Note that $\Omega^\bullet_{\widehat{\mathrm{D}}j}(\Sigma)$ and $\Omega^\bullet_{\widehat{\mathrm{N}}j}(\Sigma)$ are both subcomplexes of $\Omega^\bullet(\Sigma)$ with respect to $\dd$ and $\dd^*$, respectively. One can easily see that ultra-harmonic Dirichlet forms are ultra-Dirichlet and ultra-harmonic Neumann forms are ultra-Neumann.
Let $\Delta_\mathrm{Hodge}:=\dd\dd^*+\dd^*\dd$ denote the Hodge Laplacian.
The Hodge chain contraction is then given by $K=\dd^*/(\Delta_\mathrm{Hodge}+\pi_{\mathrm{Harm}})$, where $\pi_\mathrm{Harm}$ denotes the projection to (ultra-)harmonic forms. The propagator $\mathscr{P}$ is then a smooth form on the compactified configuration space $\overline{\mathrm{Conf}_2(\Sigma)}:=\overline{\{(u_1,u_2)\in\Sigma\times\Sigma\mid u_1\not=u_2\}}$. Moreover, consider the space $\mathfrak{D}:=\{(u_1,u_2)\in (\de_1\Sigma\times \Sigma)\cup(\de_2\Sigma\times \Sigma)\mid u_1\not=u_2\}$ (see Appendix \ref{app:configuration_spaces_on_manifolds_with_boundary} for more on configuration spaces on manifolds with boundary). Then we get that $\mathscr{P}\in\Omega^{d-1}(\overline{\mathrm{Conf}_{2}(\Sigma)},\mathfrak{D})=\{\gamma\in \Omega^{d-1}(\overline{\mathrm{Conf}_2(\Sigma)})\mid i^*_\mathfrak{D}\gamma=0\}$, where $i_\mathfrak{D}\colon \mathfrak{D}\hookrightarrow \overline{\mathrm{Conf}_2(\Sigma)}$ denotes the inclusion. We can write the propagator as a path integral 
\[
\mathscr{P}=\frac{1}{T_\Sigma}\frac{(-1)^d}{\I\hbar}\int_\calL \exp\left(\I\widehat{\calS}_\Sigma/\hbar\right)\pi_1^*\mathscr{X}\pi_2^*\mathscr{Y},
\]
where $\pi_j\colon \Sigma\times\Sigma\to \Sigma$ for $j=1,2$ denotes the projection onto the first and second factor, respectively and $T_\Sigma=\int_\calL\exp(\I\widehat{\calS}_\Sigma/\hbar)\in\C\otimes \mathrm{Dens}^\frac{1}{2}(\calV_\Sigma)/\{\pm 1\}$ is given in terms of the Ray--Singer torsion \cite{RaySinger1971} which is defined through zeta regularization which does not depend on the choice of $\calL$ since the Ray--Singer torsion does not depend on the choice of metric. 

Let $N$ be a $(d-1)$-manifold and let $\calH^n_{N,\ell}$ be the vector space of $n$-linear functionals on $\Omega^\bullet(N)[\ell]$ of the form 
\[
\Omega^\bullet(N)[\ell]\ni \mathds{D}\mapsto \int_{N^n}\gamma\pi^*_1\mathds{D}\dotsm \pi^*_n\mathds{D}
\]
up to multiplication with a term given by the Reidemeister torsion \cite{Reidemeister1935} of $N$. We have denoted by $\gamma$ a distributional form on $N^n$. 
The boundary state space $\calH^\calP_{\de\Sigma}$ is then given by
\[
\calH^\calP_{\de\Sigma}=\prod_{n_1,n_2=0}^\infty \calH^{n_1}_{\de_1\Sigma,1}\widehat{\otimes}\calH^{n_2}_{\de_2\Sigma,d-2}
\]
and $\widehat{\calH}^\calP_{\Sigma}=\calH^\calP_{\de\Sigma}\widehat{\otimes}\mathrm{Dens}^\frac{1}{2}(\calV_\Sigma)$. Perturbatively, the BV-BFV partition function is asymptotically given by 
\begin{multline*}
T_\Sigma\exp(\I\calS^\mathrm{eff}_\Sigma/\hbar)\times\\\times\sum_{j\geq 0}\hbar^j\sum_{n_1,n_2\geq 0}\,\int_{(\de_1\Sigma)^{n_1}\times(\de_2\Sigma)^{n_2}}R^j_{n_1n_2}(\mathsf{x},\mathsf{y})\pi_{1,1}^*\mathds{X}^{i_1}\dotsm \pi^*_{1,n_1}\mathds{X}^{i_n}\pi^*_{2,1}\mathds{Y}_{i_1}\dotsm \pi^*_{2,n_2}\mathds{Y}_{i_n},
\end{multline*}
where $R^j_{n_1n_2}(\mathsf{x},\mathsf{y})$ denotes distributional forms on $(\de_1\Sigma)^{n_1}\times(\de_2\Sigma)^{n_2}$ with values in $\mathrm{Dens}^\frac{1}{2}(\calV_\Sigma)$. Note that we sum over $i_1,\ldots,i_n$. Moreover, we have denoted by $\calS^\mathrm{eff}_\Sigma$ the effective action\footnote{One can easily check that the partition function $\mathsf{Z}^{\scriptscriptstyle\mathrm{BV-BFV}}_{\Sigma,\de\Sigma}=T_\Sigma\exp(\I\calS^\mathrm{eff}_\Sigma/\hbar)$, defined through the effective action $\calS^\mathrm{eff}_\Sigma$, satisfies the mQME \eqref{eq:mQME_1}} given by 
\[
\calS^\mathrm{eff}_\Sigma=(-1)^{d-1}\left(\int_{\de_2\Sigma}\sum_i\mathds{Y}_i\mathsf{x}^i-\int_{\de_1\Sigma}\sum_i\mathsf{y}_i\mathds{X}^i\right)-(-1)^{2d}\int_{\de_2\Sigma\times\de_1\Sigma}\sum_i\pi^*_1\mathds{Y}_i\mathscr{P}\pi^*_{2}\mathds{X}^i.
\]
The BFV boundary operator acting on $\calH^\calP_{\de\Sigma}$ is then given by the ordered standard quantization of the BFV boundary action $\calS^\de_{\de\Sigma}$:
\[
\Omega^\calP_{\de\Sigma}=(-1)^d\I\hbar\left(\int_{\de_2\Sigma}\sum_i\dd_{\de\Sigma}\mathds{Y}_i\frac{\delta}{\delta\mathds{Y}_i}+\int_{\de_1\Sigma}\sum_i\dd_{\de\Sigma}\mathds{X}^i\frac{\delta}{\delta\mathds{X}^i}\right).
\]
Suppose now that we have two smooth $d$-manifolds $\Sigma_1$ and $\Sigma_2$ with common boundary component $(d-1)$-manifold $N$. We can then compute the glued partition function $\mathsf{Z}^{\scriptscriptstyle\mathrm{BV}}_\Sigma$ for the glued  manifold $\Sigma=\Sigma_1\cup_N\Sigma_2$ (see Figure \ref{fig:gluing_1}) out of the BV-BFV partition functions $\mathsf{Z}^{\scriptscriptstyle\mathrm{BV-BFV}}_{\Sigma_1,\de\Sigma_1}$ on $\Sigma_1$ and $\mathsf{Z}^{\scriptscriptstyle\mathrm{BV-BFV}}_{\Sigma_2,\de\Sigma_2}$ on $\Sigma_2$. We consider the transversal polarization on $\calF^\de_{N}$. In particular, we decompose the boundaries $\de\Sigma_1=\de_1\Sigma_1\sqcup\de_2\Sigma_1$ and $\de\Sigma_2=\de_1\Sigma_2\sqcup\de_2\Sigma_2$ such that $N\subset \de_1\Sigma_1$ and $N^\mathrm{opp}\subset\de_2\Sigma_2$. Let $\mathds{X}^N_1$ and $\mathds{Y}^N_2$ be the coordinates on $\Omega^\bullet(N)[1]$ and $\Omega^\bullet(N)[d-2]$, respectively. The glued state is then given by 
\[
\mathsf{Z}^{\scriptscriptstyle\mathrm{BV}}_\Sigma=\int_{\{\mathds{X}^N_1,\mathds{Y}^N_2\}}\exp\left(\frac{\I}{\hbar}(-1)^{d-1}\int_N\sum_i(\mathds{Y}^N_2)_i(\mathds{X}^N_1)^i\right)\mathsf{Z}^{\scriptscriptstyle\mathrm{BV-BFV}}_{\Sigma_1,\de\Sigma_1}\mathsf{Z}^{\scriptscriptstyle\mathrm{BV-BFV}}_{\Sigma_2,\de\Sigma_2}.
\]
The glued partition function can be computed by considering the glued effective action and then considering $\mathsf{Z}^{\scriptscriptstyle\mathrm{BV}}_\Sigma=T_{\Sigma_1}T_{\Sigma_2}\exp(\I\calS^\mathrm{eff}_\Sigma/\hbar)$.
\begin{rem}
This construction can be generalized to the case of $BF$-like theories (perturbations of abelian $BF$ theory), see \cite{CMR2} for more details. In particular, this is important since AKSZ theories often appear as $BF$-like theories.
This will be relevant especially for treating DW theory in the BV-BFV setting by regarding it as an AKSZ theory, which will be the content of Section \ref{sec:AKSZ_formulation_of_DW_theory}. An important result regarding the gluing for $BF$-like theories is given by the following proposition.
\begin{prop}[Cattaneo--Mnev--Reshetikhin\cite{CMR2}]
Let $\Sigma$ be cut along a codimension-one submanifold $N$ into $\Sigma_1$ and $\Sigma_2$. Let $\mathsf{Z}^{\scriptscriptstyle\mathrm{BV-BFV}}_{\Sigma_1,N}$ and $\mathsf{Z}^{\scriptscriptstyle\mathrm{BV-BFV}}_{\Sigma_2,N^\mathrm{opp}}$ be the boundary states for $\Sigma_1$ and $\Sigma_2$ with a choice of residual fields and propagators and transverse ($\mathds{X}$ vs. $\mathds{Y}$) polarizations on $N$. Then the gluing of $\mathsf{Z}^{\scriptscriptstyle\mathrm{BV-BFV}}_{\Sigma_1,N}$ and $\mathsf{Z}^{\scriptscriptstyle\mathrm{BV-BFV}}_{\Sigma_2,N^\mathrm{opp}}$ is the state $\mathsf{Z}^{\scriptscriptstyle\mathrm{BV-BFV}}_{\Sigma,\de\Sigma}$ for $\Sigma$ with the consequent choice of residual fields and propagators.
\end{prop}
\end{rem}

\subsection{BV-BF$^k$V extension and shifted symplectic structures}
\label{subsec:shifted_symplectic_structures}
If we move to higher defects, the symplectic form will be shifted by degree $+1$. In particular, if we consider a codimension $k$ submanifold $N_k\subset\Sigma$, then the symplectic gauge formalism associated to $N_k$ will be $(k-1)$-shifted. We will call the theory associated to a codimension $k$ submanifold a BF$^k$V theory and the coupling for each contiguous codimension (fully extended) a BV-BF$^k$V theory associated to the $d$-manifold $\Sigma$. The underlying mathematical theory for shifted symplectic structures was developed first in \cite{PantevToenVaquieVezzosi2013} by using methods and the language of derived algebraic geometry and studied further by various people. In \cite{PantevToenVaquieVezzosi2013}, they define first a symplectic form on a smooth scheme over some base ring $k$ of characteristic zero to be a $2$-form $\omega\in H^0(X,\Omega^{2,cl}_{X/k})$ which moreover is required to be nondegenerate, i.e. it induces an isomorphism $\Theta_\omega$ between the tangent and cotangent bundles of $X$. 
Then one can define an \emph{$n$-shifted symplectic form} on a derived Artin stack $X$ (in particular, they consider $X$ to be the solution of a \emph{derived moduli problem} (see also \cite{Toen2014})) to be a closed $2$-form $\omega\in H^n(X,\bigwedge^2\mathbb{L}_{X/k})$ of degree $n$ on $X$ such that the corresponding morphism $\Theta_\omega\colon \mathbb{T}_{X/k}\to \mathbb{L}_{X/k}[n]$ is an isomorphism in the derived category of quasi-coherent sheaves $D^b\mathrm{QCoh}(X)$ on $X$. Here we have denoted by $\mathbb{L}_{X/k}$ the cotangent complex of $X$ and by $\mathbb{T}_{X/k}$ its dual, the tangent complex of $X$. Of course, it is important to mention what \emph{closed} in this setting actually means. They define closedness of general $p$-forms by interpreting sections of $\bigwedge^2\mathbb{L}_{X/k}$ as functions on the derived loop stack $\calL X$ and consider it as some type of $S^1$-equivariance property. An important result of \cite{PantevToenVaquieVezzosi2013} is an existence result concerning a derived algebraic version of the AKSZ formulation in this setting. In particular, they prove the following theorem:
\begin{thm}[Pantev--To\"en--Vaqui\'e--Vezzosi\cite{PantevToenVaquieVezzosi2013}]
Let $X$ be a derived stack endowed with an $\calO$-orientation of dimension $d$, and let $(F,\omega)$ be a derived Artin stack with an $n$-shifted symplectic structure $\omega$. Then the derived mapping stack $\mathbf{Map}(X,F)$ carries a natural $(n-d)$-shifted symplectic structure.
\end{thm}
The BV-BFV formalism, as described in Section \ref{sec:BV-BFV_formalism}, can be easily extended to higher codimensions in the classical setting, but needs more sophisticated methods in the quantum setting \cite{CMR1,Moshayedi2021}. As discussed in \cite{Moshayedi2021}, the quantum setting needs to couple the methods of deformation quantization to the methods of geometric quantization, both in the shifted setting. Fortunately, also the methods for shifted Poisson structures and shifted deformation quantization have been developed in \cite{CalaquePantevToenVaquieVezzosi2017,Safronov2017_lectures}. The setting of geometric quantization in the shifted picture has been recently considered in \cite{Safronov2020}. The methods developed in \cite{PantevToenVaquieVezzosi2013} have been considered for the setting of codimension 1 structures (BV-BFV) in \cite{Calaque2015} for the setting of AKSZ topological field theories by using Lagrangian correspondences and have been recently (fully) extended by Calaque, Haugseng and Scheimbauer \cite{CalaqueHaugsengScheimbauer}.\\
\begin{ex}[Chern--Simons theory]
As an example we want to discuss the 3-dimensional case of Chern--Simons theory (see Section \ref{subsec:Chern-Simons}). Recall that the phase space is given by the moduli space of flat $G$ connections for some compact Lie group $G$. In the codimension 0 case, we are working over a closed oriented 3-manifold $\Sigma$ (see Figure \ref{fig:3-ball_defects}) which carries the induced BV symplectic structure, i.e. a $(-1)$-shifted structure. In the codimension 1 case we are working over a closed oriented 2-manifold $N_1\subset \Sigma$ (e.g. some Riemann surface $\Sigma_g$ of genus $g$), to which we can associate a phase space endowed with the \emph{Atiyah--Bott symplectic structure} \cite{AtiyahBott1983}, which is a 0-shifted structure. In the codimension 2 case we are working over a closed oriented 1-manifold $S^1\cong N_2\subset \Sigma$, for which the phase space is given by the stack of conjugacy classes $[G/G]$ and we can consider a $1$-shifted symplectic structure by the canonical 3-form on $G$. Finally, in the codimension 3 case we consider a point $\mathrm{pt}$ (0-manifold), for which the phase space is given by the classifying stack $\mathrm{B}G=[\mathrm{pt}/G]$ endowed with the 2-shifted symplectic form given by the invariant pairing (Killing form) on the Lie algebra $\mathfrak{g}=\mathrm{Lie}(G)$.
\begin{figure}[h!]
    \centering
    \begin{tikzpicture}
    \shade[shading=ball, ball color=gray!10, draw] (0,0) circle (2);
    \draw[dashed, fill=blue!10, opacity=.3] (0,0) ellipse (2cm and .5cm);
    \draw[rotate=30,dashed,fill=green!10,opacity=.3] (0,0) ellipse (2cm and .5cm);
    \draw[rotate=150,dashed,fill=yellow!10,opacity=.3] (0,0) ellipse (2cm and .5cm);
    \end{tikzpicture}
    \caption{Examples of 2-dimensional defects sitting inside the 3-ball.}
    \label{fig:3-ball_defects}
\end{figure}
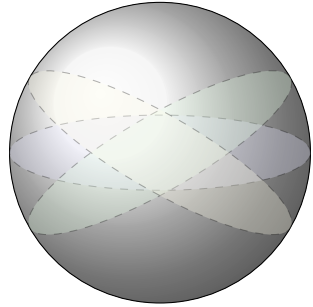
\end{ex}

\section{AKSZ formulation of DW theory}
\label{sec:AKSZ_formulation_of_DW_theory}
In this section we want to describe a way of obtaining the DW action functional as in Section \ref{subsec:field_theory_formulation} in terms of an AKSZ theory by choosing a suitable gauge-fixing Lagrangian submanifold $\calL$ \cite{BaulieuSinger1989,Ikeda2011,CattZabz2019}. 

\subsection{AKSZ data}
\label{subsec:AKSZ_data}
Let $\mathfrak{h}$ be a Lie algebra and consider its \emph{Weil model}\footnote{For a Lie algebra $\mathfrak{h}$, the Weil model is given by the algebra $\calW(\mathfrak{h}):=\bigwedge^\bullet(\mathfrak{h^*}\oplus\mathfrak{h}^*[1])$ together with a differential $\dd_\calW\colon \calW(\mathfrak{h})\to \calW(\mathfrak{h})$ such that $(\calW(\mathfrak{h}),\dd_\calW)$ is a differential graded algebra. The differential is defined to act on $\mathfrak{h}^*$ as the differential for the Chevalley--Eilenberg algebra of $\mathfrak{h}$ plus the degree shift differential. In particular $\dd_\calW=\dd_{\mathrm{CE}(\mathfrak{h})}+\dd_s$, where $\dd_s$ acts by degree shift $\mathfrak{h}^*\to \mathfrak{h}^*[1]$ for elements in $\mathfrak{h}$ and by $0$ for elements in $\mathfrak{h}^*[1]$.} $(\calW(\mathfrak{h}),\dd_\calW)$. 
Note that $\calW(\mathfrak{h})$ can be endowed with a natural symplectic form of degree $+3$ if $\mathfrak{h}$ is endowed with an invariant, nondegenerate symmetric pairing $\langle\enspace,\enspace\rangle$. We want to consider this symplectic manifold to be the target for a 4-dimensional AKSZ theory. The graded vector space $\mathfrak{h}[1]\oplus\mathfrak{h}[2]$ is endowed with the symplectic structure 
\[
\omega=\langle\delta X,\delta Y\rangle, 
\]
where $X$ denotes the coordinate of degree $+1$ and $Y$ denotes the coordinate of degree $+2$. We can define a Hamiltonian of degree $+4$ by
\[
\Theta(X,Y)=\frac{1}{2}\langle Y,Y\rangle+\frac{1}{2}\langle Y,[X,X]\rangle.
\]
One can check that the Hamiltonian vector field of $\Theta$ with respect to $\omega$ is given by the Weil differential $\dd_\calW$. We have
\begin{equation}
\label{eq:Weil_differential_relation}
\dd_\calW X=Y+\frac{1}{2}[X,X],\qquad \dd_\calW Y=[X,Y].
\end{equation}
Let now $\Sigma$ be a 4-manifold and consider the BV space of fields 
\[
\calF_\Sigma=\Map(T[1]\Sigma,\mathfrak{h}[1]\oplus \mathfrak{h}[2])\cong \Omega^\bullet(\Sigma)\otimes\mathfrak{h}[1]\oplus \Omega^\bullet(\Sigma)\otimes \mathfrak{h}[2]. 
\]
The superfields are given by 
\begin{align}
\mathbf{X}(u,\theta)&=X+A\theta^{i_1}+\chi\theta^{i_1}\theta^{i_2}+\psi^\dagger\theta^{i_1}\theta^{i_2}\theta^{i_3}+Y^\dagger\theta^{i_1}\theta^{i_2}\theta^{i_3}\theta^{i_4},\\
\mathbf{Y}(u,\theta)&=Y+\psi\theta^{i_1}+\chi^\dagger\theta^{i_1}\theta^{i_2}+A^\dagger\theta^{i_1}\theta^{i_2}\theta^{i_3}+X^\dagger\theta^{i_1}\theta^{i_2}\theta^{i_3}\theta^{i_4}.
\end{align}
for local coordinates $(u_i,\theta^i)$ on $T[1]\Sigma$.
The BV symplectic form is given by 
\[
\omega_\Sigma=\int_{T[1]\Sigma}\boldsymbol{\mu}_4\langle \delta \mathbf{X},\delta\mathbf{Y}\rangle,
\]
and the AKSZ-BV action is given by 
\[
\calS_\Sigma=\int_{T[1]\Sigma}\boldsymbol{\mu}_4\left(\langle\mathbf{Y},\dd_\Sigma\mathbf{X}\rangle+\frac{1}{2}\langle\mathbf{Y},\mathbf{Y}\rangle+\frac{1}{2}\langle \mathbf{Y},[\mathbf{X},\mathbf{X}]\rangle\right).
\]
We have denoted by $\boldsymbol{\mu}_4:=\dd^4u\dd^4\theta$ the supermeasure on $T[1]\Sigma$.
If we expand things into components and perform Berezinian integration, we get
\begin{equation}
\omega_\Sigma=\int_\Sigma\dd^4u \left(\delta X \delta X^\dagger+\delta A\delta A^\dagger+\delta\chi\delta\chi^\dagger+\delta\psi^\dagger\delta\psi+\delta Y^\dagger\delta Y\right)
\end{equation}
and
\begin{multline}
\calS_\Sigma= \int_\Sigma\dd^4u\bigg(\langle \psi,\dd_A\chi\rangle+\frac{1}{2}\langle Y,[\chi,\chi]\rangle+\langle \psi^\dagger,(\dd_A Y+[X,\psi])\rangle\\+\langle \chi^\dagger,(F_A+[X,\chi])\rangle+\langle A^\dagger,(\psi+\dd_AX)\rangle+\langle Y^\dagger,[X,Y]\rangle\\+\left\langle X^\dagger,\left(\frac{1}{2}[X,X]\right)\right\rangle+\frac{1}{2}\langle \chi^\dagger,\chi^\dagger\rangle\bigg),
\end{multline}
where $\dd_A=\dd_\Sigma+[A,\enspace]$ denotes the covariant derivative of $A$ and $F_A=\dd_\Sigma A+\frac{1}{2}[A,A]$. The BV transformations on the superfields are then given by 
\begin{align}
    Q_\Sigma \mathbf{X}&=\dd_\Sigma\mathbf{X}+\mathbf{Y}+\frac{1}{2}[\mathbf{X},\mathbf{X}],\\
    Q_\Sigma\mathbf{Y}&=\dd_\Sigma\mathbf{Y}+[\mathbf{X},\mathbf{Y}],
\end{align}
i.e., in superfield notation, the cohomological vector field is given by 
\begin{equation}
    Q_\Sigma=\int_{T[1]\Sigma}\boldsymbol{\mu}_4\Bigg(\left(\dd_{\Sigma}\mathbf{X}+\mathbf{Y}+\frac{1}{2}[\mathbf{X},\mathbf{X}]\right)\frac{\delta}{\delta\mathbf{X}}+\big(\dd_{\Sigma}\mathbf{Y}+[\mathbf{X},\mathbf{Y}]\big)\frac{\delta}{\delta\mathbf{Y}}\Bigg)
\end{equation}
In component fields and after Berezinian integration, we get the cohomological vector field
\begin{multline}
    Q_\Sigma=\int_\Sigma\dd^4u\Bigg(\left(Y+\frac{1}{2}[X,X]\right)\frac{\delta}{\delta X}+(\psi+\dd_AX)\frac{\delta}{\delta A}+(\chi^\dagger+F_A+[X,\chi])\frac{\delta}{\delta \chi}\\
    +(\dd_A\chi+A^\dagger+[X,\psi^\dagger])\frac{\delta}{\delta \psi^\dagger}+(\dd_A\psi^\dagger+X^\dagger+[X,Y^\dagger])\frac{\delta}{\delta Y^\dagger}+[X,Y]\frac{\delta}{\delta Y}\\
    +(\dd_AY+[X,\psi])\frac{\delta}{\delta \psi}+(\dd_A\psi+[X,\chi^\dagger]+[\chi,Y])\frac{\delta}{\delta \chi^\dagger}\\
    +(\dd_A\chi^\dagger+[X,A^\dagger]+[\psi^\dagger,Y]+[\chi.\psi])\frac{\delta}{\delta A^\dagger}\\
    +(\dd_AA^\dagger+[X,X^\dagger]+[Y^\dagger,Y]+[\psi^\dagger,\psi]+[\chi,\chi^\dagger])\frac{\delta}{\delta X^\dagger}\Bigg),
\end{multline}
where we have used the relation in \eqref{eq:Weil_differential_relation}.

\subsection{A suitable gauge-fixing}
\label{subsec:suitable_gauge_fixing}
Consider a Riemannian metric on $\Sigma$ such that we can take a splitting of the fields $\chi,\chi^\dagger$ into self-dual and anti self-dual parts
\begin{align}
    \chi&=\chi^++\chi^-,\\
    \chi^\dagger&=\chi^{+\dagger}+\chi^{-\dagger}.
\end{align}
We can take the gauge-fixing where $\chi^-=\chi^{+\dagger}=0$ and all other forms are coexact when using Hodge decomposition. This will imply that $X^\dagger=Y^\dagger=0$. Introduce new fields $\bar X$ and $b$ with $\deg(\bar X)=\deg(b)=0$ and $\gh(\bar X)=-1$ and $\gh(b)=0$ together with corresponding anti-fields $\bar X^\dagger$ and $b^\dagger$ with $\deg(\bar X^\dagger)=\deg(b^\dagger)=4$, such that $\gh(\bar X^\dagger)=0$ and $\gh(b^\dagger)=-1$, respectively. Moreover, we consider the term 
\[
\calS^{(1)}_\Sigma=\int_\Sigma\left(\langle \bar X^\dagger,(b+[X,\bar X])\rangle+\langle b^\dagger,([X,b]-[Y,\bar X])\rangle\right).
\]
Additionally, introduce more new fields $(\bar Y,\eta)$ with $\deg(Y)=\deg(\eta)=0$ and $\gh(\bar Y)=-2$ and $\gh(\eta)=-1$ together with corresponding anti-fields $\bar Y^\dagger$ and $\eta^\dagger$ with $\deg(\bar Y^\dagger)=\deg(\eta^\dagger)=4$, such that $\gh(\bar Y^\dagger)=+1$ and $\gh(\eta^\dagger)=0$. We consider the term 
\[
\calS^{(2)}_\Sigma=\int_\Sigma\left(\langle \bar Y^\dagger,(\eta+[X,\bar Y]\rangle+\langle \eta^\dagger,([X,\eta]+[\bar Y,Y])\rangle\right).
\]
The gauge-fixed BV action is then obtained by 
\[
\calS^\mathrm{gf}_\Sigma=\calS_\Sigma+\calS^{(1)}_\Sigma+\calS^{(2)}_\Sigma
\]
with gauge-fixing fermion
\begin{equation}
\label{eq:gauge-fixing_fermion}
\varPsi^\mathrm{gf}=\int_\Sigma\langle \bar X,\dd *A\rangle+\int_\Sigma\langle \bar Y,\dd*\psi\rangle.
\end{equation}
such that $\calS^\mathrm{gf}_\Sigma=\calS_\Sigma+Q_\Sigma\varPsi^\mathrm{gf}$.
The auxiliary fields $\chi^{-\dagger}$ have ghost number zero and can be actually integrated out. One can check that the gauge-transformations for each field squares to zero except the one for $\chi^+$. Namely, we have
\[
\delta^2\chi^+=(\dd_A\psi)^++[Y,\chi^+],
\]
which is the equation of motion for this setting. Hence, the tuple
\[
(A,Y,\psi,\chi^+,\bar Y,\eta)
\]
is the same as in the construction of Section \ref{subsec:field_theory_formulation}.

\subsection{BV-BFV formulation}
As an AKSZ theory, there a nice way of treating DW theory as a gauge-theory on 4-manifolds with boundary. Let $\Sigma$ be a 4-manifold with 3-dimensional boundary $\de\Sigma$. As obtained in Section \ref{subsec:AKSZ_data}, the BV theory of the AKSZ construction of DW theory is given by 
\[
(\calF_\Sigma,\omega_\Sigma,\calS_\Sigma).
\]
One can then easily formulate a BFV theory on the boundary $\de\Sigma$ by setting
\begin{align}
    \calF^\de_{\de\Sigma}&:=\Map(T[1]\de\Sigma,\mathfrak{h}[1]\oplus\mathfrak{h}[2])\cong \Omega^\bullet(\de\Sigma)\otimes\mathfrak{h}[1]\oplus\Omega^\bullet(\de\Sigma)\otimes \mathfrak{h}[2],\\
    \omega^\de_{\de\Sigma}&:=\int_{T[1]\de\Sigma}\bbmu_3\langle\delta\mathds{X},\delta\mathds{Y}\rangle,\\
    \calS^\de_{\de\Sigma}&:=\int_{T[1]\de\Sigma}\bbmu_3\left(\langle\mathds{Y},\dd_{\de\Sigma}\mathds{X}\rangle+\frac{1}{2}\langle \mathds{Y},\mathds{Y}\rangle+\frac{1}{2}\langle\mathds{Y},[\mathds{X},\mathds{X}]\rangle\right),
\end{align}
where $\bbmu_3:=\dd^3\mathbb{u}\dd^3\bbtheta$ denotes the supermeasure on $T[1]\de\Sigma$ for local coordinates $(\mathbb{u}_i,\bbtheta^i)$ on $T[1]\de\Sigma$.
The boundary superfields are given by 
\begin{align}
\mathds{X}(\mathbb{u},\bbtheta)&=\mathbb{X}+\mathbb{A}\bbtheta^{i_1}+\bbchi\bbtheta^{i_1}\bbtheta^{i_2}+\bbpsi^\dagger\bbtheta^{i_1}\bbtheta^{i_2}\bbtheta^{i_3}+\mathbb{Y}^\dagger\bbtheta^{i_1}\bbtheta^{i_2}\bbtheta^{i_3}\bbtheta^{i_4},\\
\mathds{Y}(\mathbb{u},\bbtheta)&=\mathbb{Y}+\bbpsi\bbtheta^{i_1}+\bbchi^\dagger\bbtheta^{i_1}\bbtheta^{i_2}+\mathbb{A}^\dagger\bbtheta^{i_1}\bbtheta^{i_2}\bbtheta^{i_3}+\mathbb{X}^\dagger\bbtheta^{i_1}\bbtheta^{i_2}\bbtheta^{i_3}\bbtheta^{i_4}.
\end{align}
The cohomological vector field $Q^\de_{\de\Sigma}$ is given by the Hamiltonian vector field of $\calS^\de_{\de\Sigma}$:
\[
\iota_{Q^\de_{\de\Sigma}}\omega^\de_{\de\Sigma}=\delta\calS^\de_{\de\Sigma}.
\]
It is easy to see that $\omega^\de_{\de\Sigma}$ is exact with primitive 1-form 
\[
\alpha^\de_{\de\Sigma}=\int_{T[1]\de\Sigma}\bbmu_3\langle\mathds{Y},\delta\mathds{X}\rangle.
\]
If we expand things into components and perform Berezinian integration, we get
\begin{equation}
\omega^\de_{\de\Sigma}=\int_{\de\Sigma}\dd^3\mathbb{u} \left(\delta \mathbb{X} \delta \mathbb{X}^\dagger+\delta \mathbb{A}\delta \mathbb{A}^\dagger+\delta\bbchi\delta\bbchi^\dagger+\delta\bbpsi^\dagger\delta\bbpsi+\delta \mathbb{Y}^\dagger\delta \mathbb{Y}\right)
\end{equation}
and
\begin{multline}
\calS^\de_{\de\Sigma}= \int_{\de\Sigma}\dd^3\mathbb{u}\bigg(\langle \bbpsi,\dd_\mathbb{A}\bbchi\rangle+\frac{1}{2}\langle \mathbb{Y},[\bbchi,\bbchi]\rangle+\langle \bbpsi^\dagger,(\dd_\mathbb{A} \mathbb{Y}+[\mathbb{X},\bbpsi])\rangle\\+\langle \bbchi^\dagger,(\mathbb{F}_\mathbb{A}+[\mathbb{X},\bbchi])\rangle+\langle \mathbb{A}^\dagger,(\bbpsi+\dd_\mathbb{A}\mathbb{X})\rangle+\langle \mathbb{Y}^\dagger,[\mathbb{X},\mathbb{Y}]\rangle\\+\left\langle \mathbb{X}^\dagger,\left(\frac{1}{2}[\mathbb{X},\mathbb{X}]\right)\right\rangle+\frac{1}{2}\langle \bbchi^\dagger,\bbchi^\dagger\rangle\bigg),
\end{multline}
where $\dd_\mathbb{A}=\dd_{\de\Sigma}+[\mathbb{A},\enspace]$ denotes the covariant derivative of $\mathbb{A}$ and $\mathbb{F}_\mathbb{A}=\dd_{\de\Sigma} \mathbb{A}+\frac{1}{2}[\mathbb{A},\mathbb{A}]$. The BV transformations on the superfields are then given by 
\begin{align}
    Q^\de_{\de\Sigma} \mathds{X}&=\dd_{\de\Sigma}\mathds{X}+\mathds{Y}+\frac{1}{2}[\mathds{X},\mathds{X}],\\
    Q^\de_{\de\Sigma}\mathds{Y}&=\dd_{\de\Sigma}\mathds{Y}+[\mathds{X},\mathds{Y}],
\end{align}
i.e., in superfield notation, the cohomological vector field is given by 
\begin{equation}
    Q^\de_{\de\Sigma}=\int_{T[1]\de\Sigma}\bbmu_3\Bigg(\left(\dd_{\de\Sigma}\mathds{X}+\mathds{Y}+\frac{1}{2}[\mathds{X},\mathds{X}]\right)\frac{\delta}{\delta\mathds{X}}+\big(\dd_{\de\Sigma}\mathds{Y}+[\mathds{X},\mathds{Y}]\big)\frac{\delta}{\delta\mathds{Y}}\Bigg)
\end{equation}
In component fields and after Berezinian integration, we get the cohomological vector field
\begin{multline}
    Q^\de_{\de\Sigma}=\int_{\de\Sigma}\dd^3\mathbb{u}\Bigg(\left(\mathbb{Y}+\frac{1}{2}[\mathbb{X},\mathbb{X}]\right)\frac{\delta}{\delta \mathbb{X}}+(\bbpsi+\dd_\mathbb{A}\mathbb{X})\frac{\delta}{\delta \mathbb{A}}+(\bbchi^\dagger+\mathbb{F}_\mathbb{A}+[\mathbb{X},\bbchi])\frac{\delta}{\delta \bbchi}\\
    +(\dd_\mathbb{A}\bbchi+\mathbb{A}^\dagger+[\mathbb{X},\bbpsi^\dagger])\frac{\delta}{\delta \bbpsi^\dagger}+(\dd_\mathbb{A}\bbpsi^\dagger+\mathbb{X}^\dagger+[\mathbb{X},\mathbb{Y}^\dagger])\frac{\delta}{\delta \mathbb{Y}^\dagger}+[\mathbb{X},\mathbb{Y}]\frac{\delta}{\delta \mathbb{Y}}\\
    +(\dd_\mathbb{A}\mathbb{Y}+[\mathbb{X},\bbpsi])\frac{\delta}{\delta \bbpsi}+(\dd_\mathbb{A}\bbpsi+[\mathbb{X},\bbchi^\dagger]+[\bbchi,\mathbb{Y}])\frac{\delta}{\delta \bbchi^\dagger}\\
    +(\dd_\mathbb{A}\bbchi^\dagger+[\mathbb{X},\mathbb{A}^\dagger]+[\bbpsi^\dagger,\mathbb{Y}]+[\bbchi,\bbpsi])\frac{\delta}{\delta \mathbb{A}^\dagger}\\
    +(\dd_\mathbb{A}\mathbb{A}^\dagger+[\mathbb{X},\mathbb{X}^\dagger]+[\mathbb{Y}^\dagger,\mathbb{Y}]+[\bbpsi^\dagger,\bbpsi]+[\bbchi,\bbchi^\dagger])\frac{\delta}{\delta \mathbb{X}^\dagger}\Bigg).
\end{multline}
By setting $\pi_\Sigma\colon \calF_\Sigma\to \calF^\de_{\de\Sigma}$ to be the restriction of the fields to the boundary, we get that the modified CME 
\[
Q_\Sigma\calS_\Sigma=\pi^*_\Sigma(2\calS^\de_{\de\Sigma}-\iota_{Q^\de_{\de\Sigma}}\alpha^\de_{\de\Sigma})
\]
is satisfied and that the cohomological vector field $Q_\Sigma$ is indeed projectable, i.e. we have  
\[
\delta\pi_\Sigma Q_\Sigma=Q^\de_{\de\Sigma}.
\]

\section{Quantization of DW theory in the BV-BFV setting}
\label{sec:quantization_of_DW_theory_in_the_BV-BFV_setting}
\subsection{The BFV boundary operator}
Assume now that we can write the boundary as a disjoint union $\de\Sigma=\de_1\Sigma\sqcup\de_2\Sigma$ where $\de_1\Sigma$ and $\de_2\Sigma$ have opposite orientation. We choose the convenient polarization $\calP$ on $\de\Sigma$ consisting of choosing the $\frac{\delta}{\delta \mathds{Y}}$-polarization ($\mathds{X}$-representation) on $\de_1\Sigma$ and the $\frac{\delta}{\delta\mathds{X}}$-polarization ($\mathds{Y}$-representation) on $\de_2\Sigma$ (see Figure \ref{fig:boundary_polarization_BF}).
As we have seen in Section \ref{subsec:BV-BFV_formalism}, we can compute the BFV boundary operator $\Omega^\calP_{\de\Sigma}$ as the ordered standard quantization of the boundary action $\calS^\de_{\de\Sigma}$ with respect to the chosen polarization. 
We get 
\[
\Omega^\calP_{\de\Sigma}=\Omega^\calP_0+\Omega^\calP_\mathrm{pert},
\]
where 
\begin{equation}
\label{eq:boundary_operator_0}
\Omega^\calP_0:=\underbrace{\I\hbar\int_{\de_1\Sigma}\sum_i\dd_{\de_1\Sigma}\mathds{X}^i\frac{\delta}{\delta\mathds{X}^i}}_{=:\Omega_0^\mathds{X}}+\underbrace{\I\hbar\int_{\de_2\Sigma}\sum_i\dd_{\de_2\Sigma}\mathds{Y}_i\frac{\delta}{\delta\mathds{Y}_i}}_{=:\Omega^\mathds{Y}_0},
\end{equation}
and 
\begin{equation}
\label{eq:boundary_operator_pert}
\Omega^\calP_\mathrm{pert}=\Omega^\mathds{X}_\mathrm{pert}+\Omega^\mathds{Y}_\mathrm{pert},
\end{equation}
with
\begin{equation}
\label{eq:boundary_operator_pert_X}
\Omega^\mathds{X}_\mathrm{pert}:=\int_{\de_1\Sigma}\left(-\frac{\hbar^2}{2}\sum_{i,j}\frac{\delta^2}{\delta\mathds{X}^i\delta\mathds{X}^j}+\frac{\I\hbar}{2}\sum_{i,j,k}c_{ij}^k\mathds{X}^i\mathds{X}^j\frac{\delta}{\delta\mathds{X}^k}\right),
\end{equation}
and 
\begin{equation}
\label{eq:boundary_operator_pert_Y}
\Omega^\mathds{Y}_\mathrm{pert}:=\int_{\de_2\Sigma}\left(\frac{1}{2}\sum_{i,j}\delta^{i}_j\mathds{Y}_i\mathds{Y}_j-\frac{\hbar^2}{2}\sum_{i,j,k}c_{ij}^k \mathds{Y}_k\frac{\delta^2}{\delta\mathds{Y}_i\delta\mathds{Y}_j}\right),
\end{equation}
where $\delta^{i}_j$ denotes the Kronecker delta and $c_{ij}^k$ the structure constants of the Lie bracket $[\enspace,\enspace]$.

\subsection{Composite fields}
In order to regularize the higher functional derivatives, one needs to introduce the concept of \emph{composite fields} as in \cite{CMR2}. In particular, we want to regard the higher functional derivates as a single first order derivative with respect to the composite field. In order to get it coherent to the naive interpretation where a higher functional derivative concentrates the fields on some diagonal, we should also understand the product of integrals as containing the diagonal contributions for the corresponding composite field. We consider the following bullet product of integrals
\begin{multline}
\label{eq:bullet_product}
\left(\int_{\de_k\Sigma}\sum_i\alpha_i\Phi^i\right)\bullet\left(\int_{\de_k\Sigma}\sum_j\beta_j\Phi^j\right):=\\
(-1)^{\mathrm{gh}(\Phi^i)(\mathrm{gh}(\beta_j)-3)+3\mathrm{gh}(\alpha_i)}\left(\int_{\overline{\mathrm{Conf}_2(\de_k\Sigma)}}\sum_{i,j}\pi_1^*\alpha_i\pi_2^*\beta_j\pi^*_1\Phi^i\pi^*_2\Phi^j+\int_{\de_k\Sigma}\sum_{i,j}\alpha_i\beta_j[\Phi^i\Phi^j]\right),
\end{multline}
where $\alpha$ and $\beta$ are smooth forms which depend on the bulk and residual fields and we denote by $[\Phi^i\Phi^j]$ the composite field. We can now interpret the operator $\int_{\de_k\Sigma}\phi^{ij}\frac{\delta^2}{\delta\Phi^i\delta\Phi^j}$ as $\int_{\de_k\Sigma}\phi^{ij}\frac{\delta}{\delta[\Phi^i\Phi^j]}$. Hence, we get 
\[
\int_{\de_k\Sigma}\phi^{ij}\frac{\delta^2}{\delta\Phi^i\delta\Phi^j}\left(\left(\int_{\de_k\Sigma}\alpha_i\Phi^i\right)\bullet\left(\int_{\de_k\Sigma}\beta_j\Phi^j\right)\right)=\int_{\de_k\Sigma}\alpha_i\beta_j\phi^{ij}.
\]
More general, let $I:=(i_1,\ldots,i_n)$ be a multi-index and note that we can replace the higher derivative $\frac{\delta^n}{\delta\Phi^{i_1}\dotsm \delta\Phi^{i_n}}$ by the first order derivative $\frac{\delta}{\delta[\Phi^I]}:=\frac{\delta}{\delta[\Phi^{i_1}\dotsm \Phi^{i_n}]}$ with respect to the composite field $[\Phi^I]:=[\Phi^{i_1}\dotsm \Phi^{i_n}]$. The higher functional derivatives appearing in the BFV boundary operator are then regularized by using composite fields and it acts on the state space given by \emph{regular functionals}, i.e. on the algebra which is generated by linear combinations of expressions of the form 
\begin{multline}
\label{eq:condition_1}
\int_{\overline{\mathrm{Conf}_{m_1}(\de_1\Sigma)}\times\overline{\mathrm{Conf}_{m_2}(\de_2\Sigma)}}\mathsf{Z}^{J_1^1\dotsm J_1^{\ell_1}J_2^1\dotsm J^{\ell_2}\dotsm}_{I_1^1\dotsm I_1^{r_1}I_2^{1}\dotsm I_2^{r_2}\dotsm}\pi_1^*\prod_{j=1}^{r_1}\left[\mathds{X}^{I_1^j}\right]\dotsm \pi_{m_1}^*\prod_{j=1}^{r_{m_1}}\left[\mathds{X}^{I^j_{m_1}}\right]\times\\
\times \pi^*_1\prod_{j=1}^{\ell_1}\left[\mathds{Y}_{J^j_1}\right]\dotsm\pi^*_{m_2}\prod_{j=1}^{\ell_{m_2}}\left[\mathds{Y}_{J^j_{m_2}}\right],
\end{multline}
where $I^j_i$ and $J^j_i$ are multi-indices and $\mathsf{Z}^{J_1^1\dotsm J_1^{\ell_1}J_2^1\dotsm J^{\ell_2}\dotsm}_{I_1^1\dotsm I_1^{r_1}I_2^{1}\dotsm I_2^{r_2}\dotsm}$ is a smooth form on the product $\overline{\mathrm{Conf}_{m_1}(\de_1\Sigma)}\times\overline{\mathrm{Conf}_{m_2}(\de_2\Sigma)}$. If we denote by $\exp_\bullet$ the exponential defined through the bullet product $\bullet$ constructed as in \eqref{eq:bullet_product}, we can see that 
\[
\int_{\de_k\Sigma}\phi^I\frac{\delta}{\delta\Phi^I}\Big\langle\exp_\bullet\left(\calS^\mathrm{res}_\Sigma+\calS^\mathrm{source}_\Sigma\right)\Big\rangle=\left\langle\int_{\de_k\Sigma}\phi^I\frac{\delta}{\delta\Phi^I}\exp_\bullet\left(\calS^\mathrm{res}_\Sigma+\calS^\mathrm{source}_\Sigma\right)\right\rangle,
\]
where 
\begin{align*}
    \calS^\mathrm{res}_\Sigma&:=\int_{\de_1\Sigma}\langle \mathsf{y},\mathds{X}\rangle-\int_{\de_2\Sigma}\langle\mathds{Y},\mathsf{x}\rangle,\\
    \calS^\mathrm{source}_\Sigma&:=\int_{\de_1\Sigma}\langle\mathscr{Y},\mathds{X}\rangle-\int_{\de_2\Sigma}\langle \mathds{Y},\mathscr{X}\rangle,
\end{align*}
The \emph{full boundary state} is then defined through the regularization of composite fields by using \eqref{eq:condition_1}:
\begin{equation}
    \label{eq:full_state}
    \mathsf{Z}^{\bullet,\scriptscriptstyle\mathrm{BV-BFV}}_{\Sigma,\de\Sigma}:=\mathsf{Z}^{\scriptscriptstyle\mathrm{BV-BFV}}_{\Sigma,\de\Sigma}\Big\langle\exp_\bullet\left(\calS^\mathrm{res}_\Sigma+\calS^\mathrm{source}_\Sigma\right)\Big\rangle,
\end{equation}
where $\big\langle\enspace\big\rangle$ denotes the expectation value with respect to $\calS_\Sigma$, i.e. for an observable $O\in\calO_\calF$, we have
\[
\big\langle O\big\rangle:=\int\exp(\I\calS_\Sigma(\mathbf{X},\mathbf{Y})/\hbar)O(\mathbf{X},\mathbf{Y})\mathscr{D}[\mathbf{X}]\mathscr{D}[\mathbf{Y}].
\]
Let $I_1,\ldots, I_{m_1}$ and $J_1,\ldots,J_{m_2}$ be multi-indices for $m_1,m_2\geq 1$.
Denote by $\Gamma_k'$ for $k=1,2$ the graphs with $m_1$ vertices on $\de_k\Sigma$, where vertex $s$ has valency $\vert I_s\vert\geq 1$, with adjacent half-edges oriented inwards and decorated with boundary fields $\left[\Phi^{I_1}_k\right],\ldots,\left[\Phi^{I_{m_1}}_k\right]$ all evaluated at the point of collapse $\mathbb{u}\in \de_k\Sigma$. We also want them to have $\vert J_1\vert+\dotsm +\vert J_{m_2}\vert$ outward leaves if $k=1$ and $\vert J_1\vert+\dotsm+\vert J_{m_2}\vert$ inward leaves if $k=2$, decorated with functional derivatives with respect to the boundary fields:
\[
\I\hbar\frac{\delta}{\delta\left[\Phi^{J_1}_k\right]},\ldots,\I\hbar\frac{\delta}{\delta\left[\Phi^{J_{m_2}}_k\right]}
\]
evaluated at the point of collapse $\mathbb{u}\in \de_k\Sigma$. Moreover, there are no outward leaves if $k=2$ and no inward leaves if $k=1$. Denote by $\sigma_{\Gamma_i'}$ the differential form given by 
\[
\sigma_{\Gamma_k'}:=\int_{\overline{\mathrm{C}_{\Gamma_k'}(\mathbb{H}^4)}}\omega_{\Gamma_k'},
\]
where $\omega_{\Gamma_k'}$ denotes the product of limiting propagators at the point of collapse $\mathbb{u}\in\de_k\Sigma$ and vertex tensors.
Then we can construct the corresponding \emph{full BFV boundary operator} through the regularization of composite fields as 
\begin{equation}
\label{eq:full_boundary_operator}
    \Omega^{\bullet,\calP}_{\de\Sigma}=\Omega^\calP_0+\Omega^{\bullet,\mathds{X}}_\mathrm{pert}+\Omega^{\bullet,\mathds{Y}}_\mathrm{pert},
\end{equation}
where 
\begin{align}
    \label{eq:principal_part_X}
    \Omega^{\bullet,\mathds{X}}_\mathrm{pert}&:=\sum_{m_1,m_2\geq 0}\,\sum_{\Gamma_1'}\frac{(-\I\hbar)^{\ell(\Gamma_1')}}{\vert\mathrm{Aut}(\Gamma_1')\vert}\int_{\de_1\Sigma}\left(\sigma_{\Gamma_1'}\right)^{J_1\dotsm J_{m_2}}_{I_1\dotsm I_{m_1}}\prod_{j=1}^{m_1}\left[\mathds{X}^{I_j}\right]\left((-1)^{4m_2}(\I\hbar)^{m_2}\frac{\delta^{\vert J_1\vert+\dotsm+\vert J_{m_2}\vert}}{\delta\left[\mathds{X}^{J_1}\dotsm \mathds{X}^{J_{m_2}}\right]}\right)\\
    \label{eq:principal_part_Y}
    \Omega^{\bullet,\mathds{Y}}_\mathrm{pert}&:=\sum_{m_1,m_2\geq 0}\,\sum_{\Gamma_2'}\frac{(-\I\hbar)^{\ell(\Gamma_2')}}{\vert\mathrm{Aut}(\Gamma_2')\vert}\int_{\de_2\Sigma}\left(\sigma_{\Gamma_2'}\right)^{J_1\dotsm J_{m_2}}_{I_1\dotsm I_{m_1}}\prod_{j=1}^{m_1}\left[\mathds{Y}_{I_j}\right]\left((-1)^{4m_2}(\I\hbar)^{m_2}\frac{\delta^{\vert J_1\vert+\dotsm+\vert J_{m_2}\vert}}{\delta\left[\mathds{Y}_{J_1}\dotsm \mathds{Y}_{J_{m_2}}\right]}\right)
\end{align}

\subsection{Feynman rules}
\label{subsec:Feynman_rules}
In order to formulate a BV-BFV quantization, i.e. for the boundary state and the BFV boundary operator, we need to consider the Feynman graphs on the source manifold $\Sigma$ (see Figure \ref{fig:Feynman_graphs_manifold} for an example) for DW theory. The graphs are determined by the Feynman rules of the theory and the corresponding degree count (see Section \ref{subsec:Degree_count_and_Feynman_graphs}).
The Feynman rules are determined by the AKSZ action functional $\calS_\Sigma$ of DW theory. In particular, the interaction vertices are given by the ones as in Figure \ref{fig:Feynman_rules_1}. The boundary vertices are given by Figure \ref{fig:Feynman_rules_2} and \ref{fig:Feynman_rules_3}.

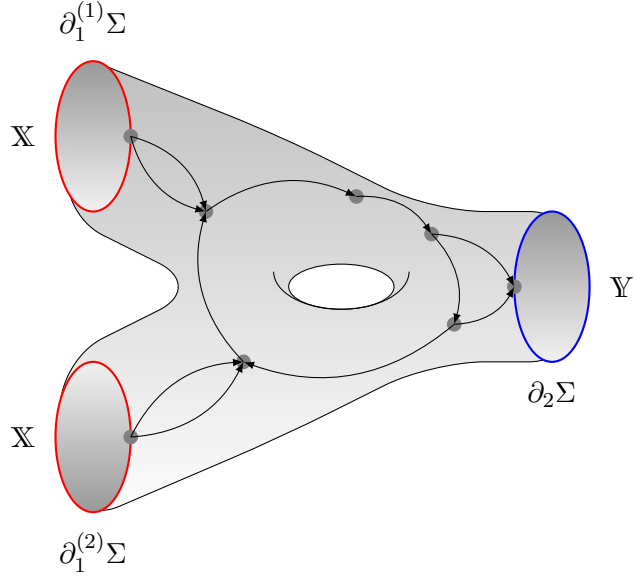
\begin{figure}[h!]
    \centering
    \begin{tikzpicture}
    \tikzset{Bullet/.style={fill=blue,draw,color=#1,circle,minimum size=3pt,scale=0.5}}
    \shadedraw[rounded corners=20pt,top color=gray!50, bottom color=gray!.5](5.5,1)--(8,0)--(10,-1)--(12,-1)--(12,-3)--(10,-3)--(8,-4)--(5.2,-5.2)--(5,-3)--(7,-2)--(5,-1)--(5.5,1);
    \draw[fill=white] (9.5,-2) arc (0:360:0.7cm and 0.3cm);
    \draw (8.8,-2.3) arc (-90:0:0.9cm and 0.5cm);
    \draw (8.8,-2.3) arc (-90:-180:0.9cm and 0.5cm);
    \shadedraw[color=white, top color=gray!80, bottom color=gray!10.5] (6,0) arc (0:360:0.5cm and 1cm);
    \draw[thick, color=red] (6,0) arc (0:360:0.5cm and 1cm);
    \shadedraw[color=white, bottom color=gray!80, top color=gray!10.5] (6,-4) arc (0:360:0.5cm and 1cm);
    \draw[thick, color=red] (6,-4) arc (0:360:0.5cm and 1cm);
    \shadedraw[color=white, top color=gray!80, bottom color=gray!10.5] (12.1,-2) arc (0:360:0.5cm and 1cm);
    \draw[thick, color=blue] (12.1,-2) arc (0:360:0.5cm and 1cm);
    \node[label=left:{$\mathds{X}$}] at (5,0){};
    \node[label=left:{$\mathds{X}$}] at (5,-4){};
    \node[label=right:{$\mathds{Y}$}] at (12.1,-2){};
    \node[label=above:{$\de_1^{(1)}\Sigma$}] at (5.5,1){};
    \node[label=below:{$\de_1^{(2)}\Sigma$}] at (5.5,-5){};
    \node[label=below:{$\de_2\Sigma$}] at (11.6,-3){};
    \node[Bullet=gray,label=above:{}] at (7,-1){};
    \node[Bullet=gray,label=left:{}] at (9,-0.8){};
    \node[Bullet=gray,label=right:{}] at (10,-1.3){};
    \node[Bullet=gray,label=right:{}] at (7.5,-3){};
    \node[Bullet=gray,label=right:{}] at (10.3,-2.5){};
    \node[Bullet=gray,label=right:{}] at (6,0){};
    \node[Bullet=gray,label=right:{}] at (6,-4){};
    \node[Bullet=gray,label=right:{}] at (11.1,-2){};
    \draw[-{latex[scale=3.0]}] (6,0) to[bend left] (7,-1);
    \draw[-{latex[scale=3.0]}] (6,0) to[bend right] (7,-1);
    \draw[-{latex[scale=3.0]}] (7,-1) to[bend left] (9,-0.8);  
    \draw[-{latex[scale=3.0]}] (9,-0.8) to[bend left] (10,-1.3);  
    \draw[-{latex[scale=3.0]}] (10,-1.3) to[bend left] (10.3,-2.5);
    \draw[-{latex[scale=3.0]}] (10.3,-2.5) to[bend left] (7.5,-3);
    \draw[-{latex[scale=3.0]}] (7.5,-3) to[bend left] (7,-1); 
    \draw[-{latex[scale=3.0]}] (6,-4) to[bend right] (7.5,-3); 
  \draw[-{latex[scale=3.0]}] (6,-4) to[bend left] (7.5,-3); 
    \draw[-{latex[scale=3.0]}] (10,-1.3) to[bend left] (11.1,-2);  
    \draw[-{latex[scale=3.0]}] (10.3,-2.5) to[bend right] (11.1,-2);  
    \end{tikzpicture}
    \caption{Example of a Feynman graph on the source manifold.}
    \label{fig:Feynman_graphs_manifold}
\end{figure}

\begin{figure}[h!]
    \centering
    \begin{tikzpicture}
    \tikzset{Bullet/.style={fill=blue,draw,color=#1,circle,minimum size=3pt,scale=0.5}}
    \node[Bullet=cyan] at (0,0){};
    \node at (0,.5){1};
    \draw[-{latex[scale=3.0]}] (0,0)--(-1,-1);
    \draw[-{latex[scale=3.0]}] (0,0)--(1,-1);
    \node[Bullet=red] at (4,0){};
    \node at (3.3,0){$c^k(\mathbf{X})$};
    \draw[-{latex[scale=3.0]}] (4,0)--(4,-1);
    \draw[-{latex[scale=3.0]}] (3,1)--(4,0);
    \draw[-{latex[scale=3.0]}] (5,1)--(4,0);
    \node at (8,1){\circled{$\mathsf{x}$}};
    \draw[-{latex[scale=3.0]}] (8.25,1)--(10,1);
    \node at (8,-1){\circled{$\mathsf{y}$}};
    \draw[-{latex[scale=3.0]}] (10,-1)--(8.28,-1);
    \end{tikzpicture}
    \caption{The Feynman rules for DW theory. The first type of vertex corresponds to the term $\frac{1}{2}\langle\mathbf{Y},\mathbf{Y}\rangle=\frac{1}{2}\sum_i\mathbf{Y}_i\mathbf{Y}_i$. There are two outgoing arrows and no incoming arrow for the first vertex type. The second type of vertex corresponds to the term $\frac{1}{2}\langle\mathbf{Y},[\mathbf{X},\mathbf{X}]\rangle=\frac{1}{2}\sum_{i,j,k}c^k_{ij}\mathbf{X}^i\mathbf{X}^j\mathbf{Y}_k$, where $c^k_{ij}$ are the structure constants of the Lie bracket and where we have denoted $c^k(\mathbf{X}):=c^k_{ij}\mathbf{X}^i\mathbf{X}^j$. There are at most two incoming arrows and exactly one outgoing arrow for the second vertex type. Finally, we also have leaves corresponding to the residual fields. Note that an arrow between vertices represents the propagator from the starting to the end point.}
    \label{fig:Feynman_rules_1}
\end{figure}
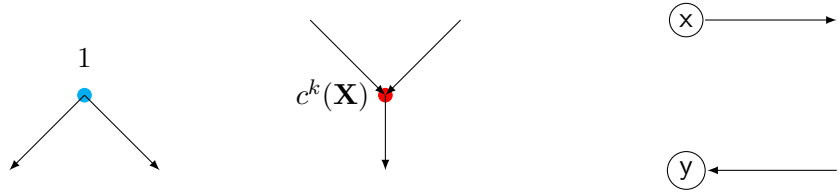

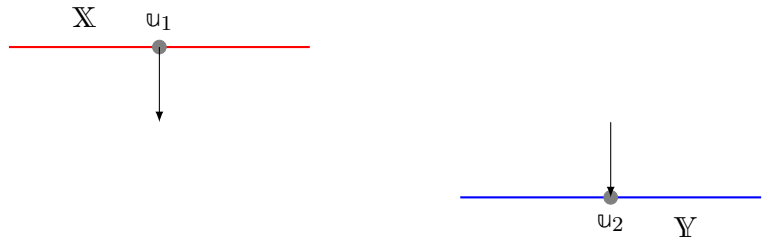
\begin{figure}[h!]
    \centering
    \begin{tikzpicture}
    \tikzset{Bullet/.style={fill=blue,draw,color=#1,circle,minimum size=3pt,scale=0.5}}
    \draw[color=red,thick] (0,0)--(4,0);
    \draw[color=blue,thick] (6,-2)--(10,-2);
    \node[Bullet=gray, label=above:{$\mathbb{u}_1$}] at (2,0){};
    \node[label=above:{$\mathds{X}$}] at (1,0){};
    \node[Bullet=gray, label=below:{$\mathbb{u}_2$}] at (8,-2){};
    \node[label=below:{$\mathds{Y}$}] at (9,-2){};
    \draw[-{latex[scale=3.0]}] (2,0)--(2,-1);
    \draw[-{latex[scale=3.0]}] (8,-1)--(8,-2);  
    \end{tikzpicture}
    \caption{The Feynman rules for vertices on the boundary}
    \label{fig:Feynman_rules_2}
\end{figure}

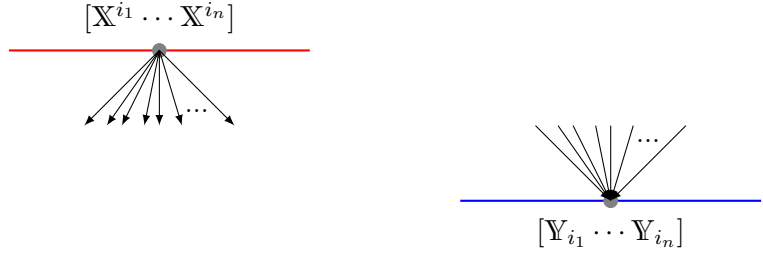
\begin{figure}[h!]
    \centering
    \begin{tikzpicture}
    \tikzset{Bullet/.style={fill=blue,draw,color=#1,circle,minimum size=3pt,scale=0.5}}
    \draw[color=red,thick] (0,0)--(4,0);
    \draw[color=blue,thick] (6,-2)--(10,-2);
    \node[Bullet=gray, label=above:{$[\mathds{X}^{i_1}\dotsm\mathds{X}^{i_n}]$}] at (2,0){};
    \node[Bullet=gray, label=below:{$[\mathds{Y}_{i_1}\dotsm\mathds{Y}_{i_n}]$}] at (8,-2){};
    \draw[-{latex[scale=3.0]}] (2,0)--(1,-1);
    \draw[-{latex[scale=3.0]}] (2,0)--(1.3,-1);    
    \draw[-{latex[scale=3.0]}] (2,0)--(1.5,-1); 
    \draw[-{latex[scale=3.0]}] (2,0)--(1.8,-1);  
    \draw[-{latex[scale=3.0]}] (2,0)--(2,-1); 
    \draw[-{latex[scale=3.0]}] (2,0)--(2.3,-1);
    \node[] at (2.5,-0.8){...};
    \draw[-{latex[scale=3.0]}] (2,0)--(3,-1);
    \draw[-{latex[scale=3.0]}] (7,-1)--(8,-2);  
    \draw[-{latex[scale=3.0]}] (7.3,-1)--(8,-2);
    \draw[-{latex[scale=3.0]}] (7.5,-1)--(8,-2);
    \draw[-{latex[scale=3.0]}] (7.8,-1)--(8,-2);
    \draw[-{latex[scale=3.0]}] (8,-1)--(8,-2);
    \draw[-{latex[scale=3.0]}] (8.3,-1)--(8,-2);
    \node[] at (8.5,-1.2){...};
    \draw[-{latex[scale=3.0]}] (9,-1)--(8,-2);
    \end{tikzpicture}
    \caption{The Feynman rules for composite field vertices on the boundary}
    \label{fig:Feynman_rules_3}
\end{figure}

\subsection{The partition function}
The partition function (or boundary state) is then formally given by the functional integral where we integrate out the fluctuations, i.e. by
\begin{equation}
\label{eq:partition_function_DW_theory}
\mathsf{Z}^{\scriptscriptstyle\mathrm{BV-BFV}}_{\Sigma,\de\Sigma}(\mathds{X},\mathds{Y},\mathsf{x},\mathsf{y};\hbar)=\int_{\calL\subset\calY'} \exp\left(\I\calS_\Sigma(\mathbf{X},\mathbf{Y})/\hbar\right)\mathscr{D}[\mathscr{X}]\mathscr{D}[\mathscr{Y}]
\end{equation}
where the Lagrangian submanifold is given by the gauge-fixing as in Section \ref{subsec:suitable_gauge_fixing}, i.e.
\[
\calL=\mathrm{graph}(\dd\varPsi^\mathrm{gf}),
\]
where $\varPsi^\mathrm{gf}$ is the gauge-fixing fermion constructed as in \eqref{eq:gauge-fixing_fermion}. The integral \eqref{eq:partition_function_DW_theory} is defined perturbatively as a formal power series in $\hbar$ in terms of the Feynman graphs given by the according Feynman rules as in Section \ref{subsec:Feynman_rules}, i.e. we can rewrite it as a perturbative expansion around critical points of $\calS_\Sigma$ as
\begin{equation}
\label{eq:perturbative_expansion}
\mathsf{Z}^{\scriptscriptstyle\mathrm{BV-BFV}}_{\Sigma,\de\Sigma}(\mathds{X},\mathds{Y},\mathsf{x},\mathsf{y};\hbar)=T_\Sigma\exp\left(\frac{\I}{\hbar}\sum_\Gamma \frac{(-\I\hbar)^{\ell(\Gamma)}}{\vert\mathrm{Aut}(\Gamma)\vert}\int_{\overline{\mathrm{Conf}_\Gamma(\Sigma)}}\omega_\Gamma(\mathds{X},\mathds{Y},\mathsf{x},\mathsf{y})\right),
\end{equation}
where we sum over connected Feynman graphs $\Gamma$ and where $\ell(\Gamma)$ denotes the number of loops of $\Gamma$. Moreover, the integral in \eqref{eq:perturbative_expansion} is over the configuration space of the vertex set of $\Gamma$ regarded as points in $\Sigma$ where we integrate a differential form $\omega_\Gamma$ depending on the boundary fields and residual fields.

\subsection{Degree count and Feynman graphs}
\label{subsec:Degree_count_and_Feynman_graphs}
Note that we want the dimension of the configuration space $\overline{\mathrm{Conf}_\Gamma(\Sigma)}$ to match the form degree of the differential form $\omega_\Gamma$ in order to perform the integrals in \eqref{eq:perturbative_expansion}. Note that the dimension of the configuration space is given by 
\[
\dim \overline{\mathrm{Conf}_\Gamma(\Sigma)}=4n+3m,
\]
where $n$ denotes the amount of vertices in the bulk and $m$ the amount of vertices on the boundary. Using the fact that the propagator $\mathscr{P}\in \Omega^3(\overline{\mathrm{Conf}_2(\Sigma)})$ is a 3-form on the configuration space of two points, the form degree of the differential form $\omega_\Gamma$ is given by $6\cdot\#I+3\cdot\#II$, where $\#I$ denotes the amount of first type vertices and $\#II$ the amount of second type vertices as in Figure \ref{fig:Feynman_rules_1}. Thus we get the system of equations
\begin{align}
\begin{split}
\label{eq:degree_count_1}
    4n+3m&=6\cdot\#I+3\cdot\#II,\\
    n&=\#I+\#II.
\end{split}
\end{align}
Now if $m=0$, we get $2\cdot\#I=\#II$. Hence, one can check that the diagrams which give a contribution are either wheels with an even amount of type $II$ vertices and $\frac{1}{2}\cdot\#I$ vertices attached to them as in Figure \ref{fig:Feynman_graphs_1}.

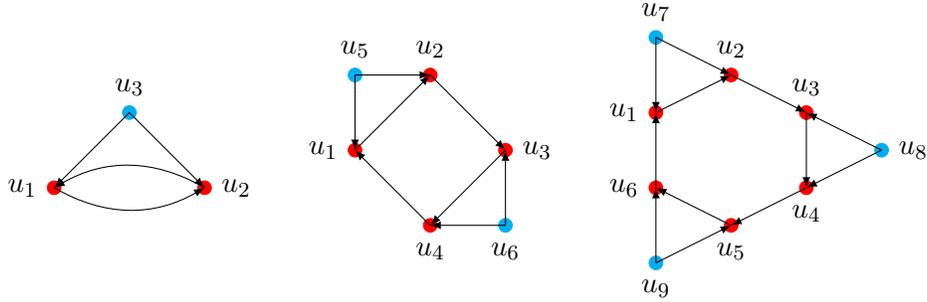
\begin{figure}[h!]
    \centering
    \begin{tikzpicture}[scale=1]
    \tikzset{Bullet/.style={fill=blue,draw,color=#1,circle,minimum size=3pt,scale=0.5}}
    \node[Bullet=red,label=left:{$u_1$}] at (1,-.5){};
    \node[Bullet=red,label=above:{$u_2$}] at (2,.5){};
    \node[Bullet=red,label=right:{$u_3$}] at (3,-.5){};
    \node[Bullet=red,label=below:{$u_4$}] at (2,-1.5){};
    \draw[-{latex[scale=3.0]}] (1,-.5)--(2,.5);
    \draw[-{latex[scale=3.0]}] (2,.5)--(3,-.5);
    \draw[-{latex[scale=3.0]}] (3,-.5)--(2,-1.5);
    \draw[-{latex[scale=3.0]}] (2,-1.5)--(1,-.5);
    \node[Bullet=cyan,label=above:{$u_5$}] at (1,.5){};
    \node[Bullet=cyan,label=below:{$u_6$}] at (3,-1.5){};
    \draw[-{latex[scale=3.0]}] (1,.5)--(1,-.5);
    \draw[-{latex[scale=3.0]}] (1,.5)--(2,.5);    
    \draw[-{latex[scale=3.0]}] (3,-1.5)--(3,-.5); 
    \draw[-{latex[scale=3.0]}] (3,-1.5)--(2,-1.5); 
    \node[Bullet=red,label=left:{$u_1$}] at (5,0){};
    \node[Bullet=red,label=above:{$u_2$}] at (6,.5){};
    \node[Bullet=red,label=above:{$u_3$}] at (7,0){};
    \node[Bullet=red,label=below:{$u_4$}] at (7,-1){};
    \node[Bullet=red,label=below:{$u_5$}] at (6,-1.5){};
    \node[Bullet=red,label=left:{$u_6$}] at (5,-1){};
    \draw[-{latex[scale=3.0]}] (5,0)--(6,0.5);
    \draw[-{latex[scale=3.0]}] (6,.5)--(7,0);    
    \draw[-{latex[scale=3.0]}] (7,0)--(7,-1);
    \draw[-{latex[scale=3.0]}] (7,-1)--(6,-1.5);
    \draw[-{latex[scale=3.0]}] (6,-1.5)--(5,-1); 
    \draw[-{latex[scale=3.0]}] (5,-1)--(5,0);
    \node[Bullet=cyan,label=above:{$u_7$}] at (5,1){};
    \node[Bullet=cyan,label=right:{$u_8$}] at (8,-.5){};
    \node[Bullet=cyan,label=below:{$u_9$}] at (5,-2){};
    \draw[-{latex[scale=3.0]}] (5,1)--(5,0);
    \draw[-{latex[scale=3.0]}] (5,1)--(6,.5);    
    \draw[-{latex[scale=3.0]}] (8,-.5)--(7,0);
    \draw[-{latex[scale=3.0]}] (8,-.5)--(7,-1);
    \draw[-{latex[scale=3.0]}] (5,-2)--(6,-1.5);
    \draw[-{latex[scale=3.0]}] (5,-2)--(5,-1);
    \node[Bullet=red,label=left:{$u_1$}] at (-3,-1){};
    \node[Bullet=red, label=right:{$u_2$}] at (-1,-1){};
    \draw[-{latex[scale=3.0]}] (-3,-1) to[bend right] (-1,-1);
    \draw[-{latex[scale=3.0]}] (-1,-1) to[bend right] (-3,-1);
    \node[Bullet=cyan,label=above:{$u_3$}] at (-2,0){};
    \draw[-{latex[scale=3.0]}] (-2,0) to (-3,-1);
    \draw[-{latex[scale=3.0]}] (-2,0) to (-1,-1);
    \end{tikzpicture}
    \caption{Example of wheel graphs appearing due to the degree count. In the first graph we have $\#II=2$ and $\#I=1$, in the second graph we have $\#II=4$ and $\#I=2$ and in the third graph we have $\#II=6$ and $\#I=3$. In fact, the first graph does not give any contribution, since by Kontsevich's lemma \cite{K} all the graphs with \emph{double edges} (i.e. graphs where there exist two vertices which have exactly two arrows (propagators) connecting them) vanish. This can be seen by using the angle form on $\mathbb{H}^4$.}
    \label{fig:Feynman_graphs_1}
\end{figure}

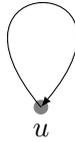
\begin{figure}[h!]
    \centering
    \begin{tikzpicture}[scale=1.2]
    \tikzset{every loop/.style={min distance=20mm,in=50,out=130,looseness=10}}
    \tikzset{Bullet/.style={fill=blue,draw,color=#1,circle,minimum size=3pt,scale=0.5}}
    \node[Bullet=gray, label=below:{$u$}] at (4,0){};
    \path[-{latex[scale=3.0]}] (4,0) edge  [loop above] node {} ();
    \end{tikzpicture}
    \caption{Short loop graphs (tadpoles) are not allowed since the propagator is singular on the diagonal. By a unimodularity condition which needs to be satisfied on $\mathfrak{h}[1]\oplus \mathfrak{h}[2]$, one can actually exclude these graphs. The condition is that the structure constants $c^k_{ij}$ of the Lie bracket on $\mathfrak{h}[1]\oplus\mathfrak{h}[2]$ satisfy $\sum_ic^i_{ij}=0$.
    In fact, the unimodularity condition can be dropped if the Euler characteristic of $\Sigma$ vanishes.}
    \label{fig:Feynman_graphs_2}
\end{figure}

When $m\not=0$, the system of equations \eqref{eq:degree_count_1} gives us $2\cdot\#I-\#II-3m=0$. Example of Feynman graphs in this setting are given in Figure \ref{fig:Feynman_graphs_3}.

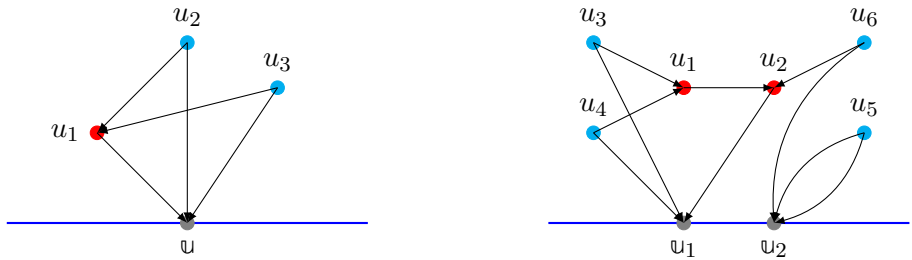
\begin{figure}[h!]
    \centering
    \begin{tikzpicture}[scale=1.2]
    \tikzset{Bullet/.style={fill=blue,draw,color=#1,circle,minimum size=3pt,scale=0.5}}
    \draw[thick,color=blue] (0,0) to (4,0);
    \draw[thick,color=blue] (6,0) to (10,0);
    \node[Bullet=red,label=left:{$u_1$}] at (1,1){};
    \node[Bullet=cyan,label=above:{$u_2$}] at (2,2){};
    \node[Bullet=cyan,label=above:{$u_3$}] at (3,1.5){};
    \node[Bullet=gray, label=below:{$\mathbb{u}$}] at (2,0){};
    \draw[-{latex[scale=3.0]}] (1,1) to (2,0);
    \draw[-{latex[scale=3.0]}] (2,2) to (1,1);
    \draw[-{latex[scale=3.0]}] (2,2) to (2,0);
    \draw[-{latex[scale=3.0]}] (3,1.5) to (2,0);
    \draw[-{latex[scale=3.0]}] (3,1.5) to (1,1);
    \node[Bullet=red,label=above:{$u_1$}] at (7.5,1.5){};
    \node[Bullet=red,label=above:{$u_2$}] at (8.5,1.5){};
    \node[Bullet=cyan,label=above:{$u_3$}] at (6.5,2){};
    \node[Bullet=cyan,label=above:{$u_4$}] at (6.5,1){};    
    \node[Bullet=cyan,label=above:{$u_5$}] at (9.5,1){};
    \node[Bullet=cyan,label=above:{$u_6$}] at (9.5,2){};
    \node[Bullet=gray, label=below:{$\mathbb{u}_1$}] at (7.5,0){};
    \node[Bullet=gray, label=below:{$\mathbb{u}_2$}] at (8.5,0){};
    \draw[-{latex[scale=3.0]}] (7.5,1.5) to (8.5,1.5);
    \draw[-{latex[scale=3.0]}] (6.5,2) to (7.5,1.5);
    \draw[-{latex[scale=3.0]}] (6.5,2) to (7.5,0);
    \draw[-{latex[scale=3.0]}] (6.5,1) to (7.5,1.5);
    \draw[-{latex[scale=3.0]}] (6.5,1) to (7.5,0);
    \draw[-{latex[scale=3.0]}] (9.5,2) to (8.5,1.5);
    \draw[-{latex[scale=3.0]}] (9.5,2) to[bend right] (8.5,0);
    \draw[-{latex[scale=3.0]}] (9.5,1) to[bend right] (8.5,0);
    \draw[-{latex[scale=3.0]}] (9.5,1) to[bend left] (8.5,0);
    \draw[-{latex[scale=3.0]}] (8.5,1.5) to (7.5,0);
    \end{tikzpicture}
    \caption{Example of Feynman graphs when $m\not=0$. In the left graph we have $m=1$, $\#I=2$ and $\#II=2$. In the right graph we have $m=2$, $\#I=4$ and $\#II=2$. Similarly as befor, the second graph does not give any contribution since by Kontsevich's lemma \cite{K} all the graphs with double edges vanish.}
    \label{fig:Feynman_graphs_3}
\end{figure}

\begin{figure}[h!]
    \centering
    \begin{tikzpicture}[scale=1.2]
    \tikzset{Bullet/.style={fill=blue,draw,color=#1,circle,minimum size=3pt,scale=0.5}}
    \node[Bullet=red,label=left:{$u_1$}] at (0,-.5){};
    \node[Bullet=red,label=left:{$u_2$}] at (1,.5){};
    \node[Bullet=red,label=above:{$u_3$}] at (2,-.5){};
    \node[Bullet=red,label=left:{$u_4$}] at (1,-1.5){};
    \draw[-{latex[scale=3.0]}] (0,-.5)--(1,.5);
    \draw[-{latex[scale=3.0]}] (1,.5)--(2,-.5);
    \draw[-{latex[scale=3.0]}] (2,-.5)--(1,-1.5);
    \draw[-{latex[scale=3.0]}] (1,-1.5)--(0,-.5);
    \node[Bullet=cyan,label=left:{$u_5$}] at (0,.5){};
    \node[Bullet=cyan,label=right:{$u_6$}] at (2,-1.5){};
    \draw[-{latex[scale=3.0]}] (0,.5)--(0,-.5);
    \node at (.5,1){\circled{$\mathsf{y}$}};
    \node at (1.5,1){\circled{$\mathsf{x}$}};
    \draw[-{latex[scale=3.0]}] (0,.5)--(.35,.85); 
    \draw[-{latex[scale=3.0]}] (1.365,.865)--(1,.5); 
    \draw[-{latex[scale=3.0]}] (2,-1.5)--(2,-.5); 
    \node at (1.5,-2){\circled{$\mathsf{y}$}};
    \node at (.5,-2){\circled{$\mathsf{x}$}};
    \draw[-{latex[scale=3.0]}] (2,-1.5)--(1.65,-1.85); 
    \draw[-{latex[scale=3.0]}] (.62,-1.8556)--(1,-1.5); 
    \node[Bullet=red,label=left:{$u_1$}] at (5,0){};
    \node[Bullet=red,label=left:{$u_2$}] at (6,.5){};
    \node[Bullet=red,label=above:{$u_3$}] at (7,0){};
    \node[Bullet=red,label=right:{$u_4$}] at (7,-1){};
    \node[Bullet=red,label=below:{$u_5$}] at (6,-1.5){};
    \node[Bullet=red,label=below:{$u_6$}] at (5,-1){};
    \draw[-{latex[scale=3.0]}] (5,0)--(6,0.5);
    \draw[-{latex[scale=3.0]}] (6,.5)--(7,0);    
    \draw[-{latex[scale=3.0]}] (7,0)--(7,-1);
    \draw[-{latex[scale=3.0]}] (7,-1)--(6,-1.5);
    \draw[-{latex[scale=3.0]}] (6,-1.5)--(5,-1); 
    \draw[-{latex[scale=3.0]}] (5,-1)--(5,0);
    \node[Bullet=cyan,label=left:{$u_7$}] at (5,1){};
    \node[Bullet=cyan,label=right:{$u_8$}] at (8,-.5){};
    \node[Bullet=cyan,label=below:{$u_9$}] at (5,-2){};
    \draw[-{latex[scale=3.0]}] (5,1)--(5,0);
    \node at (5.5,1.5){\circled{$\mathsf{y}$}};
    \node at (6.5,1.5){\circled{$\mathsf{x}$}};
    \draw[-{latex[scale=3.0]}] (5,1)--(5.35,1.35);
    \draw[-{latex[scale=3.0]}] (6.4,1.35)--(6,0.5);
    \draw[-{latex[scale=3.0]}] (8,-.5)--(7,0);
    \node at (8.5,-1.5){\circled{$\mathsf{y}$}};
    \node at (7.5,-1.5){\circled{$\mathsf{x}$}};
    \draw[-{latex[scale=3.0]}] (8,-.5)--(8.4,-1.3);    
    \draw[-{latex[scale=3.0]}] (7.4,-1.35)--(7,-1); 
    \draw[-{latex[scale=3.0]}] (5,-2)--(6,-1.5);
    \node at (4,-1.5){\circled{$\mathsf{y}$}};
    \node at (4,-.5){\circled{$\mathsf{x}$}};
    \draw[-{latex[scale=3.0]}] (5,-2)--(4.17,-1.6);  
    \draw[-{latex[scale=3.0]}] (4.15,-.6)--(5,-1);  
    \end{tikzpicture}
    \caption{Example of Feynman graphs with leaves (residual fields). Note that the $\mathsf{y}$-leaves give a contribution of form degree 2, whereas the $\mathsf{x}$-leaves give a contribution of form degree $1$. Gluing residual fields will induce a propagator of form degree $3$.}
    \label{fig:Feynman_graphs_4}
\end{figure}
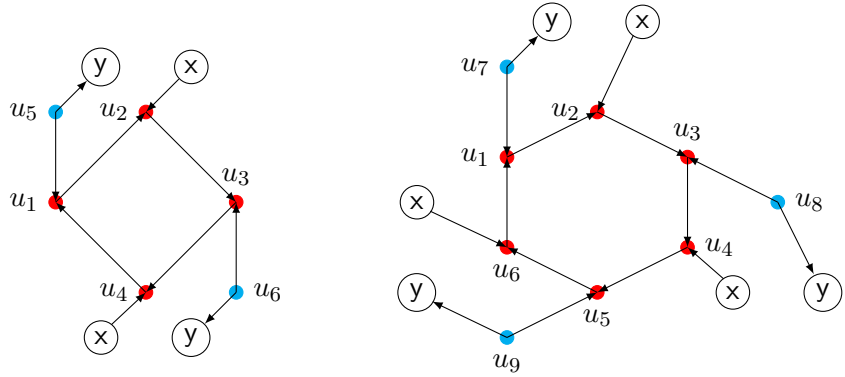

\subsection{The modified Quantum Master Equation}
We can now prove the following theorem:
\begin{thm}[mQME for DW theory]
\label{thm:mQME_DW}
The BV-BFV partition function $\mathsf{Z}^{\scriptscriptstyle\mathrm{BV-BFV}}_{\Sigma,\de\Sigma}$ for DW theory satisfies the modified Quantum Master Equation, i.e.
\[
\big(\hbar^2\Delta_{\calV_\Sigma}+\Omega^{\bullet,\calP}_{\de\Sigma}\big)\mathsf{Z}^{\bullet,\scriptscriptstyle\mathrm{BV-BFV}}_{\Sigma,\de\Sigma}=0,
\]
where $\Omega^{\bullet,\calP}_{\de\Sigma}=\Omega_0^\calP+\Omega^{\bullet,\calP}_\mathrm{pert}$ with $\Omega_0^\calP$ being the unperturbed quantization part as defined in \eqref{eq:boundary_operator_0} and $\Omega^{\bullet,\calP}_\mathrm{pert}:=\Omega^{\bullet,\mathds{X}}_\mathrm{pert}+\Omega^{\bullet,\mathds{Y}}_\mathrm{pert}$ is fully determined by the boundary configuration space integrals.
\end{thm}

\begin{proof}[Proof of Theorem \ref{thm:mQME_DW}]
The proof is on the level of graphs. First, we want to describe the construction of $\Omega^{\bullet,\calP}_\mathrm{pert}$ in terms of boundary configuration space integrals. Consider the compactified configuration space $\overline{\mathrm{Conf}}_\Gamma$ and $\omega_\Gamma$ the corresponding differential form, for some Feynman graph $\Gamma$. Then, using Stokes' theorem, we get $\int_{\overline{\mathrm{Conf}}_\Gamma}\dd\omega_\Gamma=\int_{\de\overline{\mathrm{Conf}}_\Gamma}\omega_\Gamma$. When we apply the de Rham differential $\dd$ to $\omega_\Gamma$ on the left-hand-side, it can either act on an $\mathds{X}$-field, $\mathds{Y}$-field or on the propagator $\mathscr{P}$. Clearly, the part acting on the $\mathds{X}$- or $\mathds{Y}$-fields corresponds to the action of $\frac{1}{\I\hbar}\Omega_0^\calP$, whereas the part acting on the propagator will correspond to the action of $-\I\hbar\Delta_{\calV_\Sigma}$ on the partition function $\mathsf{Z}^{\bullet,\scriptscriptstyle\mathrm{BV-BFV}}_{\Sigma,\de\Sigma}$. The integral over the boundary configuration space on the right-hand-side contains terms of different possible strata for the corresponding manifold with corners. It either contains integrals over boundary components where two vertices collapse in the bulk, which we denote as situation $(a)$, or integrals over boundary components where more than two vertices collapse in the bulk\footnote{Usually called \emph{hidden faces}.}, which we denote by situation $(b)$, or integrals over boundary components where two or more (bulk and/or boundary) vertices collapse at the boundary or one single bulk vertex collapses to the boundary, which we denoted by situation $(c)$. For situation $(a)$, we can note that since we assume the CME to hold which equivalently translates to the fact that 
\[
\frac{1}{2}\sum_{i=1}^4\pm \frac{\delta}{\delta\mathbf{X}^i}\left(\langle\mathbf{Y},\mathbf{Y}+[\mathbf{X},\mathbf{X}]\rangle\right)\cdot \frac{\delta}{\delta\mathbf{Y}_i}\left(\langle\mathbf{Y},\mathbf{Y}+[\mathbf{X},\mathbf{X}]\rangle\right)=0,
\]
the combinatorics of the Feynman graphs in the perturbative expansion leads to the cancellation of such terms when summing over all graphs. For situation $(b)$, one can use the usual vanishing theorems \cite{K,Kontsevich1993_2,Bott1996,BottTaubes1994} in order to get rid of faces where all vertices of a connected component of a Feynman graph collapse. For situation $(c)$, note that we can split such integrals into an integral over a subgraph $\Gamma'\subset\Gamma$ and an integral over the graph which can be obtained by identifying all the vertices of $\Gamma'$ and deleting all edges of $\Gamma'$ which we denote by $\Gamma/\Gamma'$. Hence, one can define the action of $\frac{\I}{\hbar}\Omega^{\bullet,\calP}_\mathrm{pert}$ as the sum of the boundary contributions coming from the subgraphs $\Gamma'\subset \Gamma$ and thus we have 
\[
\Omega^{\bullet,\calP}_\mathrm{pert}\mathsf{Z}^{\bullet,\scriptscriptstyle\mathrm{BV-BFV}}_{\Sigma,\de\Sigma}:=\sum_\Gamma\sum_{\Gamma'\subset\Gamma}\int_{\overline{\mathrm{C}_{\Gamma'}(\mathbb{H}^4)}\times \overline{\mathrm{Conf}_{\Gamma/\Gamma'}(\Sigma)}}\omega_{\Gamma'}.
\]

\end{proof}

\begin{rem}[Principal part]
It is important to note that in the setting of DW theory, there are no higher corrections for the definition of the \emph{principal part} of $\Omega^{\bullet,\calP}_{\de\Sigma}$, i.e. the term $\Omega^{\bullet,\calP}_\mathrm{pert}$ which is in general of the form \eqref{eq:principal_part_X}+\eqref{eq:principal_part_Y}. This is due to the fact that we are working with $4$-dimensional source manifolds and by the result of the following lemma. 

\begin{lem}[Cattaneo--Mnev--Reshetikhin\cite{CMR2}]
Let $\Sigma$ be a $d$-dimensional source manifold. If $d$ is even, then the principal part of $\Omega^{\bullet,\calP}_{\de\Sigma}$ is directly given by the ordered standard quantization of the boundary action $\calS^\de_{\de\Sigma}$. If $d$ is odd, the principal part of $\Omega^{\bullet,\calP}_{\de\Sigma}$ is given by the ordered standard quantization of the modified boundary action 
\[
\tilde\calS^\de_{\de\Sigma}:=\calS^\de_{\de\Sigma}-\I\hbar\sum_{j=0}^{\left[\frac{d-3}{4}\right]}\int_{\de\Sigma}\gamma_j\tr\ad_{\mathbf{X}}^{d-4j},
\]
where $\gamma_j$ is a closed $4j$-form on $\de\Sigma$ which is an invariant polynomial, with universal coefficients, of the curvature of the connection used in the construction of the propagator.
\end{lem}
\end{rem}

Moreover, we can prove the following theorem:
\begin{thm}[Flatness of BFV boundary operator]
\label{thm:square_zero}
The full BFV boundary operator $\Omega^{\bullet,\calP}_{\de\Sigma}$ for DW theory squares to zero and hence
we have a well-defined BV-BFV cohomology for DW theory, i.e. the operator $\hbar^2\Delta_{\calV_\Sigma}+\Omega^{\bullet,\calP}_{\de\Sigma}$ squares to zero.
\end{thm}

\begin{proof}[Proof of Theorem \ref{thm:square_zero}]
It is easy to see that if $\Omega^{\bullet,\calP}_{\de\Sigma}$ squares to zero, so does $\hbar^2\Delta_{\calV_\Sigma}+\Omega^{\bullet,\calP}_{\de\Sigma}$ since $\Delta_{\calV_\Sigma}$ squares to zero and $[\Delta_{\calV_\Sigma},\Omega^{\bullet,\calP}_{\de\Sigma}]=0$. Consider again a subgraph $\Gamma'\subset\Gamma$ as in the definition of $\Omega^{\bullet,\calP}_{\de\Sigma}$ in the proof of Theorem \ref{thm:mQME_DW} with corresponding differential form $\sigma_{\Gamma'}$ over the configuration space $\overline{\mathrm{C}_{\Gamma'}(\mathbb{H}^4)}$. Using Stokes' theorem, we get $\int_{\overline{\mathrm{C}_{\Gamma'}(\mathbb{H}^4)}}\dd\sigma_{\Gamma'}=\int_{\de\overline{\mathrm{C}_{\Gamma'}(\mathbb{H}^4)}}\sigma_{\Gamma'}$. Similarly as before, we can obtain that the part where the de Rham differential acts on the boundary fields in $\sigma_{\Gamma'}$ corresponds to the action of $\frac{1}{\I\hbar}\Omega_0^\calP$. Moreover, we can obtain the similar three situations as in the proof of Theorem \ref{thm:mQME_DW}. In particular, the terms of situation $(a)$ cancel out when summing over all graphs, the terms of situation $(b)$ are excluded by the vanishing theorems and the terms of situation $(c)$ lead to the action of $\Omega^{\bullet,\calP}_\mathrm{pert}$ when summing over all graphs. Hence, we have 
\[
\Omega_0^\calP\Omega^{\bullet,\calP}_\mathrm{pert}+\Omega^{\bullet,\calP}_\mathrm{pert}\Omega_0^\calP+(\Omega^{\bullet,\calP}_\mathrm{pert})^2=0,
\]
since $(\Omega_0^\calP)^2=0$. 
\end{proof}

\begin{rem}
Instead of considering the boundary operator in terms of boundary configuration space integrals, one can also take the explicit form of $\Omega^{\bullet,\calP}_{\de\Sigma}$ as in \eqref{eq:full_boundary_operator} in order to prove that it squares to zero. In particular, for the full perturbation part, one can take the sum of the corresponding full versions of \eqref{eq:boundary_operator_pert_X} and \eqref{eq:boundary_operator_pert_Y}.
\end{rem}

\begin{rem}\label{rem:Floer_group_boundary_state}
By Witten's approach \cite{Witten1988} (see also Section \ref{subsec:field_theory_approach_to_floer_homology}), the expectation value of the observable $O$ as in \eqref{eq:observable} with respect the BV-BFV partition function for the DW AKSZ-BV-action $\calS_\Sigma$ should be given by a Floer cohomology class associated to the boundary after integrating out the residual fields (see \cite{CMR2} for a possible integration theory on $\calV_\Sigma$). The resulting boundary state would then be of the form
\[
\Psi_{\de\Sigma}(\mathds{X},\mathds{Y})=\int_{\calV_\Sigma}\underbrace{\int_{\calL}\exp(\I\calS_\Sigma(\mathbf{X},\mathbf{Y})/\hbar)O(\mathbf{X},\mathbf{Y})\mathscr{D}[\mathscr{X}]\mathscr{D}[\mathscr{Y}]}_{=\left\langle\prod_{j=1}^dO^{(\gamma_j)}\right\rangle=\left\langle \prod_{j=1}^d\int_{\gamma_j}W_{k_j}\right\rangle}\in HF^\bullet(\de\Sigma)\subset\calH^\calP_{\de\Sigma}.
\]
For a particular polarization adapted to the boundary condition given by the Floer cohomology class $\Psi_{\de\Sigma}$, the boundary state space $\calH^\calP_{\de\Sigma}$ in this case should include the Floer cohomology classes and thus the BFV boundary operator $\Omega^{\bullet,\calP}_{\de\Sigma}$ should be given in terms of the Floer differential. Note that here we want $\mathfrak{h}=\mathfrak{su}(2)$. The gauge-fixing Lagrangian $\calL$ is given by the one considered in Section \ref{subsec:suitable_gauge_fixing}.
\end{rem}

\subsection{The modified differential Quantum Master Equation}
Following the construction in \cite{CMW4}, we can perform \emph{globalization} on the target of the DW AKSZ theory by using methods of \emph{formal geometry} as developed in \cite{GelfandFuks1969,GelfandFuks1970,GK,B}. For a manifold $M$ and an open neighborhood $U\subset TM$ of the zero section, we can consider a \emph{generalized exponential map} $\varphi\colon U\to M$ such that $\varphi(x,y):=\varphi_x(y)$, i.e. $\varphi$ satisfies the properties:
\begin{itemize}
    \item $\varphi_x(0)=x,\quad \forall x\in M$,
    \item $\dd\varphi_x(0)=\id_{T_xM},\quad \forall x\in M$.
\end{itemize}
Locally, we get 
\begin{equation}
    \label{eq:local_generalized_exponential_map}
    \varphi_x^i(y)=x^i+y^i+\frac{1}{2}\varphi^i_{x,jk}y^jy^k+\frac{1}{3!}\varphi^i_{x,jk\ell}y^jy^ky^\ell+\dotsm
\end{equation}
where $(x^i)$ are coordinates on the base and $(y^i)$ are coordinates on the fiber. A \emph{formal exponential map} is given by the equivalence class of generalized exponential maps with respect to the equivalence relation which identifies two generalized exponential maps when their jets agree to all orders. By abuse of notation, we will also denote formal exponential maps by $\varphi$. 
We can define a \emph{flat} connection $D$ on $\widehat{\Sym}(T^*M)$, where $\widehat{\Sym}$ denotes the \emph{completed} symmetric algebra. The connection is called \emph{Grothendieck connection} \cite{Grothendieck1968} and can be locally written as $D=\dd+R$, where $\dd$ denotes the de Rham differential on $M$ and $R$ denotes a 1-form with values in derivations of the completed symmetric algebra of the cotangent bundle\footnote{Note that flatness of $D$ is equivalent to the Maurer--Cartan equation $\dd R+\frac{1}{2}[R,R]=0$.}. For a section $\sigma\in\Gamma(\widehat{\Sym}(T^*M))$, we get that $R$ acts on $\sigma$ through the Lie derivative, i.e. we have  $R(\sigma)=L_R\sigma$. Locally, we can express $R=R_\ell\dd x^\ell$, where $R_\ell:=R^j_\ell(x,y)\frac{\de}{\de y^j}$ and 
\[
R^j_\ell(x,y):=-\frac{\de\varphi^k}{\de x^\ell}\left(\left(\frac{\de\varphi}{\de y}\right)^{-1}\right)^j_k=-\delta_\ell^j+\mathrm{O}(y).
\]
Thus, for a section $\sigma\in \Gamma(\widehat{\Sym}(T^*M))$, we have
\[
R(\sigma)=L_R\sigma=R_\ell(\sigma)\dd x^\ell=-\frac{\de\sigma}{\de y^j}\frac{\de\varphi^k}{\de y^\ell}\left(\left(\frac{\de\varphi}{\de y}\right)^{-1}\right)^j_k\dd x^\ell. 
\]
Moreover, one can extend the Grothendieck connection $D$ to the complex $\Gamma(\bigwedge^\bullet T^*M\otimes \widehat{\Sym}(T^*M))$ consisting of $\widehat{\Sym}(T^*M)$-valued forms.

\begin{prop}[\cite{Moshayedi2020}]\label{prop:formal_geometry}
A section $\sigma\in \Gamma(\widehat{\Sym}(T^*M))$ is $D$-closed if and only if $\sigma=\mathsf{T}\varphi^*f$ for some $f\in C^\infty(M)$, where $\mathsf{T}$ denotes the Taylor expansion around the fiber coordinates at zero. Moreover, the $D$-cohomology is concentrated in degree zero and is given by
\[
H_D^0(\widehat{\Sym}(T^*M))=\mathsf{T}\varphi^*C^\infty(M)\cong C^\infty(M).
\]
\end{prop}

The main part of the proof of Proposition \ref{prop:formal_geometry} uses methods from \emph{cohomological perturbation theory}. let us now consider the AKSZ construction of DW theory as in Section \ref{sec:AKSZ_formulation_of_DW_theory}. In particular, the techniques of \cite{CMW4} can be used here since the target is of the form $\mathfrak{h}[1]\oplus\mathfrak{h}[2]$. Consider a 4-dimensional source manifold (possibly with boundary) and recall the BV space of fields of DW theory
\[
\calF_{\Sigma}=\Map(T[1]\Sigma,\mathfrak{h}[1]\oplus\mathfrak{h}[2])\cong\Map(T[1]\Sigma,(T[1]\mathfrak{h})[1]).
\]
Consider then the formal exponential map given by 
\begin{align*}
    \varphi\colon T\mathfrak{h}\cong\mathfrak{h}\oplus\mathfrak{h}&\to \mathfrak{h},\\
    (x,y)&\mapsto \varphi_x(y):=x+y.
\end{align*}
Note also that a particularly easy class of solutions of AKSZ theories is given by constant ones of the type $x=(x,0)\colon T[1]\Sigma\to \mathfrak{h}[1]\oplus \mathfrak{h}[2]$, where $x=\mathrm{const}$.
Thus, we can consider the \emph{linearized} space of fields at a constant solution $x$, given by
\[
\varphi_x^*\calF_\Sigma:=\Map(T[1]\Sigma,(T[1]T_x\mathfrak{h})[1]).
\]
Let $(\mathbf{X},\mathbf{Y})\in\calF_\Sigma$. Then the corresponding lifts by $\varphi_x$ are given by 
\[
\widehat{\mathbf{X}}:=\varphi_x^{-1}(\mathbf{X}),\qquad \widehat{\mathbf{Y}}:=(\dd\varphi_x)^*\mathbf{Y}.
\]
The Grothendieck connection can be computed by noticing that $R^j_\ell(x,y)=-\delta_\ell^j$.
We can then define the \emph{formal global action} by 
\[
\widehat{\calS}_{\Sigma,x}:=\int_\Sigma\left(\widehat{\mathbf{Y}}_i\dd_\Sigma\widehat{\mathbf{X}}^i+\widehat{\Theta}_x(\widehat{\mathbf{X}},\widehat{\mathbf{Y}})-\widehat{\mathbf{Y}}_\ell\dd_{\mathcal{EL}} x^\ell\right).
\]
Here we have denoted $\widehat{\Theta}_x(\widehat{\mathbf{X}},\widehat{\mathbf{Y}}):=\mathsf{T}\widehat{\varphi}_x^*\Theta(\mathbf{X},\mathbf{Y})$, where $\widehat{\varphi}_x\colon \varphi_x^*\calF_\Sigma\to \calF_\Sigma$. Note also that $\dd_{\mathcal{EL}}$ denotes the de Rham differential on the moduli space $\mathcal{EL}$ of classical solutions, which can be identified with the target by considering constant solutions. One can then check that the \emph{differential CME} (dCME) holds:
\begin{equation}
    \label{eq:dCME}
    \dd_x\widehat{\calS}_{\Sigma,x}+\frac{1}{2}(\widehat{\calS}_{\Sigma,x},\widehat{\calS}_{\Sigma,x})=0.
\end{equation}
If $\Sigma$ has boundary, one can show that a \emph{differential} version of the mCME \eqref{eq:mCME} is also satisfied \cite{CMW4}. For the quantum case, we have the following theorem: 
\begin{thm}[mdQME for DW theory]\label{thm:mdQME}
The formal global BV-BFV partition function for DW theory, given formally by the functional integral
\[
\widehat{\mathsf{Z}}^{\bullet,\scriptscriptstyle\mathrm{BV-BFV}}_{\Sigma,\de\Sigma}=\int_{\calL}\exp(\I\widehat{\calS}_{\Sigma,x}/\hbar),
\]
satisfies the \emph{modified differential QME} (mdQME) 
\begin{equation}
    \label{eq:mdQME}
    \left(\dd_x-\I\hbar\Delta_{\calV_\Sigma}+\frac{\I}{\hbar}\Omega^{\bullet,\calP}_{\de\Sigma}\right)\widehat{\mathsf{Z}}^{\bullet,\scriptscriptstyle\mathrm{BV-BFV}}_{\Sigma,\de\Sigma}=0.
\end{equation}
\end{thm}

Moreover, using similar methods as in \cite{CMW4}, we can show the following theorem:

\begin{thm}[Flatness of qGBFV operator for DW theory]\label{thm:flat_Grothendieck}
The \emph{quantum Grothendieck BFV (qGBFV) operator}
\[
\nabla_\mathsf{G}:=\dd_x-\I\hbar\Delta_{\calV}+\frac{\I}{\hbar}\Omega^{\bullet,\calP}
\]
is flat and behaves well under change of data. Moreover, it defines a cohomology theory on the globally extended state space.
\end{thm}

\begin{proof}[Proof of Theorem \ref{thm:flat_Grothendieck}]
This proof is similar to the one in \cite{CMW4}. In particular, the flatness of $\nabla_\mathsf{G}$ is equivalent to the equation
\begin{equation}
    \I\hbar\dd_x\Omega^{\bullet,\calP}_{\de\Sigma}-\frac{1}{2}\left[\Omega^{\bullet,\calP}_{\de\Sigma},\Omega^{\bullet,\calP}_{\de\Sigma}\right]=0.
\end{equation}
Using the explicit expression of $\Omega^{\bullet,\calP}_{\de\Sigma}$ through configuration space integrals by using the perturbed parts \eqref{eq:principal_part_X} and \eqref{eq:principal_part_Y}, we can apply Stokes' theorem 
\begin{equation}
\label{eq:Stokes}
\dd_x\int_{\overline{\mathrm{Conf}_\Gamma(\Sigma)}}\omega_\Gamma=\int_{\overline{\mathrm{Conf}_\Gamma(\Sigma)}}\dd\omega_\Gamma\pm \int_{\de\overline{\mathrm{Conf}_\Gamma(\Sigma)}}\omega_\Gamma,
\end{equation}
where $\dd=\dd_x+\dd_1+\dd_2$ is the differential on $(\mathfrak{h}[1]\oplus\mathfrak{h}[2])\times \overline{\mathrm{Conf}}_\Gamma(\Sigma)$ with $\dd_1$ the differential acting on the propagator, i.e. on the residual fields, and $\dd_2$ the differential acting on the boundary fields. Then by \cite[Lemma 4.9]{CMW4}, we can observe that the boundary face where more than two bulk vertices in a subgraph $\Gamma'\subset \Gamma$ collapse to the boundary was shown to be $\frac{1}{2}\left[\Omega^{\bullet,\mathds{Y}}_\mathrm{pert},\Omega^{\bullet,\mathds{Y}}_\mathrm{pert}\right]$ (and similarly for the $\mathds{X}$-representation). For the case where exactly two bulk vertices collapse, we can observe that these faces cancel with $\dd_x\omega_{\Gamma'}$ by using the dCME \eqref{eq:dCME}. 
\end{proof}


\begin{proof}[Proof of Theorem \ref{thm:mdQME}]
Again, using the configuration space integral formulation of the partition function as in \eqref{eq:perturbative_expansion}, we can apply Stokes' theorem \eqref{eq:Stokes} in order to obtain similar relations of $\dd_1$ with $\Delta_{\calV_\Sigma}$, $\dd_2$ with $\Omega_0^\calP$ and $\dd_x$ with the boundary contribution when applied to the partition function $\widehat{\mathsf{Z}}^{\bullet,\scriptscriptstyle\mathrm{BV-BFV}}_{\Sigma,\de\Sigma}$ as in \cite{CMW4}. 
\end{proof}

\section{Nekrasov's partition function, equivariant BV formalism and equivariant Floer (co)homology}
\label{sec:Nekarsov_partition_function_and_equivariant_BV_formalism}

\subsection{Seiberg--Witten theory}
In \cite{SeibergWitten1994,SeibergWitten1994_2}, Seiberg and Witten have formulated a way of describing low-energy behaviour for special supersymmetric gauge field theories. Moreover, they have formulated topological invariants of 4-manifolds which can be shown to be equivalent to the Donaldson polynomials but in general easier to compute. 
The gauge group is fixed to be $\mathrm{SU}(2)$ in this setting.
Recall that an $\calN=2$ \emph{chiral multiplet} (or \emph{vector multiplet}) includes a gauge field $A_\mu$, two Weyl fermions\footnote{Recall that a \emph{Weyl fermion} (or \emph{Weyl spinor}) is a spinor which satisfies the \emph{Weyl equations} $\sigma^\mu\de_\mu\psi=0$, where $\sigma^\mu$ denotes the Pauli matrices for $\mu=0,\ldots,3$.} $\lambda,\psi$ and a scalar field $\phi$ in the adjoint representation, i.e. we have something of the form 
\[
(\!(A_\mu,(\lambda,\psi),\phi)\!),
\]
where the brackets $(\!(\enspace)\!)$ here indicate that we have singlets $A_\mu,\phi$ and a doublet $(\lambda,\psi)$. When $\calN=1$, we can express this in one vector multiplet $W_\alpha$ (containing $A_\mu$ and $\lambda$) and a chiral multiplet $\Phi$ (containing $\phi$ and $\psi$).
In $\calN=1$ superspace, one can express the fermionic low-energy effective action in terms of a holomorphic function $\mathsf{F}$ as 
\begin{equation}
    \label{eq:Seiberg-Witten_action}
    S^\mathrm{SW,eff}=\frac{1}{4\pi}\mathrm{Im}\,\left(\int\dd^4\theta \frac{\de\mathsf{F}(A)}{\de A}\bar A+\int \dd^2\theta\frac{1}{2}\frac{\de^2\mathsf{F}(A)}{\de A^2}W_\alpha W^{\alpha}\right),
\end{equation}
where $\theta$ denotes the angle for the global $\mathrm{SU}(2)$-action, $A$ denotes the $\calN=1$ chiral multiplet in the $\calN=2$ vector multiplet whose scalar component is given by some complex parameter which labels the vacua. The holomorphic function $\mathsf{F}$ is in fact the free energy which can be expressed in terms of the Seiberg--Witten formula for periods\footnote{Recall that a \emph{period} of a closed differential form $\omega$ over some $n$-cycle $C$ is given by $\int_C\omega$.} of some differential $\dd \mathsf{S}$ on some algebraic curve $C$ (\emph{Seiberg--Witten curve}).

\subsection{Nekrasov's partition function}
In \cite{Nekrasov2003}, Nekrasov constructed a regularized partition function $\mathsf{Z}^\mathrm{Nek}$ for the supersymmetric gauge field theories appearing in the construction of Seiberg--Witten, i.e. based on $\calN=2$ supersymmetric Yang--Mills theory. The idea was to use \emph{equivariant integration} with respect to a natural symmetry group for a long-distance cut-off regularization with parameter $\varepsilon$. Moreover, he proposed that for $\varepsilon\to 0$, asymptotically we get $\log\mathsf{Z}^\mathrm{Nek}\sim -\frac{1}{\varepsilon^2}\mathsf{F}$. The proof was given by Nekrasov and Okounkov for $\mathrm{U}(r)$-gauge theories with matter fields in fundamental and adjoint representations of the gauge group and for 5-dimensional theories compactified on the circle.
Consider a connection $A$ on a trivial rank $r>1$ bundle over $\Sigma=\R^4$ such that the Yang--Mills action functional $S^{\mathrm{YM}}_\Sigma(A)=\int_\Sigma\|F_A\|^2<\infty$. In this case, as we have seen in Section \ref{subsec:Moduli_space_of_anti_self-dual_connections}, it is bounded from below 
\[
8\pi^2k_A\leq S^\mathrm{YM}_{\Sigma}(A)=\int_\Sigma\|F_A\|^2.
\]
Moreover, as we have also seen in Section \ref{subsec:Moduli_space_of_anti_self-dual_connections}, if $A$ is an instanton (anti self-dual connection), we get equality. Donaldson showed in \cite{Donaldson1984} that the moduli space of anti self-dual connections $\calM_\mathrm{ASD}$ is equal to the moduli space of holomorphic bundles on $\Sigma=\C^2\cong \R^4$ which are trivial at infinity. One can describe Nekrasov's theory as a $(\Aut(\Sigma)\times \GL(r))$-equivariant integration over a certain partial compactification 
\[
\calM_\mathrm{ASD}\subset \overline{\calM}_r=\{\text{framed torsion-free sheaves on $\Sigma$ of rank $r$}\}.
\]
Note that each element of $\mathrm{Aut}(\Sigma)\times \mathrm{GL}(r)$ can always be considered in the form 
\[
\begin{pmatrix}\varepsilon_1&\\& \varepsilon_2\end{pmatrix}\times \begin{pmatrix}a_1& & \\&\ddots&\\&&a_r\end{pmatrix}. 
\]
In particular, $\varepsilon_1$ and $\varepsilon_2$ can be considered as two rotations of $\C^2$ regarded as the generators of a torus $\mathbb{T}$. 
Let $\calM$ be any smooth algebraic variety and consider a torus $\mathbb{T}$ acting on it and $\calE$ a $\mathbb{T}$-equivariant coherent sheaf on $\calM$. Furthermore, let $\overline{\Sigma}$ be a projective surface, e.g. $\mathbb{P}^2$, and choose an embedding $\Sigma\hookrightarrow \overline{\Sigma}$ with a framing\footnote{Note that $\mathrm{GL}(r)\cong\mathrm{Aut}(\calO^{\oplus r}_D)$ acts on $\phi$ which is the same as the action of constant gauge transformations on instantons.} $\phi\colon\calE\vert_D\to \calO_D^{\oplus r}$ of $\calE$ where $D=\overline{\Sigma}\setminus \Sigma$. Then one can consider the localization formula \cite{ChrissGinzburg1997}, given by 
\[
\chi(\calM,\mathcal{E})=\chi\big(\calM^\mathbb{T},\mathcal{E}\vert_{\calM^\mathbb{T}}\otimes \Sym\big(N^*\big(\calM/\calM^\mathbb{T}\big)\big)\big)\in \mathrm{K}_\mathbb{T}(\mathrm{pt})\left[\frac{1}{1-\varepsilon^\nu}\right]\subset \Q(\mathbb{T}),
\]
where $\varepsilon^\nu$ denote the weights of the normal bundle $N\big(\calM/\calM^\mathbb{T}\big)$ and where $\Sym$ denotes the symmetric algebra. Here we have denoted by $\mathrm{K}_\mathbb{T}(\mathrm{pt})$ the $\mathbb{T}$-equivariant $K$-theory over a point and $\Q(\mathbb{T})$ denotes the representation ring of $\mathbb{T}$.
\begin{rem}
If a point $p\in \calM^\mathbb{T}$ is isolated, and if the tangent space at that point is given by $T_p\calM=\sum_i \varepsilon^{\nu_i}$ as a $\mathbb{T}$-module, then 
\[
\Sym\big(N^*\big(\calM/p\big)\big)=\prod_i \frac{1}{1-\varepsilon^{\nu_i}}.
\]
This product is given by the character of the $\mathbb{T}$-action on functions on the formal neighborhood of $p\in\calM$.
\end{rem}
We will only consider the case where all fixed points are isolated, hence $\chi\big(\calM_r^\mathbb{T}\big)$ is a finite sum over $r$-tuples of partitions. 
Consider the \emph{Hilbert scheme} (see also \cite{Nakajima1999} for the definition of a Hilbert scheme)
\[
\mathrm{Hilb}(\Sigma,k):=\{\calI\subset \C[u_1,u_2] \mid \text{$\calI$ ideal of codimension $k$},\,\, u_1,u_2\in\Sigma\}
\]
and let $\varepsilon:=(\varepsilon_1,\varepsilon_2)$ acting on it by $f(u)\mapsto f(\varepsilon^{-1}\cdot u)$. In particular, if $\calI$ is fixed, we get\footnote{A partition of a number $k$ is a monotone sequence $\lambda=(\lambda_1\geq\lambda_2\geq\dotsm \lambda_{\ell(\lambda)}\geq 0)$ consisting of nonnegative integers whose sum is equal to $k$. Most of the times, we denote by $\vert\lambda\vert:=\sum_{j=1}^{\ell(\lambda)}\lambda_j=k$ the \emph{size} and $\ell(\lambda)$ the \emph{length} of a partition $\lambda$.} 
\begin{align*}
\mathrm{Hilb}(\Sigma,k)^\varepsilon&=\{\calI\subset \C[u_1,u_2]\mid \text{$\calI$ monomial ideal of codimension $k$},\,\, u_1,u_2\in\Sigma\}\\
&\cong\{\text{partitions of $k$}\}.
\end{align*}
Table \ref{tab:table} gives a good illustration for the relation between these ideals and the partitions.

\begin{table}[h!]
\begin{center}
  \begin{tabular}{ !{\vrule width 2pt} l !{\vrule width 1.5pt} c | c | c | c | c | c | c | c | r !{\vrule width 2pt} }
    \specialrule{.2em}{.1em}{.1em}
    1 & $u_1$ & $u_1^2$ & $u_1^3$ & $u_1^4$ & $u_1^5$ & $u_1^6$ & $u_1^7$ & \cellcolor{gray!30}$\boldsymbol{u_1^8}$ & \cellcolor{gray!15}$u_1^9$\\[6pt] \specialrule{.1em}{.05em}{.05em}
    $u_2$ & $u_1u_2$ & $u_1^2u_2$ & $u_1^3u_2$ & $u_1^4u_2$ & $u_1^5u_2$ & \cellcolor{gray!30}$\boldsymbol{u_1^6u_2}$ & \cellcolor{gray!15}$u_1^7u_2$ & \cellcolor{gray!15}$u_1^8u_2$ & \cellcolor{gray!15}$u_1^9u_2$\\[6pt] \hline
    $u_2^2$ & $u_1u_2^2$ & $u_1^2u_2^2$ & $u_1^3u_2^2$ & \cellcolor{gray!30}$\boldsymbol{u_1^4u_2^2}$ & \cellcolor{gray!15}$u_1^5u_2^2$ & \cellcolor{gray!15}$u_1^6u_2^2$ & \cellcolor{gray!15}$u_1^7u_2^2$ & \cellcolor{gray!15}$u_1^8u_2^2$ & \cellcolor{gray!15}$u_1^9u_2^2$\\[6pt] \hline
    $u_2^3$ & $u_1u_2^3$ & $u_1^2u_2^3$ & \cellcolor{gray!30}$\boldsymbol{u_1^3u_2^3}$ & \cellcolor{gray!15}$u_1^4u_2^3$ & \cellcolor{gray!15}$u_1^5u_2^3$ & \cellcolor{gray!15}$u_1^6u_2^3$ & \cellcolor{gray!15}$u_1^7u_2^3$ & \cellcolor{gray!15}$u_1^8u_2^3$ & \cellcolor{gray!15}$u_1^9u_2^3$\\[6pt] \hline   
    $u_2^4$ & \cellcolor{gray!30}$\boldsymbol{u_1u_2^4}$ & \cellcolor{gray!15}$u_1^2u_2^4$ & \cellcolor{gray!15}$u_1^3u_2^4$ & \cellcolor{gray!15}$u_1^4u_2^4$ & \cellcolor{gray!15}$u_1^5u_2^4$ & \cellcolor{gray!15}$u_1^6u_2^4$ & \cellcolor{gray!15}$u_1^7u_2^4$ & \cellcolor{gray!15}$u_1^8u_2^4$ & \cellcolor{gray!15}$u_1^9u_2^4$\\[6pt] \hline
    $u_2^5$ & \cellcolor{gray!15}$u_1u_2^5$ & \cellcolor{gray!15}$u_1^2u_2^5$ & \cellcolor{gray!15}$u_1^3u_2^5$ & \cellcolor{gray!15}$u_1^4u_2^5$ & \cellcolor{gray!15}$u_1^5u_2^5$ & \cellcolor{gray!15}$u_1^6u_2^5$ & \cellcolor{gray!15}$u_1^7u_2^5$ & \cellcolor{gray!15}$u_1^8u_2^5$ & \cellcolor{gray!15}$u_1^9u_2^5$\\[6pt] \hline
    \cellcolor{gray!30}$\boldsymbol{u_2^6}$ & \cellcolor{gray!15}$u_1u_2^6$ & \cellcolor{gray!15}$u_1^2u_2^6$ & \cellcolor{gray!15}$u_1^3u_2^6$ & \cellcolor{gray!15}$u_1^4u_2^6$ & \cellcolor{gray!15}$u_1^5u_2^6$ & \cellcolor{gray!15}$u_1^6u_2^6$ & \cellcolor{gray!15}$u_1^7u_2^6$ & \cellcolor{gray!15}$u_1^8u_2^6$ & \cellcolor{gray!15}$u_1^9u_2^6$\\[6pt] 
    \specialrule{.2em}{.1em}{.1em}
  \end{tabular}
  \caption{Illustration for the ideal $\calI_\lambda\subset \mathcal{O}_\Sigma$ generated by monomials $u_1^8,u_1^6u_2,\ldots,u_2^6$ which corresponds to the partition $\lambda=(8,6,4,3,1,1)$ of $k=23$. The entries in the table for $\lambda$ correspond to a basis of $\mathcal{O}_\Sigma/\calI_\lambda$.}
  \label{tab:table}
\end{center}
\end{table}

Consider the collection  
\[
\mathcal{E}:=\calI_{\lambda^{(1)}}\oplus\dotsm \oplus \calI_{\lambda^{(r)}}\in \calM,
\]
where $\calI_{\lambda^{(j)}}$ denotes the ideal corresponding to the partition $\lambda^{(j)}$.
To compute the tangent space $T_\mathcal{E}\calM$, we can use the modular interpretation of $\calM$. Generally, we can compute the tangent space to the moduli space of (coherent) sheaves by using the $\mathrm{Ext}^1$-groups\footnote{Recall first that, for $D=\mathrm{Spec}(k[t]/t^2)$, a \emph{deformation} over $D$ of a coherent sheaf $\calE$ over a scheme $X$ is defined to be a coherent sheaf $\calE'$ on $X':=X\times D$, flat over $D$, together with a homomorphism $\calE'\to\calE$ such that the induced map $\calE'\otimes_Dk\to \calE$ is an isomorphism. Then there is a theorem (see e.g. \cite{Hartshorne2010}) which says that deformations of $\calE$ over $D$ are in natural one-to-one correspondence with elements of $\mathrm{Ext}^1_X(\calE,\calE)$, with the zero element corresponding to the trivial deformation. Finally, by the universal property of the moduli space $\calM$, its tangent space at $\calE$ consists of the deformations of $\calE$ over $D$.}. So we get 
\begin{align}
\begin{split}
    T_\mathcal{E}\calM&=\mathrm{Ext}^1_{\overline{\Sigma}}(\mathcal{E},\mathcal{E}(-D))\\
    &=\bigoplus_{1\leq i,j\leq r} a_j/a_i\otimes \mathrm{Ext}^1_{\overline{\Sigma}}\big(\calI_{\lambda^{(i)}},\calI_{\lambda^{(j)}}(-D)\big).
\end{split}
\end{align}
In fact, we get 
\[
\Sym(T^*_\mathcal{E}\calM)=\prod_{i,j}\E\big(\lambda^{(i)},\lambda^{(j)},a_j/a_i\big)^{-1},
\]
with 
\begin{equation}
\label{eq:character}
\E(\lambda,\mu,u):=\prod_{\text{$w$ weights of $u\otimes \mathrm{Ext}^1_{\overline{\Sigma}}\big(\calI_\lambda,\calI_\mu(-D)\big)$}}(1-w^{-1}).
\end{equation}

\begin{lem}
For a partition $\lambda$, consider the generating function 
\[
\mathbb{G}_\lambda=\chi^{\C^2}(\mathcal{O}_\Sigma/\calI_\lambda)=\sum_{u_i^au_j^b\not\in \calI_\lambda}\varepsilon_1^{-a}\varepsilon_2^{-b},
\]
where the sum corresponds to the boxes $\Box=(a+1,b+1)$ in the table of $\lambda$. Moreover, define\footnote{Note that these numbers will be negative for $\Box\not\in \lambda$.} the \emph{arm-length} for the partition $\lambda$ of a box $\Box=(j,i)$ by $a_\lambda(\Box):=\lambda_i-j$ and the \emph{leg-length} by $\ell_\lambda(\Box):=\lambda_j'-i$, where $\lambda'$ denotes the transposed table of $\lambda$. 
Then the character in \eqref{eq:character} is given by 
\begin{align}
    \label{eq:character_1}
    \mathrm{Ext}^1_{\overline{\Sigma}}(\calI_\lambda,\calI_\mu(-D))&=\mathbb{G}_\mu+\varepsilon_1\varepsilon_2\overline{\mathbb{G}}_\lambda-(1-\varepsilon_1)(1-\varepsilon_2)\mathbb{G}_\mu\overline{\mathbb{G}}_\lambda\\
    \label{eq:character_2}
    &=\sum_{\Box\in\mu}\varepsilon_1^{-a_\mu(\Box)}\varepsilon_2^{\ell_\lambda(\Box)+1}+\sum_{\Box\in \lambda}\varepsilon_1^{a_\lambda(\Box)+1}\varepsilon_2^{-\ell_\mu(\Box)}.
\end{align}
\end{lem}
\begin{proof}
The proof of \eqref{eq:character_1} is given in \cite[Lemma 3.1]{Okounkov2019} and the proof of \eqref{eq:character_2} is given in \cite[Lemma 3]{CarlssonOkounkov2012}.
\end{proof}
\begin{rem}
Note that $\overline{\varepsilon^d}=\varepsilon^{-d}$ denotes the usual duality for representations and characters. Moreover, in particular, there are $\vert\lambda\vert+\vert\mu\vert$ factors in \eqref{eq:character}.
\end{rem}

Define now 
\begin{equation}
    \overline{\calM}:=\prod_i \overline{\calM}_{r_i}\supset \calM_\mathrm{ASD}(\Sigma,\mathrm{U}(r_i)).
\end{equation}
Then we can define a \emph{preliminary} partition function as
\begin{align}
\label{eq:preliminary_partition_function}
\begin{split}
\mathsf{Z}^{\mathrm{pre}}_\Sigma&:=\chi\left(\overline{\calM},\prod_i z_i^{c_2(\mathcal{E}_i)}\chi_\mathrm{top}(\mathbb{M})\right)\\
&=\sum_{\text{$r$-tuple of partitions}}z^{\#\Box}\prod_{\eta,\nu \text{ $r$-tuples}\atop\text{interactions with mass $u_k$}}\E(\eta,\nu,u_k)^{\pm 1},
\end{split}
\end{align}
where $r=\sum_i r_i$ denotes the total rank and 
\[
z^{\#\Box}:=\prod_i z_i^{c_2(\mathcal{E}_i)}=z_1^{\sum_i\vert \lambda^{(i)}\vert}z_2^{\sum_j\vert\mu^{(j)}\vert}\dotsm 
\]
We have denoted by $\mathbb{M}$ in \eqref{eq:preliminary_partition_function} the matter that is described by fermions which are defined through their representations of the gauge group of the form $(\C^{r_i})^*\otimes \C^{r_j}$. 


Using the discussions before, we can define \emph{Nekrasov's partition function} by
\begin{equation}
    \label{eq:Nekrasov_partition_function}
    \mathsf{Z}^{\mathrm{Nek}}_\Sigma:=\mathsf{Z}^\mathrm{pert}\mathsf{Z}^{\mathrm{pre}}_\Sigma\Big|_{\E\mapsto \Hat{\E},}
\end{equation}
where $\E$ is defined as in \eqref{eq:character}, 
\[
\Hat{\E}(\lambda,\mu,u):=\prod_{\text{$w$ weights of $u\otimes \mathrm{Ext}^1_{\overline{\Sigma}}\big(\calI_\lambda,\calI_\mu(-D)\big)$}\atop\text{same $w$}}(w^{1/2}-w^{-1/2}) 
\]
and $\mathsf{Z}^\mathrm{pert}$ denotes a perturbation factor given by 
\begin{equation}
\label{eq:Nekrasov_perturbative}
\mathsf{Z}^\mathrm{pert}(\varepsilon,a_1,\ldots,a_r)\risingdotseq \prod_{1\leq i_1,i_2\leq r}\,\,\, \prod_{j_1,j_2\geq 1}\I(a_{i_1}-a_{i_2}+\varepsilon(j_1-j_2)), 
\end{equation}
where $\risingdotseq$ denotes the equality up to regularization of the product on the right-hand-side.
A suitable regularization is given by (see e.g. \cite{NekrasovOkounkov2006,Okounkov2006}) using \emph{Barne's double $\Gamma$-function} \cite{Ruijsenaars2000}. Define
\begin{equation}
    \zeta_2(s,w\mid c_1,c_2):=\frac{1}{\Gamma(s)}\int_0^\infty\frac{\dd t}{t}t^s\frac{\exp(-wt)}{\prod_i(1-\exp(-c_it))},\quad c_1,c_2\in\R,\quad \mathrm{Re}\,(w)\gg0,
\end{equation}
which has a meromorphic extension in $s$ with poles at $s=1,2$. Define now 
\[
\Gamma_2(w\mid c_1,c_2):=\exp\left(\frac{\dd}{\dd s}\zeta_2(s,w\mid c_1,c_2)\right)\Bigg|_{s=0}.
\]
Using the difference equation 
\[
w\Gamma_2(w)\Gamma_2(w+c_1+c_2)=\Gamma_2(w+c_1)\Gamma_2(w+c_2),
\]
we can see that it extends to a meromorphic function of $w$. Moreover, define
\begin{equation}
    \label{eq:Lambda}
    \Lambda:=\exp(-4\pi^2\beta/r)
\end{equation}
for some parameter $\beta>0$ and the \emph{scaled} perturbation factor 
\begin{equation}
    \mathsf{Z}^\mathrm{pert}(\varepsilon,a_1,\ldots,a_r,\Lambda):=\prod_{1\leq i_1,i_2\leq r}\Gamma_2\left(\frac{\I(a_{i_1}-a_{i_2})}{\Lambda}\,\Bigg|\,\frac{\I\varepsilon}{\Lambda},\frac{-\I\varepsilon}{\Lambda}\right)^{-1},
\end{equation}
where $\Gamma_2$ is analytically continued to imaginary arguments by using
\[
\Gamma_2(xw\mid xc,-xc)=x^{\frac{w^2}{2c^2}-\frac{1}{12}}\Gamma_2(w\mid c,-c),\quad \forall x\not\in (-\infty,0].
\]
Note also that 
\[
\Gamma_2(0\mid 1,1)=\exp(-\zeta'(-1)).
\]

The instanton equations (ASD equations) are conformally invariant and can be translated to a punctured 4-sphere $S^4=\R^4\cup\{\infty\}$ by stereographic projection. Recall that Uhlenbeck's theorem for removing singularities (Theorem \ref{thm:Uhlenbeck_removable_singularities}) implies that any instanton on $\R^4$ extends to one on $S^4$. Hence, there is an interpretation for an instanton at $\infty$. Consider the group
\[
\mathcal{G}_0:=\{g\colon S^4\to \mathrm{U}(r)\mid g(\infty)=1\}
\]
and define $\calM(r,k):=\calM^k_\mathrm{ASD}(S^4,\mathrm{U}(r))/\mathcal{G}_0$. Note that $\dim_\R\calM(r,k)=4rk<\infty$ but nevertheless $\calM(r,k)$ is not compact and thus has infinite volume. This is basically due to two problems: An element of $\calM(r,k)$ is, roughly speaking, a nonlinear superposition of $k$ instantons with charge $+1$. Some can be point-like\footnote{This means that their curvature might be concentrated in a peak with respect to the Dirac delta function.}, others can go to $\infty$. Uhlenbeck's construction will only solve the first problem by replacing point-like instantons by points in $\R^4$, but it does not solve the second problem, while Nekrasov's idea of using \emph{equivariant integration} (see Appendix \ref{app:Equivariant_localization}) gives a suitable regularization procedure for those instantons.
We can rewrite Nekrasov's partition function as 
\begin{equation}
\label{eq:Nekrasov_partition_function_explicit}
\mathsf{Z}^\mathrm{Nek}_{\Sigma\to S^4}(\varepsilon,a_1,\ldots,a_r,\Lambda)=\mathsf{Z}^\mathrm{pert}(\varepsilon,a_1,\ldots,a_r,\Lambda)\sum_{k\geq 0}\Lambda^{2rk}\int_{\overline{\calM}(r,k)}1,
\end{equation}
where the integral can be computed via equivariant localization (see Appendix \ref{app:Equivariant_localization}, Equation \eqref{eq:ABDH_localization}) by using the formula
\[
\int_\Sigma1=\sum_{u\in \Sigma^{\C^\times}}\frac{1}{\det \xi|_{T_u\Sigma}}.
\]
Here we consider the group $\C^\times$ acting on a complex manifold $\Sigma$ and $\Sigma^{\C^\times}$ denotes the isolated fixed points with respect to this action. Moreover, $\xi=(\mathrm{diag}(-\I\varepsilon,\I\varepsilon),\mathrm{diag}(\I a_1,\ldots, \I a_r))\in \mathrm{SU}(2)\times \mathrm{SU}(r)$ and $\mathsf{Z}^{\mathrm{pert}}$ is defined as in \eqref{eq:Nekrasov_perturbative}.

\begin{rem}[Gluing]
A way of expressing the gluing of Nekrasov's partition function for 5-manifolds with boundary has been studied in \cite{QiuTizzianoWindingZabzine2015}. 
\end{rem}

\subsection{Equivariant BV formalism}
\label{subsec:equivariant_BV_formalism}
The equivariant formulation of DW theory and Nekrasov's construction are a main motivation to formulate a gauge formalism which deals with field theories in the equivariant setting. The equivariant extension of the BV formalism, which has been considered in \cite{CattZabz2019}, is a suitable method to deal with the AKSZ construction of DW theory as in Section \ref{sec:AKSZ_formulation_of_DW_theory} by extending it naturally to equivariant solutions of the CME, as it was also shown in \cite{CattZabz2019}. We want to give the main ideas of this approach. 
\subsubsection{Equivariant (co)homology}
consider a Lie algebra $\mathfrak{g}$ and define a $\mathfrak{g}$-differential algebra $\calO$ to be a dg algebra $(\calO,\dd)$ with Lie derivative and contraction $L_\xi\in\mathrm{Der}^0\calO$ and $\iota_\xi\in\mathrm{Der}^{-1}\calO$, respectively for all $\xi\in\mathfrak{g}$ which is linear in $\xi$ and satisfies the usual equations in Cartan's calculus. Define the subalgebra $\calO_b:=\{O\in \calO\mid L_\xi O=\iota_{\xi}O=0,\,\, \forall \xi\in\mathfrak{g}\}$. Choosing a basis $(e_a)$ of $\mathfrak{g}$, we can define its Weil model $\calW(\mathfrak{g}):=\left(\bigwedge^\bullet\mathfrak{g}^*\otimes \Sym(\mathfrak{g}^*),\dd_\calW\right)$ where the differential is given by 
\begin{align}
    \dd_\calW\theta^i&:=u^i+\frac{1}{2}[\theta,\theta]^i,\\
    \dd_\calW u^i&:= [\theta,u]^i.
\end{align}
Here, $\theta$ denote the odd coordinates of degree $+1$ on $\bigwedge^\bullet\mathfrak{g}^*$ and $u$ the even coordinates of degree $+2$ on $\Sym(\mathfrak{g}^*)$. Moreover, we have $\iota_{e_i}=\frac{\de}{\de\theta^i}$ and $L_{e_i}=\{\iota_{e_i},\dd_\calW\}$. Note that that the subcomplex $\calO_b$ is given by $\mathfrak{g}$-invariant polynomials in $u$ endowed with the differential given by $\dd_\calW$ restricted to $\calO_b$, i.e. we have $(\C[u]^\mathfrak{g}\cong \Sym(\mathfrak{g}^*)^\mathfrak{g},\dd_\calW\vert_{\calO_b})$. Let $H^\bullet_G(\calO)$ the cohomology of the subcomplex $\calO_G:=(\calO\otimes \calW(\mathfrak{g}))_b$. This is called the Weil model for $H^\bullet_G(\calO)$. Similarly, we can consider the Cartan model for $H^\bullet_G(\calO)$ by considering the graded algebra $\calO[u]:=\calO\otimes \Sym(\mathfrak{g}^*)$ together with the differential $\dd_G:=\dd-u^i\iota_{e_i}$ and the diagonal $\mathfrak{g}$-action. It is easy to check that $(\dd_G)^2=u^iL_{e_i}$. Thus, we have a dg algebra $(\calO[u]^\mathfrak{g},\dd_G)$ with $\calO[u]^\mathfrak{g}:=\{O_u\in\calO[u]\mid L_\xi O_u=0,\,\, \forall \xi\in\mathfrak{g}\}$. In fact, we have an isomorphism between the cohomology of $(\calO[u]^\mathfrak{g},\dd_G)$ and $H^\bullet_G(\calO)$. Indeed, one can check that the map 
\[
\mathrm{Exp}^\mathfrak{\mathfrak{g}}:=\exp\left(-(\iota_{e_a}\otimes\theta^a)\right)\colon \calO\otimes \calW(\mathfrak{g})\to \calO\otimes \calW(\mathfrak{g})
\]
can be restricted to an isomorphism of dg-algebras $\mathrm{Exp}^{\mathfrak{g}}\colon \calO_G\to \calO[u]^\mathfrak{g}$.

\subsubsection{Equivariant extension for AKSZ-BV theories} 
We can formulate an equivariant extension of the AKSZ construction as follows: Let $\Sigma$ be a $d$-manifold togeher with a Lie algebra $\mathfrak{g}$ acting on it by the vector fields $v_\xi$ for some $\xi\in\mathfrak{g}$. Let $\calM$ be a graded manifold together with a symplectic form $\omega$ of degree $d-1$ and a smooth Hamiltonian function $\Theta$ on $\calM$ of degree $d$. The Hamiltonian vector field of $\Theta$ is denoted by $D_\Theta$. The space of fields is given by $\calF_\Sigma=\Map(T[1]\Sigma,\calM)$ and the cohomological vector field by
\[
Q_\Sigma=\hat{\dd}_\Sigma+\hat D_\Theta=(\calS_\Sigma,\enspace),
\]
where $\calS_\Sigma=\calS_0+\calS_\Theta$ with $\calS_0$ and $\calS_\Theta$ being the Hamiltonian functions of $\hat{\dd}_\Sigma$ and $\hat D_\Theta$, respectively. The ``hat'' denotes the lift to a vector field on the mapping space $\Map(T[1]\Sigma,\calM)$ of the corresponding vector field either on the source or on the target. It is easy to check that since $(D_\Theta)^2=0$ and $(Q_\Sigma)^2=0$, the CME $(\calS_\Sigma,\calS_\Sigma)=0$ holds. Denote by $\calO_{\calF_\Sigma}$ the $\mathfrak{g}$-dg algebra of functionals on $\calF_\Sigma$ endowed with the differential $\hat{\dd}_\Sigma+\hat Q$. Consider then the contraction $\hat\iota_{v_\xi}$ and Lie derivative $\hat L_{v_\xi}$ on $\calO_{\calF_\Sigma}$ for $\xi\in\mathfrak{g}$ and their corresponding Hamiltonian functions $\calS_{\hat\iota_{v_\xi}}$ and $\calS_{\hat L_{v_\xi}}$, respectively. Define also $\calO_{\calF_\Sigma}[u]:=\calO_{\calF_\Sigma}\otimes \Sym(\mathfrak{g}^*)$. If we choose a basis $(e_i)$ of $\mathfrak{g}$, we can consider the equivariant Cartan model of the BV action as 
\[
\calS_\Sigma^\mathrm{Cartan}=\calS_\Sigma-u^i\calS_{\hat\iota_{v_i}},
\]
and therefore we get an equivariant model for the cohomological vector field as 
\[
Q_{\Sigma}^\mathrm{Cartan}=(\calS_\Sigma^\mathrm{Cartan},\enspace)=\hat{\dd}_\Sigma+\hat{D}_\Sigma-u^i\hat{\iota}_{v_i}.
\]
Define now $L_\xi:=-\xi^ic^k_{ij}u^j\frac{\de}{\de u^k}$ and $\calL_\xi:=L_\xi+\hat{L}_{v_\xi}$ and note that then $\calL_\xi\calS_\Sigma^\mathrm{Cartan}=0$. This implies that $\calS^\mathrm{Cartan}_\Sigma\in\calO[u]^\mathfrak{g}$ and it satisfies an equivariant version of the CME 
\begin{equation}
\label{eq:equivariant_master_equation}
\frac{1}{2}(\calS^\mathrm{Cartan}_\Sigma,\calS^\mathrm{Cartan}_\Sigma)+u^i\calS_{\hat{L}_{v_i}}=0. 
\end{equation}
We further define
\[
\calT_{\Sigma}:=\frac{1}{2}(\calS^\mathrm{Cartan}_\Sigma,\calS^\mathrm{Cartan}_\Sigma)-\I\hbar\Delta \calS^\mathrm{Cartan}_\Sigma=-u^i\calS_{\hat{L}_{v_i}}-\I\hbar \Delta \calS^\mathrm{Cartan}_\Sigma,
\]
such that $\calT_\Sigma=-u^i(\calS_{\hat{L}_{v_i}}+\I\hbar\Delta\calS_{\hat{\iota}_{v_i}})-\I\hbar\Delta\calS^\mathrm{Cartan}_\Sigma$. If we apply $\Delta$ to $\calS_0$ and $\calS_{\hat{\iota}_{v_i}}$, we can use the fact that they are quadratic in the fields and thus we get 
\begin{equation}
\label{eq:equivariant_BV_1}
(\Delta\calS_0,\enspace)=(\Delta\calS_{\hat{\iota}_{v_\xi}},\enspace)=0,\quad\forall X\in\mathfrak{g}.
\end{equation}
This gives us that 
\[
\Delta\calS_{\hat{L}_{v_\xi}}=\Delta(\calS_0,\calS_{\hat{\iota}_{v_\xi}})=(\Delta\calS_0,\calS_{\hat{\iota}_{v_\xi}})\pm (\calS_0,\Delta\calS_{\hat{\iota}_{v_\xi}})=0,
\]
which needs to hold for consistency with the definition of a BV algebra as in Section \ref{subsec:BV_algebras}. Hence, $\Delta\calT_\Sigma=0$ and thus $[\Delta,\hat{L}_{v_\xi}]=0$. We can then check that $Q^\mathrm{Cartan}_\Sigma\calT_\Sigma=0$. 

Let us briefly consider a more general case. For an observable $O$, note that we want the equation $\Delta O+\frac{\I}{\hbar}Q_\Sigma O=0$, with $Q_\Sigma=(\calS_\Sigma,\enspace)$, to be satisfied in order for point (2) of Theorem \ref{thm:BV} to hold. This equation can be rephrased by using the twisted BV Laplacian $\Delta_{\calS_\Sigma}:=\exp(-\I\calS_\Sigma/\hbar)\Delta\exp(\I\calS_\Sigma/\hbar)$ according to the definition of the BV Laplacian on functions on the space of fields by Equation \eqref{eq:twisting_BV_Laplacian}. 
In this case, $O$ is called a \emph{BV observable}. However, without assuming the Quantum Master Equation, we can define 
\[
\calT_\Sigma:=\left(\frac{\I}{\hbar}\right)^2\exp(-\I\calS_\Sigma/\hbar)\Delta \exp(\I\calS_\Sigma/\hbar)=\frac{1}{2}(\calS_\Sigma,\calS_\Sigma)-\I\hbar\Delta \calS_\Sigma
\]
and note that the twisted operator is given by 
\[
\Delta_{\calS_\Sigma,\calT_\Sigma}:=\exp(-\I\calS_\Sigma/\hbar)\Delta(\exp(\I\calS_\Sigma/\hbar)O)=\Delta O+\frac{\I}{\hbar}Q_\Sigma O+\left(\frac{\I}{\hbar}\right)^2\calT_\Sigma O.
\]
The observation is that then 
\[
\Delta\calT_\Sigma+\frac{\I}{\hbar}Q_\Sigma\calT_\Sigma=0,
\]
since $\Delta_{\calS_\Sigma,\calT_\Sigma}\calT_\Sigma=0$ and $(\calT_\Sigma)^2=0$. 

Moreover, if we apply $\Delta$ to $(\calS_{\hat{L}_{v_i}},\calS_{\hat{\iota}_{v_j}})$ we get $c^k_{ij}\Delta\calS_{\hat{\iota}_{v_k}}=0$. Using \eqref{eq:equivariant_BV_1}, we also get $(\calT_\Sigma,O)=(\calT_\Sigma',O)$ with $\calT_\Sigma':=-u^i\calS_{\hat{L}_{v_i}}-\I\hbar\Delta\calS_\Theta$. Define now 
\[
\calN_{\calT_\Sigma}:=\{O\in \calO_{\calF_\Sigma}[u]\mid (\calT_\Sigma',O)\in\calI_{\calT_\Sigma}\},
\]
where $\calI_{\calT_\Sigma}$ denotes the ideal generated by $\calT_\Sigma$ in $\calO_{\calF_\Sigma}[u]$ and consider the subalgebra
\[
\calN_{\calT_\Sigma}':=\{O\in\calO_{\calF_\Sigma}[u]\mid \calL_\xi O=0,\, \forall \xi\in\mathfrak{g}\}\subset\calN_{\calT_\Sigma}.
\]
Then, one can show that if \eqref{eq:equivariant_BV_1} holds, $\calN_{\calT_\Sigma}'$ is a Poisson subalgebra which is invariant under both $Q^\mathrm{Cartan}_\Sigma$ and $\Delta_{\calS_{\Sigma},\calT_\Sigma}$ and $\calT_\Sigma\in \calN_{\calT_\Sigma}'$. Moreover, it is easy to see that the ideal $\calI_{\calT_\Sigma}'\subset \calN_{\calT_\Sigma}'$ generated by $\calT_\Sigma$ is a $\Delta_{\calS_\Sigma,\calT_\Sigma}$-invariant Poisson ideal. Now one can define the algebra of \emph{quantum equivariant preobservables} as $\calO_{\calF_\Sigma}':=\calN_{\calT_\Sigma}'/\calI_{\calT_\Sigma}'$ together with the induced differential 
\begin{equation}
    \label{eq:equivariant_BV_2}
    \Delta_{\calS_\Sigma,\calT_\Sigma}[O]:=[\Delta_{\calS_\Sigma,\calT_\Sigma}O]=\left[\left(\Delta+\frac{\I}{\hbar}Q^\mathrm{Cartan}_\Sigma\right)O\right].
\end{equation}
A \emph{quantum equivariant observable} can then be defined as a quantum equivariant preobservable which is additionally $\Delta_{\calS_\Sigma,\calT_\Sigma}$-closed. 

\begin{rem}[Gauge-fixing for (graded) linear targets]
If the target $\calM$ is a graded vector space $V$, we can write the space of fields as $\calF_\Sigma=\Omega^\bullet(\Sigma)\otimes V$. Consider an invariant metric on $\Sigma$ and define a the submanifold $\calL:=\Omega^{\bullet}_\mathrm{coex}(\Sigma)\otimes V$, where $\Omega^\bullet_\mathrm{coex}(\Sigma)$ denotes the subspace of \emph{coexact} forms. Note that, in general, $\calL$ is only isotropic due to harmonic forms. However, by invariance of the metric, we have $[L_{v_\xi},\dd^*]=0$, and thus $i_\calL^*\calS_{\hat{L}_{v_\xi}}=0$, where $i_\calL\colon\calL\hookrightarrow \calF_\Sigma$. The foliation then defined by $(\calT_\Sigma,\enspace)$ is the same as the infinitesimal $\mathfrak{g}$-action and hence we have to require that the Lie group $G$ is compact. 
\end{rem}


\subsection{Equivariant DW theory}
Consider the set-up of Section \ref{subsec:AKSZ_data} and let $\mathfrak{g}$ be the Lie algebra of a Lie group $G$ acting on the 4-manifold $\Sigma$ with vector fields $v_\xi$ for $\xi\in\mathfrak{g}$. If we replace now $\dd_\Sigma$ with the equivariant differential $\dd_G:=\dd_\Sigma-u^i\iota_{v_i}$, we get the equivariant extension of the action as 
\begin{align}
    \begin{split}
        \calS^\mathrm{Cartan}_\Sigma&=\int_{T[1]\Sigma}\boldsymbol{\mu}_4\left(\langle \mathbf{Y},\dd_G\mathbf{X}\rangle+\frac{1}{2}\langle\mathbf{Y},\mathbf{Y}\rangle+\frac{1}{2}\langle\mathbf{Y},[\mathbf{X},\mathbf{X}]\rangle\right)\\
        &=\calS_\Sigma-u^i\int_\Sigma\left(\psi\iota_{v_i}Y^\dagger+\chi^\dagger\iota_{v_i}\psi^\dagger+A^\dagger\iota_{v_i}\chi+X^\dagger\iota_{v_i}A\right).
    \end{split}
\end{align}
The BV transformations of the superfields are then given by 
\begin{align}
    Q^\mathrm{Cartan}_\Sigma\mathbf{X}&=\dd_G\mathbf{X}+\mathbf{Y}+\frac{1}{2}[\mathbf{X},\mathbf{X}],\\
    Q^\mathrm{Cartan}_\Sigma\mathbf{Y}&=\dd_G\mathbf{Y}+[\mathbf{X},\mathbf{Y}],
\end{align}
i.e. in superfield notation, the equivariant cohomological vector field is given by 
\begin{equation}
    \label{eq:equivariant_cohomological_vector_field}
    Q^\mathrm{Cartan}_{\Sigma}=\int_{T[1]\Sigma}\boldsymbol{\mu}_4\left(\left(\dd_G\mathbf{X}+\mathbf{Y}+\frac{1}{2}[\mathbf{X},\mathbf{X}]\right)\frac{\delta}{\delta\mathbf{X}}+\left(\dd_G\mathbf{Y}+[\mathbf{X},\mathbf{Y}]\right)\frac{\delta}{\delta\mathbf{Y}}\right).
\end{equation}
In component fields and after Berezinian integration, we get the equivariant cohomological vector field
\begin{multline}
    Q^\mathrm{Cartan}_\Sigma=\int_\Sigma\Bigg(\left(Y+\frac{1}{2}[X,X]-u^i\iota_{v_i}A\right)\frac{\delta}{\delta X}+(\psi+\dd_AX-u^i\iota_{v_i}\chi)\frac{\delta}{\delta A}\\+(\chi^\dagger+F_A+[X,\chi]-u^i\iota_{v_a}\psi^\dagger)\frac{\delta}{\delta\chi}+([X,Y]-u^i\iota_{v_i}\psi)\frac{\delta}{\delta Y}\\+(\dd_AY+[X,\psi]-u^i_{\iota_{v_i}}\chi^\dagger)\frac{\delta}{\delta \psi}\Bigg).
\end{multline}

\subsubsection{Relation to Nekrasov's partition function}
We can construct then a gauge-fixed solution for the equivariant CME \eqref{eq:equivariant_master_equation} similarly as before whenever the chosen metric on $\Sigma$ is invariant. In particular, we have 
\begin{multline}
\label{eq:equivariant_gauge-fixing}
    \calS^\mathrm{Cartan, gf}_\Sigma=\calS^\mathrm{Cartan}_\Sigma+\int_\Sigma\left(\langle \bar X^\dagger,(b+[X,\bar X])\rangle+\langle b^\dagger,(L_v\bar X+[X,b]-[Y,\bar X])\rangle\right)\\
    +\int_\Sigma\left(\langle \bar Y^\dagger,(\eta+[X,\bar Y])\rangle+\langle \eta^\dagger,(L_v\bar Y+[X,\eta]+[\bar Y,Y])\rangle\right).
\end{multline}

If we redefine the fields, we can see that the transformations are the same as in \cite{Nekrasov2003,NekrasovOkounkov2006}. In particular, one can take $v^i$ to be the two rotations of $\C^2$ and consider instead of $u^i$ the parameters\footnote{This is true whenever the evaluation map which evaluates $u^i$ at $\varepsilon^i$ commutes with the path integral with the additional assumption that $\varepsilon^i$ are not on the support of the equivariant cohomology class induced by the path integral.} $\varepsilon^i$. In fact, the (regularized) path integral quantization of this equivariant theory leads to Nekrasov's partition function \eqref{eq:Nekrasov_partition_function_explicit}, i.e. we have
\[
\mathsf{Z}^\mathrm{BV}_{\Sigma}=\int_\calL\exp(\I\calS^\mathrm{Cartan}_\Sigma(\mathbf{X},\mathbf{Y})/\hbar)\mathscr{D}[\mathscr{X}]\mathscr{D}[\mathscr{Y}]=\mathsf{Z}^\mathrm{Nek}_{\Sigma},
\]
where $\calL$ denotes the gauge-fixing Lagrangian as in \eqref{eq:equivariant_gauge-fixing}.
It is important to note that in the mentioned construction, we are considering the perturbation around constant solutions rather than general instantons. Moreover, the methods presented in this section are only considered for the finite-dimensional case.  

\begin{rem}
Similar equivariant extensions have been studied e.g. in \cite{CattaneoFelder2010,Mosh1} for a (global) $S^1$-equivariant setting of the Poisson sigma model on the disk, and in \cite{Getzler2019} where a classical BV-equivariance under source diffeomorphisms for 1-dimensional systems is considered. Generalizations of this construction would lead to equivariant cohomological methods similarly as for the methods which appear in \cite{Pestun2012}. 
\end{rem}

\subsection{Equivariant Floer (co)homology and boundary structure}
\label{subsec:equivariant_Floer_homology}
Since the perturbative expectation value for DW theory of an observable as in \cite{Witten1988} reproduces the Floer (co)homology groups as the boundary states on the corresponding boundary 3-manifold (see also Section \ref{subsec:field_theory_formulation}), we expect to produce similarly in a perturbative way an equivariant version of Floer (co)homology as the boundary state when considering equivariant DW theory on a 4-manifold with boundary. An equivariant extension of Floer (co)homology groups in the setting of 4-manifold topology was considered in \cite{AustinBraam1996}. The convenient model chosen there is the Cartan model similarly as discussed in Section \ref{subsec:equivariant_BV_formalism}.
In particular, for $G=\mathrm{SU}(2)$ and $\mathfrak{g}=\mathrm{Lie}(G)=\mathfrak{su}(2)$ we can construct equivariant Floer homology and cohomology as follows: For any $G$-manifold $N$, define complexes 
\begin{align*}
    \Omega_{G,\bullet}(N)&:=(\Sym(\mathfrak{g})\otimes \Omega^{\dim N-\bullet}(N))^G,\\
    \Omega_G^\bullet(N)&:=(\Sym(\mathfrak{g^*})\otimes \Omega^\bullet(N))^G,
\end{align*}
endowed with the boundary and coboundary operators $\de_G$ and $\dd_G$, respectively. They define equivariant homology and cohomology $H_{G,\bullet}(N)$ and $H^\bullet_G(N)$, respectively. Consider the $\mathrm{SU}(2)$-invariant Chern--Simons action functional $S^\mathrm{CS}_N\colon \calA/\tilde\calG_0\to S^1$, where $\tilde\calG_0$ denotes the group of based gauge transformations of degree 0 and denote by $A$ and $B$ critical orbits of the perturbed action $S^\mathrm{CS}_N+f$ for which the moduli space of gradient lines between $A$ and $B$ then becomes a smooth $\mathrm{SU}(2)$-manifold $\widetilde{\calM}_0(A,B)$ (here $f$ denotes an element of a suitable class of perturbations). For $0<p<\infty$ and $\delta\in\R$, denote by $L^p_\delta(E)$, for some bundle $E\to \Sigma$ over a 4-manifold $\Sigma$, the Banach space completion of smooth and compactly supported sections of $E$ by using the norm 
\[
\|\sigma\|_{L^p_\delta}:=\left(\int_\Sigma\dd\mathrm{vol}_\Sigma(\exp(\tau\delta)\vert \sigma\vert^p)\right)^{1/p},
\]
where $\tau\colon \End(\Sigma)\to (0,\infty)$ denotes a distance function on $\Sigma$ defined as in \cite{Taubes1987}. Furthermore, for $1\leq p<\infty$, $0\leq \ell<\infty$ and $\delta\in\R$, define the weighted Sobolev space $L^p_{\ell,\delta}(E)$ to be the completion of smooth and compactly supported sections of $E$ with respect to the norm 
\[
\|\sigma\|_{L^p_{\ell,\delta}}:=\left(\int_\Sigma\dd\mathrm{vol}_\Sigma\left(\exp(\tau\delta)\sum_{0\leq k\leq \ell}\vert\nabla^k\sigma\vert^p\right)\right)^{1/p}, 
\]
where $\nabla^k:=\nabla\nabla\dotsm \nabla$ $k$-times and $\nabla\colon \Gamma^\infty_0(E\otimes_q T^*\Sigma)\to \Gamma^\infty_{0}(E\otimes_{q+1}T^*\Sigma)$ is the covariant derivative from the \emph{end-periodic}\footnote{The definition of \emph{end-periodic} connections is rather technical and we refer to \cite{Taubes1987} for a definition of \emph{end-periodic} connections and manifolds. In fact, these objects are only defined for end-periodic 4-manifolds.} connections on $E$ and $T^*\Sigma$. We have denoted by $\Gamma^\infty_0$ the space of smooth sections with compact support.

We can then consider the Atiyah--Hitchin--Singer deformation theory \cite{AtiyahHitchinSinger1977,AtiyahHitchinSinger1978} to the moduli space $\calM_\varepsilon(A,B)$ for $\varepsilon<0$, which is locally a smooth $(\mathcal{G}(N)^A\times\calG(N)^B)$-manifold whenever the cokernel of the following deformation operator vanishes.
Define the \emph{deformation} and \emph{gauge-fixing operator} at some instanton $A_t$ as
\[
a\mapsto \left(\frac{\dd}{\dd t}a+*\dd_{A_t}a+\de^2f_{A_t}a,\,\,\frac{\dd}{\dd t}a+\dd^*_{A_t}a\right)
\]
and note that this indeed defines a Fredholm operator 
\[
L^p_{\ell,\varepsilon}(\Omega^1(N\times \R,\ad P))\to L^p_{\ell-1,\varepsilon}((\Omega^0\oplus\Omega^1)(N\times \R,\ad P)).
\]
Denote by $\mathrm{sf}(A,B)$ the \emph{index}\footnote{This is again the spectral flow as defined in Section \ref{subsec:Instanton_Floer_homology}.} of this operator and note that we then get 
\[
\dim\calM_0(A,B)=\mathrm{sf}(A,B)+3-\dim \mathcal{G}(N)^A-\dim\calG(N)^B.
\]
Define then the \emph{index} of a critical orbit $A\in \calA/\tilde\calG_0$ of the perturbed Chern--Simons action functional $S^\mathrm{CS}_N+f$ as 
\[
i_N(A):=-\mathrm{sf}(\theta,A)+\dim\calG(N)^A,
\]
where $\theta$ denotes a preferred trivial connection in $\calA/\tilde\calG_0$ and $\varepsilon$ is chosen to be negative and sufficiently small. 
Fix a homotopy of the corresponding principal $\mathrm{SU}(2)$-bundle $P$ over $N$ and note that
\[
\dim\widetilde{\calM}_0(A,B)=i_N(A)-i_N(B)+\dim A-1.  
\]
For the 3-manifold $N$, we can then define graded vector spaces
\begin{align*}
    CF_{\calG,k}&:=\bigoplus_{i_N(A)+j=k}\Omega_{G,j}(A),\\
    CF^k_\calG&:=\bigoplus_{i_N(A)+j=k}\Omega_G^j(A),
\end{align*}
with boundary and coboundary operators given by 
\begin{equation}
    (\delta_\calG)_{A,B}\Psi:=\begin{cases}\de_G\Psi,& \text{if $A=B$},\\ (-1)^{\deg\Psi}(e^-_{B})_*(e_A^+)^*\Psi,&\text{if $i_N(A)>i_N(B)$},\\ 0,&\text{otherwise}.\end{cases}
\end{equation}
and 
\begin{equation}
    (D_\calG)_{A,B}\Psi:=\begin{cases}\dd_G\Psi,& \text{if $A=B$},\\ (-1)^{\deg\Psi}(e^+_{A})_*(e_B^-)^*\Psi,&\text{if $i_N(A)>i_N(B)$},\\ 0,&\text{otherwise}.\end{cases}
\end{equation}
We have used the \emph{endpoint maps}
\begin{align*}
    e_A^+\colon\widetilde{\calM}_0(A,B)\to A,\\
    e_B^-\colon\widetilde{\calM}_0(A,B)\to B.
\end{align*}
It is easy to see that the equivariant Floer complexes are actually filtered by index
\begin{align}
    CF^{(m)}_{\calG,k}&:=\bigoplus_{i_N(A)+j=k\atop i_N(A)<m}\Omega_{G,j}(A),\\
    CF^{k,(m)}_\calG&:=\bigoplus_{i_N(A)+j=k\atop i_N(A)\geq m}\Omega^j_G(A),
\end{align}
with the corresponding relative equivariant (co)homology groups
\begin{align}
    HF^{(m)}_{\calG,\bullet}&:=H_\bullet(CF_{\calG,\bullet}/CF^{(m)}_{\calG,\bullet}),\\
    HF^{\bullet,(m)}_{\calG}&:=H^\bullet(CF^\bullet_\mathcal{G}/CF^{\bullet,(m)}_\calG).
\end{align}
\begin{thm}[Austin--Braam\cite{AustinBraam1996}]
The operators $\delta_\calG$ and $D_\calG$ square to zero and define equivariant Floer homology (resp. cohomology) $HF_{\calG,\bullet}(N)$ (resp. $HF^\bullet_\calG(N)$) of the 3-manifold $N$. Moreover, there is a pairing $HF_{\calG,\bullet}(N)\times HF^\bullet_\calG(N)\to \R$ which is defined on the level of chains by $\bigoplus_A\langle\enspace,\enspace\rangle_A$. Furthermore, these groups are independent of the metric on $N$ and the choice of perturbation withing a class of perturbations, up to natural isomorphisms. An orientation preserving diffeomorphism of $N$ induces a natural map on equivariant Floer (co)homology.
\end{thm}

\begin{rem}
Equivariant Floer homology and cohomology are both modules over equivariant polynomials $\Sym(\mathfrak{g}^*)^G\cong\R[u]$ and the action of $\Sym(\mathfrak{g}^*)^G$ is symmetric with respect to the pairing. Moreover, the complexes defining equivariant homology and cohomology are filtered by index and there is an associated spectral sequence, whose $E^1$ and $E_1$ terms are the equivariant (co)homology of the critical locus.
\end{rem}

\begin{ex}[3-sphere]
The equivariant Floer groups for $S^3$ are simple since there is only one flat connection, the trivial one. Hence, we get
\begin{align*}
HF_{\calG,\bullet}(S^3)&=H_{G,\bullet}(\mathrm{pt})=\R[u],\\
HF^\bullet_\calG(S^3)&=H^\bullet_G(\mathrm{pt})=\R[\hat{u}].
\end{align*}
\end{ex}

\begin{ex}[Poincar\'e $3$-sphere]
In the case of the Poincar\'e 3-sphere $S^3_P$, there are two irreducible flat connections $A$ and $B$. These are indexed by $i_{S^3_P}(A)=0$ and $i_{S^3_P}(B)=4$ by using the standard orientation induced from $S^3$. The first term in the spectral sequence of the equivariant Floer cohomology complex has then generators endowed with index 0 and 4 coming from the connections $A$ and $B$, together with a tower of $H^\bullet_G(\mathrm{pt})$ at index 0 which represents the trivial connection. Then the dimension of the nonzero terms will give lead to the fact that all higher differentials vanish and thus, after taking the quotient with respect to the $\Z$-translations, we get 
\begin{align*}
    HF^{\bullet}_\calG(S^3_P)&=H^\bullet_G(\mathrm{pt})\oplus H^\bullet(\mathrm{pt})\oplus H^{\bullet-4}(\mathrm{pt}),\\
    HF_{\calG,\bullet}(S^3_P)&=H_{G,\bullet}(\mathrm{pt})\oplus H_\bullet(\mathrm{pt})\oplus H_{\bullet-4}(\mathrm{pt}).
\end{align*}
\end{ex}

Consider now a $\mathrm{U}(2)$-bundle $P\to \Sigma=\Sigma_1\cup_N\Sigma_2$ over an oriented $4$-manifold $\Sigma$ with $\de\Sigma_1=N$ and $\de\Sigma_2=N^\mathrm{opp}$. Moreover, consider the moduli space $\calM_0(\Sigma,\mathrm{U}(2))$ of projectively anti self-dual connections on $P$ with fixed central parts. Moreover, assume that $\calM_0(\Sigma,\mathrm{U}(2))$ is compact. Define now $\calM_0(P_1,A)$ and $\calM_0(P_2,A)$ to be the moduli spaces of anti self-dual connections modulo gauge transformations asymptotic to the flat connection defined by $A$. Then we have \begin{align}
    \dim \calM_0(P_1,A)&=C_{\Sigma_1}-i_N(A),\\
    \dim \calM_0(P_2,A)&=C_{\Sigma_2}+3-i_N(A),
\end{align}
for some constants $C_{\Sigma_1}$ and $C_{\Sigma_2}$. For the framed moduli space, we get the $\mathrm{SU}(2)$-invariant endpoint maps
\begin{align*}
    e_A^1&\colon \calM_0(P_1,A)\to A,\\
    e_A^2&\colon \calM_0(P_2,A)\to A.
\end{align*}
If $P_1$ and $P_2$ are trivial bundles such that $c_2(P_1)+c_2(P_2)=c_2(P)$, we get a gluing map with open image
\[
\calM_0(P_1,A)\times \calM_0(P_2,A)\to \calM_0(\Sigma,\mathrm{U}(2)).
\]

\begin{thm}[Austin--Braam\cite{AustinBraam1996}]
Let $\calM_0(\Sigma,\mathrm{U}(2))$ be compact and let $P_1,P_2$ run over bundles such that $c_2(P_1)+c_2(P_2)=c_2(P)$. Assume that for a generic one parameter family of metrics on $\Sigma_1$ no reducible connections exist on $P_1$ for any $A$. Assume further that $b_+^2(\Sigma_2)=0$ or $b^2_+(\Sigma_2)>1$.
\begin{enumerate}
    \item Let $a\in H^\bullet_G(\calA(P_1)/\calG_0(P_1))$ and denote by $\gamma:=\left[\calM_0(P_1,A)\right]_G\in H_{G,\bullet}(\calM_0(P_1,A))$ the fundamental class. Then 
    \[
    \Psi_{\Sigma_1}(a):=\bigoplus_{i_N(A)\geq m}(e_A^1)_*(\gamma\cap a)\in CF_{\calG,\bullet}(N)
    \]
    defines an element of $HF^{(m)}_{\calG,\bullet}(N)$ whenever $m>C_{\Sigma_1}-8$; that is, it is closed and, up to boundaries, independent of the choice of metric on $\Sigma_1$.
    \item Let $b\in H^\bullet_G(\calA(P_2)/\calG_0(P_2))$. The integration over the fiber gives an element 
    \[
    \calD(\Sigma_2)(b):=\bigoplus_{A}(e^2_A)_*b\in CF^\bullet_\calG(N)
    \]
    which is coclosed and up to coboundaries independent of the metric on $\Sigma_2$. Moreover, it also defines an element of $HF^{\bullet,(m)}_\calG(N)$ for all $m<-3-C_{\Sigma_2}$:
    \[
    \int_{\calM_0(\Sigma,\mathrm{U}(2))}a\cup b=\langle \Psi_{\Sigma_1}(a),\calD(\Sigma_2)(b)\rangle_{HF_{\calG}^{(m)}}.
    \]
\end{enumerate}
\end{thm}

\subsubsection{Irreducible connection on one side} 
Assume now that the $2d$-dimensional moduli space of anti self-dual connections on $P$ splits uniquely as 
\[
\calM_0(\Sigma,\mathrm{U}(2))=\calM_0(P_1,A)\times_A \calM_0(P_2,A),
\]
along a $\Z/2$ connection $A$. Then, we get that $H_G^\bullet(A,\R)=\R[u]$ such that $u$ is of degree 4. Moreover, suppose that $\calM_0(P_2,A)\simeq S^2$ only consists of a single reducible connection which defines a line bundle reduction $\mathbb{L}\oplus \mathbb{L}^\mathrm{opp}\to \Sigma_2$. Recall from Section \ref{subsec:Donaldson_polynomials} that the Donaldson polynomials are given by
\[
\calD(\Sigma)([C],\ldots,[C])=\int_{\calM_0(\Sigma,\mathrm{U}(2))}\mu([C])^d,
\]
for a $2$-dimensional homology class $[C]\in H_2(\Sigma,\Z)$. Let us first look at the $\Sigma_2$ side. There, we get that $\calM_0(P_2,A)\simeq S^2$ corresponds to a single reducible connection with endpoint map $e_A^2\colon \calM_0(P_2,A)\to A$ which sends $S^2$ to one single point. We can extend the map $\mu$ to equivariant cohomology $\mu\colon H_2(\Sigma_2,\Z)\to H^2_G(\calA(P_2)/\calG_0(P_2))$ and hence 
\[
\mu([C])=-2\int_{C}c_1(\mathbb{L})v\in H^2_G(\calM_0(P_2,A))=\R[u],
\]
where $v$ is of degree 2 and such that $v^2=u$. Note that we have a push-forward $(e^2_A)_*\colon \R[v]\to \R[u]$ such that 
\[
\begin{cases}v^k\mapsto u^{(k-1)/2},&\text{$k$ odd},\\ 0, &\text{otherwise}.\end{cases}
\]
Finally, the relative Donaldson polynomial in the equivariant Floer cohomology of $N$ defined by $\Sigma_2$ is given by 
\begin{equation}
    \label{eq:Donaldson_equivariant_Sigma_2}
    \calD(\Sigma_2)(\mu([C])^d)=\begin{cases}(-2)^d\left({\displaystyle\int_Cc_1(\mathbb{L})}\right)^du^{(d-1)/2},&\text{$d$ odd},\\ 0,&\text{otherwise}.\end{cases}
\end{equation}

Now we look at the $\Sigma_1$ side. There we need $\Psi_{\Sigma_1}(\boldsymbol{1}):=(e^1_A)_*(\gamma)\in CF_{\calG,\bullet}(N)$ in the equivariant Floer homology of $N$. In particular, we get 
\[
\Psi_{\Sigma_1}(\boldsymbol{1})=\int_{\calM_0(P_1,A)}\gamma=\frac{2^{(d-1)/2}}{(d-1)!}p_1^{(d-1)/2}\hat{u}^{(d-1)/2}\in H_{G,\bullet}(A)=\R[\hat{u}],
\]
where $H_{G,\bullet}(\mathrm{pt})=\R[\hat u]$. In fact, the framed moduli space of anti self-dual connections on $\Sigma_1$ is an $\mathrm{SO}(3)$-bundle and thus has a Pontryagin class $p_1$.

\begin{rem}
Note that the latter construction only holds if $d<4$, otherwise the compactness of the moduli spaces can be violated. Moreover, as it was mentioned in \cite{AustinBraam1996}, this construction also coincides with the Friedman--Morgan construction \cite{FriedmanMorgan1994} if $\Sigma$ is the blow-up of an algebraic surface, in particular, $\Sigma_2=\C \mathbb{P}^2$ and $N=S^3$, thus the moduli space on $\Sigma$ has to split along the trivial connection.
\end{rem}

\begin{rem}
Similarly as mentioned in Remark \ref{rem:Floer_group_boundary_state}, this construction is expected to be consistent with a possible equivariant extension of the BV-BFV formalism in parallel to the BV case for closed manifolds. In the same way as the BV partition function of equivariant DW theory on closed manifolds yields Nekrasov's partition function, one should be able to produce the equivariant Floer groups as a BV-BFV partition function for equivariant DW theory in the presence of boundary. This yields a boundary state of the form
\[
\Psi^\mathrm{eq.}_{\de\Sigma}(\mathds{X},\mathds{Y})=\int_{\calV_\Sigma}\underbrace{\int_{\calL}\exp(\I\calS^\mathrm{Cartan}_\Sigma(\mathbf{X},\mathbf{Y})/\hbar)\widetilde{O}(\mathbf{X},\mathbf{Y})\mathscr{D}[\mathscr{X}]\mathscr{D}[\mathscr{Y}]}_{=\left\langle\prod_{j=1}^d\widetilde{O}^{(\gamma_j)}\right\rangle_\mathrm{eq.}=\left\langle \prod_{j=1}^d\int_{\widetilde{\gamma}_j}\widetilde{W}_{k_j}\right\rangle_\mathrm{eq.}}\in HF_\calG^\bullet(\de\Sigma)\subset\calH^\calP_{\de\Sigma},
\]
where $\widetilde{O},\widetilde{\gamma}_j$ and $\widetilde{W}_{k_j}$ are equivariant extensions of the respective objects as in Section \ref{subsec:field_theory_approach_to_floer_homology} and $\langle\enspace\rangle_\mathrm{eq.}$ denotes the expectation for the BV-BFV partition function with respect to $\calS^\mathrm{Cartan}_\Sigma$.
In particular, the boundary state space $\calH^\calP_{\de\Sigma}$ in this case should include the equivariant Floer cohomology and thus a possible equivariant extension of the BFV boundary operator $\Omega^{\bullet,\calP}_{\de\Sigma}$ should be given in terms of the equivariant Floer differential. Moreover, the polarization needs to be adapted to the boundary condition defined by the state $\Psi^\mathrm{eq.}_{\de\Sigma}$. Note that here we want $\mathfrak{h}=\mathfrak{su}(2)$. 
\end{rem}

\section{Relation to Donaldson--Thomas theory}
\label{sec:Relation_to_Donaldson-Thomas_theory}
Donaldson--Thomas theory, starting first in \cite{DonaldsonThomas1998}, was motivated by formulating Donaldson's construction of invariants of 4-manifolds in higher dimensions.  
In \cite{Thomas2000}, Thomas constructed a holomorphic \emph{Casson invariant} (see Definition\ref{defn:Casson_invariant}) counting bundles on a Calabi-Yau 3-fold, by using the holomorphic Chern--Simons action functional (see Section \ref{subsec:Chern-Simons}). He developed the deformation theory necessary to obtain certain virtual moduli cycles in moduli spaces of stable sheaves whose higher obstruction groups vanish, which gives Donaldson- and Gromov--Witten- (see Section \ref{sec:Relation_to_Gromov-Witten_theory}) like invariants of Fano 3-folds. Moreover, he defined the holomorphic Casson invariant of a Calabi--Yau 3-fold $X$ and proved that it is deformation invariant. 
In particular, when considering real Chern--Simons theory as in Section \ref{subsec:Chern-Simons}, one can show that the Hessian of the action functional is symmetric and hence the deformation complex is defined by its critical locus, given by flat connections, is self-dual. Thus, by Poincar\'e duality, we get 
\[
H^i(\ad E, A)\cong (H^{3-i}(\ad E,A))^*
\]
and hence the virtual dimension of the moduli space of flat connections $\calM_\mathrm{flat}$ is given by 
\[
\mathrm{vdim}\,\calM_\mathrm{flat}=\sum_{i=0}^3(-1)^{i+1}\dim H^i(\ad E,A)=0.
\]
A similar observation is also true when considering holomorphic Chern--Simons theory (see Section \ref{subsec:Chern-Simons}). In particular, instead of Poincar\'e duality one can use Serre duality to obtain that the virtual dimension of the moduli space of holomorphic bundles is zero. This observation in fact leads to the countig of bundles over Calabi--Yau 3-folds and more generally curves in algebraic 3-folds.  

\subsection{DT invariants} 
Let $X$ be a Calabi--Yau 3-fold. Define
\[
\mathrm{Hilb}_\beta(X,k):=\{X_0\subset X\mid [X_0]=\beta,\, \chi_\mathrm{hol}(\mathcal{O}_{X_0})=k\}
\]
to be the Hilbert scheme depending on the 1-dimensional subschemes in the curve class $\beta\in H_2(X,\mathbb{Z})$ with \emph{holomorphic}\footnote{Recall that this is defined through the alternating sum of the dimension of sheaf cohomology $\chi_\mathrm{hol}(\mathcal{O}_X):=\sum_{i=0}^{\dim X}(-1)^i\dim H^i(X,\mathcal{O}_X)$.} Euler characteristic $k\in\mathbb{Z}$. By Behrend's construction \cite{Behrend2009}, one can define the \emph{Donaldson--Thomas (DT) invariants} as a weighted Euler characteristic of the underlying moduli space. In particular, it is defined as
\begin{equation}
\label{eq:DT_invariants}
\mathrm{DT}_{\beta,k}^X:=\chi_\mathrm{top}\left(\mathrm{Hilb}_{\beta}(X,k),\nu^X_{\beta,k}\right):=\sum_{l\in\Z}l\chi_\mathrm{top}\big(\big(\nu_{\beta,k}^X\big)^{-1}(l)\big),
\end{equation}
where $\chi_\mathrm{top}$ denotes the \emph{topological} Euler characteristic and 
\begin{equation}
\label{eq:Behrend_function}
\nu_{\beta,k}^X\colon \mathrm{Hilb}_{\beta}(X,k)\to \Z
\end{equation}
denotes Behrend's constructible\footnote{That is to say for a quasi-projective proper moduli scheme with a symmetric obstruction theory, in our case $\mathrm{Hilb}_{\beta}(X,k)$, the weighted Euler characteristic $\chi_\mathrm{top}(\mathrm{Hilb}_{\beta}(X,k),\nu_{\beta,k}^X)$ is the degree of the \emph{virtual fundamental class} $\left[\mathrm{Hilb}_{\beta}(X,k)\right]^{\mathrm{vir}}$ (see Footnote \ref{foot:virtual_fundamental_class} for a definition of the virtual fundamental class).} function \cite{Behrend2009}. 

\begin{rem}
A result of Brav, Bussi, Dupont and Joyce was to prove that the coarse moduli space of simple perfect complexes of coherent sheaves, with fixed determinant, on a Calabi--Yau 3-fold admits, locally, for the analytic topology, a potential, i.e. it is isomorphic to the critical locus of a function. As it was shown in \cite{PandharipandeThomas2014}, such a result is not true for general symmetric obstruction theories. Thus, the existence of such local potentials will crucially depend on the existence of a $(-1)$-shifted symplectic structure (in the setting of \cite{PantevToenVaquieVezzosi2013}, see also Section \ref{subsec:shifted_symplectic_structures}) on the derived moduli stack of simple perfect complexes on a Calabi--Yau 3-fold $X$. As it was shown in \cite{BenBassatBravBussiJoyce2015,BravBussiJoyce2019}, there is an \'etale local version of Darboux's theorem for $k$-shifted symplectic structures with $k<0$. However, extending such a local structure theorem to general $n$-shifted symplectic structures, especially for the case of interest where $n\geq0$, might lead to the existence of DT invariants for higher dimensional Calabi--Yau manifolds. Moreover, the results of \cite{BenBassatBravBussiJoyce2015} might lead to a categoryfied version of DT theory and a motivic version of DT invariants similarly as in \cite{KontsevichSoibelman2008_2}.
\end{rem}

\subsubsection{Relation to instanton Floer homology}
DT invariants can be interpreted as a complex version of a certain invariant of homology 3-spheres, called \emph{Casson invariant}, and then related to the instanton Floer homology construction as it was shown in \cite{Taubes1990}. Let us first recall the definition of a Casson invariant.

\begin{defn}[Casson invariant\cite{Saveliev1999}]
\label{defn:Casson_invariant}
Let $\mathscr{S}$ be the class of oriented integral homology $3$-spheres. A \emph{Casson invariant} is a map $\lambda\colon \mathscr{S}\to \Z$ such that:
\begin{enumerate}
    \item $\lambda(S^3)=0$, and $\lambda(\mathscr{S})$ is not contained in any proper subgroup of $\Z$.
    \item For any homology sphere $N$ and knot $k\subset N$, the difference
    \begin{equation}
    \lambda'(k):=\lambda\left(N+\frac{1}{m+1}\cdot k\right)-\lambda\left(N+\frac{1}{m}\cdot k\right),\quad m\in\Z,
    \end{equation}
    is independent of $m$.
    \item Let $k\cup \ell$ be a link in a homology sphere $N$ with linking number $\mathrm{lk}(k,\ell)=0$ and note that then for any integers $m,n$ the manifold 
    \[
    N+\frac{1}{m}\cdot k+\frac{1}{n}\cdot \ell
    \]
    is a homology 3-sphere. We have 
    \begin{equation}
    \lambda''(k,\ell):=\lambda'\left(\ell\subset N+\frac{1}{m+1}\cdot k\right)-\lambda'\left(\ell\subset N+\frac{1}{m}\cdot k\right)=0,
    \end{equation}
    for any boundary link $k\cup \ell$ in a homology sphere $N$.
\end{enumerate}
\end{defn}

\begin{thm}[Casson]
\label{thm:Casson}
There is a Casson invariant $\lambda$ which is unique up to sign. Moreover, it has the following properties:
\begin{enumerate}[$(i)$]
    \item $\lambda'(k_{3_1})=\pm 1$, where $k_{3_1}$ denotes the trefoil knot.
    \item $\lambda'(k\subset N)=\frac{1}{2}\Delta_{k\subset N}''(1)\cdot \lambda'(k_{3_1})$ for any knot $k\subset N$,
    \item $\lambda(-N)=-\lambda(N)$ where $-N$ stands for $N$ with reversed orientation,
    \item $\lambda(N_1\#N_2)=\lambda(N_1)+\lambda(N_2)$,
    \item $\lambda(N)=\mu(N)$ mod 2, where $\mu$ denotes the \emph{Rohlin invariant}\footnote{For an oriented integral homology 3-sphere $N$, there exists a smooth simply-connected oriented 4-manifold $\Sigma$ with even intersection form such that $\de\Sigma=N$. Then the signature of $\Sigma$ is divisible by 8 and $\mu(N):=\frac{1}{8}\mathrm{sign}\,\Sigma$ mod 2 is independent of the choice of $\Sigma$ \cite{Saveliev1999}. The number $\mu(N)$ is called \emph{Rohlin invariant} of $N$.}.
\end{enumerate}
\end{thm}

In \cite{Taubes1990}, Taubes has shown that the Casson invariant as in Theorem \ref{thm:Casson} can be defined by using an infinite-dimensional generalization of the topological Euler characteristic using instanton Floer homology for an oriented homology 3-sphere. Let $G$ be a finitely presented group and $N$ a homology 3-sphere. Denote by $R(G)=\Hom(G,\mathrm{SU}(2))$ where $G$ is assumed to have the discrete topology. Note that one can turn $R(G)$ into a real algebraic variety by regarding $\mathrm{SU}(2)$ as $S^4$. The quotient $\calR(G):=R(G)/\mathrm{SU}(2)$ under the $\mathrm{SU}(2)$-action given by conjugation, is the $\mathrm{SU}(2)$-character variety of $G$. Denote the open set of irreducible representations by $R^*(G)$. Then we can consider the quotient $\calR^*(G):=R^*(G)/\mathrm{SO}(3)\subset \calR(G)$. Now let $\mathcal{R}^*(N):=\calR^*(\pi_1(N))$, where $\pi_1(N)$ denotes the \emph{fundamental group} of $N$. Moreover, for some admissible perturbation function $f\colon \calA^*/\mathcal{G}\to \R/8\pi^2\Z$, define the \emph{perturbed moduli space}
\[
\calR_f(N):=\{A \in \calA\mid *F_A-4\pi^2(\nabla f)(A)=0\}/\mathcal{G}.
\]
Finally, define $\calR^*_f(N)$ to be the subset of $\calR_f(N)$ consisting of the orbits of irreducible perturbed flat connections.
Then, we have the following theorem;
\begin{thm}[Taubes\cite{Taubes1990}] 
The Euler characteristic
\[
\chi_\mathrm{top}(N)=\frac{1}{2}\sum_{A\in \mathcal{R}^*_f(N)}(-1)^{\mu(A)}
\]
with \emph{Floer index} $\mu(A):=\mathrm{sf}(\theta,A)\,\,\mathrm{mod}\,8$, where $\theta$ denotes the trivial connection, is independent of the holonomy perturbation $f$ and the metric on $N$ and equals up to an overall sign the Casson invariant of $N$.
\end{thm}
However, working with a homology 3-sphere does simplify things tremendously since it prevents the existence of non-trivial reducible flat connections before and after small perturbations. In fact, this allows to use the standard cobordism argument to show that the alternating sum is indeed independent of the choice of metric. We refer also to \cite{Saveliev2001} for an excellet overview.

\subsection{DT partition function}
\label{subsec:DT_partition_function}
Let $X$ be a Calabi--Yau 3-fold. Using the DT invariants \eqref{eq:DT_invariants}, we can define the \emph{DT partition function} as a generating function by 
\begin{equation}
\label{eq:DT_partition)_function_1}
\mathsf{Z}^\mathrm{DT}_X(q):=\sum_{(\beta,k)\in H_2(X,\mathbb{Z})\oplus\Z}\mathrm{DT}_{\beta,k}^Xq^\beta(-z)^k,
\end{equation}
where $q^\beta:=q_1^{\beta_1}\dotsm q_n^{\beta_n}$ with $\beta=\beta_1C_1+\dotsm +\beta_nC_n$ where $\{C_1,\ldots,C_n\}$ denotes a basis of $H_2(X,\Z)$ which is chosen such that $\beta_i\geq 0$ for any effective curve class. Note that $\mathsf{Z}^\mathrm{DT}_X(q)$ is a formal power series in $q_1,\ldots,q_n$ with coefficients in formal Laurent series in the complex variable $z$. In particular, since $X$ is a 3-fold, we can split the Hilbert scheme into components according to the degree and the arithmetic genus of $C$ when $\beta=[C]$. We get
\[
([C],\chi_\mathrm{hol}(\mathcal{O}_C))\in H_2(X,\Z)\oplus H_0(X,\Z).
\]
Consider a smooth divisor $D\subset X$ which algebraically takes the place of the boundary when $X$ is a smooth manifold. We can define the relative DT partition function corresponding to the pair $(X,D)$ as 
\begin{equation}
\label{eq:DT_partition_function_2}
\mathsf{Z}^\mathrm{DT}_{X/D}(q):=\sum_{(\beta,k)\in H_2(X,\mathbb{Z})\oplus H_0(X,\Z)}q^\beta(-z)^k\int_{\mathrm{Hilb}_{\beta}(X,k)}\Xi(D),
\end{equation}
where $\Xi(D)$ denotes a collection of forms constructed from the boundary conditions. We need to make sense of the integral in \eqref{eq:DT_partition_function_2}.
The integral localizes to a \emph{virtual fundamental class}\footnote{\label{foot:virtual_fundamental_class}Following \cite{Behrend2009}, for a scheme (or Deligne--Mumford stack) $X$ with cotangent complex $\mathbb{L}_X$, the virtual fundamental class is defined through a perfect obstruction theory of $E\to \mathbb{L}_X$, where $E\in D(\calO_X)$ (here $D(\calO_X)$ denotes the derived category of sheaves of $\calO_X$-modules). In particular, $E$ defines a vector bundle stack $\mathfrak{C}$ over $X$. The \emph{virtual fundamental class} $[X]^\mathrm{vir}\in A_{\mathrm{rk}\,E}(X)$ is then defined as the intersection of the fundamental class of the \emph{intrinsic normal cone} $[\mathfrak{C}_X]$ (the perfect obstruction theory $E\to\mathbb{L}_X$ induces a closed immersion of cone stacks $\mathfrak{C}_X\hookrightarrow\mathfrak{C}$) with the zero section of $\mathfrak{C}$, i.e. $[X]^\mathrm{vir}:=0^!_{\mathfrak{C}}[\mathfrak{C}_X]$. Note that we have denoted by $A_r(X)$ the \emph{Chow group} of $r$-cycles modulo rational equivalence on $X$ with values in $\Z$.} $[\calM_\mathrm{DT}]^\mathrm{vir}\in H_{2\mathrm{vdim}}(X)$, where $\calM_\mathrm{DT}$ denotes the moduli space of \emph{supersymmetric configurations}. In particular, for the Hilbert scheme, we get the virtual dimension\footnote{Typically computed by using some version of the Riemann--Roch theorem.}
\[
\mathrm{vdim}\,\mathrm{Hilb}_{\beta}(X,k)=\int_\beta c_1(X).
\]
Using this construction, we can formulate a more precise version of \eqref{eq:DT_partition_function_2} by 
\begin{equation}
\label{eq:DT_partition_function_3}
\mathsf{Z}^\mathrm{DT}_{X/D}(q)=\sum_{(\beta,k)\in H_2(X,\mathbb{Z})\oplus H_0(X,\Z)}q^\beta(-z)^k\int_{[\mathrm{Hilb}_{\beta}(X,k)]^\mathrm{vir}}[\Xi](D),
\end{equation}
where $[\Xi](D)$ denotes the collection of cohomology classes constructed from the boundary conditions. 

\begin{rem}
Note that in fact, using the virtual fundamental class, we can equivalently express the DT invariants \eqref{eq:DT_invariants} as a \emph{virtual count} by 
\[
\mathrm{DT}^X_{\beta,k}=\int_{\left[\mathrm{Hilb}_\beta(X,k)\right]^\mathrm{vir}}1.
\]
\end{rem}

\subsection{Gluing of DT partition functions}
We want to describe the gluing procedure for partition functions in DT theory. Let $X$ be given by two parts $X_1$ and $X_2$. The Hilbert scheme of the singular variety 
\[
X_0:=X_1\cup_DX_2
\]
is not nice (although well-defined) since if $C$ intersects $D$ in a nontransverse way, it will lead to different problems including the failure for the obstruction theory. However, one can resolve this by using expanded degenerations. There we allow $X_0$ to insert \emph{bubbles} $\mathbb{B}$ such that 
\[
X_0[\ell]=X_1\cup_{D_1}\underbrace{\mathbb{B}\cup_{D_2}\dotsm \cup_{D_{\ell}}\mathbb{B}}_{\ell}\cup_{D_{\ell+1}}X_2,
\]
where $\mathbb{B}:=\mathbb{P}(N(X_1/D)\oplus\mathcal{O}_D)$ is a $\mathbb{P}^1$-bundle over $D$ associated to the rank 2 vector bundle $N(X_1/D)\oplus\mathcal{O}_D$, where $N(X_1/D)$ denotes the normal bundle of $D$ in $X_1$. The divisors $D_1\cong D_2\cong\dotsm \cong D_{\ell+1}\cong D$ are all copies of $D$ which appear together with the bubbles $\mathbb{B}$. We consider now subschemes $C\subset X_0[\ell]$ of the form $C=C_0\cup C_1\cup \dotsm \cup C_\ell\cup C_{\ell+1}$ with components $\{C_i\}$ all being transverse to the divisors $\{D_i\}$. We can glue them together by the setting 
\[
C_i\cap D_{i+1}=D_i\cap C_{i+1}. 
\]
This produces a new bubble each time the intersection does not appear to be transversal. The moduli spaces for different $\ell$ can be captured into one single orbifold
\[
\mathrm{Hilb}(X_0[\bullet])=\bigcup_{\ell\geq 0}\mathrm{Hilb}(X_0[\ell])_{\mathrm{semistable}}\Big/(\mathbb{C}^\times)^\ell,
\]
where $\C^\times$ acts on the $\mathbb{P}^1$-bundle. Note that when talking about the smooth setting for a manifold $\Sigma$, we consider tubular neighborhoods of the boundary $\de\Sigma\times [0,1]$ in order to glue the two boundary pieces $\de\Sigma\times\{0\}$ and $\de\Sigma\times\{1\}$, which in the algebraic setting corresponds to the gluing of two distinguished divisors, namely $D_0=\mathbb{P}(N(X/D))\subset\mathbb{B}$ and $D_\infty=\mathbb{P}(\mathcal{O}_X)\subset \mathbb{B}$. In fact, the $\mathbb{C}^\times$-action on $\mathbb{P}^1$ preserves the divisors $D_0$ and $D_\infty$. The DT gluing formula is then given by setting
\begin{equation}
\label{eq:DT_gluing_1}
\int_{[\mathrm{Hilb}(X)]^\mathrm{vir}}[\Xi](D)=\int_{[\mathrm{Hilb}(X_0[\bullet])]^\mathrm{vir}}[\Xi](D),
\end{equation}
where $[\Xi](D)$ denotes a cohomology class defined by using boundary conditions away from $D$ or in any other way such that it is well-defined as a cohomology class on the whole family of the Hilbert schemes. Additionally, we need to require that the integral on the right-hand-side of \eqref{eq:DT_gluing_1} can be computed in terms of DT counts in $X_1$ and $X_2$ relative to the divisor $D$. In particular, this means 
\begin{equation}
    \label{eq:DT_gluing_2}
    \mathsf{Z}^\mathrm{DT}_X=\left\langle\mathsf{Z}^\mathrm{DT}_{X_1/D},(-z)^{\vert\,\cdot\,\vert}\mathsf{Z}^\mathrm{DT}_{X_2/D}\right\rangle_{H_\bullet(\mathrm{Hilb}(D))},
\end{equation}
where $\mathsf{Z}^\mathrm{DT}_{X_i/D}$ for $i=1,2$ is defined as in \eqref{eq:DT_partition_function_2} and $\vert\,\cdot\,\vert$ denotes the grading on $H_\bullet(\mathrm{Hilb}(D))=\bigoplus_kH_\bullet(\mathrm{Hilb}(D,k))$. We put the extra weight $(-z)^{\vert\,\cdot\,\vert}$ since 
\[
\chi_\mathrm{hol}(\mathcal{O}_C)=\chi_\mathrm{hol}(\mathcal{O}_{C_1})+\chi_\mathrm{hol}(\mathcal{O}_{C_2})-\chi_\mathrm{hol}(\mathcal{O}_{C_1\cap D})
\]
whenever $C=C_1\cup_{C_1\cap D}C_2$ is a transverse union along the common intersection with $D$.


\section{Relation to Gromov--Witten theory}
\label{sec:Relation_to_Gromov-Witten_theory}

\subsection{The moduli space of stable maps}
Gromov--Witten (GW) theory \cite{Gromov1985,Witten1991,Behrend1997} of a non-singular projective variety $X$ deals with the \emph{moduli spaces of stable maps}\footnote{Recall that a stable map on a non-singular projective variety $X$ is a morphism $f$ from a pointed nodal curve $C$ to $X$ such that if $f$ is constant on any component of $C$, then that component is required to have at least three distinguished points. The distinguished points are either marked points, or points lying over nodes in the normalization of $C$.} constructed by Kontsevich (first appeared in \cite{Kontsevich1995,KontsevichManin1994}), which is a generalization for the moduli space $\overline{\calM}_{g,n}$ of $n$-pointed stable curves\footnote{This is a complete algebraic curve whose automorphism group, as an $n$-pointed curve, is finite.} of genus $g$ (see also \cite{Pandharipande1999_2,FultonPandharipande1997,HZKKTVPV2003}). As a set, $\overline{\calM}_{g,n}$ is the set of isomorphism classes of $n$-pointed stable curves of genus $g$. It actually turns out that $\overline{\calM}_{g,n}$ is a quasi-projective variety of dimension $3g-3+n$ and is given as a projective compacification\footnote{One can show that $\calM_{g,n}$ is not compact by noticing that in general a family of smooth curves over a noncomplete base is not extendable to a family of smooth curves over a completion of it. The completion can be realized by allowing fibers that have nodes at worst and using the \emph{stable reduction theorem}, hence compactness of $\overline{\calM}_{g,n}$.} of the moduli space given by the isomorphism classes of smooth $n$-pointed stable curves of genus $g$, which is usually denoted by $\calM_{g,n}$.

In particular, the moduli space of stable maps from $n$-pointed nodal curves\footnote{Recall that a \emph{nodal curve} is a complete algebraic curve such that each of its points is either smooth or locally complex-analytically isomorphic to a neighborhood of the origin within the locus of the equation $xy=0$ in $\C^2$ \cite{ArbarelloCornalbaGriffiths2011}.} of genus $g$ to a non-singular projective variety $X$ representing\footnote{Any stable map $f\colon C\to X$ represents a homology class $\beta\in H_2(X,\mathbb{Z})$ whenever $f_*[C]=\beta$.} the class $\beta\in H_2(X,\mathbb{Z})$ is denoted by $\overline{\calM}_{g,n}(X,\beta)$. In fact, $\overline{\calM}_{g,n}(X,\beta)$ is a Deligne--Mumford stack. It is easy to see that if $X$ is a point, we get that 
\[
\overline{\calM}_{g,n}(X,0)=\overline{\calM}_{g,n}.
\]
A natural cohomology class on the moduli space of stable maps is given by the pullback of $X$ (see below), i.e. $\mathrm{ev}^*_i(\gamma)$ where $\gamma\in H^\bullet(X,\Z)$ for $i=1,\ldots, n$. At each point $[C,p_1,\ldots,p_n,f]\in\overline{\calM}_{g,n}(X,\beta)$, the cotangent space to $C$ at each point $p_i$ is a 1-dimensional vector space. Gluing all these spaces together, we get a line bundle $\mathbb{L}_i$ called \emph{$i$-th tautological line bundle}. Denote by $\psi_i:=c_1(\mathbb{L}_i)$ its first Chern class. Then we have the \emph{string equation} for $\overline{\calM}_{g,n}$ 
\begin{equation}
    \label{eq:string_eq}
    \int_{\overline{\calM}_{g,n+1}}\prod_{i=1}^n\psi_i^{\beta_i}=\sum_{i=1}^n\int_{\overline{\calM}_{g,n}}\psi_1^{\beta_1}\dotsm \psi_i^{\beta_i-1}\dotsm \psi_n^{\beta_n}.
\end{equation}
One can then easily prove the \emph{dilaton equation} for $\overline{\calM}_{g,n}$
\begin{equation}
    \label{eq:dilaton_eq}
    \int_{\overline{\calM}_{g,n+1}}\psi_{n+1}\prod_{i=1}^n\psi_i^{\beta_i}=(2g-2+n)\int_{\overline{\calM}_{g,n}}\prod_{i=1}^n\psi_i^{\beta_i}
\end{equation}
for $2g-2+n>0$.
Moreover, define 
\[
\langle \tau_{\beta_1},\ldots,\tau_{\beta_n}\rangle_g:=\int_{\overline{\calM}_{g,n}}\prod_{i=1}^n\psi_i^{\beta_i}.
\]
The virtual dimension of the moduli space of stable maps \cite{BehrendFantechi1997} is given by 
\begin{equation}
    \label{eq:virtual_dim_moduli_of_stable_maps}
    \mathrm{vdim}\, \overline{\calM}_{g,n}(X,\beta)=\int_\beta c_1(X)+(\dim X-3)(1-g)+n.
\end{equation}
Similarly as for the DT construction, one can also consider the virtual fundamental class
\[
\Big[\overline{\calM}_{g,n}(X,\beta)\Big]^\mathrm{vir}\in H_{2\mathrm{vdim}}(X,\mathbb{Q}).
\]
In particular, if there is no obstruction for all stable maps, we get that the virtual fundamental class is equal to the ordinary fundamental class. As already mentioned, for $i=1,...,n$ there are evaluation maps $\mathrm{ev}_i\colon \overline{\calM}_{g,n}(X,\beta)\to X$ such that $\mathrm{ev}_i(f)=f(x_i)$ for $x_i\in C$. Thus, the classes $\gamma\in H^\bullet(X,\Z)$ can be pulled back to classes in $H^\bullet(\overline{\calM}_{g,n}(X,\beta),\Q)$ by the map $\mathrm{ev}_i^*\colon H^\bullet(X,\Z)\to H^\bullet(\overline{\calM}_{g,n}(X,\beta),\Q)$.

\subsection{GW invariants}
Given cohomology classes $\gamma_1,\ldots, \gamma_n\in H^\bullet(X,\Z)$, we can define the \emph{GW invariants} by 
\begin{equation}                    
\label{eq:GW_invariants}
\mathrm{GW}_{g,\beta}^X(\gamma_1,\ldots,\gamma_n):=\int_{\big[\overline{\calM}_{g,n}(X,\beta)\big]^\mathrm{vir}}\prod_{i=1}^n\mathrm{ev}^*_i(\gamma_i).
\end{equation}
\begin{rem}
According to Witten's notation, we may sometimes also write $\left\langle \prod_{i=1}^n\gamma_i\right\rangle^X_{g,\beta}=\langle\gamma_1\dotsm\gamma_n\rangle_{g,\beta}^X$ instead of $\mathrm{GW}_{g,\beta}^X(\gamma_1,\ldots,\gamma_n)$. 
\end{rem}
One can also define a generalized version of the GW invariants, called the \emph{gravitational descendant invariants} or just \emph{descendant invariants}, defined by 
\begin{equation}
    \label{eq:gd_invariants}
    \left\langle\prod_{i=1}^n\tau_{a_i}(\gamma_i)\right\rangle_{g,\beta}^X:=\int_{\big[\overline{\calM}_{g,n}(X,\beta)\big]^\mathrm{vir}}\prod_{i=1}^n\mathrm{ev}^*_i(\gamma_i) \psi_{i}^{a_i},
\end{equation}
where $\gamma_i\in H^\bullet(X,\Z)$ and $a_i\in\mathbb{Z}_{>0}$.
We can extend the string \eqref{eq:string_eq} and dilaton \eqref{eq:dilaton_eq} equations to the moduli space of stable maps $\overline{\calM}_{g,n}(X,\beta)$. The string equation is given by 
\begin{equation}
    \label{eq:string_eq2}
    \left\langle\prod_{i=1}^n\tau_{a_i}(\gamma_i)\boldsymbol{1}\right\rangle_{g,\beta}^X
    =\sum_{j=1}^n\left\langle\prod_{i=1}^{j-1}\tau_{a_i}(\gamma_i)\tau_{a_j-1}(\gamma_j)\prod_{i=j+1}^n\tau_{a_{i}}(\gamma_{i})\right\rangle_{g,\beta}^X
\end{equation}
where $\boldsymbol{1}\in H^\bullet(X,\Z)$ denotes the unit. The dilaton equation is given by 
\begin{equation}
    \label{eq:dilaton_eq2}
    \left\langle\prod_{i=1}^n\tau_{a_i}(\gamma_i)\tau_1(\boldsymbol{1})\right\rangle_{g,\beta}^X=(2g-2+n)\left\langle \prod_{i=1}^n\tau_{a_i}(\gamma_i)\right\rangle_{g,\beta}^X.
\end{equation}
It is easy to see that \eqref{eq:string_eq2} reduces to \eqref{eq:string_eq} and \eqref{eq:dilaton_eq2} to \eqref{eq:dilaton_eq} if $X$ is a point. 

\subsection{GW partition function}
One can then define the \emph{GW partition function} as 
\begin{align}
    \label{eq:GW_partition_function}
    \begin{split}
    \mathsf{Z}^\mathrm{GW}_X(q,u)&:=\sum_{(\beta,g)\in H_2(X,\mathbb{Z})\oplus \Z}\mathrm{GW}^X_{\beta,g}(\gamma_1,\ldots,\gamma_n)q^\beta u^{2g-2}\\
    &=\sum_{(\beta,g)\in H_2(X,\mathbb{Z})\oplus \Z}q^\beta u^{2g-2}\int_{\big[\overline{\calM}_{g,n}(X,\beta)\big]^\mathrm{vir}}\prod_{i=1}^n\mathrm{ev}^*_i(\gamma_i)
    \end{split}
\end{align}

\begin{rem}
GW theory can be regarded as some $A$-mirror since the invariants can be thought of an $A$-model with path integral over the space of fields given through pseudo-holomorphic curves as for the construction of Lagrangian Floer homology (see Section \ref{subsec:Lagrangian_Floer_homology}). We refer to the excellent reference \cite{HZKKTVPV2003} for more details on GW theory and its connection to mirror symmetry.
\end{rem}

\begin{rem}[Korteweg-de-Vries (KdV) hirarchy]
\label{rem:Witten_conjecture}
Consider the generating series 
\begin{equation}
    \label{eq:generating_series_KdV}
    \mathsf{F}(t_0,t_1,\ldots):=\sum_{g\geq 0,n\geq0\atop 2g-2+n\geq 0}\frac{1}{n!}\sum_{\beta_1,\ldots,\beta_n\geq 0}\left(\int_{\overline{\calM}_{g,n}}\prod_{i=1}^n\psi_i^{\beta_i}\right)t_0^{\beta_1}t_1^{\beta_2}\dotsm t_n^{\beta_n}.
\end{equation}
In \cite{Witten1991}, Witten conjectured that the function $u:=\frac{\de^2\mathsf{F}}{\de t_0^2}$ satisfies the \emph{KdV hirarchy}, i.e. 
\begin{align}
\begin{split}
\label{eq:KdV_hirarchy}
    u_{t_1}&=uu_x+\frac{1}{12}u_{xxx},\\
    u_{t_2}&=\frac{1}{2}u^2u_x+\frac{1}{12}(2u_xu_xx+uu_{xxx})+\frac{1}{240}u_{xxxxx}\\
    &\vdots
\end{split}
\end{align}
where $x:=t_0$. Using the string equation \eqref{eq:string_eq}, we get 
\begin{equation}
\label{eq:KdV_string_equation}
\frac{\de\mathsf{F}}{\de t_0}=\sum_{i\geq 0}t_{i+1}\frac{\de\mathsf{F}}{\de t_i}+\frac{t_0^2}{2}.
\end{equation}
If $\mathsf{F}$ satisfies condition \eqref{eq:KdV_hirarchy} together with \eqref{eq:KdV_string_equation}, one can uniquely determine the generating series $\mathsf{F}$. Witten's conjecture was proven by Kontsevich in \cite{Kontsevich1992}.
\end{rem}





\subsection{Topological recursion}\label{subsec:topological_recursion}
Topological recursion started with the work of Eynard--Orantin in \cite{EynardOrantin2009} which uses the recursive methods of \emph{symplectic invariants} to solve matrix model loop equations. In particular, as an application they show how one can obtain e.g. a proof of Witten's conjecture (see Remark \ref{rem:Witten_conjecture}), Mirzakhani's recursive methods of deriving Weil--Petersson volumes \cite{Mirzakhani2007} and certain constructions of topological string theories. 
In particular, one is interested in the data given by a spectral curve $C=(\Sigma_g,x,y,B)$, where $\Sigma_g$ denotes a Riemann surface of genus $g\geq 0$, $x,y$ are two meromorphic functions on $\Sigma_g$ and $B$ is some 2-form on $\Sigma_g\times\Sigma_g$. Assume moreover that the zeros of $\dd x$ are all simple and do not coincide with the zeros of $\dd y$. Then, as a result of topological recusrsion, one obtains symmetric multidifferentials $\omega_{g,n}\in H^0(K_{\Sigma_g}(*D)^{\otimes n},(\Sigma_g)^n)^{S_n}$ for $g\geq 0$ and $n\geq 1$ such that $2g-2+n> 0$. These differentials are usually called \emph{correlators}. We have denoted by $D$ the divisor of zeros of $\dd x=0$ and by $S_n$ the symmetric group of order $n$. Note that then the multidifferentials $\omega_{g,n}$ will be holomorphic outside of $\dd x=0$ and can have poles of any order when the given variables approach the divisor $D$. The key is to define the correlators recursively. In particular, we define first the exceptional cases
\begin{align}
    \omega_{0,1}(p_1)&=y(p_1)\dd x(p_1),\\
    \omega_{0,2}(p_1,p_2)&=B(p_1,p_2).
\end{align}
Then, the recursive relation for the correlators in general is given by 
\begin{multline}
\label{eq:multidifferentials}
    \omega_{g,n}(p_1,\boldsymbol{p}_I)=\sum_{\dd x(\alpha)=0}\mathrm{Res}_{p=\alpha} K(p_1,p)\times\\
    \times\left\{\omega_{g-1,n+1}(p,\sigma_\alpha(p),\boldsymbol{p}_I)+\widehat{\sum_{h+h'=g\atop J\sqcup J'=I}^{\omega_{0,1}}}\omega_{h,1+\vert J\vert}(p,\boldsymbol{p}_J)\omega_{h',1+\vert J'\vert}(\sigma_\alpha(p),\boldsymbol{p}_{J'})\right\}.
\end{multline}
We have used the notation $I=\{2,3,\ldots,n\}$ and $\boldsymbol{p}_J=\{p_{j_1},\ldots,p_{j_k}\}$ for $J=\{j_1,j_2,\ldots,j_k\}\subseteq I$. Note that the holomorphic function $p\mapsto \sigma_\alpha(p)$ is the non-trivial involution which is locally defined at the ramification point $\alpha$ and satisfies $x(\sigma_\alpha(p))=x(p)$. Moreover, we have denoted by $\widehat{\sum^{\omega_{0,1}}}$ in the bracket the sum which leaves out the terms involving $\omega_{0,1}$. The \emph{recursion kernel} is given by
\begin{equation}
    K(p_1,p):=\frac{1}{2}\frac{1}{(y(p)-y(\sigma_\alpha(p)))\dd x(p)}\int_{\sigma_\alpha(p)}^p\omega_{0,2}(p_1,\enspace).
\end{equation}
Note that the recursion will only depend on the local behaviour of $y$ near the zeros of $\dd x$ up to functions which are even with respect to the involution $\sigma_\alpha$ and thus it only depends on $\dd y$. 
If $2g-2+n>0$, one can observe that $\omega_{g,n}(p_1,\ldots, p_n)$ are meromorphic forms on $(\Sigma_g)^n$ with poles at $p_i\in D$. Therefore, they can be expressed in terms of polynomials for a basis of 1-forms with poles only in $D$ and divergent part being odd with respect to each local involution $\sigma_\alpha$. If we denote such a basis by $(\xi^\alpha_k)$, we can write down a partition function for the spectral curve $C$ as 
\begin{equation}
    \label{eq:partition_function_spectral_curve}
    \mathsf{Z}_C(\{t_k^\alpha\};\hbar):=\exp\left(\sum_{g\geq 0,\, n\geq 1\atop 2g-2+n>0}\frac{\hbar^{g-1}}{n!}\omega_{g,n}\big|_{\xi^\alpha_k=t^\alpha_k}\right).
\end{equation}

\begin{rem}
The relation to GW theory was considered in \cite{Dunin-BarkowskiOrantinShadrinSpitz2014} and studied further in \cite{BorotNorbury2019}. In particular, in \cite{Dunin-BarkowskiOrantinShadrinSpitz2014}, they show how this construction fits into the cohomological field theory encoding GW invariants of $\mathbb{P}^1$ which corresponds to the spectral curve given by 
\[
C_{\mathbb{P}^1}=\left(\mathbb{P}^1,x=z+\frac{1}{z},\dd y=\frac{\dd z}{z},B=\frac{\dd z_1\land \dd z_2}{(z_1-z_2)^2}\right).
\]
\end{rem}

\begin{rem}
An interesting approach to the formulation of topological recursion would be to combine it with the methods of \cite{PantevToenVaquieVezzosi2013} (see Section \ref{subsec:shifted_symplectic_structures}) in order to talk about a \emph{shifted} version of the multidifferentials \eqref{eq:multidifferentials}, i.e. to consider the behaviour of forms in $\omega_{g,n}\in H^k(K_{\Sigma_g}(*D)^{\otimes n},(\Sigma_g)^n)^{S_n}$, or more generally a shifted version of the symplectic invariants as defined in \cite{EynardOrantin2009}. This would lead to a relation of the BV-BF$^k$V formalism by considering enumerative methods, such as GW invariants, for stratified spaces in the setting of perturbative gauge theories. 
\end{rem}

\begin{rem}
In \cite{BorotBouchardChidambaramCreutzig2021}, it was recently also shown how one can obtain Nekrasov's partition function by using methods of topological recursion. 
\end{rem}

\subsection{GW/DT correspondence}
Based on a duality construction involving topological strings in the limit of the large string coupling constant described in terms of a classical statistical mechanical model of crystal melting considered in \cite{OkounkovReshetikhinVafa2006}, Maulik, Nekrasov, Okounkov and Panharipande have formulated a correspondence between GW and DT partition functions in \cite{MaulikNekreasovOkounkovPandharipande2006_1,MaulikNekreasovOkounkovPandharipande2006_2}. In particular, in \cite{OkounkovReshetikhinVafa2006} one considers the crystal to be a discretization of the toric base of some Calabi--Yau manifold. Moreover, a more general duality involving the $A$-model string on a toric 3-fold was proposed. 
For this, let $X$ be a nonsingular projective 3-fold and define the generating function 
\begin{equation}
\label{eq:GW_generating_function}
\widetilde{\mathsf{Z}}^\mathrm{GW}_X\left(u, \prod_{i=1}^n\tau_{a_i}(\gamma_i)\right)_\beta:=\sum_{g\in\Z}\left\langle\!\left\langle\prod_{i=1}^n\tau_{a_i}(\gamma_i)\right\rangle\!\right\rangle_{g,\beta}^Xu^{2g-2}
\end{equation}
for a fixed nontrivial class $\beta\in H_2(X,\mathbb{Z})$. Here we have defined
\[
\left\langle\!\left\langle\prod_{i=1}^n\tau_{a_i}(\gamma_i)\right\rangle\!\right\rangle_{g,\beta}^X:=\int_{\big[\overline{\calM}'_{g,n}(X,\beta)\big]^{\mathrm{vir}}}\prod_{i=1}^n\mathrm{ev}^*_i(\gamma_i)\psi_i^{a_i},
\]
with $\overline{\calM}'_{g,n}(X,\beta)$ the moduli space of maps with possibly disconnected domain curves $C$ of genus $g$ where there are no collapsed components, hence each component $C$ represents a nontrivial\footnote{In particular, $C$ has to represent a nonzero class.} class $\beta\in H_2(X,\Z)$. Let $I_k(X,\beta)$ be the moduli space of ideal sheaves\footnote{These are torsion-free sheaves of rank 1 with trivial determinant.} satisfying $\chi_\mathrm{hol}(\mathcal{O}_Y)=k$ and $[Y]=\beta\in H_2(X,\Z)$ for some subscheme $Y\subset X$.
The descendent fields $\tilde\tau_a(\gamma)$ in DT theory correspond to $(-1)^{a+1}\mathrm{ch}_{a+2}(\gamma)$, where 
\begin{equation}
\label{eq:chern_class}
\mathrm{ch}_{a+2}(\gamma)\colon H_\bullet(I_k(X,\beta),\Q)\to H_{\bullet-2a+2-\ell}(I_k(X,\beta),\Q)
\end{equation}
for $\gamma\in H^\ell(X,\Z)$. One can express the map $\mathrm{ch}_{a+2}(\gamma)$ by considering the projections $\pi_1$ and $\pi_2$ to the first and second factor of $I_k(X,\beta)\times X$, respectively and the universal sheaf $\mathfrak{I}\to I_k(X,\beta)\times X$. One can show that there is\footnote{This is due to the fact that $\mathfrak{I}$ is $\pi_1$-flat and $X$ is nonsingular.} a finite resolution of $\mathfrak{I}$ by locally free sheaves on $I_k(X,\beta)\times X$. This will guarantee that the Chern classes of $\mathfrak{I}$ are well-defined and thus one can express the map \eqref{eq:chern_class} by 
\[
\mathrm{ch}_{a+2}(\gamma)(\xi)=(\pi_1)_*\big(\mathrm{ch}_{a+2}(\mathfrak{I})(\pi_2)^*(\gamma)\cap(\pi_1)^*(\xi)\big).
\]
We can then define the descendent invariants by 
\begin{equation}
\label{eq:descendent_invariants}
\left\langle\prod_{i=1}^n\tilde\tau_{a_i}(\gamma_i)\right\rangle_{k,\beta}^X=\int_{\big[I_k(X,\beta)\big]^\mathrm{vir}}\prod_{i=1}^n(-1)^{a_i+1}\mathrm{ch}_{a_i+2}(\gamma_i).
\end{equation}

Using this construction, and noticing that the moduli space $I_k(X,\beta)$ is canonically isomorphic to the Hilbert scheme $\mathrm{Hilb}_{\beta}(X,k)$, we can define the DT partition function by 
\[
\mathsf{Z}^\mathrm{DT}_X\left(q,\prod_{i=1}^n \tilde\tau_{a_i}(\gamma_i)\right)_\beta=\sum_{k\in\Z}\left\langle \prod_{i=1}^n\tilde\tau_{a_i}(\gamma_i)\right\rangle_{k,\beta}^Xq^k.
\]
We then define the \emph{reduced} DT partition function by formally removing the degree zero contributions
\[
\widetilde{\mathsf{Z}}^\mathrm{DT}_X\left(q,\prod_{i=1}^n\tilde\tau_{a_i}(\gamma_i)\right)_\beta:=\frac{\mathsf{Z}^\mathrm{DT}_X\left(q,\prod_{i=1}^n \tilde\tau_{a_i}(\gamma_i)\right)_\beta}{\mathsf{Z}^\mathrm{DT}_X(q)_0}.
\]
For simplicity, we state the GW/DT correspondence for \emph{primary fields} $\tau_0(\gamma)$ and $\tilde\tau_0(\gamma)$. 
\begin{Conj}[Maulik--Nekrasov--Okounkov--Pandharipande\cite{MaulikNekreasovOkounkovPandharipande2006_1}]\label{conj:MNOP_1}
Considering the change-of-variables $\exp(\I u)=-q$, we have 
\[
(-\I u)^d\widetilde{\mathsf{Z}}^\mathrm{GW}_X\left(u,\prod_{i=1}^n\tau_0(\gamma_i)\right)_\beta=(-q)^{-d/2}\widetilde{\mathsf{Z}}^\mathrm{DT}_X\left(q,\prod_{i=1}^n\tilde\tau_0(\gamma_i)\right)_\beta,
\]
where $d:=\int_\beta c_1(X)$.
\end{Conj}

\begin{rem}
Conjecture \ref{conj:MNOP_1} has been proven for $X$ being a toric 3-fold in \cite{MaulikOblomkovOkounkovPandharipande2011}.
\end{rem}
\subsubsection{Relative GW/DT correspondence}
An important concept for us is the case for \emph{relative theories} which corresponds to the algebraic case when considering defects of a manifold in the smooth setting. Let $X$ be a nonsingular projective 3-fold and let $D\subset X$ be a nonsingular divisor. The \emph{relative} GW invariants are then defined by 
\begin{equation}
    \label{eq:relative_GW_invariants}
    \left\langle\!\left\langle\prod_{i=1}^n\tau_{a_i}(\gamma_i)\,\Bigg|\,\eta\right\rangle\!\right\rangle^{X/D}_{g,\beta}:=\frac{1}{\vert\mathrm{Aut}\,(\eta)\vert}\int_{\big[\overline{\calM}'_{g,n}(X/D,\beta,\eta)\big]^\mathrm{virt}}\prod_{i=1}^n\mathrm{ev}^*_i(\gamma_i)\psi_i^{a_i}\prod_{j=1}^m\mathrm{ev}^*_j(\delta_j),
\end{equation}
where $\beta\in H_2(X,\Z)$ is such that $\int_\beta[D]\geq0$, $\eta=(\eta_j)$ a partition whose components satisfy $\sum_j\eta_j=\int_\beta[D]$ together with a certain ordering condition\footnote{In fact, we can consider a \emph{weighted} partition $\eta$ which consists of tuples $(\eta_j,\delta_{\ell_j})$. One then chooses the \emph{standard order}, i.e. $(\eta_j,\delta_{\ell_j})>(\eta_{j'},\delta_{\ell_{j'}})$ if $(\eta_j>\eta_{i'})$ or if $\eta_i=\eta_{i'}$ and $\ell_i>\ell_{i'}$. Then $\eta$ is basically the partition $(\eta_1,\ldots,\eta_m)$ that is obtained from the standard order.} with respect to a basis $\delta_1,\ldots, \delta_m$ of $H^\bullet(D,\Q)$, $\overline{\calM}'_{g,n}(X/D,\beta,\eta)$ denotes the moduli of stable relative maps with possibly disconnected domains and relative multiplicities determined by $\eta$, and $\mathrm{ev}_j\colon \overline{\calM}'_{g,n}(X/D,\beta,\eta)\to D$ are determined by the relative points. Hence, the \emph{relative} GW partition function is given by
\begin{equation}
    \label{eq:relative_GW_partition_function}
    \widetilde{\mathsf{Z}}^\mathrm{GW}_{X/D}\left(u, \prod_{i=1}^n\tau_{a_i}(\gamma_i)\right)_{\beta,\eta}=\sum_{g\in\Z}\left\langle\!\left\langle\prod_{i=1}^n\tau_{a_i}(\gamma_i)\,\Bigg|\, \eta\right\rangle\!\right\rangle^{X/D}_{g,\beta}u^{2g-2}.
\end{equation}
Similarly as we have seen in Section \ref{subsec:DT_partition_function}, we can construct a relative version of the DT partition function by integration over the moduli space of relative ideal sheaves. Let $\eta$ be a cohomology weighted partition with respect to the basis $\delta_1,\ldots,\delta_m\in H^\bullet(D,\Q)$ and let $\vert\eta\vert:=\sum_j\eta_j$.
The relative moduli space\footnote{We can construct a proper moduli space $I_k(X/D,\beta)$ consisting of stable ideal sheaves relative on the degenerations $X[\ell]$ of $X$. We define an ideal sheaf on $X[\ell]$ to be \emph{predeformable} if for each singular divisor $D_l\subset X[\ell]$, the induces map 
\[
\mathscr{I}\otimes_{\mathcal{O}_{X[\ell]}}\mathcal{O}_{D_l}\to \mathcal{O}_{X[\ell]}\otimes_{\mathcal{O}_{X[\ell]}}\mathcal{O}_{D_l}
\]
is injective. Moreover, we call an ideal sheaf $\mathscr{I}$ on $X[\ell]$ relative to $D_l$ \emph{stable} if $\mathrm{Aut}(\mathscr{I})$ is finite. We can then define the moduli space $I_k(X/D,\beta)$ by the parametrization of stable, predeformable, ideal sheaves $\mathscr{I}$ on degenerations $X[\ell]$ relative to $D_l$ satisfying $\chi_\mathrm{hol}(\mathcal{O}_Y)=k$ and $\pi_*[Y]=\beta\in H_2(X,\Z)$ with $\pi\colon X[\ell]\to X$ denoting the canonical stabilization map. In particular, $I_k(X/D)$ is a complete Deligne--Mumford stack together with a canonical perfect obstruction theory.} $I_k(X/D,\beta)$ in fact has still dimension $\int_\beta c_1(X)$. The relative DT partition function is then given by 
\begin{equation}
    \label{eq:relative_DT_partition_function}
    \mathsf{Z}^\mathrm{DT}_{X/D}\left(q,\prod_{i=1}^n\tilde\tau_{a_i}(\gamma_i)\right)_{\beta,\eta}=\sum_{k\in\Z}\left\langle\prod_{i=1}^n\tilde\tau_{a_i}(\gamma_i)\,\Bigg|\, \eta\right\rangle_{k,\beta}q^k,
\end{equation}
where the descendent invariants in the relative DT theory are given by 
\[
\left\langle\prod_{i=1}^n\tilde\tau_{a_i}(\gamma_i)\,\Bigg|\,\eta\right\rangle_{k,\beta}=\int_{\big[I_k(X/D,\beta\big]^\mathrm{vir}}\prod_{i=1}^n(-1)^{a_i+1}\mathrm{ch}_{a_i+2}(\gamma_i)\cap \epsilon^*C_\eta
\]
with the canonical intersection map $\epsilon\colon I_k(X/D,\beta)\to \mathrm{Hilb}\left(D,\int_\beta[D]\right)$ to the Hilbert scheme of points and 
\[
C_\eta:=\frac{1}{\prod_i\eta_i\vert\mathrm{Aut}(\eta)\vert}\prod_{j=1}^mP_{\delta_j}[\eta_j]\cdot\boldsymbol{1}\in H^\bullet\left(\mathrm{Hilb}\left(D,\vert\eta\vert\right),\Q\right),
\]
where we follow the notation of \cite{Nakajima1999}. The collection $(C_\eta)_{\vert\eta\vert=k}$ is called \emph{Nakajima basis} of the cohomology of the Hilbert scheme $\mathrm{Hilb}(D,k)$. The Nakajima basis is \emph{orthogonal} with respect to the Poincar\'e pairing on the cohomology
\[
\int_{\mathrm{Hilb}(D,k)}C_\eta\cup C_\nu=\frac{(-1)^{k-\ell(\eta)}}{\prod_{i}\eta_i\vert\mathrm{Aut}(\eta)\vert}\delta_{\nu,\eta^\lor},
\]
where $\eta^\lor$ denotes the dual\footnote{We define the dual partition by taking the Poincar\'e dual of the cohomology weights.} partition to the weighted partition $\eta$ and $\ell(\eta)$ is the length of the partition $\eta$.
We can define the \emph{reduced} relative DT partition function by 
\begin{equation}
    \label{eq:reduced_relative_DT_partition_function}
    \widetilde{\mathsf{Z}}^\mathrm{DT}_{X/D}\left(q,\prod_{i=1}^n\tilde\tau_{a_i}(\gamma_i)\right)_{\beta,\eta}:=\frac{\mathsf{Z}^\mathrm{DT}_{X/D}\left(q,\prod_{i=1}^n \tilde\tau_{a_i}(\gamma_i)\right)_{\beta,\eta}}{\mathsf{Z}^\mathrm{DT}_{X/D}(q)_0}.
\end{equation}
Again, for simplicity, we restrict the relative GW/DT correspondence to primary fields. 
\begin{Conj}[Maulik--Nekrasov--Okounkov--Pandharipande\cite{MaulikNekreasovOkounkovPandharipande2006_2}]
Considering the change-of-variables $\exp(\I u)=-q$, we have
\[
(-\I u)^{d+\ell(\eta)-\vert\eta\vert}\widetilde{\mathsf{Z}}^\mathrm{GW}_{X/D}\left(u,\prod_{i=1}^n\tau_0(\gamma_i)\right)_{\beta,\eta}=(-q)^{-d/2}\widetilde{\mathsf{Z}}^\mathrm{DT}_{X/D}\left(q,\prod_{i=1}^n\tilde\tau_0(\gamma_i)\right)_{\beta,\eta},
\]
where $d:=\int_\beta c_1(X)$.
\end{Conj}
\subsection{GW invariants on (graded) supermanifolds}
Recall that in the BV-BFV formalism we consider the space of boundary fields given by the BFV space $(\calF^\de_{\de M},\omega^\de_{\de M})$ associated to a smooth variety $M$, which is a graded symplectic supermanifold. It would be interesting to understand the GW invariants on (graded) supermanifolds. Recently, an approach to this aim has been done in \cite{KesslerSheshmaniYau2020} for \emph{super Riemann surfaces}, i.e. a complex supermanifold $M$ of dimension $1\mid 1$ endowed with a holomorphic distribution $\calD\subset TM$ of rank $0\mid 1$ such that $\calD\otimes \calD\cong TM/\calD$ (the isomorphism should be induced from the commutator of vector fields). Denote by $I$ the almost complex structure on $TM$. They define the the notion of \emph{super $J$-holomorphic curve} (recall Section \ref{subsec:Lagrangian_Floer_homology} for the case of usual $J$-holomorphic curves) to be a map $\Phi\colon M\to N$ for some almost K\"ahler manifold $(N,\omega,J)$, if the differential $\bar D_J\Phi:=\frac{1}{2}(1+I\otimes J)\dd\Phi\Big|_{\calD}\in \Gamma(\calD^*\otimes \Phi^*TN)$ vanishes. 
It was shown in \cite{KesslerSheshmaniYau2020} that under certain assumptions the moduli space of super $J$-holomorphic curves $\Phi\colon M\to N$, denoted by $\calM^{\mathrm{super}}([\Phi],J)$ is a smooth supermanifold of dimension 
\[
2n(1-g)+2\int_Cc_1(N)\,\,\,\Big|\,\,\, 2\int_C c_1(N),
\]
where $g$ denotes the topological genus of $M$ and $C\in H_2(N)$ denotes the homology class of the image of the reduction $\Phi_\mathrm{red}\colon M_\mathrm{red}\to N$ (see \cite{KesslerSheshmaniYau2020} for a detailed discussion). Using the construction of GW invariants through the moduli space of stable curves, it is expected that also there we would be able to obtain a smooth supermanifold structure for some non-singular projective variety $X$ by using similar methods.


\appendix

\section{Equivariant localization}
\label{app:Equivariant_localization}
In this appendix we want to recall the most important notions on equivariant localization techniques \cite{Szabo2000,BerlineVergne1982,BerlineVergne1983}, notions of equivariant (co)homolgy according to \cite{AtiyahBott1984} and how they fit into the field theory setting, especially to the BV formalism \cite{Neressian1993,Szabo2000}.

\subsection{The ABBV method}
A way of treating (finite-dimensional) path integrals in QFT, similarly as through the saddle point approximation, is given by a construction that uses the symplectic structure of the underlying manifold, which is known today as the Atiyah--Bott--Berline--Vergne (ABBV) equivariant localization \cite{AtiyahBott1984,BerlineVergne1982,BerlineVergne1983} based on the Duistermaat--Heckman theorem \cite{DuistermaatHeckman1982} for symplectic manifolds (see also \cite{Szabo2000}). Formally, as already seen, we describe the partition function as a functional integral of the form 
\[
\mathsf{Z}(\hbar)=\int_\calM \frac{\omega^n}{n!}\exp(\I \Theta/\hbar),
\]
where we consider a Hamiltonian $G$-space $(\calM,\omega,\Theta,G)$ . Assume moreover that the Hamiltonian $\Theta$ is a Morse function as in Section \ref{subsec:Morse_homology} and let $v_\Theta$ be the Hamiltonian vector field of $\Theta$. Using the \emph{stationary phase expansion formula}, we can expand $\mathsf{Z}$ around critical points of $\Theta$ given by 
\[
\mathsf{Z}(\hbar)\sim\left(2\pi\I\hbar\right)^n\sum_{p\in\calM\atop\text{$p$ critical point of $\Theta$}}(-\I)^{\mathrm{ind}_p\Theta}\exp(\I \Theta(p)/\hbar)\sqrt{\frac{\det \omega(p)}{\det \de^2\Theta(p)}}+\mathrm{O}(\hbar^{n+1}),
\]
where $\de^2\Theta(x)=\left(\frac{\de^2 \Theta}{\de x^i\de x^j}(x)\right)$ denotes the Hessian of $\Theta$.
In \cite{DuistermaatHeckman1982}, Duistermaat and Heckman found a class of Hamiltonian $G$-spaces for which the $\mathrm{O}(\hbar^{n+1})$ term vanishes in the above formula. Suppose that $\calM$ is a compact symplectic manifold of dimension $2n$ endowed with a Riemannian metric and suppose that $v_\Theta$, the Hamiltonian vector field of $\Theta$, generates the global Hamiltonian action of a torus $\mathbb{T}$ on $\calM$ (we can also work with the circle $S^1$ for simplicity) and is the Killing vector for the metric. The simplicity of considering the circle action implies that $\omega+\Theta$ is the equivariant extension of the symplectic form, i.e. it is closed with respect to the equivariant differential $\dd_{v_\Theta}:=\dd+\iota_{v_\Theta}$. Then the partition function can be written by 
\[
\mathsf{Z}(\hbar)=\int_\calM\alpha(\hbar),
\]
where 
\[
\alpha(\hbar):=(-\I\hbar)^n\exp(\I (\omega +\Theta)/\hbar)=(-\I\hbar)^n\exp(\I \Theta/\hbar)\sum_{k=0}^n
\left(\frac{\I}{\hbar}\right)^k\frac{\omega^k}{k!}.
\]
Note that since $\omega+\Theta$ is equivariantly closed, we have $\dd_{v_\Theta}\alpha=0$. This allows us to apply the Berline--Vergne localization formula \cite{BerlineVergne1982,BerlineVergne1983} which gives
\begin{equation}
    \label{eq:Berline-Vergne}
    \int_\calM\alpha(\hbar)=\left(2\pi\I\hbar\right)^n\sum_{\text{$p$ critical point of $\Theta$}}\frac{\exp(\I \Theta(p)/\hbar)}{\mathrm{Pfaff}\,\dd V(p)},
\end{equation}
where $\dd V(p):= \omega^{-1}(p)\de^2\Theta(p)$. Consider now the action of $\mathbb{C}^\times$ on some complex manifold $\Sigma$ and denote by $\Sigma^{\mathbb{C}^\times}$ the set of isolated fixed points. Moreover, assume that the action of $\mathrm{U}(1)\subset \mathbb{C}^\times$ is generated by some Hamiltonian $\Theta$ with respect to a symplectic form $\omega$.
Then the Atiyah--Bott--Duistermaat--Heckman equivariant localization formula is given by 
\begin{equation}
    \label{eq:ABDH_localization}
    \int_\Sigma \exp(\omega-2\pi\xi\Theta)=\sum_{u\in \Sigma^{\mathbb{C}^\times}}\frac{\exp(-2\pi \xi\Theta(u))}{\det\xi\vert_{T_u\Sigma}}.
\end{equation}
Here $\xi$ is considered to be an element of the Lie algebra of $\mathbb{C}^\times$ given by $\mathbb{C}$, so it can act in the complex tangent space $T_u\Sigma$ to some fixed point $u\in\Sigma$. An important remark is that \eqref{eq:ABDH_localization} can also work for non-compact manifolds $\Sigma$. 

\subsection{Relation to the BV-BFV formalism}
We have seen in Section \ref{subsec:BFV_formalism} that in the BFV formalism we are considering a $\mathbb{Z}$-graded supermanifold $\calF^\de$ together with a symplectic form $\omega^\de$ of degree $0$, thus we have a usual symplectic manifold $(\calF^\de,\omega^\de)$. Hence, we can use the equivariant localization construction without corrections. In the case of the BV formalism the symplectic structure $\omega$ is odd of degree $-1$ inducing the anti-bracket $(\enspace,\enspace)$ which requires the equivariant localization construction to be adapted to this case. We refer to \cite{Neressian1993,Szabo2000} for such a BV formulation. 

\section{Configuration spaces on manifolds with boundary}
\label{app:configuration_spaces_on_manifolds_with_boundary}
In this appendix we want to recall some of the most important notions on configuration spaces on manifolds with boundary and their compactification. 
We refer to \cite{CamposIdrissiLambrechtsWillwacher2018,BottTaubes1994,Bott1996,Kontsevich1993_2} for an excellent introduction on this subject.

\subsection{Open configuration spaces}
Let $\Sigma$ be a closed $d$-manifold. Denote the open configuration space of $n$ points in $\Sigma$ by 
\begin{equation}
\label{eq:open_configuration_space}
\mathrm{Conf}_n(\Sigma):=\{(u_1,\ldots,u_n)\in \Sigma^n\mid u_i\not=u_j,\, \text{for $i\not=j$}\}.
\end{equation}
This is a manifold with corners of dimension $d\cdot n$. If $\Sigma$ has boundary, we consider can consider the configuration space of $n$ points in the bulk and $m$ points on the boundary given by 
\begin{multline}
\label{eq:open_configuration_space_boundary}
\mathrm{Conf}_{n,m}(\Sigma):=\{(u_1,\ldots,u_n,\mathbb{u_1},\ldots,\mathbb{u}_m)\in \Sigma^n\times (\de\Sigma)^m\mid u_i\not=u_j,\, \text{for $i\not=j$ with $1\leq i,j\leq n$} \\
\text{and }\mathbb{u}_\ell\not=\mathbb{u}_k,\, \text{for $\ell\not=k$ with $1\leq\ell,k\leq m$}\}.  
\end{multline}
Moreover, we have 
\[
\dim \mathrm{Conf}_{n,m}(\Sigma)=d\cdot n+(d-1)\cdot m. 
\]
\subsection{Local formulation} 
A special case is when we consider the configuration of points locally on $\R^d$. In particular, there is a symmetry group acting on $\mathrm{Conf}_n(\R^d)$, which is given by the $(d+1)$-dimensional Lie group $G^{(d+1)}$ consisting of scaling and translation, i.e. $u\mapsto au+b$ with $a\in\R$ and $b\in \R^d$. The quotient $\mathrm{C}_n(\R^d):=\mathrm{Conf}_n(\R^d)/G^{(d+1)}$ is then a manifold of dimension $d\cdot n+(d-1)\cdot m-(d+1)$. 
For the case with boundary, i.e. when considering the $d$-dimensional upper half-space $\mathbb{H}^d=\{(u_1,\ldots,u_d)\in\R^d\mid u_d\geq 0\}$, one can see that there is a $d$-dimensional Lie group $G^{(d)}$ acting on $\mathrm{Conf}_{n,m}(\mathbb{H}^d)$ also by scaling and translation, i.e. $u\mapsto au+b$ with $a\in\R$ and $b\in \R^{d-1}$. The quotient $\mathrm{C}_{n,m}(\mathbb{H^d}):=\mathrm{Conf}_{n,m}(\mathbb{H}^d)/G^{(d)}$ is then a manifold of dimension $d\cdot n+(d-1)\cdot m -d$.

\subsection{Compactification}
There is a natural way of compactifying the open configuration space. The compactification construction was first formulated for algebraic varieties by Fulton--MacPherson \cite{FulMacPh} and later adapted to the smooth setting of manifolds by Axelrod--Singer \cite{AS,AS2}. This compactification is often called \emph{FMAS compactification} (see also \cite{Sinha2004} for an introduction for the approach on manifolds). 
Let us give some ideas of the Fulton--MacPherson construction. Let $S$ be a finite set and consider the space $\Map(S,\Sigma)$ of maps from $S$ to $\Sigma$. Moreover, consider the smooth blow up $\mathrm{B}\ell (\Map(S,\Sigma),\Delta_S)$, where $\Delta_S$ denotes the diagonal $\Delta_S\subset\Map(S,\Sigma)$, consisting of constant maps $S\to \Sigma$. Denote by $\mathrm{Conf}_S(\Sigma)$ the space of embeddings of $S$ into $\Sigma$. One can then observe that for every inclusion $K\subset S$ there are natural projections $\Map(S,\Sigma)\to \Map(K,\Sigma)$ and corresponding arrows $\mathrm{Conf}_S(\Sigma)\to \mathrm{Conf}_K(\Sigma)$ by restriction of maps from $S$ to $K$ as a functorial approach. Further, one can show that the inclusions $\mathrm{Conf}_S(\Sigma)\subset \Map(S,\Sigma)$ can be lifted to inclusions $\mathrm{Conf}_S(\Sigma)\subset \mathrm{B}\ell (\Map(S,\Sigma),\Delta_S)$ since these sets avoid all diagonals. Thus, for a finite set $X$, we have a canonical inclusion 
\[
\mathrm{Conf}_X(\Sigma)\hookrightarrow \bigotimes_{S\subset X\atop \vert S\vert\geq2}\mathrm{B}\ell (\Map(S,\Sigma),\Delta_S)\times \Map(S,\Sigma).
\]
The Fulton--MacPherson compactification, denoted by $\overline{\mathrm{Conf}_{X}(\Sigma)}$, is then defined as the closure of $\mathrm{Conf}_X(\Sigma)$ in this embedding.
It turns out that the compactified configuration space is a \emph{manifold with corners} and comes with equivariant functorial properties under embeddings and that the propagators do indeed extend smoothly to this compactification in certain important cases, e.g. when $\Sigma=\R^3$, and $\mathscr{P}_{ij}\colon \mathrm{C}_n(\R^3)\to S^2$, $\mathscr{P}_{ij}(x_1,\ldots,x_n):=\frac{x_j-x_i}{\|x_j-x_i\|}$ for $1\leq i,j\leq n$, such an extensions holds, which is important for different aspects in the theory of Vassiliev's knot invariants arising from configuration space integrals \cite{Kontsevich1993_2,Kontsevich1994,Bott1996,BC}.

\begin{rem}[Graphs]
If we consider a graph $\Gamma$ in $\Sigma$, i.e. a graph whose vertex set $V(\Gamma)$ is contained in $\Sigma$, we will write $\mathrm{Conf}_\Gamma(\Sigma):=\mathrm{Conf}_{V(\Gamma)}(\Sigma)$ for the configuration space of points in $\Sigma$ which are vertices of $\Gamma$. Moreover, if $V(\Gamma)=V_b(\Gamma)\cup V_\de(\Gamma)$ with $V_b(\Gamma)$ the set of vertices of $\Gamma$ lying in the bulk and $V_\de(\Gamma)$ de set of vertices of $\Gamma$ lying on the boundary, we have $\mathrm{Conf}_\Gamma(\Sigma):=\mathrm{Conf}_{\vert V_b(\Gamma)\vert,\vert V_\de(\Gamma)\vert}(\Sigma)$.
\end{rem}

\subsection{Boundary structure}
Since $\overline{\mathrm{Conf}_{n,m}(\Sigma)}$ is a manifold with corners, we can consider its boundary. The boundary is given by the collapsing of points in different situations, which we call \emph{boundary strata}. In particular, since we have a manifold with corners, the boundary of the configuration space will consist of two different types of strata:
\begin{itemize}
    \item (Strata of type SI) These are strata where $i\geq 2$ points in the bulk collapse to a point in the bulk. Elements of such a stratum can be described as points in $\overline{\mathrm{Conf}_i(\Sigma)}\times\overline{\mathrm{C}_{n-i+1,m}(\mathbb{H}^4)}$,
    \item (Strata of type S2) These are strata where $i>0$ points in the bulk and $j>0$ points on the boundary with $2i+j-2\geq 0$ collapse to a point on the boundary. Elements of such a stratum can be described as points in $\overline{\mathrm{Conf}_{i,j}(\Sigma)}\times\overline{\mathrm{C}_{n-i,m-j+1}(\mathbb{H}^4)}$.
\end{itemize}


\printbibliography

\end{document}